\documentclass[acmsmall,screen]{acmart}
\AtEndPreamble{%
    \theoremstyle{acmdefinition}
    \newtheorem{remark}[theorem]{Remark}}

\includecomment{FULLVERSION}

\usepackage{bussproofs}
\usepackage{cleveref}
\usepackage{tabulary}
\usepackage{listings}
\usepackage{subcaption}
\usepackage{mathpartir}
\usepackage{placeins}
\usepackage{paralist}
\usepackage[only,rightarrowtriangle]{stmaryrd}

\usepackage{tikz}
\usetikzlibrary{automata, positioning, arrows}
\tikzset{
->, 
node distance=2cm, 
every state/.style={thick, fill=gray!10}, 
initial text=$ $, 
}

\usepackage{todonotes}

\makeatletter
\newif\ifdraftComments
\draftCommentstrue
\def\mkDraftComment#1#2{%
    \expandafter\def\csname#1\endcsname##1{%
        \ifdraftComments
            \textcolor{#2}{[#1: ##1]}%
        \else
            \@bsphack\@esphack
        \fi
    }%
}
\makeatother
\mkDraftComment{YN}{red}

\usepackage{bcprules}
\usepackage{myrules,yn/mythm}

\newenvironment{ottdefnblock}[3][]{ \framebox{\mbox{#2}} \quad #3 \\[0pt]}{}

\newcommand{\ottnt}[1]{\mathit{#1}}
\newcommand{\ottmv}[1]{\mathit{#1}}
\newcommand{\ottkw}[1]{\mathbf{#1}}
\newcommand{\ottsym}[1]{#1}

\renewcommand{\ottkw}[1]{\mathtt{#1} }
\renewcommand{\ottnt}[1]{#1}
\renewcommand{\ottmv}[1]{#1}
\renewcommand{\ottsym}[1]{#1}
\newcommand{\GJOIN}[2]{#1 * #2}

\providecommand{\eqdef}{\mathrel{\overset{\textrm{def}}{=}}}
\newcommand{\NAT}{\mathbb{N}}
\providecommand{\Coloneqq}{\mathrel{::=}}
\providecommand{\eqdef}{\mathrel{\overset{\textrm{def}}{=}}}
\providecommand{\lBrack}{\lbrack\!\lbrack}
\providecommand{\rBrack}{\rbrack\!\rbrack}

\newcommand{\INTRP}[1]{\lBrack #1 \rBrack}
\newcommand{\REP}[1]{\lceil #1 \rceil}

\definecolor{darkviolet}{rgb}{0.5,0,0.4}
\definecolor{darkgreen}{rgb}{0,0.4,0.2}
\definecolor{darkblue}{rgb}{0.1,0.1,0.9}
\definecolor{darkgrey}{rgb}{0.5,0.5,0.5}
\definecolor{lightblue}{rgb}{0.4,0.4,1}

\lstdefinestyle{eclipse}{
  breaklines=true,
  basicstyle=\sffamily\small,
  emphstyle=\color{red}\bfseries,
  keywordstyle=\color{darkviolet}\bfseries,
  commentstyle=\color{darkgreen},
  stringstyle=\color{darkblue},
  numberstyle=\color{darkgrey},
  emphstyle=\color{red},
  showstringspaces=false,
}

\setcopyright{rightsretained}
\acmDOI{10.1145/3689763}
\acmYear{2024}
\acmJournal{PACMPL}
\acmVolume{8}
\acmNumber{OOPSLA2}
\acmArticle{323}
\acmMonth{10}
\acmSubmissionID{oopslab24main-p498-p}
\received{2024-04-05}
\received[accepted]{2024-08-18}

\begin{document}

\title{Law and Order for Typestate with Borrowing}

\author{Hannes Saffrich}
\orcid{0009-0004-7014-754X}
\affiliation{%
  \institution{University of Freiburg}
  \city{Freiburg}
  \country{Germany}
}
\email{saffrich@informatik.uni-freiburg.de}

\author{Yuki Nishida}
\orcid{0000-0001-5941-6770}
\affiliation{%
  \institution{Kyoto University}
  \city{Kyoto}
  \country{Japan}
}
\email{nishida@fos.kuis.kyoto-u.ac.jp}

\author{Peter Thiemann}
\orcid{0000-0002-9000-1239}
\affiliation{%
  \institution{University of Freiburg}
  \city{Freiburg}
  \country{Germany}
}
\email{thiemann@acm.org}

\begin{abstract}

Typestate systems are notoriously complex as they require
sophisticated machinery for tracking aliasing.
We propose a new, transition-oriented foundation for typestate in the
setting of impure functional programming.
Our approach relies on ordered types for simple alias tracking and
its formalization draws on work on bunched implications. 
Yet, we support a flexible notion of borrowing in the presence of typestate.

Our core calculus comes with a notion of resource types indexed by an
ordered partial monoid that models abstract state transitions. 
We prove syntactic type soundness with respect to a
resource-instrumented semantics.
We give an algorithmic version of our type system and prove its
soundness. Algorithmic typing facilitates a
simple surface language that does not expose tedious details of
ordered types. We implemented a typechecker for the surface language
along with an interpreter for the core language.


\end{abstract}



\maketitle


\section{Introduction}
\label{sec:introduction}


Typestate \cite{DBLP:journals/tse/StromY86,
  DBLP:journals/tse/StromY93} is a programming concept that incorporates information about the state of an
object into its type. In traditional type systems, types describe the structure of
objects, but they do not directly capture behavioral information
about the possible states an object can be in during its
lifetime. Typestate, on the other hand, encodes the current state of
an object in its type and the state transitions in the types of functions and methods.

Typestate systems contribute to writing safe, understandable, and
maintainable code, particularly in systems with  complex behavior
of objects. The code is safe because the
type checker prevents accessing an object in an invalid state or
performing operations in the wrong order. The code is understandable
because types make the objects' state explicit. Taken
together, typestate systems improve maintainability
because they eliminate out-of-band uses of objects at compile
time.


Early typestate systems require full ownership of all objects under
the typestate regime. Here, ownership means that at every point of
time there is a single reference to the object and all operations on
the object are conveyed via that reference. In particular, these
systems disallow aliases and passing the object to a function
transfers its ownership from caller to callee.


Borrowing refers to the act of temporarily loaning a reference to a
resource (e.g., an object) without
transferring ownership. The concept currently receives a lot of attention due
to its widespread use in the Rust language
\cite{DBLP:conf/sigada/MatsakisK14, team24:_rust_progr_languag}: When
you borrow a reference to a resource in 
Rust, the borrower (typically a function or a block of code) can read
or mutate the resource, but it does not take ownership of it. In
this context ownership means the responsibility for any cleanup
actions and deallocation of the resource.

The combination of typestate and borrowing is challenging, as
mutable borrows can change the typestate of the
borrowed object.
Problems like this have been addressed in various ways and are related
to constraining and reasoning about
aliasing. \citet{DBLP:series/lncs/7850} presents an overview of several
approaches, but they can be very complex. For example, fractional
permissions \cite{boyland03:_check_inter_fract_permis} provide a
fine-grained way to split an owning reference into several
aliases with different permissions and combine them later to regain full ownership.

In this paper, we propose a new angle of attack that is inspired by
non-commutative logic and ordered type systems. Non-commutative logic
is a substructural logic where neither weakening, nor contraction, nor
exchange of assumptions is allowed. Thus, it is a kind of linear
logic where assumptions (i.e., resources) must be consumed in the order
in which they were introduced. Such systems have been investigated
previously in the context of  linguistics and logic, starting with the Lambek
calculus \cite{lambek58:_mathem_senten_struc,
  lambek60:_calcul_syntac_types} and further studied by
\citet{DBLP:journals/mlq/Abrusci90a}. In a programming context,
Polakow, Pfenning, and DeYoung
studied ordered logic, mainly exploring
the proof theory \cite{DBLP:conf/tlca/PolakowP99, DBLP:journals/entcs/PolakowP99}, applications to linear logic programming
\cite{DBLP:conf/ppdp/Polakow00}, to exceptions
\cite{DBLP:conf/flops/PolakovY01}, and applications to concurrency theory 
\cite{deyoung20:_session_typed_order_logic_specif}.
None of these works study the impact of ordering in an effectful call-by-value
language as we do.

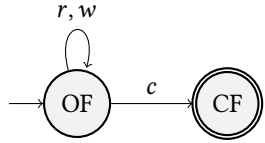
\begin{figure}[t]
  \centering
  \begin{tikzpicture}
    \node[state,initial] (q0) {OF};
    \node[state,accepting,right of=q0] (q1) {CF};

    \draw
    (q0) edge[loop above] node{$r,w$} (q0)
    (q0) edge [above] node{$c$} (q1);
  \end{tikzpicture}
  \caption{Simple life cycle of file handles; OF=$(r|w)^*c$ (open file);
    CF=$\varepsilon$ (closed file)}
  \label{fig:read-write-close}
\end{figure}
For illustration, let us consider the well-known example
of file handles (cf.\ \Cref{fig:read-write-close}). Opening a
file yields a file handle that is ready for read, write, or close
operations. After a read or write operation, we can read, write, or
close again, but no further operations are 
possible after a close.

Our system can model this example with regular languages for
the different states. In particular, the typestate of a file handle is
the language of all admissible traces of operations, that is, a
language on the symbols \texttt{r}, \texttt{w}, and \texttt{c}. We use
regular expressions to denote the languages and we assume that read
and write operations work on single bytes.
\begin{lstlisting}
  open   : FileName -o Handle[ (r|w)^*c ]
  read   : Handle[ r ] -o Byte
  write  : Handle[ w ] ox Byte -o ()
  close  : Handle[ c ] -o ()
\end{lstlisting}
This example indicates ownership with a language ending in
\texttt{c}; any other handles indexed with languages not ending in
\texttt{c} are borrowed. That is, the \texttt{read} and \texttt{write}
operations actually work on borrows; only \texttt{close} requires
ownership. A function that first reads and then writes to a borrowed
handle takes an argument of type \lstinline|Handle[ rw ]|.
Handles are generally linear in our. A handle of type \lstinline|Handle[ eps ]| (for the
empty trace) is exhausted: no further operations are possible and the
reference can be discarded. 

We express borrowing with a \lstinline|split| primitive that consumes
a resource \texttt{h} and returns two (aliased) 
references \texttt{bh} (``borrow of \texttt{h}'') and \texttt{h} of
the same underlying resource. These references may
have different typestates and we have to make sure that
the borrow \texttt{bh} is processed to the end before resuming
operations on \texttt{h}. At this point, the \emph{ordered aspect of our
  type system} comes into play:
\begin{lstlisting}
  split@m1/m2 : Handle[ m0 ] -o (Handle[ m1 ] .o Handle[ m2 ])
\end{lstlisting}
The \lstinline{split} consumes a handle in state $m_0$ and returns an \emph{ordered
  pair} of handles in states $m_1$ and $m_2$, respectively. This split is subject to
the correctness condition that $m_1 \cdot m_2 \subseteq m_0$. This
condition imposes that any trace of operations on the
\lstinline|Handle[ m1 ]| object followed by any trace of
operations on the \lstinline|Handle[ m2 ]| object yields a trace on
\lstinline|Handle[ m0 ]|. The elimination rule of the
ordered pair makes sure that the $m_1$ handle is processed first
before operations on the $m_2$ handle can start.

Here are some viable split instances (just the types without the
annotation) for the example. The result consists of the
new borrow (left) and the remaining resource (right).
\begin{lstlisting}
  split : Handle[ (r|w)^*c ] -o Handle[ (r|w)^* ] .o Handle[ (r|w)^*c ]
\end{lstlisting}
The classic borrow: \lstinline/Handle[(r|w)*]/ is a file handle for
reading and writing, which cannot be closed. The owner retains the
right to do anything and then close.
\begin{lstlisting}
  split : Handle[ (r|w)^* ] -o Handle[ r^* ] .o Handle[ (r|w)^* ]
\end{lstlisting}
Read-only re-borrow: the incoming handle is already borrowed as
there is no right for closing. We split into a re-borrow that is read-only (\texttt{r*}) and return a handle in the original read-write state.
\begin{lstlisting}
  split : Handle[ (r|w)^*c ] -o Handle[ c ] .o Handle[ eps ]
\end{lstlisting}
This split is needed to apply the \texttt{close} operation: we split a
file handle with all capabilities into one owning reference to close
and a dangling reference that can be discarded.
\begin{lstlisting}
  split : Handle[ (r|w)^*c ] -o Handle[ r^*c ] .o Handle[ eps ]
\end{lstlisting}
The final example demonstrates that the split can be used for more
than just borrowing. This split implements a safe cast from an owning
read-write handle to an owning read-only handle.  

The actual calculus that we present handles states more abstractly. We
only demand that states are drawn from an \emph{ordered partial
  monoid} (OPM) (see \Cref{sec:order-part-mono}). In our example, the set of regular languages
$\{ L \mid L \subseteq (r|w)^*c^? \}$ is an OPM under language
concatenation and ordered by language 
inclusion, as we will show in \Cref{example:opm-files}. 

As we remarked above, there is plenty of work on typestate as well as
on controlling aliasing. In the context of Java,
\citet{DBLP:conf/oopsla/BierhoffA07} discuss a closely related typestate system
that supports aliasing. Their system is
fundamentally different from ours. It is \emph{state-based} as it specifies a
transition as a pair of input and output state, whereas our
system is \emph{transition-based} as it specifies a transition by an element of an OPM, that is, an abstraction
of a set of traces of operations performed by the method.
Moreover, their system supports refined state transitions that depend on the
run-time result of a method call. Thus, each method is annotated with a
state-change formula. Access to a resource is governed
by a permission which controls the right to write and
change the state. Their  notion of splitting applies to
the permission when creating an alias. Our system does not
have built-in support for controlling read/write permissions (like
immutable borrows), but we can model that notion with a suitably
chosen OPM (cf.\ the read-only borrow in the above file example).

The calculus in \citet{DBLP:journals/toplas/GarciaTWA14} is very close
to the one by
\citet{DBLP:conf/oopsla/BierhoffA07}, but simplifies the handling of
states, presents further refinements to fine-tune the handling of
typestate, and proposes gradual typing to include run-time checks for
typestates. Their model can control interleaved operations on aliases
to some extent; our system allows for multiple aliases, but
requires that they are processed in order.

\subsection*{Contributions}
\label{sec:contributions}

\begin{enumerate}
\item We identify the notion of an ordered partial monoid as a minimal
  specification for typestate with borrowing (\Cref{sec:motivation}).
\item We define a typed lambda calculus that supports
  intuitionistic as well as linear and ordered constructions (\Cref{sec:language});
  type constructions like function space and product come
  in several variants that differ in their ordering properties; the
  calculus requires a tree-like typing environment as in the
  logic of bunched implications
  \cite{DBLP:journals/bsl/OHearnP99}. 
\item We establish the syntactic metatheory for type soundness, which
  incorporates protocol adherence for
  objects with typestate and borrowing (\Cref{sec:language:properties}).
\item We present a syntax-driven version of the type system
  (\Cref{sec:algorithmic-typing}) as a basis for our prototype
  implementation of the type checker. Type checking is decidable if
  certain operations on the underlying OPM are decidable.
\end{enumerate}



\section{Motivation}
\label{sec:motivation}

In this section, we formally introduce OPMs and illustrate them with
examples. Then we move on to consider splitting and investigate some
ramifications of our type system with ordered resources. We illustrate
the issues by dissecting two representative constructs:  ordered
product elimination and let expressions. 

The full system in \Cref{sec:language} supports three kinds of types:
unrestricted types, linear types,
and orderer types for resources. This distinction governs the way that
type assumptions of variables are treated in typing contexts:
variables of unrestricted type are subject to exchange, contraction,
and weakening; linear variables are subject to exchange; ordered
variables are subject to nothing.

For simplicity, we focus on linear and
ordered types in this section. Unrestricted types, such as
\texttt{Int} and \texttt{Bool}, introduce additional
complexity that we want to avoid in this initial discussion.

\subsection{Ordered Partial Monoids}
\label{sec:order-part-mono}

In the introduction, we already claimed that ordered partial monoids
are an appropriate abstraction for resource management with
borrowing. Let us first define this notion.
\begin{definition}
  A \emph{partial monoid} is a structure $(M, \odot, \varepsilon)$ such that
  \begin{itemize}
  \item multiplication $\odot : M \times M \hookrightarrow M$ is a partial function;
  \item there is a neutral element $\varepsilon \in M$ with respect to
    $\odot$, that is, $\forall m \in M$,
    $m \odot \varepsilon = \varepsilon \odot m = m$ where
    multiplication with the neutral element is always defined;
  \item multiplication is associative: $\forall m, n, o \in M$,  $(m
    \odot n) \odot o = m \odot (n \odot o)$ where the left side is
    defined if and only if the right side is defined.
  \end{itemize}
  A partial monoid is \emph{ordered} (OPM), if there is a preorder ${\le}
  \subseteq M \times M$ such that multiplication is monotone and
  downwards closed:
  $\forall m, m', n, n'$, if $m \le m'$ and $n \le n'$ and $m' \odot
  n'$ is defined, then $m \odot n$ is defined and
  $m \odot n \le m' \odot n'$.
  \qed
\end{definition}

This definition is not without precedent. 
Partially-ordered \emph{commutative} monoids are used for
constructing Kripke models for BI
\cite{DBLP:journals/bsl/OHearnP99}. \citet{DBLP:journals/mscs/GalmicheMJP05}
even makes commutativity optional. Both derive their ordering from the
multiplication, rather than requiring a separate order relation.

The basic idea is that we regard the elements of the OPM as
abstractions of admissible sequences of operations on a
resource. Thus, a resource type is indexed with some OPM element $m$,
an operation on a resource expects a resource with a type index that
abstracts the singleton sequence of just this operation, taken
together we expect that the indexes of all operations on the resource must
multiply up to some $m'$ with $m'\le m$. 

Once a resource has index $m$ that contains $\varepsilon$ (i.e.,
$\varepsilon \le m$), we
consider it \emph{droppable}; a resource with index $\varepsilon$ is
\emph{exhausted} in the sense that no further operation is applicable. We choose to keep resources
linear so that our calculus provides an explicit $\ottkw{drop}$ operation
to discard exhausted resources. We omit $\ottkw{drop}$ in our
examples, assuming that the typechecker inserts it
automatically.\footnote{Alternatively, the type system could change
  the substructural properties of a resource with index $\varepsilon$
  to affine, so that it could be dropped by weakening.}

\begin{example}
  \label{example:opm-files}
  The OPM for file permissions in the introduction is given by the
  set of regular languages, which are subsets of a prefix language of
  the \emph{envelope language} $E = (r|w)^*c$. That is, $M = \{ R \mid
  R \subseteq (r|w)^*c^? , R \text{ regular}\}$. 
  \begin{itemize}
  \item Multiplication is the usual language concatenation, with the
    restriction that the result must be in $M$. For example, the product
    $rc \odot (r|w)^*$ is not defined in $M$, because $rc(r|w)^*$
    contains the word $cw$ which is not a prefix of an element of $E$.
  \item The neutral element is $\{\varepsilon\}$, the language of the
    empty word.
  \item Multiplication is associative and the definedness condition is
    fulfilled.
  \item The ordering is language inclusion; monotonicity of
    multiplication is a standard result in formal language theory;
    downwards closedness is easy to see.
  \end{itemize}
  Technically, the languages involved need not be regular. We phrase
  the example in terms of regular languages because it makes all
  operations and relations (and hence type checking) decidable.
\end{example}
\begin{figure}[tp]
  \begin{subfigure}{0.6\textwidth}
  \begin{displaymath}
  \begin{array}{r|ccc}
    \odot  &\varepsilon & b & \ast \\
    \hline
    \varepsilon &\varepsilon&b & \ast \\
    b &b&b& \ast \\
    \ast &\ast &-&- \\
  \end{array}
  \qquad
  \begin{array}{r|ccc}
    \le  &\varepsilon & b & \ast \\
    \hline
    \varepsilon &\bullet &\bullet & - \\
    b &-&\bullet& - \\
    \ast &- &-&\bullet \\
  \end{array}
\end{displaymath}
  \caption{OPM for ownership and mutable borrows}
  \label{fig:opm-ownership-borrows}
\end{subfigure}
\begin{subfigure}{0.39\textwidth}
\begin{lstlisting}
new   : () -o Ref[ * ]
read  : Ref[ b ] -o X
write : Ref[ b ] -o X -o ()
free  : Ref[ * ] -o ()
\end{lstlisting}
\caption{API for ownership and mutable borrows}
\label{fig:api-ownership-borrows}
\end{subfigure}
\caption{OPM and API for ownership and mutable borrows}
\label{fig:opm-api-ownership-borrows}
\end{figure}
\begin{example}
  Here is an OPM for ownership and mutable borrows.
  Let $M = \{ \varepsilon, b, \ast \}$ (discardable, borrowed, owned
  reference) with multiplication and ordering defined by the tables in
  \Cref{fig:opm-ownership-borrows}. The intuition of multiplication
  derives from splitting: if the left multiplication with $\ast$ were
  defined for any element but $\varepsilon$, a split like $\ast \odot
  x$ for $x\ne\varepsilon$ would be valid. But that would mean, we
  could first destroy the resource (exerting full ownership), but afterwards
  we would still hold a valid handle at state $x$, which is impossible.
  Right multiplication with $\ast$ or $b$ controls borrowing:
  We can split borrows from owning references as $ b \odot \ast \le \ast$. We
  can also split borrows from borrows as $ b \odot b \le b$.
  An API for references would look similar to the file API (see
  \Cref{fig:api-ownership-borrows}).

\end{example}
\begin{example}
  Session types \cite{DBLP:journals/csur/HuttelLVCCDMPRT16} is a
  generic term for behavioral type systems that 
  guarantee communication properties like type soundness, protocol fidelity, and
  deadlock freedom for concurrent systems. In a nutshell, a session
  type is ascribed to a communication channel and it describes the
  protocol on this channel, i.e., the possible sequences of
  interaction on this channel.

  Context-free session types (CFST) have been proposed as an extension of
  traditional session types embedded in a functional language, but
  with fewer restrictions on recursion
  \cite{DBLP:journals/iandc/AlmeidaMTV22}. The original system has
  been extended in several ways including subtyping
  \cite{DBLP:conf/concur/SilvaMV23}. The latter publication describes
  the syntax for session types like this:
  \begin{align*}
    S &::= \mathtt{skip} \mid S;S \mid {!T} \mid {?T}
        \mid \oplus\{ l_i: S_i\} \mid \&\{ l_i: S_i \}
        \mid \mu X.S \mid X & \text{open}\\
    S^\bot & ::= \mathtt{end!} \mid \mathtt{end?} \mid S;S^\bot
             \mid \oplus\{ l_i: S^\bot_i\} \mid \&\{ l_i: S^\bot_i \}
             \mid S^\bot;\mathtt{skip} & \text{terminated}
  \end{align*}
  Briefly spoken, a communication channel is created with a terminated type of
  the form $S^\bot$. The type \texttt{skip} indicates a noop; the type
  $S_1;S_2$ stands for the sequential composition of the protocols
  $S_1$ and $S_2$; ${!T}$ and ${?T}$ are protocols that send or
  receive one item of type $T$, the type of messages which we leave
  unspecified here; the selection type $\oplus\{ l_i: S_i\}$ lets the
  sender choose a label $l_j$ and continue with the corresponding
  protocol $S_j$; dually, the branch type $\&\{ l_i: S_i \}$ receives
  a label $l_j$ and continues with the chosen $S_j$; the type $\mu
  X.S$ is an equirecursive session type with type variable $X$; and
  $X$ refers to a variable bound by a recursive type.
  The nonterminal $S^\bot$ characterizes the possible endings of a
  session type. The types \texttt{end!} and \texttt{end?} stand for
  two different ways of closing a channel (active and passive); the
  remaining cases are just special cases of $S$ tailored to the
  situation at the end. Importantly, all continuations
  have to be uniform in selection and branching, either open or terminating.

  This syntax is subject to laws that make CFSTs a partial monoid with
  neutral element \texttt{skip} and sequential composition as
  multiplication. (Equi-) Recursive types come with the usual guardedness
  restriction that disallows empty types like $\mu X.X$. Moreover,
  composition distributes from the right 
  over selection $\oplus$ and branching $\&$. 

  A suitable partial order for this OPM is the subtyping relation proposed by
  \citet{DBLP:conf/concur/SilvaMV23}. Sequential composition is
  monotone with respect to subtyping as this requirement is an
  inference rule in the definition of subtyping.
  Another suitable partial order is, of course, the equivalence
  relation on types defined by the monoid laws.

  The literature \cite{DBLP:journals/iandc/AlmeidaMTV22} advocates an
  interface based on channel passing, which is concisely described using
  polymorphism on the continuation session \lstinline{alpha}:
\begin{lstlisting}
  send : forall alpha. T -o Chan[ !T;alpha ] -o Chan[ alpha ]
  recv : forall alpha. Chan[ ?T;alpha ] -o (T ox Chan[ alpha ])
\end{lstlisting}
  The functions \lstinline{send} and
  \lstinline{recv} both consume a channel of suitable (linear) type
  \lstinline{Chan[ S ]}, they modify it according to the operation,
  and return the channel in its updated state and type. In this
  system, all session types have the form $S^\bot$ as channel references are linear, they
  imply ownership, and consequently they must be closed in the end.

  Combining CFST with our approach to borrowing, we find good use for types of
  the form $S$: they correspond to the type of \emph{borrowed channels} whereas $S^\bot$
  characterizes an owning reference of a channel. 
  The resulting concise API uses borrows
  instead of polymorphism on sessions.
\begin{lstlisting}
  send_borrow : T -o Chan[ !T ] -o ()
  recv_borrow : Chan[ ?T ] -o T
\end{lstlisting}
  Comparing the types, we see that the channel-passing operations peel
  off the first operation from the type of the owning reference. In
  contrast, the API using borrowing does not require ownership and it
  is not concerned with managing the continuation: it just gets the
  right to send or receive a single \lstinline{T} on the channel.
\end{example}

\subsection{Splitting and Order}
\label{sec:splitting-order}

The key operation to create a borrow is the \texttt{split} shown in
the introduction:\footnote{We drop the prefix \lstinline{Handle} as
  the discussion is generic for all resources.}
\begin{lstlisting}
  split@m1/m2 : [ m0 ] -o ([ m1 ] .o [ m2 ])
\end{lstlisting}
Its indices are values from the OPM: $m_1$ and $m_2$. The first
$m_1$ describes the traces that we want to grant to the borrow, the
second $m_2$ describes the remaining traces that are admissible after
we have exercised $m_1$ on the borrow. The soundness condition is $m_1
\odot m_2 \le m_0$, which states that every composition of traces from
$m_1$ and $m_2$ is still a valid trace on the original resource
described by $m_0$.

As a linear function, \texttt{split} consumes the argument of type
${[ m_0 ]}$ and returns an \emph{ordered pair} of type
$[m_1] \odot [m_2]$. An ordered pair comes with the
obligation to \emph{process all resources in the first component before
  starting to process any resource in the second component}.

An elimination rule for ordered pairs has to reflect this ordering in
the typing context $\Gamma$.
We express ordering contraints by having two distinct context formers. 
This idea is inspired by the logic of bunched implications (BI)
\cite{DBLP:journals/bsl/OHearnP99}, where one context former allows
for weakening and contraction; but the other does not.
Our context formers have a different interpretation: neither allows
for weakening and contraction; one allows for exchange, and the other
does not.
\begin{align*}
  \Gamma \Coloneqq\enspace &  \cdot  && \text{empty context} \\
                  & \ottmv{x}  \mathord:  \ottnt{T} && \text{singleton binding;} \\
  & \Gamma_{{\mathrm{1}}}  \ottsym{,}  \Gamma_{{\mathrm{2}}} && \text{no exchange: bindings from $\Gamma_{{\mathrm{1}}}$
                 must be processed before bindings from $\Gamma_{{\mathrm{2}}}$;} \\
  & \Gamma_{{\mathrm{1}}}  \parallel  \Gamma_{{\mathrm{2}}} && \text{exchange: bindings from $\Gamma_{{\mathrm{1}}}$ are
  independent from those in $\Gamma_{{\mathrm{2}}}$; no order imposed.}
\end{align*}
As in BI, the context formers can be arbitrarily nested. 
Both context formers are subject to the monoid laws with unit $ \cdot $ and
parallel composition is commutative.
We also define \emph{context patterns} $\mathcal{G}$ as contexts with a hole
that can be filled with a context $\Gamma$.

To illustrate the definition, we discuss some typing rules of our
system in simplified form, sound but more restrictive than necessary. Recall the rules for \emph{linear pairs}. In a
linear context, all binders are mutually independent so contexts are
built solely with $\|$.
\begin{center}
\begin{minipage}[t]{0.54\linewidth}
  \infrule[\ensuremath{\otimes}-E]
  {\Gamma_1 \vdash M : T_1 \otimes T_2 \\
    \Gamma_{{\mathrm{2}}}  \parallel  \ottmv{x_{{\mathrm{1}}}}  \mathord:  \ottnt{T_{{\mathrm{1}}}}  \parallel  \ottmv{x_{{\mathrm{2}}}}  \mathord:  \ottnt{T_{{\mathrm{2}}}} \vdash N : T} {\Gamma_{{\mathrm{1}}}  \parallel  \Gamma_{{\mathrm{2}}} \vdash
    \mathtt{let}~(x_1\otimes x_2) = M~\mathtt{in}~N : T}
\end{minipage}
\begin{minipage}[t]{0.45\linewidth}
  \infrule[\ensuremath{\otimes}-I] {\Gamma_{{\mathrm{1}}} \vdash M : S \andalso
    \Gamma_{{\mathrm{2}}} \vdash N : T} {\Gamma_{{\mathrm{1}}}  \parallel  \Gamma_{{\mathrm{2}}} \vdash M \otimes N : S
    \otimes T}
\end{minipage}
\end{center}
For linear types, the notation $\Gamma_1\| \Gamma_2$ splits the
incoming typing context $\Gamma$ into contexts with disjoint domains,
that is, $ \Gamma \ottsym{=} \Gamma_{{\mathrm{1}}}  \parallel  \Gamma_{{\mathrm{2}}} $. The
resources in $\Gamma_1$ are consumed by evaluating $M$. After 
elimination the new resources obtained by splitting are added to the
remaining resources in $\Gamma_2$. The introduction rule is also standard.

For a left-to-right ordered product, we have to keep in mind that the
system is intended for left-to-right call-by-value evaluation.
We have to split the incoming context
sequentially, indicated by the operator ``,''. The reading of the assumption $\Gamma_1 \vdash M : T_1
\odot T_2$ is that the resources in $\Gamma_1$ have been
converted to the resource $T_1 \odot T_2$, but they have not
necessarily been fully processed. To ensure that they are processed
before $\Gamma_2$, we have to put them in front of the context. 

\smallskip{}
\begin{minipage}{0.9\linewidth}
  \infrule[\ensuremath{\odot}-E1] {\Gamma_{{\mathrm{1}}} \vdash M : T_1 \odot T_2
    \andalso \ottmv{x_{{\mathrm{1}}}}  \mathord:  \ottnt{T_{{\mathrm{1}}}}  \ottsym{,}  \ottmv{x_{{\mathrm{2}}}}  \mathord:  \ottnt{T_{{\mathrm{2}}}}  \ottsym{,}  \Gamma_{{\mathrm{2}}} \vdash N : T} {\Gamma_{{\mathrm{1}}}  \ottsym{,}  \Gamma_{{\mathrm{2}}} \vdash
    \mathtt{let}~(x_1\odot x_2) = M~\mathtt{in}~N : T}
\end{minipage}

\smallskip
This rule is sound, but restrictive in two respects. First, we cannot
assume that the incoming context $\Gamma$ has the sequential form $\Gamma_{{\mathrm{1}}}  \ottsym{,}  \Gamma_{{\mathrm{2}}}$ expected by the rule. It is always possible to weaken the context $\Gamma$ to
match the expectation by imposing extra (unnecessary) sequencing constraints.
Second, the rule sequentializes uses of $x_1$
and $x_2$ with uses of potentially unrelated resources in $\Gamma_{{\mathrm{2}}}$. If
two resources are independent (e.g., two different 
file handles or two distinct communication channels), then they
should not unduly constrain one another.

Hence, we restrict this elimination rule to
terms $M$ that do not use any operations on resources.
To this end, we extend typing with a simple effect system that only
checks the existence of resource operations: the effect $0$
characterizes a pure term, whereas effect $1$ indicates that an
operation may be performed.
With this extra information, the elimination rule can be much more
liberal: It is sufficient if we can find the resources $\Gamma$ needed
for $M$ somewhere in the context as prescribed by context pattern
$\mathcal{G}$. Then we can replace those resources by the components of the pair.

\smallskip{}
\begin{minipage}{0.9\linewidth}
  \infrule[\ensuremath{\odot}-E2] { \Gamma \vdash M : T_1 \odot T_2
    \mid 0 \andalso \mathcal{G}  \ottsym{[}  \ottmv{x_{{\mathrm{1}}}}  \mathord:  \ottnt{T_{{\mathrm{1}}}}  \ottsym{,}  \ottmv{x_{{\mathrm{2}}}}  \mathord:  \ottnt{T_{{\mathrm{2}}}}  \ottsym{]} \vdash N : T \mid e } {\mathcal{G}  \ottsym{[}  \Gamma  \ottsym{]} \vdash \mathtt{let}~(x_1\odot x_2) = M~\mathtt{in}~N : T
    \mid e}
\end{minipage}

\smallskip
To make this rule generally applicable, we have to move resource operations
out of the elimination as done in A-normal form  \cite{DBLP:conf/pldi/FlanaganSDF93}. That is,
we transform $\mathtt{let}~(x_1\odot x_2) = M~\mathtt{in}~N$ to
$\mathtt{let}~w = M~\mathtt{in}~\mathtt{let}~(x_1\odot x_2) =
w~\mathtt{in}~N$. This transformation, called let-insertion, shifts
our concerns about resource operations 
to the typing of the normal \texttt{let} expression, where we can
solve them once and for all.

Before we discuss the \texttt{let} expression, it remains to discuss
the introduction rule for the ordered product. Here, the easiest
option is to require that the components are effect-free. A
preprocessing step that performs let-insertion for effect-ful
components guarantees general applicability. 

\smallskip
\begin{minipage}{0.9\linewidth}
  \infrule[\ensuremath{\odot}-I] {\Gamma_{{\mathrm{1}}}  \vdash  \ottnt{M}  :  \ottnt{S}  \mid  \ottsym{0} \andalso \Gamma_{{\mathrm{2}}}  \vdash  \ottnt{N}  :  \ottnt{T}  \mid  \ottsym{0}} {\Gamma_{{\mathrm{1}}}  \ottsym{,}  \Gamma_{{\mathrm{2}}}  \vdash  \ottnt{M}  \odot  \ottnt{N}  :  \ottnt{S}  \odot  \ottnt{T}  \mid  \ottsym{0}}
\end{minipage}
\smallskip


Several typing rules are possible for \texttt{let}. 
Here are the two extreme cases.
\begin{center}
  \begin{minipage}{0.54\linewidth}
      \infrule[let-\ensuremath{\|}]
      {\Gamma_{{\mathrm{1}}}  \vdash  \ottnt{M}  :  \ottnt{S}  \mid  \ottnt{e_{{\mathrm{1}}}} \\
        \Gamma_{{\mathrm{2}}}  \parallel  \ottmv{x}  \mathord:  \ottnt{S}  \vdash  \ottnt{N}  :  \ottnt{T}  \mid  \ottnt{e_{{\mathrm{2}}}}}
      {\Gamma_{{\mathrm{1}}}  \parallel  \Gamma_{{\mathrm{2}}}  \vdash  \ottkw{let} \, \ottmv{x}  \ottsym{=}  \ottnt{M} \, \ottkw{in} \, \ottnt{N}  :  \ottnt{T}  \mid   \ottnt{e_{{\mathrm{1}}}}  \sqcup  \ottnt{e_{{\mathrm{2}}}} }
  \end{minipage}
  \begin{minipage}{0.45\linewidth}
    \infrule[let-\ensuremath{,}]
    {\Gamma_{{\mathrm{1}}}  \vdash  \ottnt{M}  :  \ottnt{S}  \mid  \ottnt{e_{{\mathrm{1}}}} \\
      \ottmv{x}  \mathord:  \ottnt{S}  \ottsym{,}  \Gamma_{{\mathrm{2}}}  \vdash  \ottnt{N}  :  \ottnt{T}  \mid  \ottnt{e_{{\mathrm{2}}}}}
    {\Gamma_{{\mathrm{1}}}  \ottsym{,}  \Gamma_{{\mathrm{2}}}  \vdash  \ottkw{let} \, \ottmv{x}  \ottsym{=}  \ottnt{M} \, \ottkw{in} \, \ottnt{N}  :  \ottnt{T}  \mid   \ottnt{e_{{\mathrm{1}}}}  \sqcup  \ottnt{e_{{\mathrm{2}}}} }
  \end{minipage}
\end{center}
Rule (\textsc{let-}\ensuremath{\|}) requires disjointness of the resources of
$M$ and $N$. Hence, the resources in $x$ are also disjoint to those of
$N$.
Rule (\textsc{let-}\ensuremath{,}) forces a sequential split of the incoming
resources and puts the processing of $x$ in front of any other
resource operation in $N$.

While these rules are sound, rule (\textsc{let-}\ensuremath{\|}) is only
applicable if resources are disjoint. Rule (\textsc{let-}\ensuremath{,}) is
always applicable, but it may enforce more sequentiality than
necessary. To see that, consider $\mathtt{let}~z = x_1~\mathtt{in}~N$ in the context
\begin{displaymath}
   \Gamma \ottsym{=} \ottsym{(}  \ottmv{x_{{\mathrm{1}}}}  \mathord:  \ottnt{S_{{\mathrm{1}}}}  \ottsym{,}  \ottmv{x_{{\mathrm{2}}}}  \mathord:  \ottnt{S_{{\mathrm{2}}}}  \ottsym{)}  \parallel  \ottsym{(}  \ottmv{y_{{\mathrm{1}}}}  \mathord:  \ottnt{T_{{\mathrm{1}}}}  \ottsym{,}  \ottmv{y_{{\mathrm{2}}}}  \mathord:  \ottnt{T_{{\mathrm{2}}}}  \ottsym{)} 
  .
\end{displaymath}
Rule (\textsc{let-}\ensuremath{\|}) is not applicable because $x_2$ cannot be
put in parallel to $x_1$.
Rule (\textsc{let-}\ensuremath{,}) is not directly applicable to
context $\Gamma$, but we have to decompose it to obtain a more
restrictive context
\begin{displaymath}
   \Gamma_{{\mathrm{1}}} \ottsym{=} \ottsym{(}  \ottmv{x_{{\mathrm{1}}}}  \mathord:  \ottnt{S_{{\mathrm{1}}}}  \ottsym{)}  \qquad  \Gamma_{{\mathrm{2}}} \ottsym{=} \ottsym{(}  \ottmv{x_{{\mathrm{2}}}}  \mathord:  \ottnt{S_{{\mathrm{2}}}}  \ottsym{)}  \parallel  \ottsym{(}  \ottmv{y_{{\mathrm{1}}}}  \mathord:  \ottnt{T_{{\mathrm{1}}}}  \ottsym{,}  \ottmv{y_{{\mathrm{2}}}}  \mathord:  \ottnt{T_{{\mathrm{2}}}}  \ottsym{)} 
  .
\end{displaymath}
The problem is that $N$ is now typechecked in context
$\ottmv{z}  \mathord:  \ottnt{S_{{\mathrm{1}}}}  \ottsym{,}  \Gamma_{{\mathrm{2}}} = \ottsym{(}  \ottmv{z}  \mathord:  \ottnt{S_{{\mathrm{1}}}}  \ottsym{)}  \ottsym{,}  \ottsym{(}  \ottsym{(}  \ottmv{x_{{\mathrm{2}}}}  \mathord:  \ottnt{S_{{\mathrm{2}}}}  \ottsym{)}  \parallel  \ottsym{(}  \ottmv{y_{{\mathrm{1}}}}  \mathord:  \ottnt{T_{{\mathrm{1}}}}  \ottsym{,}  \ottmv{y_{{\mathrm{2}}}}  \mathord:  \ottnt{T_{{\mathrm{2}}}}  \ottsym{)}  \ottsym{)}$
so that we get an artificial sequential dependency stating that $z$ is
before $y_1$ and $y_2$. 

We address this restriction by defining a \emph{left context} $\mathcal{L}$ as follows.
\begin{center}
  \begin{tabulary}{\linewidth}{@{}r@{\enspace}c@{\enspace}J<{\hfill\null}@{}}
    $\mathcal{L}$ & $\Coloneqq$ & $\ottsym{[]} ~\mid~ \mathcal{L}  \ottsym{,}  \Gamma ~\mid~ \Gamma  \parallel  \mathcal{L}
    ~\mid~ \mathcal{L}  \parallel  \Gamma$
  \end{tabulary}
\end{center}
This context pattern finds the usage prefixes of resources. In our
example, we would use $\mathcal{L} = \ottsym{(}  \ottsym{[]}  \ottsym{,}  \ottmv{x_{{\mathrm{2}}}}  \mathord:  \ottnt{S_{{\mathrm{2}}}}  \ottsym{)}  \parallel  \ottsym{(}  \ottmv{y_{{\mathrm{1}}}}  \mathord:  \ottnt{T_{{\mathrm{1}}}}  \ottsym{,}  \ottmv{y_{{\mathrm{2}}}}  \mathord:  \ottnt{T_{{\mathrm{2}}}}  \ottsym{)}$ so that $\Gamma = \mathcal{L}  \ottsym{[}  \ottmv{x_{{\mathrm{1}}}}  \mathord:  \ottnt{S_{{\mathrm{1}}}}  \ottsym{]}$ and then we can typecheck $N$
with $\Gamma = \mathcal{L}  \ottsym{[}  \ottmv{z}  \mathord:  \ottnt{S_{{\mathrm{1}}}}  \ottsym{]}$, that is, without imposing additional
sequential dependencies. The resulting rule subsumes both
[let-\ensuremath{\|}] and  [let-\ensuremath{,}]:
\infrule[let]
{\Gamma  \vdash  \ottnt{M}  :  \ottnt{S}  \mid  \ottnt{e_{{\mathrm{1}}}} \\
  \mathcal{L}  \ottsym{[}  \ottmv{x}  \mathord:  \ottnt{S}  \ottsym{]}  \vdash  \ottnt{N}  :  \ottnt{T}  \mid  \ottnt{e_{{\mathrm{2}}}}}
{\mathcal{L}  \ottsym{[}  \Gamma  \ottsym{]}  \vdash  \ottkw{let} \, \ottmv{x}  \ottsym{=}  \ottnt{M} \, \ottkw{in} \, \ottnt{N}  :  \ottnt{T}  \mid   \ottnt{e_{{\mathrm{1}}}}  \sqcup  \ottnt{e_{{\mathrm{2}}}} }

Depending on the resources needed by the let header $M$, we may have
to sequentialize the incoming context as in this variation of our
example.
Consider 
$\mathtt{let}~z = (x_1 \otimes y_1)~\mathtt{in}~N$ in the same context
$\Gamma$ as before. Before we can match with some $\mathcal{L}$, we have to
sequentialize $\Gamma$ to 
$ \Gamma' \ottsym{=} \ottsym{(}  \ottmv{x_{{\mathrm{1}}}}  \mathord:  \ottnt{S_{{\mathrm{1}}}}  \parallel  \ottmv{y_{{\mathrm{1}}}}  \mathord:  \ottnt{T_{{\mathrm{1}}}}  \ottsym{)}  \ottsym{,}  \ottsym{(}  \ottmv{x_{{\mathrm{2}}}}  \mathord:  \ottnt{S_{{\mathrm{2}}}}  \parallel  \ottmv{y_{{\mathrm{2}}}}  \mathord:  \ottnt{T_{{\mathrm{2}}}}  \ottsym{)} $ and then we can
choose $\mathcal{L} = \ottsym{[]}  \ottsym{,}  \ottsym{(}  \ottmv{x_{{\mathrm{2}}}}  \mathord:  \ottnt{S_{{\mathrm{2}}}}  \parallel  \ottmv{y_{{\mathrm{2}}}}  \mathord:  \ottnt{T_{{\mathrm{2}}}}  \ottsym{)}$.
Our type system has a suitable weakening rule to allow for this
sequentialization in the context.

\subsection{Borrowing vs. Resource-passing}
\label{sec:funct-vs-imper}

The operations on resources can be defined using borrows or
by passing resources explicitly. 
The discussion so far concentrated on borrowing interfaces. As an
example, we consider an operation that consumes a 
parameter of type \lstinline|A| and a resource of type
\lstinline|[m1]| to return a value of type
\lstinline|B|.\footnote{\lstinline{b/op} stands for borrowing version
  of operation \lstinline{op}.}
\begin{lstlisting}
  b/op : A ox [ m1 ] -o B
\end{lstlisting}

A resource-passing version of the same operation consumes a resource
of type \lstinline|[m0]| and returns it in its next state
\lstinline|[m2]| that ``peels off'' the \lstinline|m1|$ = $\lstinline|m|$(op)$ denoting the single
application of \lstinline{op} from the \lstinline|m0|, if that is defined:
\begin{lstlisting}
  f/op : A ox [ m0 ] -o B ox [ m2 ]
\end{lstlisting}
Here \lstinline|m2|$ = \delta_{op} ($\lstinline|m0|$)$ is the (partial) transformation of
\texttt{op} on the OPM. It is partial as \texttt{op} need not be defined for all \lstinline|m0|.
We can transform the resource-passing interface into the
borrowing one if $\delta_{op}$ can be expressed by multiplication with
an element \lstinline|m1| in the monoid:
\lstinline|m1|$ \odot $\lstinline|m2|$ \le $\lstinline|m0|, the side
condition for \texttt{split}.
The reverse transformation is also possible as the resource returned
from \texttt{f/op} can be considered (formally) exhausted because
\lstinline|m1|$ \odot \varepsilon \le $\lstinline|m1|. Hence, we obtain the typing $r_0 : [\varepsilon]$.
\begin{center}
  \begin{minipage}{0.53\linewidth}
\begin{lstlisting}
let (b ox r2) = f/op (a ox r) in ...
----->
let (r1 .o r2) = split@m1 r in
let b = b/op (a ox r1)  in ...
\end{lstlisting}
  \end{minipage}
  \vline
  \begin{minipage}{0.46\linewidth}
\begin{lstlisting}
 let b = b/op (a ox r) in ...
 ----->
 let (b ox r0) = f/op (a ox r) in
 let () = drop r0 in ...
\end{lstlisting}
  \end{minipage}
\end{center}

\subsection{Aliasing and Permissions}
\label{sec:example-from-garcia}

\citet[page 5]{DBLP:journals/toplas/GarciaTWA14}  use access
permissions to manage the interaction of typestate and
borrows in their language FT. Their motivating discussion involves the
following function in the setting of an API similar to the file handles from the introduction.
\begin{lstlisting}
  void m(full OF >> full CF f, full OF >> full OF g) {    /* FT code */
    f.close(); print(g.file_desc.pos);             }
\end{lstlisting}
This function takes two file handles \texttt{f} and \texttt{g}; it
closes the handle \texttt{f} first, and then obtains information from
\texttt{g}. FT's parameter notation \texttt{S >> T f} says that \texttt{f} is expected
to be in typestate \texttt{S} when called and in typestate \texttt{T}
when \texttt{m} returns. Here \texttt{OF} and \texttt{CF} are short
for \texttt{OpenFile} and \texttt{ClosedFile}. The qualifier
\texttt{full} indicates that the 
reference must have exclusive write permission. This
function type avoids calls with aliased parameters as in
\texttt{m(h, h)}. Such a call would not be safe as the function
cannot obtain information from a closed file. In FT, a \texttt{full}
permission cannot split in two \texttt{full} permissions, so that it is impossible to have two
references with full permission to the same object.

In our setting and assuming that \lstinline{file_desc.pos} is a read
operation, the function type would be 
\begin{lstlisting}
  void m(f : Handle[ c ] .o g : Handle[ r ])
\end{lstlisting}
Suppose  we wish to pass some \lstinline|h : Handle[ rc ]| to both
parameters of \texttt{m}. To this end, we would have to 
split \texttt{h} at \texttt{r} into an ordered pair:
\begin{lstlisting}
  let (h1 .o h2) = split@r/c h in ...
\end{lstlisting}
In the body of the let, we have that \lstinline|h1 : Handle[ r ]|, \lstinline|h2 : Handle[ c ]|
and elimination of the ordered pair prescribes that \texttt{h1} and
\texttt{h2} must be consumed in this order. In particular, we \textbf{cannot} pass them
in the reverse order to \texttt{m}.

However, if we flip the parameters and the statements in the function
body, we obtain a function that we can call with the above split of \texttt{h}:
\begin{lstlisting}
  void ma(g : Handle[ r ] .o f : Handle[ c ]) {
   print(g.file_desc.pos); f.close();         }
\end{lstlisting}
This function can be called both with aliased or unaliased resources.

To obtain the same result in FT, we would have to revise the function
signature using the read-only permission  \texttt{pure}:
\begin{lstlisting}
  void ma(pure OF >> pure CF g, full OF >> full CF f)        /* FT code */
\end{lstlisting}
This FT signature faces a dilemma. If we use \texttt{pure CF} as the type of \texttt{g} after
the call, then we cannot call the function with unaliased arguments.  If
we use \texttt{pure OF} as the after-type of \texttt{g}, we cannot
call \texttt{ma} with aliased arguments.


\section{Language}
\label{sec:language}

In this section, we formalize a simply typed lambda calulus with
resource operations, ordered pair and function types, and
call-by-value evaluation. Compared to the linear calculus in the
informal presentation (\Cref{sec:motivation}), the calculus has an
unrestricted fragment with values that can be freely duplicated and
discarded.

Our formal semantics for resource operations does not perform actual
effects (like manipulating files, etc), it just tracks the trace of
operations as an element of an OPM and blocks if an operation is not
admitted by the specification. Consequently, our type soundness result is
actually trace soundness, that is, operations only occur in an order
admitted by their specification.

\subsection{Syntax}
\label{sec:language:syntax}

\begin{figure}
    \setlength{\extrarowheight}{\smallskipamount}
\begin{tabulary}{\linewidth}{@{}r@{\enspace}c@{\enspace}J<{\hfill\null}@{}}
    $\ottnt{c}$ & $\Coloneqq$ & $\ottkw{unit} \mid  \ottkw{new} _{ \ottnt{m} }  \mid  \ottkw{op} _{ \ottnt{m} }  \mid  \ottkw{split} _{ \ottnt{m_{{\mathrm{1}}}} , \ottnt{m_{{\mathrm{2}}}} }  \mid \ottkw{drop}$ \\
    $\ottnt{M}, \ottnt{N}$ & $\Coloneqq$ & $\ottnt{c} \mid \ottmv{l} \mid \ottmv{x} \mid \lambda  \ottmv{x}  \ottsym{.}  \ottnt{M} \mid \lambda^\circ  \ottmv{x}  \ottsym{.}  \ottnt{M} \mid \lambda^>  \ottmv{x}  \ottsym{.}  \ottnt{M} \mid \lambda^<  \ottmv{x}  \ottsym{.}  \ottnt{M} \mid \ottnt{M} \, \ottnt{N} \mid \ottnt{M}  {}^\circ  \ottnt{N} \mid \ottnt{M}  {}^>  \ottnt{N} \mid \ottnt{M}  {}^<  \ottnt{N} \mid \ottnt{M}  \otimes  \ottnt{N} \mid \ottnt{M}  \odot  \ottnt{N} \mid \ottkw{let} \, \ottmv{x}  \otimes  \ottmv{y}  \ottsym{=}  \ottnt{M} \, \ottkw{in} \, \ottnt{N} \mid \ottkw{let} \, \ottmv{x}  \odot  \ottmv{y}  \ottsym{=}  \ottnt{M} \, \ottkw{in} \, \ottnt{N}$
\end{tabulary}

    \caption{Syntax of expressions}
    \label{fig:syntax}
\end{figure}

\Cref{fig:syntax} defines the syntax, where elements of an {OPM} are
ranged over by $\ottnt{m}$, \emph{constants} by $\ottnt{c}$, \emph{expressions} by $\ottnt{M}$ and $\ottnt{N}$, \emph{locations} by $\ottmv{l}$, and \emph{variables} by $\ottmv{x}$ and $\ottmv{y}$.

Expressions include constants, locations, variables, abstractions, applications, and constructors and destructors for pairs.
Constants consist of $\ottkw{unit}$, the unique value of the type
$ \mathtt{Unit} $, and resource managing functions: $ \ottkw{new} _{ \ottnt{m} } $ to create a
new resource whose usage is specified by $\ottnt{m}$, $ \ottkw{op} _{ \ottnt{m} } $ to
perform the operation denoted by $\ottnt{m}$ on a resource, $ \ottkw{split} _{ \ottnt{m_{{\mathrm{1}}}} , \ottnt{m_{{\mathrm{2}}}} } $ to split the usage of a given resource into the head part
$\ottnt{m_{{\mathrm{1}}}}$ and the rest part $\ottnt{m_{{\mathrm{2}}}}$, and $\ottkw{drop}$ to dispose of a resource.

Locations, which only occur at run time, designate a resource.
We always refer to a resource via its location from a program, meaning
the resource managing functions operate on locations.
Abstractions vary in how they capture resources.
A \emph{non-capture abstraction} $\lambda  \ottmv{x}  \ottsym{.}  \ottnt{M}$ captures no resources.
On the other hand, \emph{capture abstractions}, $\lambda^\circ  \ottmv{x}  \ottsym{.}  \ottnt{M}$,
$\lambda^>  \ottmv{x}  \ottsym{.}  \ottnt{M}$, and $\lambda^<  \ottmv{x}  \ottsym{.}  \ottnt{M}$, may capture resources.
All abstractions bind variable $\ottmv{x}$ with scope $\ottnt{M}$.
The superscripts of $\lambda$ in capture abstractions specify a resource usage order between the resources in a given argument and those captured.
An \emph{unordered capture abstraction} $\lambda^\circ  \ottmv{x}  \ottsym{.}  \ottnt{M}$ has no order between them.
A \emph{right ordered capture abstraction} $\lambda^>  \ottmv{x}  \ottsym{.}  \ottnt{M}$ consumes
captured resources first and then uses those of the argument.
Conversely, a \emph{left ordered capture abstraction} $\lambda^<  \ottmv{x}  \ottsym{.}  \ottnt{M}$ consumes resources in the argument first.
Each abstraction comes with a corresponding applications $\ottnt{M} \, \ottnt{N}$,
$\ottnt{M}  {}^\circ  \ottnt{N}$, $\ottnt{M}  {}^>  \ottnt{N}$, and $\ottnt{M}  {}^<  \ottnt{N}$, that is, $\ottnt{M}  {}^\circ  \ottnt{N}$
expects evaluating $\ottnt{M}$ results in $\lambda^\circ  \ottmv{x}  \ottsym{.}  \ottnt{M_{{\mathrm{1}}}}$, and so on.

Pairs vary in whether elements have a usage order between them.
An \emph{unordered pair} $\ottnt{M}  \otimes  \ottnt{N}$ has no usage order between the values obtained from $\ottnt{M}$ and $\ottnt{N}$.
A typical situation is that either or both values have no resources, but it is also possible that both have resources that are independent.
An \emph{ordered pair} $\ottnt{M}  \odot  \ottnt{N}$, on the other hand, has a usage order.
We need to consume the resources in the value obtained from $\ottnt{M}$ before those from $\ottnt{N}$.
There are two corresponding elimination constructs $\ottkw{let} \, \ottmv{x}  \otimes  \ottmv{y}  \ottsym{=}  \ottnt{M} \, \ottkw{in} \, \ottnt{N}$ and $\ottkw{let} \, \ottmv{x}  \odot  \ottmv{y}  \ottsym{=}  \ottnt{M} \, \ottkw{in} \, \ottnt{N}$ that decompose the pair obtained
from $\ottnt{M}$ and bind $\ottmv{x}$ and $\ottmv{y}$ with scope $N$ accordingly.
Elimination of the unordered pair imposes no restriction on uses of
$\ottmv{x}$ and $\ottmv{y}$ in $\ottnt{N}$, but for the ordered pair $\ottmv{x}$ must
be consumed before $\ottmv{y}$.

\begin{remark}[Let expressions]\label{sec:language:letexpression}
    We can derive several \texttt{let} expressions with different constraints on resource usage, which appear as typing constraints.
    Here are some of the derived \texttt{let} expressions.
    \begin{align*}
        \ottkw{let} \, \ottmv{x}  \ottsym{=}  \ottnt{M} \, \ottkw{in} \, \ottnt{N}   \eqdef \ottsym{(}  \lambda^\circ  \ottmv{x}  \ottsym{.}  \ottnt{N}  \ottsym{)}  {}^\circ  \ottnt{M}  \qquad
        \mathtt{let}^>  \ottmv{x}  \ottsym{=}  \ottnt{M} \, \ottkw{in} \, \ottnt{N}  \eqdef \ottsym{(}  \lambda^<  \ottmv{x}  \ottsym{.}  \ottnt{N}  \ottsym{)}  {}^<  \ottnt{M} 
    \end{align*}
    The first one corresponds to the typing rule (\textsc{let-}$\parallel$), and the second one to (\textsc{let-}$,$) in \Cref{sec:splitting-order}.
\end{remark}

\subsection{Operational Semantics}
\label{sec:language:semantics}

We define the operational semantics as a small-step reduction relation
between \emph{configurations}---pairs of expressions and
\emph{heaps}.

\begin{definition}[Heap]
    A \emph{heap}, $\mathcal{H}$, is a finite map from locations
    to triples $\ottsym{(}  \ottnt{n}  \ottsym{,}  \ottnt{m_{{\mathrm{0}}}}  \ottsym{,}  \ottnt{m}  \ottsym{)}$, where $\ottnt{n}$ is a reference count
    (with $0$ denoting one reference), $\ottnt{m_{{\mathrm{0}}}}$ is the usage of the
    resource specified at its creation, and $\ottnt{m}$ records the trace of operations
    applied to the resource so far.
    We write $ \ottkw{dom} ( \mathcal{H} ) $ for the domain of $\mathcal{H}$
    and $ \mathcal{H} \cup \ottsym{\{} \ottmv{l}  \mapsto  \ottsym{(}  \ottnt{n}  \ottsym{,}  \ottnt{m_{{\mathrm{0}}}}  \ottsym{,}  \ottnt{m}  \ottsym{)} \ottsym{\}} $ for the updated heap that maps
    $\ottmv{l}$ to $\ottsym{(}  \ottnt{n}  \ottsym{,}  \ottnt{m_{{\mathrm{0}}}}  \ottsym{,}  \ottnt{m}  \ottsym{)}$ and other $\ottmv{l'} \ne \ottmv{l}$ to $\mathcal{H}  \ottsym{(}  \ottmv{l'}  \ottsym{)}$.
\end{definition}

\begin{figure}
    \setlength{\extrarowheight}{\smallskipamount}
\begin{tabulary}{\linewidth}{@{}r@{\enspace}c@{\enspace}J<{\hfill\null}@{}}
    $\ottnt{V}$ & $\Coloneqq$ & $\ottnt{c} \mid \ottmv{l} \mid \lambda  \ottmv{x}  \ottsym{.}  \ottnt{M} \mid \lambda^\circ  \ottmv{x}  \ottsym{.}  \ottnt{M} \mid \lambda^>  \ottmv{x}  \ottsym{.}  \ottnt{M} \mid \lambda^<  \ottmv{x}  \ottsym{.}  \ottnt{M} \mid \ottnt{V_{{\mathrm{1}}}}  \otimes  \ottnt{V_{{\mathrm{2}}}} \mid \ottnt{V_{{\mathrm{1}}}}  \odot  \ottnt{V_{{\mathrm{2}}}}$ \\
    $\mathcal{E}$ & $\Coloneqq$ & $\ottsym{[]} \mid \mathcal{E} \, \ottnt{M} \mid \ottnt{V} \, \mathcal{E} \mid \mathcal{E}  {}^\circ  \ottnt{M} \mid \ottnt{V}  {}^\circ  \mathcal{E} \mid \mathcal{E}  {}^>  \ottnt{M} \mid \ottnt{V}  {}^>  \mathcal{E} \mid \mathcal{E}  {}^<  \ottnt{V} \mid \ottnt{M}  {}^<  \mathcal{E} \mid \mathcal{E}  \otimes  \ottnt{M} \mid \ottnt{V}  \otimes  \mathcal{E} \mid \mathcal{E}  \odot  \ottnt{M} \mid \ottnt{V}  \odot  \mathcal{E} \mid \ottkw{let} \, \ottmv{x}  \otimes  \ottmv{y}  \ottsym{=}  \mathcal{E} \, \ottkw{in} \, \ottnt{M} \mid \ottkw{let} \, \ottmv{x}  \odot  \ottmv{y}  \ottsym{=}  \mathcal{E} \, \ottkw{in} \, \ottnt{M}$
\end{tabulary}

    \caption{Values and evaluation contexts}
    \label{fig:value}
\end{figure}

We use a call-by-value strategy for the semantics, which is organized
by using \emph{values}, ranged over by $\ottnt{V}$, and \emph{evaluation contexts}, $\mathcal{E}$, shown in \Cref{fig:value}.
Values are a subset of expressions comprising constants, locations, abstractions, and pairs of values.
An evaluation context is an expression with one \emph{hole} $\ottsym{[]}$.
The hole specifies the position of the sub-expression of an expression to be evaluated in the next step.
The evaluation order is generally left-to-right except for left
application, namely $\mathcal{E}  {}^<  \ottnt{V}$ and $\ottnt{M}  {}^<  \mathcal{E}$, which is evaluated right-to-left.

As usual, we write $\mathcal{E}  \ottsym{[}  \ottnt{M}  \ottsym{]}$ for the expression obtained by filling
the hole in $\mathcal{E}$ with $\ottnt{M}$
and $\ottnt{M}  \ottsym{[}  \ottnt{V}  \ottsym{/}  \ottmv{x}  \ottsym{]}$ for the expression obtained by substituting
$\ottnt{V}$ for $\ottmv{x}$ in $\ottnt{M}$.

\begin{figure}
    \setlength{\ruleSep}{.2in}
    \setlength{\premiseSep}{.2in}
    \setlength{\labelSep}{2pt}
    \input{figures/operation}
    \caption{Operational semantics}
    \label{fig:operation}
\end{figure}

We define the semantics by three relations: \emph{expression reduction}, denoted by $\ottnt{M}  \longrightarrow_\beta  \ottnt{M'}$, \emph{configuration reduction}, denoted by $\ottnt{M}  \mid  \mathcal{H}  \longrightarrow_\gamma  \ottnt{M'}  \mid  \mathcal{H}'$, and \emph{contextual reduction}, denoted by $\ottnt{M}  \mid  \mathcal{H}  \longrightarrow  \ottnt{M'}  \mid  \mathcal{H}'$.
\Cref{fig:operation} contains their defining rules.

The expression reduction $\ottnt{M}  \longrightarrow_\beta  \ottnt{M'}$ states that $\ottnt{M}$ reduces $\ottnt{M'}$.
The rules consist of the usual $\beta$-reduction for applications and \texttt{let}-reduction for pair eliminations.
One thing we need to care about is consistency: the kind of
application and the abstraction in its function part and the kind of
let-expression and the pair in its decomposed part have to agree.

The configuration reduction $\ottnt{M}  \mid  \mathcal{H}  \longrightarrow_\gamma  \ottnt{M'}  \mid  \mathcal{H}'$ organizes
resource manipulation. It reads: expression $\ottnt{M}$ steps to $\ottnt{M'}$ under the heap $\mathcal{H}$, and the heap is changed into $\mathcal{H}'$.
Rule \ruleref{RC-Ne} picks a fresh location $\ottmv{l}$, returns it, and
adds the map from $\ottmv{l}$ to the triple $\ottsym{(}  \ottsym{0}  \ottsym{,}  \ottnt{m}  \ottsym{,}  \varepsilon  \ottsym{)}$ to the
current heap: the reference counter is initialized to $\ottsym{0}$ and the
operation trace to $ \varepsilon $ as no operations have been performed, yet.
Rule \ruleref{RC-Op} implements the resource-passing interface: it appends the new operation to the trace for the resource at the specified location and returns the same location.
From the heap we know that the resource was allocated with behavior
$\ottnt{m_{{\mathrm{0}}}}$ and uses up to now accumulate to $\ottnt{m'}$. The side condition
checks if the operation
$\ottnt{m}$ is admissibile: the extension $ \ottnt{m'} \odot \ottnt{m} $ must be defined and
it must be possible to extend  $ \ottnt{m'} \odot \ottnt{m} $ to $\ottnt{m_{{\mathrm{0}}}}$ by some
$\ottnt{m''}$.

The rules \ruleref{RC-Cl1} and \ruleref{RC-Cl2} dispose of the location and return $\ottkw{unit}$.
Rule \ruleref{RC-Cl1} is for a location with an active alias in the
program, as seen from the reference count.
So, it just decrements the reference counter and does not dispose of the resource itself.
Rule \ruleref{RC-Cl2} applies to the last location pointing to a resource.
So, \ruleref{RC-Cl2} removes the resource from the heap, with the side
condition checking that the resource is fully used up, that is, the recorded operations belong to the original usage.
The premise $ \ottmv{l}  \notin   \ottkw{dom} ( \mathcal{H} )  $ in \ruleref{RC-Cl2} prevents the
resulting heap from having a resource at the disposed location.
At this point we can recover the borrowing interface to resources as it is
defineable by \texttt{i/op M}$ = \ottkw{drop} \, \ottsym{(}   \ottkw{op} _{ \ottnt{m} }  \, \ottnt{M}  \ottsym{)}$.

Rule \ruleref{RC-Sp} duplicates a location, increments the reference counter of the resource pointed to by the location, and returns the duplicated locations as a pair.
The reader may wonder why the rule does not touch the resource usage, although we have explained that $ \ottkw{split} _{ \ottnt{m} , \ottnt{m'} }  \, \ottmv{l}$ splits the resource usage for $\ottmv{l}$ by $\ottnt{m}$ and $\ottnt{m'}$.
In \Cref{sec:type-system}, we will see that we manage the resource usage solely by the type of the location.
Each of the three occurrences of  $\ottmv{l}$ in the expression may have a
different type, but this fact is not visible in the operational rule.

Contextual reduction $\ottnt{M}  \mid  \mathcal{H}  \longrightarrow  \ottnt{M'}  \mid  \mathcal{H}'$ defines reduction for
a whole program by reverting to either expression reduction or configuration
reduction in a suitable evaluation context.

\subsection{Type System}\label{sec:type-system}
\label{sec:language:type-system}

The primary objective of the type system is to prevent any operation
on a resource that would violate the usage specification at its creation.
This property is crucial for type safety---a well-typed program does not
get stuck---since \ruleref{RC-Op} and \ruleref{RC-Cl2} expect that
resources are used according to their specification.

\begin{figure}
    \setlength{\extrarowheight}{\smallskipamount}
\begin{tabulary}{\linewidth}{@{}r@{\enspace}c@{\enspace}J<{\hfill\null}@{}}
    $\ottnt{e}$ & $\Coloneqq$ & $\ottsym{0} \mid \ottsym{1} \mid  \ottnt{e_{{\mathrm{1}}}}  \sqcup  \ottnt{e_{{\mathrm{2}}}} $ \\
    $\ottnt{T}, \ottnt{S}$ & $\Coloneqq$ & $ \mathtt{Unit}  \mid \ottsym{[}  \ottnt{m}  \ottsym{]} \mid  \ottnt{S} \rightarrow _{ \ottnt{e} } \ottnt{T}  \mid  \ottnt{S} \rightarrowtriangle _{ \ottnt{e} } \ottnt{T}  \mid  \ottnt{S} \twoheadrightarrow _{ \ottnt{e} } \ottnt{T}  \mid  \ottnt{S} \rightarrowtail _{ \ottnt{e} } \ottnt{T}  \mid \ottnt{S}  \otimes  \ottnt{T} \mid \ottnt{S}  \odot  \ottnt{T}$
\end{tabulary}

    \caption{Syntax of types}
    \label{fig:types}
\end{figure}

\Cref{fig:types} defines the syntax of types, where $\ottnt{e}$ ranges
over \emph{effects} and  $\ottnt{T}$ and $\ottnt{S}$ over \emph{types}.
Effects are a crude way of keeping track whether resource operations might
happen. The effect $\ottsym{0}$ indicates that there are no operations;
effect $\ottsym{1}$ indicates that an operation might happen; effect
composition $ \ottnt{e_{{\mathrm{1}}}}  \sqcup  \ottnt{e_{{\mathrm{2}}}} $ just takes the maximum of $\ottnt{e_{{\mathrm{1}}}}$ and $\ottnt{e_{{\mathrm{2}}}}$.

The unit type $ \mathtt{Unit} $ is the type for $\ottkw{unit}$.
A trace type $\ottsym{[}  \ottnt{m}  \ottsym{]}$ classifies a resource location, which is good
for operations according to $\ottnt{m}$.
As customary in an ordered type system, there are several arrow types
corresponding to the different kinds of abstraction: that is, $ \ottnt{S} \rightarrow _{ \ottnt{e} } \ottnt{T} $ for non-capture abstractions; $ \ottnt{S} \rightarrowtriangle _{ \ottnt{e} } \ottnt{T} $ for unordered (linear)
capture abstractions; $ \ottnt{S} \twoheadrightarrow _{ \ottnt{e} } \ottnt{T} $ for right ordered capture
abstractions; and $ \ottnt{S} \rightarrowtail _{ \ottnt{e} } \ottnt{T} $ for left ordered capture
abstractions.
In each case, $\ottnt{S}$ is the parameter type, $\ottnt{T}$ is the return
type, and the subscript $\ottnt{e}$ denotes the latent effect that happens when the abstraction is invoked.
The last two types are product types: $\ottnt{S}  \otimes  \ottnt{T}$ for unordered pairs
and $\ottnt{S}  \odot  \ottnt{T}$ for ordered pairs, where $\ottnt{S}$ and $\ottnt{T}$ are the
types of the first and second projection.

Each type is classified as either \emph{unrestricted} or
\emph{ordered}, where an unrestricted type does not contain a
resource, a linear function, or an ordered function. For that reason,
values of unrestricted type can be safely duplicated.


\begin{definition}[Unrestricted and ordered types]
    We call $\ottnt{T}$ \emph{unrestricted} (\emph{ordered}) if and only if $\ottkw{unr} \, \ottnt{T}$ ($\ottkw{ord} \, \ottnt{T}$) is derivable by the following rules, respectively.
    \setlength{\ruleSep}{.2in}
    \setlength{\premiseSep}{.2in}
    \vskip\abovedisplayskip
    \begin{ynrules}
    \yninfer{
    }{
        $\ottkw{unr} \,  \mathtt{Unit} $
    }
    \yninfer{
    }{
        $\ottkw{unr} \, \ottsym{(}   \ottnt{S} \rightarrow _{ \ottnt{e} } \ottnt{T}   \ottsym{)}$
    }
    \yninfer{
        $\ottkw{unr} \, \ottnt{S}$ \\
        $\ottkw{unr} \, \ottnt{T}$
    }{
        $\ottkw{unr} \, \ottsym{(}  \ottnt{S}  \otimes  \ottnt{T}  \ottsym{)}$
    }
    \yninfer{
        $\ottkw{unr} \, \ottnt{S}$ \\
        $\ottkw{unr} \, \ottnt{T}$
    }{
        $\ottkw{unr} \, \ottsym{(}  \ottnt{S}  \odot  \ottnt{T}  \ottsym{)}$
    }
\end{ynrules}

\begin{ynrules}
    \yninfer{
    }{
        $\ottkw{ord} \, \ottsym{[}  \ottnt{m}  \ottsym{]}$
    }
    \yninfer{
    }{
        $\ottkw{ord} \, \ottsym{(}   \ottnt{S} \rightarrowtriangle _{ \ottnt{e} } \ottnt{T}   \ottsym{)}$
    }
    \yninfer{
    }{
        $\ottkw{ord} \, \ottsym{(}   \ottnt{S} \twoheadrightarrow _{ \ottnt{e} } \ottnt{T}   \ottsym{)}$
    }
    \yninfer{
    }{
        $\ottkw{ord} \, \ottsym{(}   \ottnt{S} \rightarrowtail _{ \ottnt{e} } \ottnt{T}   \ottsym{)}$
    }
    \yninfer{
        $\ottkw{ord} \, \ottnt{S}$
    }{
        $\ottkw{ord} \, \ottsym{(}  \ottnt{S}  \otimes  \ottnt{T}  \ottsym{)}$
    }
    \yninfer{
        $\ottkw{ord} \, \ottnt{T}$
    }{
        $\ottkw{ord} \, \ottsym{(}  \ottnt{S}  \otimes  \ottnt{T}  \ottsym{)}$
    }
    \yninfer{
        $\ottkw{ord} \, \ottnt{S}$
    }{
        $\ottkw{ord} \, \ottsym{(}  \ottnt{S}  \odot  \ottnt{T}  \ottsym{)}$
    }
    \yninfer{
        $\ottkw{ord} \, \ottnt{T}$
    }{
        $\ottkw{ord} \, \ottsym{(}  \ottnt{S}  \odot  \ottnt{T}  \ottsym{)}$
    }
\end{ynrules}

\end{definition}


\subsubsection{Typing Contexts}
\label{sec:language:type-system:typing-environments}


A novel part of our type system is the formulation of \emph{typing contexts}.
We define the syntax of typing contexts as follows, where
\emph{bindings} are ranged over by $\ottnt{b}$, typing contexts by $\Gamma$, and \emph{context patterns}  by $\mathcal{G}$.
\begin{align*}
    \ottnt{b} \Coloneqq \ottmv{x}  \mathord:  \ottnt{T} ~\mid~ \ottmv{l}  \mathord:  \ottsym{[}  \ottnt{m}  \ottsym{]}
    \qquad
    \Gamma \Coloneqq  \cdot  ~\mid~ \ottnt{b} ~\mid~ \Gamma_{{\mathrm{1}}}  \ottsym{,}  \Gamma_{{\mathrm{2}}} ~\mid~ \Gamma_{{\mathrm{1}}}  \parallel  \Gamma_{{\mathrm{2}}}
    \qquad
    \mathcal{G} \Coloneqq \ottsym{[]} ~\mid~ \Gamma  \ottsym{,}  \mathcal{G} ~\mid~ \mathcal{G}  \ottsym{,}  \Gamma ~\mid~ \Gamma  \parallel  \mathcal{G} ~\mid~ \mathcal{G}  \parallel  \Gamma
\end{align*}
A typing context assigns the type of each variable in an expression by a \emph{variable binding} $\ottmv{x}  \mathord:  \ottnt{T}$ and each location by a \emph{location binding} $\ottmv{l}  \mathord:  \ottsym{[}  \ottnt{m}  \ottsym{]}$ during typing.
The other context formers have been introduced in \Cref{sec:splitting-order}.
A typing context also represents the usage count and order of variables and locations.
For instance, the same variable may be bound several times and the contexts $\ottsym{(}  \ottmv{x}  \mathord:  \ottnt{T}  \ottsym{,}  \ottmv{x}  \mathord:  \ottnt{T}  \ottsym{)}$ and $\ottsym{(}  \ottmv{x}  \mathord:  \ottnt{T}  \ottsym{)}$ are different.


\begin{definition}[Notations for typing contexts]
    The \emph{domain} of $\Gamma$ is the set of variables and locations bound in $\Gamma$, denoted by $ \ottkw{dom} ( \Gamma ) $.
    We write  $ \ottsym{\#} _{ \ottnt{b} }( \Gamma ) $ for the number of occurences of $\ottnt{b}$ in $\Gamma$.
    We write $ \ottnt{b}  \in  \Gamma $ if and only if $ \ottsym{\#} _{ \ottnt{b} }( \Gamma )  > 0$.
    We write $\Gamma  \ottsym{[}  \Gamma'  \ottsym{/}  \ottnt{b}  \ottsym{]}$ to express the typing context obtained
    from $\Gamma$ by replacing all bindings $\ottnt{b}$ in $\Gamma$ with
    $\Gamma'$ and we write $ \Gamma ^{- \ottnt{b} } $ if $ \Gamma' \ottsym{=}  \cdot  $.
    We lift the notions of unrestricted and ordered types to bindings, denoted by $\ottkw{unr} \, \ottnt{b}$ and $\ottkw{ord} \, \ottnt{b}$, according to the type part of a binding.
    Furthermore, we lift the notion unrestricted to typing contexts, denoted by $\ottkw{unr} \, \Gamma$, meaning that $\ottkw{unr} \, \ottnt{b}$ for all $ \ottnt{b}  \in  \Gamma $.
    We use the same notations for context patterns
    $\mathcal{G}$, mutatis mutandis.
\end{definition}

As we have seen, a typing context may have multiple bindings for the same variable.
However, we do not use arbitrary typing contexts but only \emph{well-formed} ones.

\begin{definition}[Well-formed typing context]
    We call $\Gamma$ \emph{well-formed} if and only if
    \begin{itemize}
        \item $ \ottnt{S} \ottsym{=} \ottnt{T} $ for any $ \ottmv{x}  \mathord:  \ottnt{S}  \in  \Gamma $ and $ \ottmv{x}  \mathord:  \ottnt{T}  \in  \Gamma $, and
        \item $ \ottsym{\#} _{ \ottmv{x}  \mathord:  \ottnt{T} }( \Gamma )  \le 1$ for any $\ottmv{x}$ and $\ottkw{ord} \, \ottnt{T}$.
    \end{itemize}
\end{definition}

Informally, $\Gamma$ is well-formed if and only if each variable is assigned a unique type, and only unrestricted bindings can have multiple occurrences.
This definition respects the meaning of unrestricted and ordered types.
Ordered bindings can occur at most once.

\subsubsection{Semantics of Typing Contexts}

Many syntactically different typing contexts have the same meaning---for instance $\Gamma  \ottsym{,}   \cdot $ and $\Gamma$.
To precisely capture the meaning of a typing context, we interpret a typing context as a directed acyclic graph ({DAG}) of ordered bindings and a set of unrestricted bindings.
In the graph, vertices represent bindings and edges represent the ordering constraints among the bindings.
This structure captures the strength of an ordering constraint as the number of edges in the graph interpretation.
Using the interpretation, we can define equivalence and a subcontext relation between typing contexts semantically.

We start by recalling some definitions from graph theory.

\begin{definition}[Graph representation]
    A \emph{graph representation}, denoted by $\mathfrak{G}$, is a triple
    $(n, v, E)$, where $n\ge0$ is the number of
    vertices in the graph,
    $v$ is a map from $\NAT_{< n}$
    to bindings labeling each vertex by a binding, and $E$ is a binary
    relation between $\NAT_{< n}$ and $\NAT_{< n}$ representing the
    directed edges of the graph.
    As we mentioned, a graph representation must be acyclic.
\end{definition}



\begin{figure}
    \input{figures/ginterp}
    \caption{Join and union of graph representations}
    \label{fig:language:graph-join-union}
\end{figure}

\begin{definition}[Join and union of graph representations]
    Let $\mathfrak{G}_{{\mathrm{1}}} = (n_1, v_1, E_1)$ and $\mathfrak{G}_{{\mathrm{2}}} = (n_2, v_2, E_2)$.
    The join $ \GJOIN{ \mathfrak{G}_{{\mathrm{1}}} }{ \mathfrak{G}_{{\mathrm{2}}} } $ and the union $\mathfrak{G}_{{\mathrm{1}}}  \cup  \mathfrak{G}_{{\mathrm{2}}}$ are defined in \Cref{fig:language:graph-join-union}.
\end{definition}

\begin{definition}[Isomorphic graph representation]
    Let $\mathfrak{G}_{{\mathrm{1}}} = (n, v_1, E_1)$ and $\mathfrak{G}_{{\mathrm{2}}} = (n, v_2, E_2)$.
    We call $\mathfrak{G}_{{\mathrm{1}}}$ and $\mathfrak{G}_{{\mathrm{2}}}$ \emph{isomorphic} and write $\mathfrak{G}_{{\mathrm{1}}}  \simeq  \mathfrak{G}_{{\mathrm{2}}}$ if and only if there is a bijection $f$ from $\NAT_{< n}$ to $\NAT_{< n}$ satisfying $E_1 = \{(f(n), f(m)) \mid (n, m) \in E_2\}$ and $v_1 \circ f = v_2$.
\end{definition}

\begin{definition}[Spanning graph representations]
    Let $\mathfrak{G}_{{\mathrm{1}}} = (n, v, E_1)$ and $\mathfrak{G}_{{\mathrm{2}}} = (n, v, E_2)$.
    We call $\mathfrak{G}_{{\mathrm{1}}}$ a \emph{spanning graph representation} of $\mathfrak{G}_{{\mathrm{2}}}$, denoted by $\mathfrak{G}_{{\mathrm{1}}}  \mathrel{<_\rightarrow}  \mathfrak{G}_{{\mathrm{2}}}$, if and only if $E_1 \subseteq E_2$.
\end{definition}

\begin{definition}[Topological ordering of a graph representation]
    A \emph{topological ordering} of $\mathfrak{G} = (n, v, E)$ is a bijection $f$ from $\NAT_{< n}$ to $\NAT_{< n}$ satisfying $f^{-1}(n_1) < f^{-1}(n_2)$ for all $(n_1, n_2) \in E$.
    The bijective function defines the ordering of the vertices of the representation.
\end{definition}

Now, we define the interpretation of a typing context and then define
equivalence and the subcontext relation using the interpretation.

\begin{definition}[Context interpretation]
    We define the \emph{interpretation} of $\Gamma$, denoted by
    $ \lBrack \Gamma \rBrack $, as a pair of a graph and a set of unrestricted bindings.
    \begin{align*}
         \lBrack  \cdot  \rBrack         & = (0, \emptyset, \emptyset); \emptyset                                                                                                       \\
         \lBrack \ottnt{b} \rBrack         & = \begin{cases}
                                    (0, \emptyset, \emptyset); \{b\}               & \text{if $\ottkw{unr} \, \ottnt{b}$} \\
                                    (1, \{0 \mapsto \ottnt{b}\}, \emptyset); \emptyset & \text{if $\ottkw{ord} \, \ottnt{b}$}
                                \end{cases} \\
         \lBrack \Gamma_{{\mathrm{1}}}  \ottsym{,}  \Gamma_{{\mathrm{2}}} \rBrack    & =  \GJOIN{ \mathfrak{G}_{{\mathrm{1}}} }{ \mathfrak{G}_{{\mathrm{2}}} } ;  \mathbb{S}_{{\mathrm{1}}} \cup \mathbb{S}_{{\mathrm{2}}}  \qquad \text{where} \quad \text{$  \lBrack \Gamma_{{\mathrm{1}}} \rBrack  \ottsym{=} \mathfrak{G}_{{\mathrm{1}}}  \ottsym{;}  \mathbb{S}_{{\mathrm{1}}} $ and $  \lBrack \Gamma_{{\mathrm{2}}} \rBrack  \ottsym{=} \mathfrak{G}_{{\mathrm{2}}}  \ottsym{;}  \mathbb{S}_{{\mathrm{2}}} $}                   \\
         \lBrack \Gamma_{{\mathrm{1}}}  \parallel  \Gamma_{{\mathrm{2}}} \rBrack  & = \mathfrak{G}_{{\mathrm{1}}}  \cup  \mathfrak{G}_{{\mathrm{2}}};  \mathbb{S}_{{\mathrm{1}}} \cup \mathbb{S}_{{\mathrm{2}}}  \qquad \text{where} \quad \text{$  \lBrack \Gamma_{{\mathrm{1}}} \rBrack  \ottsym{=} \mathfrak{G}_{{\mathrm{1}}}  \ottsym{;}  \mathbb{S}_{{\mathrm{1}}} $ and $  \lBrack \Gamma_{{\mathrm{2}}} \rBrack  \ottsym{=} \mathfrak{G}_{{\mathrm{2}}}  \ottsym{;}  \mathbb{S}_{{\mathrm{2}}} $}
    \end{align*}
\end{definition}





\begin{definition}[Equivalent typing contexts]
    We say $\Gamma_{{\mathrm{1}}}$ is equivalent to $\Gamma_{{\mathrm{2}}}$ and write $\Gamma_{{\mathrm{1}}}  \simeq  \Gamma_{{\mathrm{2}}}$ if and only if $\mathfrak{G}_{{\mathrm{1}}}  \simeq  \mathfrak{G}_{{\mathrm{2}}} \wedge S_1 = S_2$ for $\mathfrak{G}_{{\mathrm{1}}}; S_1 =  \lBrack \Gamma_{{\mathrm{1}}} \rBrack $ and $\mathfrak{G}_{{\mathrm{2}}}; S_2 =  \lBrack \Gamma_{{\mathrm{2}}} \rBrack $.
\end{definition}

In \Cref{sec:splitting-order} we said that the context formers
$\Gamma_{{\mathrm{1}}}  \parallel  \Gamma_{{\mathrm{2}}}$ and $\Gamma_{{\mathrm{1}}}  \ottsym{,}  \Gamma_{{\mathrm{2}}}$ are monoids with unit $ \cdot $ and that
unordered composition is commutative. Moreover, ordered composition
with an unrestricted binding $\ottnt{b}$ is the same as unordered
composition with $\ottnt{b}$ (demotion). The graph interpretation of a
context yields a semantic characterization of these algebraic laws.
\begin{proposition}\label{prop:type-system/graph-cmp}
    Contexts $\Gamma_{{\mathrm{1}}}$ and $\Gamma_{{\mathrm{2}}}$ are related by applying the monoid laws,
    commutativity of unordered composition, and demotion with
    unrestricted bindings, if and only if  $\Gamma_{{\mathrm{1}}}  \simeq  \Gamma_{{\mathrm{2}}}$.
\begin{FULLVERSION}
    \proof See Appendix~\ref{sec:graph-cmp}.
\end{FULLVERSION}
\end{proposition}

\begin{definition}[Subcontexts]
    We say $\Gamma_{{\mathrm{1}}}$ is a subcontext of $\Gamma_{{\mathrm{2}}}$ and write $\Gamma_{{\mathrm{1}}}  \lesssim  \Gamma_{{\mathrm{2}}}$ if and only if there exist $\mathfrak{G}'_{{\mathrm{1}}}$ and $\mathfrak{G}'_{{\mathrm{2}}}$ such that $\mathfrak{G}_{{\mathrm{1}}} \sim \mathfrak{G}'_{{\mathrm{1}}}  \mathrel{<_\rightarrow}  \mathfrak{G}'_{{\mathrm{2}}} \sim \mathfrak{G}_{{\mathrm{2}}} \wedge S_1 \supseteq S_2$ for $\mathfrak{G}_{{\mathrm{1}}}; S_1 =  \lBrack \Gamma_{{\mathrm{1}}} \rBrack $ and $\mathfrak{G}_{{\mathrm{2}}}; S_2 =  \lBrack \Gamma_{{\mathrm{2}}} \rBrack $.
\end{definition}

The subcontext relation
adds more ordering constraints between the bindings in a typing context.
If $\Gamma_{{\mathrm{1}}}  \lesssim  \Gamma_{{\mathrm{2}}}$, a typing under $\Gamma_{{\mathrm{2}}}$ is more constrained on variables and locations use.
So, we design our type system so that an expression is typeable under $\Gamma_{{\mathrm{1}}}$ if we can type an expression under $\Gamma_{{\mathrm{2}}}$.
The spanning condition characterizes the additional edges that
sequentialize more bindings.

The subcontext relation cannot be used like classic weakening. For
instance, we will see that we can apply rule \textsc{[T-UPair]} to construct
$\ottmv{x_{{\mathrm{1}}}}  \otimes  \ottmv{x_{{\mathrm{2}}}}$  given the context $\ottmv{x_{{\mathrm{1}}}}  \mathord:  \ottnt{T_{{\mathrm{1}}}}  \parallel  \ottmv{x_{{\mathrm{2}}}}  \mathord:  \ottnt{T_{{\mathrm{2}}}}$. But this is
not possible in the subcontext  $\ottmv{x_{{\mathrm{1}}}}  \mathord:  \ottnt{T_{{\mathrm{1}}}}  \ottsym{,}  \ottmv{x_{{\mathrm{2}}}}  \mathord:  \ottnt{T_{{\mathrm{2}}}}$.


Unrestricted bindings enjoy all substructural rules---exchange,
weakening, and contraction---, which is managed by their treatment as
a set in the interpretation.

\subsubsection{Typing Rules}
\label{sec:language:type-system:typing-rules}

\begin{figure}
    \setlength{\ruleSep}{.2in}
    \setlength{\premiseSep}{.2in}
    \setlength{\labelSep}{2pt}
    \input{figures/typing}
    \caption{Typing rules}
    \label{fig:type/system}
\end{figure}

The typing judgment $\Gamma  \vdash  \ottnt{M}  :  \ottnt{T}  \mid  \ottnt{e}$ says that $\ottnt{M}$ has type
$\ottnt{T}$ under context $\Gamma$ with effect $\ottnt{e}$.  The judgment is
defined by the inference rules in \Cref{fig:type/system}.

The typing rules for constants reflect their semantics and have no effect.
A constant $ \ottkw{op} _{ \ottnt{m} } $ only has an effect on its application as
indicated by the subscript $\ottsym{1}$ of the arrow.
As explained earlier, we only consider resource operations as an
effect; effects are not related to heap modifications.
The typing rules for variables and locations simply pick the type in the typing context.
Rule \ruleref{T-Weaken} enables us to weaken the typing context
according to the subcontext relation and to weaken the effect (a
standard treatment for subeffecting).

The most intriguing rules are ones for abstractions and applications,
which control the usage order of resources that do not follow the
syntactic order.
Rule \ruleref{T-Abs} handles non-capture abstractions, which must not capture a resource.
The premise $\ottkw{unr} \, \Gamma$ expresses this condition.
The parameter type $\ottnt{S}$ can be ordered as an actual parameter is passed from outside.
The corresponding application rule is \ruleref{T-App}.
The typing context of the conclusion is supposed to be $\Gamma_{{\mathrm{1}}}  \ottsym{,}  \Gamma_{{\mathrm{2}}}$, and $\ottnt{M}$ is typed under $\Gamma_{{\mathrm{1}}}$ and $\ottnt{N}$ is typed under $\Gamma_{{\mathrm{2}}}$ because $\ottnt{M}$ is evaluated first.
The function part and the argument part can have an effect, as neither $\ottnt{e_{{\mathrm{1}}}}$ nor $\ottnt{e_{{\mathrm{2}}}}$ is constrained.
The argument part can have an effect because the function part becomes a non-capture abstraction; that means all usages assigned for $\Gamma_{{\mathrm{1}}}$ are consumed, and thus, $\ottnt{N}$ can perform operations on resources in $\Gamma_{{\mathrm{2}}}$.

In the typing rules for capture abstractions, $\Gamma$ may contain
ordered bindings, but the specific rules for concatenation of the binding for the parameter differs for each abstraction.
The difference affects how the typing context is split in the corresponding application rule.
For example, the typing environment of rule \ruleref{T-UApp} is supposed to be $\Gamma_{{\mathrm{1}}}  \parallel  \Gamma_{{\mathrm{2}}}$, which corresponds to the typing environment $\Gamma  \parallel  \ottmv{x}  \mathord:  \ottnt{S}$ of the premise in \ruleref{T-UAbs}.
This correspondence maintains unconsumed resources in the argument part to fit into the original usage position when an application happens.
Another remark is the effects of each subexpression.
In rule \ruleref{T-UApp}, both subexpressions can have an effect.
The evaluation of $\ottnt{M}$ may not consume resources in $\Gamma_{{\mathrm{1}}}$ fully,
which are then captured by the resulting abstraction, but $\ottnt{N}$ can still use a resource in $\Gamma_{{\mathrm{2}}}$ because $\Gamma_{{\mathrm{1}}}$ and $\Gamma_{{\mathrm{2}}}$ are order independent.
On the other hand, one of the subexpressions in rules \ruleref{T-RApp}
and \ruleref{T-LApp} must not have an effect because, in the rule,
$\Gamma_{{\mathrm{1}}}$ and $\Gamma_{{\mathrm{2}}}$ has an ordering constraint.
We also note that $\Gamma_{{\mathrm{1}}}$ and $\Gamma_{{\mathrm{2}}}$ are swapped in the premises of
rule \ruleref{T-LApp}, which reflects the evaluation order of the left application.

The typing rules for pairs are organized similarly to abstractions and applications.
Regarding the elimination rules---\ruleref{T-ULet} and
\ruleref{T-OLet}, we do not consider variations in how to introduce
the bound variables for the typing context for $\ottnt{N}$, which we can
have, but it can be handled as explained in
\Cref{sec:splitting-order}.

\begin{example}\label{example:lettyping}
    The following are derivations for the desugaring of \texttt{let}-expressions presented in \Cref{sec:language:letexpression}, which clarify the typing rules for them (we have anticipated it in \Cref{sec:splitting-order}).
    \begin{prooftree}
        \EnableBpAbbreviations
        \AXC{$\Gamma_{{\mathrm{2}}}  \parallel  \ottmv{x}  \mathord:  \ottnt{S}  \vdash  \ottnt{N}  :  \ottnt{T}  \mid  \ottnt{e_{{\mathrm{2}}}}$}
        \RL{\ruleref{T-UAbs}}
        \UIC{$\Gamma_{{\mathrm{2}}}  \vdash  \lambda^\circ  \ottmv{x}  \ottsym{.}  \ottnt{N}  :   \ottnt{S} \rightarrowtriangle _{ \ottnt{e_{{\mathrm{2}}}} } \ottnt{T}   \mid  \ottsym{0}$}
        \AXC{$\Gamma_{{\mathrm{1}}}  \vdash  \ottnt{M}  :  \ottnt{S}  \mid  \ottnt{e_{{\mathrm{1}}}}$}
        \RL{\ruleref{T-UApp} and \ruleref{T-Weaken}}
        \BIC{$\Gamma_{{\mathrm{1}}}  \parallel  \Gamma_{{\mathrm{2}}}  \vdash  \ottsym{(}  \lambda^\circ  \ottmv{x}  \ottsym{.}  \ottnt{N}  \ottsym{)}  {}^\circ  \ottnt{M}  :  \ottnt{T}  \mid   \ottnt{e_{{\mathrm{1}}}}  \sqcup  \ottnt{e_{{\mathrm{2}}}} $}
    \end{prooftree}
    \begin{prooftree}
        \EnableBpAbbreviations
        \AXC{$\ottmv{x}  \mathord:  \ottnt{S}  \ottsym{,}  \Gamma_{{\mathrm{2}}}  \vdash  \ottnt{N}  :  \ottnt{T}  \mid  \ottnt{e_{{\mathrm{2}}}}$}
        \RL{\ruleref{T-LAbs}}
        \UIC{$\Gamma_{{\mathrm{2}}}  \vdash  \lambda^<  \ottmv{x}  \ottsym{.}  \ottnt{N}  :   \ottnt{S} \rightarrowtail _{ \ottnt{e_{{\mathrm{2}}}} } \ottnt{T}   \mid  \ottsym{0}$}
        \AXC{$\Gamma_{{\mathrm{1}}}  \vdash  \ottnt{M}  :  \ottnt{S}  \mid  \ottnt{e_{{\mathrm{1}}}}$}
        \RL{\ruleref{T-LApp}}
        \BIC{$\Gamma_{{\mathrm{1}}}  \ottsym{,}  \Gamma_{{\mathrm{2}}}  \vdash  \ottsym{(}  \lambda^<  \ottmv{x}  \ottsym{.}  \ottnt{N}  \ottsym{)}  {}^<  \ottnt{M}  :  \ottnt{T}  \mid   \ottnt{e_{{\mathrm{1}}}}  \sqcup  \ottnt{e_{{\mathrm{2}}}} $}
    \end{prooftree}
\end{example}

\begin{example}
    Let's examine typing for the file permission example introduced in \Cref{example:opm-files} on the formal calculus.
    Consider the following typeable expression, where $\ottnt{M}; \ottnt{N}$ is the abbreviation of $\mathtt{let}^>  \ottmv{x}  \ottsym{=}  \ottnt{M} \, \ottkw{in} \, \ottnt{N}$ for some fresh $\ottmv{x}$.
    This expression creates the \emph{thunk} $f$ that keeps \texttt{read} operation on the opened file and uses it after that.
    So, we need to take care of the usage order of $f$, too.
    \begin{align*}
         & \texttt{let} \, x = \texttt{new}_{(r\vert w)^{*}c} \, \ottkw{unit} \, \texttt{in}                        \\
         & \texttt{let} \, \ottmv{x_{{\mathrm{1}}}}  \odot  \ottmv{x_{{\mathrm{2}}}} = \texttt{split}_{r,(r\vert w)^{*}c} \, x \, \texttt{in}                \\
         & \texttt{let}^{>} \, f = \lambda^\circ z. \texttt{drop} \, (\texttt{op}_{r} \, \ottmv{x_{{\mathrm{1}}}}) \, \texttt{in} \\
         & f^\circ\ottkw{unit}; \texttt{drop}\,(\texttt{op}_{c}\,\ottmv{x_{{\mathrm{2}}}})
    \end{align*}
    Let's see how our type system handles the usage order of $f$.
    For ease of writing, we abbreviate subexpressions as follows.
    \begin{align*}
        \ottnt{M}  & \Coloneqq \texttt{let} \, x = \texttt{new}_{(r\vert w)^{*}c} \, \ottkw{unit} \, \texttt{in} \, \ottnt{M_{{\mathrm{1}}}}                        \\
        \ottnt{M_{{\mathrm{1}}}} & \Coloneqq \texttt{let} \, \ottmv{x_{{\mathrm{1}}}}  \odot  \ottmv{x_{{\mathrm{2}}}} = \texttt{split}_{r,(r\vert w)^{*}c} \, x \, \texttt{in} \, \ottnt{M_{{\mathrm{2}}}}                \\
        \ottnt{M_{{\mathrm{2}}}} & \Coloneqq \texttt{let}^{>} \, f = \lambda^\circ z. \texttt{drop} \, (\texttt{op}_{r} \, \ottmv{x_{{\mathrm{1}}}}) \, \texttt{in} \, \ottnt{M_{{\mathrm{3}}}} \\
        \ottnt{M_{{\mathrm{3}}}} & \Coloneqq f^\circ\ottkw{unit}; \texttt{drop}\,(\texttt{op}_{c}\,\ottmv{x_{{\mathrm{2}}}})
    \end{align*}
    Then the expression is typed by the following derivations---it is still not complete, but we can fill the omitted subderivations.
    \small
    \begin{prooftree}
        \EnableBpAbbreviations
        \AXC{$f \mathord:   \mathtt{Unit}  \rightarrowtriangle _{ \ottsym{1} }  \mathtt{Unit}   \vdash f^\circ\ottkw{unit} :  \mathtt{Unit}  \mid 1$}
        \AXC{$\ottmv{x_{{\mathrm{2}}}} \mathord: [(r\vert w)^{*}c] \vdash \texttt{drop}\,(\texttt{op}_{c}\,\ottmv{x_{{\mathrm{2}}}}) :  \mathtt{Unit}  \mid 1$}
        \BIC{$f \mathord:   \mathtt{Unit}  \rightarrowtriangle _{ \ottsym{1} }  \mathtt{Unit}   , \ottmv{x_{{\mathrm{2}}}} \mathord: [(r\vert w)^{*}c] \vdash \ottnt{M_{{\mathrm{3}}}} :  \mathtt{Unit}  \mid 1$}
    \end{prooftree}
    \begin{prooftree}
        \EnableBpAbbreviations
        \AXC{$\ottmv{x_{{\mathrm{1}}}}\mathord:[r] \vdash \lambda^\circ z. \texttt{drop} \, (\texttt{op}_{r} \, \ottmv{x_{{\mathrm{1}}}}) :   \mathtt{Unit}  \rightarrowtriangle _{ \ottsym{1} }  \mathtt{Unit}   \mid 0$}
        \AXC{$f \mathord:   \mathtt{Unit}  \rightarrowtriangle _{ \ottsym{1} }  \mathtt{Unit}   , \ottmv{x_{{\mathrm{2}}}} \mathord: [(r\vert w)^{*}c] \vdash \ottnt{M_{{\mathrm{3}}}} :  \mathtt{Unit}  \mid 1$}
        \BIC{$\ottmv{x_{{\mathrm{1}}}}\mathord:[r], \ottmv{x_{{\mathrm{2}}}}\mathord:[(r\vert w)^{*}c] \vdash \ottnt{M_{{\mathrm{2}}}} :  \mathtt{Unit}  \mid 1$}
    \end{prooftree}
    \begin{prooftree}
        \EnableBpAbbreviations
        \AXC{$x\mathord:[(r\vert w)^{*}c] \vdash \texttt{split}_{r,(r\vert w)^{*}c} \, x : [r]  \odot  [(r\vert w)^{*}c] \mid 0$}
        \AXC{$\ottmv{x_{{\mathrm{1}}}}\mathord:[r], \ottmv{x_{{\mathrm{2}}}}\mathord:[(r\vert w)^{*}c] \vdash \ottnt{M_{{\mathrm{2}}}} :  \mathtt{Unit}  \mid 1$}
        \BIC{$x\mathord:[(r\vert w)^{*}c] \vdash \ottnt{M_{{\mathrm{1}}}} :  \mathtt{Unit}  \mid 1$}
    \end{prooftree}
    \begin{prooftree}
        \EnableBpAbbreviations
        \AXC{$ \cdot  \vdash \texttt{new}_{(r\vert w)^{*}c} \, \ottkw{unit} : [(r\vert w)^{*}c] \mid 0$}
        \AXC{$x\mathord:[(r\vert w)^{*}c] \vdash \ottnt{M_{{\mathrm{1}}}} :  \mathtt{Unit}  \mid 1$}
        \BIC{$ \cdot  \vdash \ottnt{M} :  \mathtt{Unit}  \mid 1$}
      \end{prooftree}
      \normalsize
    The subexpression at position $\ottnt{M_{{\mathrm{3}}}}$ is typed in context $f \mathord:   \mathtt{Unit}  \rightarrowtriangle _{ \ottsym{1} }  \mathtt{Unit}   , \ottmv{x_{{\mathrm{2}}}} \mathord: [(r\vert w)^{*}c]$.
    The type system rejects this misuse of thunk $f$ as in $\texttt{drop}\,(\texttt{op}_{c}\,\ottmv{x_{{\mathrm{2}}}}); f^\circ\ottkw{unit}$.
\end{example}

\subsection{Type Soundness}
\label{sec:language:properties}

In this subsection we sketch the proof of type soundness. It is
structured in the standard way by the following two propositions,
preservation and progress, where $\ottnt{C}$ denotes a \emph{run-time
    context}, which assigns a resource usage to each label, and $\ottnt{C}  \vdash  \mathcal{H}$ is a \emph{heap typing}; both are defined below.

\begin{proposition}[Preservation]
    If $\ottnt{M}  \mid  \mathcal{H}  \longrightarrow  \ottnt{M'}  \mid  \mathcal{H}'$, $\ottnt{C}  \vdash  \ottnt{M}  :  \ottnt{T}  \mid  \ottnt{e}$, and $\ottnt{C}  \vdash  \mathcal{H}$, then there exists $\ottnt{C'}$ such that $\ottnt{C'}  \vdash  \ottnt{M'}  :  \ottnt{T}  \mid  \ottnt{e}$ and $\ottnt{C'}  \vdash  \mathcal{H}'$.

\end{proposition}

\begin{proposition}[Progress]
    If $\ottnt{C}  \vdash  \ottnt{M}  :  \ottnt{T}  \mid  \ottnt{e}$ and $\ottnt{C}  \ottsym{,}  \ottnt{C'}  \vdash  \mathcal{H}$, then
    \begin{enumerate}
        \item $\ottnt{M}$ is a value, namely $ \ottnt{M} \ottsym{=} \ottnt{V} $ for some $\ottnt{V}$, or
        \item $\ottnt{M}  \mid  \mathcal{H}  \longrightarrow  \ottnt{M'}  \mid  \mathcal{H}'$ for some $\ottnt{M'}$ and $\mathcal{H}'$.
    \end{enumerate}
\end{proposition}

Roughly speaking, preservation says the usage assignment by $\ottnt{C}$,
which changes by reducing an expression, stays correct against the
specification in the heap, and progress says that a well-typed expression does not get stuck.



We start by defining run-time contexts and heap typing.

\begin{definition}[Run-time context]
    We call a typing context \emph{run-time context}, denoted by $\ottnt{C}$, if and only if the typing context has no variable bindings.
    We abuse $ \lBrack \ottnt{C} \rBrack $ to denote only the graph part of the interpretation because the set part is always empty.
\end{definition}

\begin{definition}[Run-time context pattern]
    We call a typing context pattern \emph{run-time context pattern}, denoted by $\mathcal{C}$, if and only if the pattern has no variable bindings.
\end{definition}

\begin{definition}[Focus on location]
    A \emph{focus} on location $\ottmv{l}$ in $\ottnt{C}$, denoted by $ \langle \ottnt{C} \rangle_{ \ottmv{l} } $, is the run-time context obtained from $\ottnt{C}$ by replacing all bindings $ \ottmv{l'}  \mathord:  \ottsym{[}  \ottnt{m}  \ottsym{]}  \in  \ottnt{C} $ such that $ \ottmv{l'} \neq \ottmv{l} $ with $ \cdot $.
    We abuse the notation for run-time context patterns as intended.
\end{definition}

\begin{definition}[Order-defined run-time environment]
    We call $\ottnt{C}$ \emph{order-defined} if and only if a topological ordering of $ \lBrack  \langle \ottnt{C} \rangle_{ \ottmv{l} }  \rBrack $ is unique for any $\ottmv{l}$.
\end{definition}

\begin{definition}[Usage projection]
    Suppose a topological ordering of $ \lBrack \ottnt{C} \rBrack $ is unique.
    We write $ \overline{ \ottnt{C} } $ to express $\ottnt{m_{{\mathrm{1}}}} \odot \dots \odot \ottnt{m_{\ottmv{k}}}$ where $\ottmv{l_{{\mathrm{1}}}}  \mathord:  \ottsym{[}  \ottnt{m_{{\mathrm{1}}}}  \ottsym{]}, \dots, \ottmv{l_{\ottmv{k}}}  \mathord:  \ottsym{[}  \ottnt{m_{\ottmv{k}}}  \ottsym{]}$ is the unique topological ordering of $ \lBrack \ottnt{C} \rBrack $.
    We typically use the projection for $ \langle \ottnt{C} \rangle_{ \ottmv{l} } $.
    In such a case, $\ottmv{l_{{\mathrm{1}}}} = \dots = \ottmv{l_{\ottmv{k}}} = \ottmv{l}$.
\end{definition}

\begin{definition}[Heap typing]
    We say $\ottnt{C}$ \emph{types} $\mathcal{H}$, denoted by $\ottnt{C}  \vdash  \mathcal{H}$, if and only if
    \begin{itemize}
        \item $\ottnt{C}$ is order-defined,
        \item $  \ottkw{dom} ( \ottnt{C} )  \ottsym{=}  \ottkw{dom} ( \mathcal{H} )  $, and
        \item for any $ \ottmv{l}  \in   \ottkw{dom} ( \mathcal{H} )  $ and $ \ottsym{(}  \ottnt{n}  \ottsym{,}  \ottnt{m_{{\mathrm{0}}}}  \ottsym{,}  \ottnt{m}  \ottsym{)} \ottsym{=} \mathcal{H}  \ottsym{(}  \ottmv{l}  \ottsym{)} $:
              \begin{enumerate}
                  \item $ |  \lBrack  \langle \ottnt{C} \rangle_{ \ottmv{l} }  \rBrack  |_{\bullet}  = n + 1$,
                  \item $  \ottnt{m} \odot  \overline{  \langle \ottnt{C} \rangle_{ \ottmv{l} }  }   \le \ottnt{m_{{\mathrm{0}}}} $.
              \end{enumerate}
    \end{itemize}
\end{definition}

The first condition of heap typing says there is no ambiguity in the usage order specified by $\ottnt{C}$.
The second condition demands strong agreement between the domain of $\ottnt{C}$ and $\mathcal{H}$, by which the type system guarantees all resources are disposed of properly.
The third condition maintains the consistency of reference countings and resource usage; that is, the number of bindings for $\ottmv{l}$ in $\ottnt{C}$ coincides with the corresponding reference count, and the multiplication of performed operations so far and remaining usage is defined and respected the original usage.

Now, we sketch the proof of type soundness by introducing important lemmas.


For the pure lambda calculus parts, we need a substitution lemma as the usual proof strategy.

\begin{lemma}[Value substitution]
    If $\mathcal{G}  \ottsym{[}  \ottmv{x}  \mathord:  \ottnt{S}  \ottsym{]}  \vdash  \ottnt{M}  :  \ottnt{T}  \mid  \ottnt{e_{{\mathrm{1}}}}$, $ \ottmv{x}  \notin   \ottkw{dom} ( \mathcal{G}  \ottsym{[}   \cdot   \ottsym{]} )  $, and $\ottnt{C}  \vdash  \ottnt{V}  :  \ottnt{S}  \mid  \ottnt{e_{{\mathrm{2}}}}$, then $\mathcal{G}  \ottsym{[}  \ottnt{C}  \ottsym{]}  \vdash  \ottnt{M}  \ottsym{[}  \ottnt{V}  \ottsym{/}  \ottmv{x}  \ottsym{]}  :  \ottnt{T}  \mid   \ottnt{e_{{\mathrm{1}}}}  \sqcup  \ottnt{e_{{\mathrm{2}}}} $.

\begin{FULLVERSION}
    \proof See \propref{typing/subst}.
\end{FULLVERSION}
\end{lemma}


Using the substitution lemma, we can show the preservation for expression reduction as follows.

\begin{lemma}[Preservation for $ \longrightarrow_\beta $]
    If $\ottnt{M}  \longrightarrow_\beta  \ottnt{M'}$ and $\ottnt{C}  \vdash  \ottnt{M}  :  \ottnt{T}  \mid  \ottnt{e}$, then $\ottnt{C}  \vdash  \ottnt{M'}  :  \ottnt{T}  \mid  \ottnt{e}$.

\begin{FULLVERSION}
    \proof See \propref{subj/beta}.
\end{FULLVERSION}
\end{lemma}




The remaining part of the preservation is configuration reduction.
We prove each configuration reduction rule preserves the correctness as follows.
Each item tracks the change of a run-time context and a heap for each reduction rule.

\begin{lemma}
    \noindent
    \begin{enumerate}
        \item If $\mathcal{C}  \ottsym{[}   \cdot   \ottsym{]}  \vdash  \mathcal{H}$ and $ \ottmv{l}  \notin   \ottkw{dom} ( \mathcal{H} )  $, then $\mathcal{C}  \ottsym{[}  \ottmv{l}  \mathord:  \ottsym{[}  \ottnt{m}  \ottsym{]}  \ottsym{]}  \vdash   \mathcal{H} \cup \ottsym{\{} \ottmv{l}  \mapsto  \ottsym{(}  \ottsym{0}  \ottsym{,}  \ottnt{m}  \ottsym{,}  \varepsilon  \ottsym{)} \ottsym{\}} $.
        \item If $\mathcal{C}  \ottsym{[}  \ottmv{l}  \mathord:  \ottsym{[}  \ottnt{m'}  \ottsym{]}  \ottsym{]}  \vdash   \mathcal{H} \cup \ottsym{\{} \ottmv{l}  \mapsto  \ottsym{(}  \ottnt{n}  \ottsym{,}  \ottnt{m_{{\mathrm{0}}}}  \ottsym{,}  \ottnt{m}  \ottsym{)} \ottsym{\}} $ and $  \ottnt{m_{{\mathrm{1}}}} \odot \ottnt{m_{{\mathrm{2}}}}  \le \ottnt{m'} $,\\ then $\mathcal{C}  \ottsym{[}  \ottmv{l}  \mathord:  \ottsym{[}  \ottnt{m_{{\mathrm{1}}}}  \ottsym{]}  \ottsym{,}  \ottmv{l}  \mathord:  \ottsym{[}  \ottnt{m_{{\mathrm{2}}}}  \ottsym{]}  \ottsym{]}  \vdash   \mathcal{H} \cup \ottsym{\{} \ottmv{l}  \mapsto  \ottsym{(}  \ottnt{n}  \ottsym{+}  \ottsym{1}  \ottsym{,}  \ottnt{m_{{\mathrm{0}}}}  \ottsym{,}  \ottnt{m}  \ottsym{)} \ottsym{\}} $.
        \item If $\ottmv{l}  \mathord:  \ottsym{[}  \ottnt{m'}  \ottsym{]}  \ottsym{,}  \ottnt{C}  \vdash   \mathcal{H} \cup \ottsym{\{} \ottmv{l}  \mapsto  \ottsym{(}  \ottnt{n}  \ottsym{,}  \ottnt{m_{{\mathrm{0}}}}  \ottsym{,}  \ottnt{m}  \ottsym{)} \ottsym{\}} $ and $  \ottnt{m_{{\mathrm{1}}}} \odot \ottnt{m_{{\mathrm{2}}}}  \le \ottnt{m'} $,\\ then $\ottmv{l}  \mathord:  \ottsym{[}  \ottnt{m_{{\mathrm{2}}}}  \ottsym{]}  \ottsym{,}  \ottnt{C}  \vdash   \mathcal{H} \cup \ottsym{\{} \ottmv{l}  \mapsto  \ottsym{(}  \ottnt{n}  \ottsym{,}  \ottnt{m_{{\mathrm{0}}}}  \ottsym{,}   \ottnt{m} \odot \ottnt{m_{{\mathrm{1}}}}   \ottsym{)} \ottsym{\}} $.
        \item If $\mathcal{C}  \ottsym{[}  \ottmv{l}  \mathord:  \ottsym{[}  \ottnt{m'}  \ottsym{]}  \ottsym{]}  \vdash   \mathcal{H} \cup \ottsym{\{} \ottmv{l}  \mapsto  \ottsym{(}  \ottnt{n}  \ottsym{+}  \ottsym{1}  \ottsym{,}  \ottnt{m_{{\mathrm{0}}}}  \ottsym{,}  \ottnt{m}  \ottsym{)} \ottsym{\}} $ and $ \varepsilon \le \ottnt{m'} $, then $\mathcal{C}  \ottsym{[}   \cdot   \ottsym{]}  \vdash   \mathcal{H} \cup \ottsym{\{} \ottmv{l}  \mapsto  \ottsym{(}  \ottnt{n}  \ottsym{,}  \ottnt{m_{{\mathrm{0}}}}  \ottsym{,}  \ottnt{m}  \ottsym{)} \ottsym{\}} $.
        \item If $\mathcal{C}  \ottsym{[}  \ottmv{l}  \mathord:  \ottsym{[}  \ottnt{m'}  \ottsym{]}  \ottsym{]}  \vdash   \mathcal{H} \cup \ottsym{\{} \ottmv{l}  \mapsto  \ottsym{(}  \ottsym{0}  \ottsym{,}  \ottnt{m_{{\mathrm{0}}}}  \ottsym{,}  \ottnt{m}  \ottsym{)} \ottsym{\}} $ and $ \ottmv{l}  \notin   \ottkw{dom} ( \mathcal{H} )  $, then $\mathcal{C}  \ottsym{[}   \cdot   \ottsym{]}  \vdash  \mathcal{H}$.
    \end{enumerate}
\begin{FULLVERSION}
    \proof See \propref{heap/track/new}, \propref{heap/track/split}, \propref{heap/track/op}, \propref{heap/track/cl1}, and \propref{heap/track/cl2}.
\end{FULLVERSION}
\end{lemma}

To show the preservation for the configuration reduction, we show the preservation for effect-free expressions first, and then we show the preservation for arbitrary expressions.
We use the first lemma to show the second lemma.

\begin{lemma}[Preservation for $ \longrightarrow_\gamma $ with effect-free expressions]
    If $\ottnt{M}  \mid  \mathcal{H}  \longrightarrow_\gamma  \ottnt{M'}  \mid  \mathcal{H}'$, $\ottnt{C}  \vdash  \ottnt{M}  :  \ottnt{T}  \mid  \ottsym{0}$, and $\mathcal{C}  \ottsym{[}  \ottnt{C}  \ottsym{]}  \vdash  \mathcal{H}$, then there exists $\ottnt{C'}$ such that $  \ottsym{(}    \ottkw{dom} ( \ottnt{C'} )  \setminus  \ottkw{dom} ( \ottnt{C} )    \ottsym{)} \cap  \ottkw{dom} ( \mathcal{C}  \ottsym{[}   \cdot   \ottsym{]} )   \ottsym{=} \emptyset $, $\ottnt{C'}  \vdash  \ottnt{M'}  :  \ottnt{T}  \mid  \ottsym{0}$, and $\mathcal{C}  \ottsym{[}  \ottnt{C'}  \ottsym{]}  \vdash  \mathcal{H}'$.

\begin{FULLVERSION}
    \proof See \propref{subj/cmp/noeffect}.
\end{FULLVERSION}
\end{lemma}

\begin{lemma}[Preservation for $ \longrightarrow_\gamma $]
    If $\ottnt{M}  \mid  \mathcal{H}  \longrightarrow_\gamma  \ottnt{M'}  \mid  \mathcal{H}'$, $\ottnt{C}  \vdash  \ottnt{M}  :  \ottnt{T}  \mid  \ottnt{e}$, and $\ottnt{C}  \ottsym{,}  \ottnt{C''}  \vdash  \mathcal{H}$, then there exists $\ottnt{C'}$ such that $  \ottsym{(}    \ottkw{dom} ( \ottnt{C'} )  \setminus  \ottkw{dom} ( \ottnt{C} )    \ottsym{)} \cap  \ottkw{dom} ( \ottnt{C''} )   \ottsym{=} \emptyset $, $\ottnt{C'}  \vdash  \ottnt{M'}  :  \ottnt{T}  \mid  \ottnt{e}$, and $\ottnt{C'}  \ottsym{,}  \ottnt{C''}  \vdash  \mathcal{H}'$.
\begin{FULLVERSION}
    \proof See \propref{subj/cmp/effect}.
\end{FULLVERSION}
\end{lemma}

Summing up the preservation lemmas for expression reduction and configuration reduction, we have the preservation lemma.

\begin{lemma}[Preservation]
    If $\ottnt{M}  \mid  \mathcal{H}  \longrightarrow  \ottnt{M'}  \mid  \mathcal{H}'$, $\ottnt{C}  \vdash  \ottnt{M}  :  \ottnt{T}  \mid  \ottnt{e}$, and $\ottnt{C}  \vdash  \mathcal{H}$, then there exists $\ottnt{C'}$ such that $\ottnt{C'}  \vdash  \ottnt{M'}  :  \ottnt{T}  \mid  \ottnt{e}$ and $\ottnt{C'}  \vdash  \mathcal{H}'$.
\begin{FULLVERSION}
    \proof See \propref{subj}.
\end{FULLVERSION}
\end{lemma}

Progress can be shown rather straightforwardly.
We need to show the progress for effect-free expressions first and then the objective lemma, too.

\begin{lemma}[Progress of effect-free expressions]
    If $\ottnt{C}  \vdash  \ottnt{M}  :  \ottnt{T}  \mid  \ottsym{0}$ and $\mathcal{C}  \ottsym{[}  \ottnt{C}  \ottsym{]}  \vdash  \mathcal{H}$, then
    \begin{enumerate}
        \item $\ottnt{M}$ is a value, namely $ \ottnt{M} \ottsym{=} \ottnt{V} $ for some $\ottnt{V}$, or
        \item for any $\mathcal{H}$, $\ottnt{M}  \mid  \mathcal{H}  \longrightarrow  \ottnt{M'}  \mid  \mathcal{H}'$ for some $\ottnt{M'}$ and $\mathcal{H}'$.
    \end{enumerate}
\begin{FULLVERSION}
    \proof See \propref{progress/noeffect}.
\end{FULLVERSION}
\end{lemma}

\begin{lemma}[Progress]
    If $\ottnt{C}  \vdash  \ottnt{M}  :  \ottnt{T}  \mid  \ottnt{e}$ and $\ottnt{C}  \ottsym{,}  \ottnt{C'}  \vdash  \mathcal{H}$, then
    \begin{enumerate}
        \item $\ottnt{M}$ is a value, namely $ \ottnt{M} \ottsym{=} \ottnt{V} $ for some $\ottnt{V}$, or
        \item $\ottnt{M}  \mid  \mathcal{H}  \longrightarrow  \ottnt{M'}  \mid  \mathcal{H}'$ for some $\ottnt{M'}$ and $\mathcal{H}'$.
    \end{enumerate}
\begin{FULLVERSION}
    \proof See \propref{progress/effect}.
\end{FULLVERSION}
\end{lemma}

We compile the preservation and progress lemma as the following single proposition.

\begin{proposition}[Soundness]
    If $ \cdot   \vdash  \ottnt{M}  :  \ottnt{T}  \mid  \ottnt{e}$ and $\ottkw{unr} \, \ottnt{T}$, then
    \begin{enumerate}
        \item $\ottnt{M}  \mid  \ottsym{\{}  \ottsym{\}}  \longrightarrow^*  \ottnt{V}  \mid  \ottsym{\{}  \ottsym{\}}$, or
        \item the evaluation of $\ottnt{M} \mid \ottsym{\{}  \ottsym{\}}$ diverges.
    \end{enumerate}
\end{proposition}


\newcommand*\HSExprA{M}
\newcommand*\HSExprB{N}
\newcommand*\HSInfExprA{\HSExprA^{\uparrow}}
\newcommand*\HSInfExprB{\HSExprB^{\uparrow}}
\newcommand*\HSChkExprA{\HSExprA^{\downarrow}}
\newcommand*\HSChkExprB{\HSExprB^{\downarrow}}
\newcommand*\HSUnit{\texttt{unit}}
\newcommand*\HSNew[1]{\mathop{\texttt{new}_{#1}}}
\newcommand*\HSOp[1]{\mathop{\texttt{op}_{#1}}}
\newcommand*\HSSplit[1]{\mathop{\texttt{split}_{#1}}}
\newcommand*\HSDrop{\mathop{\texttt{drop}}}
\newcommand*\HSAbs[1]{\mathop{\lambda{#1}.}}
\newcommand*\HSVarA{x}
\newcommand*\HSVarB{y}
\newcommand*\HSApp[2]{\mathop{#1}{#2}}
\newcommand*\HSPair[2]{({#1}\mathop{,}{#2})}
\newcommand*\HSLetPair[4]{\mathop{\texttt{let}}{#1},{#2}={#3}\mathop{\texttt{in}}{#4}}
\newcommand*\HSAnn[2]{{#1}:{#2}}
\newcommand*\HSEmbed[1]{#1}
\newcommand*\HSTypeA{T}
\newcommand*\HSTypeB{S}

\newcommand*\HSCtx{\Gamma}
\newcommand*\HSEff{e}
\newcommand*\HSEffNo{0}
\newcommand*\HSEffYes{1}
\newcommand*\HSCtxCtx{\mathcal{G}}
\newcommand*\HSTyping[6]{{#1} \vdash {#2} {#6} {#3} \mid {#4} \Rightarrow {#5}}
\newcommand*\HSInfTyping[5]{\HSTyping {#1} {#2} {#3} {#4} {#5} {\hspace{1mm}\uparrow\hspace{1mm}}}
\newcommand*\HSChkTyping[5]{\HSTyping {#1} {#2} {#3} {#4} {#5} {\hspace{1mm}\downarrow\hspace{1mm}}}
\newcommand*\HSLogTyping[4]{{#1} \vdash {#2} : {#3} \mid {#4}}
\newcommand*\HSRestrCtx[2]{{#1}\mathop{\downarrow_{#2}}}
\newcommand*\HSSplitCtx[2]{{#1}\mathop{\Downarrow_{#2}}}
\newcommand*\HSFv[1]{\texttt{fv}({#1})}
\newcommand*\HSSubCtx{\lesssim}
\newcommand*\HSIsoCtx{\simeq}
\newcommand*\HSCtxEmpty{\cdot}
\newcommand*\HSPar{\mathop{\|}}
\newcommand*\HSSeq{\mathop{,}}
\newcommand*\HSLub{\mathop{\sqcup}}
\newcommand*\HSDom[1]{\texttt{dom}({#1})}
\newcommand*\HSDomOrd[1]{\texttt{dom}^{\texttt{ord}}({#1})}
\newcommand*\HSUnr[1]{\mathop{\texttt{unr}}{#1}}
\newcommand*\HSOrd[1]{\mathop{\texttt{ord}}{#1}}
\newcommand*\HSSubEff{\leq}
\newcommand*\HSMonOp{\odot}
\newcommand*\HSMonLeq{\leq}
\newcommand*\HSMonNeu{\varepsilon}
\newcommand*\HSCtxCtxHole{[]}
\newcommand*\HSRightCtxCtx{\mathop{\texttt{right}}}
\newcommand*\HSLeftCtxCtx{\mathop{\texttt{left}}}
\newcommand*\HSParCtxCtx{\mathop{\texttt{par}}}
\newcommand*\HSClosedCtxCtx{\mathop{\texttt{closed}}}


\newcommand*\HSRuleAbsUnr{
  \inferrule[AT-Abs]{
    \HSUnr \HSCtx \\
    \HSVarA \not\in \HSDom{\HSCtx} \\
    \HSChkTyping {\HSCtx \HSSeq \HSVarA : \ottnt{S}} {\HSChkExprA} {\ottnt{T}} {\HSEff_1} {\ottnt{M}} \\
    \HSEff_1 \HSSubEff \HSEff_2
  }{
    \HSChkTyping {\HSCtx} {(\HSAbs \HSVarA \HSChkExprA)} {( \ottnt{S} \rightarrow _{ \ottnt{e_{{\mathrm{2}}}} } \ottnt{T} )}
                 {\HSEffNo} {(\lambda  \ottmv{x}  \ottsym{.}  \ottnt{M})}
  }
}

\newcommand*\HSRuleAbsLin{
  \inferrule[AT-UAbs]{
    \HSVarA \not\in \HSDom{\HSCtx} \\
    \HSChkTyping {\HSCtx \HSPar \HSVarA : \ottnt{S}} {\HSChkExprA} {\ottnt{T}} {\HSEff_1} {\ottnt{M}} \\
    \HSEff_1 \HSSubEff \HSEff_2
  }{
    \HSChkTyping {\HSCtx} {(\HSAbs \HSVarA \HSChkExprA)} {( \ottnt{S} \rightarrow _{ \ottnt{e_{{\mathrm{2}}}} } \ottnt{T} )}
                 {\HSEffNo} {(\lambda^\circ  \ottmv{x}  \ottsym{.}  \ottnt{M})}
  }
}

\newcommand*\HSRuleAbsLeft{
  \inferrule[AT-RAbs]{
    \HSVarA \not\in \HSDom{\HSCtx} \\
    \HSChkTyping {\HSCtx \HSSeq \HSVarA : \ottnt{S}} {\HSChkExprA} {\ottnt{T}} {\HSEff_1} {\ottnt{M}} \\
    \HSEff_1 \HSSubEff \HSEff_2
  }{
    \HSChkTyping {\HSCtx} {(\HSAbs \HSVarA \HSChkExprA)} {( \ottnt{S} \rightarrowtail _{ \ottnt{e_{{\mathrm{2}}}} } \ottnt{T} )}
                 {\HSEffNo} {(\lambda^>  \ottmv{x}  \ottsym{.}  \ottnt{M})}
  }
}

\newcommand*\HSRuleAbsRight{
  \inferrule[AT-LAbs]{
    \HSVarA \not\in \HSDom{\HSCtx} \\
    \HSChkTyping {\HSVarA : \ottnt{S} \HSSeq \HSCtx} {\HSChkExprA} {\ottnt{T}} {\HSEff_1} {\ottnt{M}} \\
    \HSEff_1 \HSSubEff \HSEff_2
  }{
    \HSChkTyping {\HSCtx} {(\HSAbs \HSVarA \HSChkExprA)} {( \ottnt{S} \twoheadrightarrow _{ \ottnt{e_{{\mathrm{2}}}} } \ottnt{T} )}
                 {\HSEffNo} {(\lambda^<  \ottmv{x}  \ottsym{.}  \ottnt{M})}
  }
}

\newcommand*\HSRuleChkInf{
  \inferrule[AT-ChkInf]{
    \HSInfTyping {\HSCtx} {\HSInfExprA} {\ottnt{S}} {\HSEff} {\ottnt{M}} \\
    \ottnt{S} = \ottnt{T}
  }{
    \HSChkTyping {\HSCtx} {\HSInfExprA} {\ottnt{T}} {\HSEff} {\ottnt{M}}
  }
}


\newcommand*\HSRuleUnit{
  \inferrule[AT-Unit]{
    \HSCtx \HSSubCtx \HSCtxEmpty
  }{
    \HSInfTyping {\HSCtx} {\HSUnit} { \mathtt{Unit} } {\HSEffNo} {\ottkw{unit}}
  }
}

\newcommand*\HSRuleNew{
  \inferrule[AT-New]{
    \HSCtx \HSSubCtx \HSCtxEmpty
  }{
    \HSInfTyping {\HSCtx} {\HSNew m} {\ottsym{[}  \ottnt{m}  \ottsym{]}} {\HSEffNo} { \ottkw{new} _{ \ottnt{m} } }
  }
}

\newcommand*\HSRuleOp{
  \inferrule[AT-Op]{
    \HSInfTyping {\HSCtx} {\HSInfExprA} {\ottsym{[}  \ottnt{m}  \ottsym{]}} {\HSEff} {\ottnt{M}} \\
    m_1 \HSMonOp m_2 \HSMonLeq m
  }{
    \HSInfTyping {\HSCtx} {\HSOp {m_1} \HSInfExprA} {\ottsym{[}  \ottnt{m_{{\mathrm{2}}}}  \ottsym{]}} {\HSEff} {\ottsym{(}   \ottkw{op} _{ \ottnt{m_{{\mathrm{1}}}} }   \ottsym{)} \, \ottnt{M}}
  }
}

\newcommand*\HSRuleSplit{
  \inferrule[AT-Split]{
    \HSInfTyping {\HSCtx} {\HSInfExprA} {\ottsym{[}  \ottnt{m}  \ottsym{]}} {\HSEff} {\ottnt{M}} \\
    m_1 \HSMonOp m_2 \HSMonLeq m
  }{
    \HSInfTyping {\HSCtx} {\HSSplit {m_1} \HSInfExprA} {\ottsym{[}  \ottnt{m_{{\mathrm{1}}}}  \ottsym{]}  \odot  \ottsym{[}  \ottnt{m_{{\mathrm{2}}}}  \ottsym{]}} {\HSEff} {\ottsym{(}   \ottkw{split} _{ \ottnt{m_{{\mathrm{1}}}} , \ottnt{m_{{\mathrm{2}}}} }   \ottsym{)} \, \ottnt{M}}
  }
}

\newcommand*\HSRuleDrop{
  \inferrule[AT-Drop]{
    \HSInfTyping {\HSCtx} {\HSInfExprA} {\ottsym{[}  \ottnt{m}  \ottsym{]}} {\HSEff} {\ottnt{M}} \\
    \HSMonNeu \HSMonLeq m
  }{
    \HSInfTyping {\HSCtx} {\HSDrop \HSInfExprA} { \mathtt{Unit} } {\HSEff} {\ottsym{(}  \ottkw{drop}  \ottsym{)} \, \ottnt{M}}
  }
}

\newcommand*\HSRuleVar{
  \inferrule[AT-Var]{
    \HSCtx \HSSubCtx \HSVarA : \ottnt{T}
  }{
    \HSInfTyping {\HSCtx} {\HSVarA} {\ottnt{T}} {\HSEffNo} {\ottmv{x}}
  }
}

\newcommand*\HSRuleAppUnr{
  \inferrule[AT-App]{
    \HSCtx_1 = \HSRestrCtx{\HSCtx}{\HSFv{\HSInfExprA}} \and
    \HSCtx_2 = \HSRestrCtx{\HSCtx}{\HSFv{\HSInfExprB}} \and
    \HSCtx \HSSubCtx \HSCtx_1 \HSSeq \HSCtx_2 \\
    \HSInfTyping {\HSCtx_1} {\HSInfExprA} {( \ottnt{S} \rightarrow _{ \ottnt{e} } \ottnt{T} )} {\HSEff_1} {\ottnt{M}} \\
    \HSInfTyping {\HSCtx_2} {\HSInfExprB} {\ottnt{S}} {\HSEff_2} {\ottnt{N}}
  }{
    \HSInfTyping {\HSCtx} {(\HSApp \HSInfExprA \HSInfExprB)} {\ottnt{T}}
                 {(\HSEff \HSLub \HSEff_1 \HSLub \HSEff_2)} {(\ottnt{M} \, \ottnt{N})}
  }
}
\newcommand*\HSRuleAppLin{
  \inferrule[AT-UApp]{
    \HSCtx_1 = \HSRestrCtx{\HSCtx}{\HSFv{\HSInfExprA}} \and
    \HSCtx_2 = \HSRestrCtx{\HSCtx}{\HSFv{\HSInfExprB}} \and
    \HSCtx \HSSubCtx \HSCtx_1 \HSPar \HSCtx_2 \\
    \HSInfTyping {\HSCtx_1} {\HSInfExprA} {( \ottnt{S} \rightarrowtriangle _{ \ottnt{e} } \ottnt{T} )} {\HSEff_1} {\ottnt{M}} \\
    \HSInfTyping {\HSCtx_2} {\HSInfExprB} {\ottnt{S}} {\HSEff_2} {\ottnt{N}}
  }{
    \HSInfTyping {\HSCtx} {(\HSApp \HSInfExprA \HSInfExprB)} {\ottnt{T}}
                 {(\HSEff \HSLub \HSEff_1 \HSLub \HSEff_2)} {(\ottnt{M}  {}^\circ  \ottnt{N})}
  }
}
\newcommand*\HSRuleAppLeft{
  \inferrule[AT-RApp]{
    \HSCtx_1 = \HSRestrCtx{\HSCtx}{\HSFv{\HSInfExprA}} \and
    \HSCtx_2 = \HSRestrCtx{\HSCtx}{\HSFv{\HSInfExprB}} \and
    \HSCtx \HSSubCtx \HSCtx_1 \HSSeq \HSCtx_2 \\
    \HSInfTyping {\HSCtx_1} {\HSInfExprA} {( \ottnt{S} \rightarrowtail _{ \ottnt{e} } \ottnt{T} )} {\HSEff_1} {\ottnt{M}} \\
    \HSInfTyping {\HSCtx_2} {\HSInfExprB} {\ottnt{S}} {\HSEffNo} {\ottnt{N}}
  }{
    \HSInfTyping {\HSCtx} {(\HSApp \HSInfExprA \HSInfExprB)} {\ottnt{T}}
                 {(\HSEff \HSLub \HSEff_1)} {(\ottnt{M}  {}^>  \ottnt{N})}
  }
}
\newcommand*\HSRuleAppRight{
  \inferrule[AT-LApp]{
    \HSCtx_1 = \HSRestrCtx{\HSCtx}{\HSFv{\HSInfExprA}} \and
    \HSCtx_2 = \HSRestrCtx{\HSCtx}{\HSFv{\HSInfExprB}} \and
    \HSCtx \HSSubCtx \HSCtx_2 \HSSeq \HSCtx_1 \\
    \HSInfTyping {\HSCtx_1} {\HSInfExprA} {( \ottnt{S} \twoheadrightarrow _{ \ottnt{e} } \ottnt{T} )} {\HSEffNo} {\ottnt{M}} \\
    \HSInfTyping {\HSCtx_2} {\HSInfExprB} {\ottnt{S}} {\HSEff_2} {\ottnt{N}}
  }{
    \HSInfTyping {\HSCtx} {(\HSApp \HSInfExprA \HSInfExprB)} {\ottnt{T}}
                 {(\HSEff \HSLub \HSEff_2)} {(\ottnt{M}  {}^<  \ottnt{N})}
  }
}

\newcommand*\HSRulePairLin{
  \inferrule[AT-UPair]{
    \HSCtx_1 = \HSRestrCtx{\HSCtx}{\HSFv{\HSInfExprA}} \and
    \HSCtx_2 = \HSRestrCtx{\HSCtx}{\HSFv{\HSInfExprB}} \and
    \HSCtx \HSSubCtx \HSCtx_1 \HSPar \HSCtx_2 \\\\
    \HSInfTyping {\HSCtx_1} {\HSInfExprA} {\ottnt{S}} {\HSEff_1} {\ottnt{M}} \and
    \HSInfTyping {\HSCtx_2} {\HSInfExprB} {\ottnt{T}} {\HSEff_2} {\ottnt{N}}
  }{
    \HSInfTyping {\HSCtx} {\HSPair \HSInfExprA \HSInfExprB} {(\ottnt{S}  \otimes  \ottnt{T})}
                 {(\HSEff_1 \HSLub \HSEff_2)} {(\ottnt{M}  \otimes  \ottnt{N})}
  }
}
\newcommand*\HSRulePairLeft{
  \inferrule[AT-OPair]{
    \HSCtx_1 = \HSRestrCtx{\HSCtx}{\HSFv{\HSInfExprA}} \and
    \HSCtx_2 = \HSRestrCtx{\HSCtx}{\HSFv{\HSInfExprB}} \and
    \HSCtx \HSSubCtx \HSCtx_1 \HSSeq \HSCtx_2 \\\\
    \HSCtx \not\HSSubCtx \HSCtx_1 \HSPar \HSCtx_2 \and
    \HSOrd {\ottnt{S}} \texttt{ implies } \HSEff_2 = \HSEffNo \\\\
    \HSInfTyping {\HSCtx_1} {\HSInfExprA} {\ottnt{S}} {\HSEff_1} {\ottnt{M}} \and
    \HSInfTyping {\HSCtx_2} {\HSInfExprB} {\ottnt{T}} {\HSEff_2} {\ottnt{N}}
  }{
    \HSInfTyping {\HSCtx} {\HSPair \HSInfExprA \HSInfExprB} {(\ottnt{S}  \odot  \ottnt{T})}
                 {(\HSEff_1 \HSLub \HSEff_2)} {(\ottnt{M}  \odot  \ottnt{N})}
  }
}

\newcommand*\HSRuleLetPairLin{
  \inferrule[AT-ULet]{
    \HSSplitCtx{\HSCtx}{\HSFv{\HSInfExprA}} = (\HSCtxCtx, \HSCtx') \\
    \ottsym{\{}  \ottmv{x}  \ottsym{\}}  \uplus  \ottsym{\{}  \ottmv{y}  \ottsym{\}}  \uplus   \ottkw{dom} ( \mathcal{G}  \ottsym{[}  \Gamma'  \ottsym{]} )  \\\\
    \HSInfTyping {\HSCtx'} {\HSInfExprA} {(\ottnt{S_{{\mathrm{1}}}}  \otimes  \ottnt{S_{{\mathrm{2}}}})} {\HSEffNo} {\ottnt{M}} \and
    \HSInfTyping {\HSCtxCtx[\HSVarA : S_1 \HSPar \HSVarB : S_2]} {\HSInfExprB} {\ottnt{T}} {\HSEff} {\ottnt{N}}
  }{
    \HSInfTyping {\HSCtx} {\HSLetPair \HSVarA \HSVarB \HSInfExprA \HSInfExprB} {\ottnt{T}}
                 {\HSEff} {(\ottkw{let} \, \ottmv{x}  \otimes  \ottmv{y}  \ottsym{=}  \ottnt{M} \, \ottkw{in} \, \ottnt{N})}
  }
}
\newcommand*\HSRuleLetPairLeft{
  \inferrule[AT-OLet]{
    \HSSplitCtx{\HSCtx}{\HSFv{\HSInfExprA}} = (\HSCtxCtx, \HSCtx') \\
    \ottsym{\{}  \ottmv{x}  \ottsym{\}}  \uplus  \ottsym{\{}  \ottmv{y}  \ottsym{\}}  \uplus   \ottkw{dom} ( \mathcal{G}  \ottsym{[}  \Gamma'  \ottsym{]} )  \\\\
    \HSInfTyping {\HSCtx'} {\HSInfExprA} {(\ottnt{S_{{\mathrm{1}}}}  \odot  \ottnt{S_{{\mathrm{2}}}})} {\HSEffNo} {\ottnt{M}} \and
    \HSInfTyping {\HSCtxCtx[\HSVarA : S_1 \HSSeq \HSVarB : S_2]} {\HSInfExprB} {\ottnt{T}} {\HSEff} {\ottnt{N}}
  }{
    \HSInfTyping {\HSCtx} {\HSLetPair \HSVarA \HSVarB \HSInfExprA \HSInfExprB} {\ottnt{T}}
                 {\HSEff} {(\ottkw{let} \, \ottmv{x}  \odot  \ottmv{y}  \ottsym{=}  \ottnt{M} \, \ottkw{in} \, \ottnt{N})}
  }
}

\newcommand*\HSRuleAnn{
  \inferrule[AT-Ann]{
    \HSChkTyping {\HSCtx} {\HSChkExprA} {\ottnt{T}} {\HSEff} {\ottnt{M}}
  }{
    \HSInfTyping {\HSCtx} {(\HSAnn \HSChkExprA {\ottnt{T}})} {\ottnt{T}} {\HSEff} {\ottnt{M}}
  }
}


\section{Algorithmic Typing and Implementation}
\label{sec:algorithmic-typing}

The logical typing relation from
\Cref{sec:language:type-system:typing-rules} has many rules
that are not deterministic. For example, the application rule
\textsc{T-UApp} requires the typing context to be split into
two subcontexts $\HSCtx_1\HSPar\HSCtx_2$.
To implement a type checker, an efficient way to compute these
splittings is required.
Furthermore, the need to syntactically distinguish between terms with
different ordering constraints poses a burden on the programmer and
introduces syntactic noise, e.g., either $\ottnt{M} \, \ottnt{N}$ or $\ottnt{M}  {}^\circ  \ottnt{N}$ has
to be used for function application depending on whether an
unrestricted or a linear function is applied.

In this section, we present an algorithmic typing, which splits
contexts deterministically and does not require annotating terms with ordering constraints.
The algorithmic typing is formulated using bidirectional
typing~\cite{DBLP:conf/popl/PierceT98} to minimize the number of required type
annotations, and translates terms from a surface syntax to
terms in the internal syntax from \Cref{fig:syntax}.

\begin{figure}
  \centering
  \setlength{\extrarowheight}{\smallskipamount}
  \begin{tabulary}{\linewidth}{@{}r@{\enspace}c@{\enspace}J<{\hfill\null}@{}}
    $\HSInfExprA, \HSInfExprB$ & $\Coloneqq$ & $
      \HSUnit \mid
      \HSNew m \mid
      \HSOp m \HSInfExprA \mid
      \HSSplit m \HSInfExprA \mid
      \HSDrop \HSInfExprA 
                                               $\\
                               & $\mid$ & $
      \HSVarA \mid
      \HSApp \HSInfExprA \HSInfExprB \mid
      \HSPair \HSInfExprA \HSInfExprB \mid
      \HSLetPair \HSVarA \HSVarB \HSInfExprA \HSInfExprB \mid
      \HSAnn \HSChkExprA \HSTypeA                $\\
    $\HSChkExprA, \HSChkExprB$ & $\Coloneqq$ & $
      \HSAbs x \HSChkExprA \mid
      \HSEmbed \HSInfExprA$
  \end{tabulary}
  \caption{Surface Syntax}
  \label{fig:surface-syntax}
\end{figure}

The surface syntax is shown in \Cref{fig:surface-syntax}.
As customary with bidirectional typing,
we distinguish between \emph{inferable} expressions $\HSInfExprA$ and
\emph{checkable} expressions $\HSChkExprA$.
An inferable expression can also be checked, but a checkable
expression $\HSChkExprA$ can only be inferred by wrapping it in a type
annotation $\HSAnn \HSChkExprA \HSTypeA$.
Lambda abstractions are the only non-inferable terms, thus they require a type annotation.
Function constants related to resource management are bundled with their arguments, e.g.,
instead of the constant $ \ottkw{op} _{ \ottnt{m} } $, there is the expression $\HSOp m \HSInfExprA$.
This bundling enables inferring the types of resource operations without introducing resource polymorphism.
The introduction and elimination forms of function and product types
are collapsed into unannotated terms, e.g., the surface expression $\HSApp\HSInfExprA\HSInfExprB$ represents
all four internal expressions $\ottnt{M} \, \ottnt{N}$, $\ottnt{M}  {}^\circ  \ottnt{N}$, $\ottnt{M}  {}^>  \ottnt{N}$, and $\ottnt{M}  {}^<  \ottnt{N}$.

\begin{figure}
  \centering
  \begin{mathpar}
    \HSRulePairLin \and
    \HSRuleVar \and
    \HSRulePairLeft \and
    \HSRuleUnit \and
    \HSRuleLetPairLeft \and
  \end{mathpar}
  \caption{Algorithmic Typing (Selected Rules)}
  \label{fig:algorithmic-typing-selected}
\end{figure}

\Cref{fig:algorithmic-typing-selected} shows selected rules from the algorithmic typing.
Having decidability in mind, we think of the checking relation
$\HSChkTyping\HSCtx\HSChkExprA{\ottnt{T}}\HSEff{\ottnt{M}}$ as taking $\HSCtx$, $\HSChkExprA$, and $\ottnt{T}$ as inputs
and producing $\HSEff$ and $\ottnt{M}$ as outputs, and of the inference relation
$\HSInfTyping\HSCtx\HSInfExprA{\ottnt{T}}\HSEff{\ottnt{M}}$ as taking $\HSCtx$ and $\HSInfExprA$ as inputs
and producing $\ottnt{T}$, $\HSEff$, and $\ottnt{M}$ as outputs.

The \textsc{AT-UPair} rule describes how to infer the type of an unordered pair.
Here the input context $\HSCtx$ needs to be rearranged into 
two parallel subcontexts, i.e.\ $\HSCtx\HSSubCtx\HSCtx_1\HSPar\HSCtx_2$,
which are used to infer the types of the subexpressions.
As each subexpression requires bindings for precisely its free variables, we can compute the subcontexts
by restricting $\HSCtx$ to the free variables of the respective subexpression.
This is achieved with the \emph{context restriction} function $\HSRestrCtx{\HSCtx}{X}$, which replaces
each variable binding $x : T$ in $\HSCtx$, where $x \not\in X$, with the empty context.
To check that $\HSCtx$ indeed decomposes to $\HSCtx_1\HSPar\HSCtx_2$, we require
$\HSCtx_1\HSSubCtx\HSCtx_1\HSPar\HSCtx_2$, which is efficiently decidable.

The \textsc{AT-OPair} rule describes how to infer the type of an ordered pair.
It is analogous to \textsc{AT-UPair}, but additionally requires $\HSCtx\not\HSSubCtx\HSCtx_1\HSPar\HSCtx_2$.
This way, typing is deterministic and a pair expression is only typed as an ordered pair
if it cannot be typed as an unordered pair.

The \textsc{AT-OLet} rule describes inference for elimination of ordered pairs.
Here the input context $\HSCtx$ needs to rearranged into a subcontext
and a context pattern,
i.e., $\HSCtx\HSSubCtx\HSCtxCtx[\HSCtx']$, where
$\HSCtx'$ and $\HSCtxCtx$ are used to infer the types of the first and
second subexpression, respectively.
Ideally, we would like to compute $\HSCtx'$ using context restriction
as shown for the \textsc{AT-UPair} rule, and compute $\HSCtxCtx$ by
replacing the occurence of $\HSCtx'$ in $\HSCtx$ with a hole.
However, $\HSCtx'$ does not need to occur in $\HSCtx$ as a whole,
but the bindings of $\HSCtx'$ could be interleaved with other bindings.
Consequently, the task at hand is to figure out if $\HSCtx$ can be rearranged
to some context $\overline{\HSCtx}$ with $\HSCtx\HSSubCtx\overline{\HSCtx}$
such that $\overline{\HSCtx}$ contains the bindings for the free variables of the first
subexpression as one subcontext.
This is achieved with the \emph{context decomposition} function
$\HSSplitCtx{\HSCtx}{X}$, which either fails or returns some $\HSCtxCtx$ and
$\HSCtx'$ such that $\HSCtx\HSSubCtx\HSCtxCtx[\HSCtx']$,
$\HSDom{\HSCtx'} \subseteq X$, and $X$ is disjoint from the variables
with ordered types in $\HSCtxCtx$.
Essentially, this function systematically applies context isomorphisms $\HSIsoCtx$
and the subcontext rule
$(\HSCtx_1 \HSSeq \HSCtx_2) \HSPar (\HSCtx_3 \HSSeq \HSCtx_4) \HSSubCtx
(\HSCtx_1 \HSPar \HSCtx_3) \HSSeq (\HSCtx_2 \HSPar \HSCtx_4)$
to separate bindings with variables in $X$ from other bindings, and ensures
that unrestricted bindings appear both in $\HSCtx'$ and $\HSCtxCtx$.

The \textsc{AT-Var} and \textsc{AT-Unit} rules describe how to infer
the types of variables and the \texttt{unit} constant.
These expressions are leaves in the syntax tree, so we need to ensure
that no restricted variables are dropped
by requiring
$\HSCtx\HSSubCtx\HSVarA:\ottnt{T}$ for variables $\HSVarA$, and
$\HSCtx\HSSubCtx\HSCtxEmpty$ for \texttt{unit}.


The rules discussed above are representative, and the omitted rules
work similarly.
The formal definition of context restriction, context decomposition,
and the omitted typing rules can be found in the supplementary
material and the Agda formalization.

\subsection{Properties}
We have proved decidability of algorithmic typing and soundness with
respect to the logical typing in the theorem prover
Agda~\cite{Agda}.

\begin{theorem}[Soundness]
  If $\HSInfTyping\HSCtx\HSInfExprA{\ottnt{T}}\HSEff{\ottnt{M}}$
  or $\HSChkTyping\HSCtx\HSChkExprA{\ottnt{T}}\HSEff{\ottnt{M}}$
  then $\HSLogTyping\HSCtx{\ottnt{M}}{\ottnt{T}}\HSEff$.
\end{theorem}
\begin{proof}
  Mutual rule induction on the typing derivations.
  A detailed proof can be found in our Agda formalization
  in the file \texttt{LawAndOrder/Algorithmic/TypingSound.agda}.
\end{proof}

Decidability of typing requires that certain operations of the OPM are decidable.
\begin{definition}[Decidable OPM]
  Let $\mathcal{M} = (M, \odot, \HSMonNeu, \HSMonLeq)$ be an OPM.
  We say $\mathcal{M}$ is \emph{decidable},
  iff for all $x,y \in M$ it is decidable whether
  $x = y$, $x \HSMonLeq y$, and $\exists z \in M.\ x \HSMonOp z \HSMonLeq y$.
\end{definition}

\begin{theorem}[Decidability]
  Let $\mathcal{M}$ be a decidable OPM, then
  \begin{enumerate}
  \item 
    for all $\HSCtx$ and $\HSInfExprA$ it is decidable if
    there exist $\ottnt{T}$, $\HSEff$, and $\ottnt{M}$ such that
    $\HSInfTyping\HSCtx\HSInfExprA{\ottnt{T}}\HSEff{\ottnt{M}}$.
  \item 
    for all $\HSCtx$, $\HSChkExprA$ and \ottnt{T} it is decidable if
    there exist $\HSEff$ and $\ottnt{M}$ such that
    $\HSChkTyping\HSCtx\HSChkExprA{\ottnt{T}}\HSEff{\ottnt{M}}$.
  \end{enumerate}
\end{theorem}
\begin{proof}
  Mutual rule induction on the typing derivations.
  A detailed proof can be found in our Agda formalization
  in the file \texttt{LawAndOrder/Algorithmic/TypingDecidable.agda}.
\end{proof}



\subsection{Implementation}
\label{sec:implementation}

We have implemented a type-checker and interpreter for our
language in Rust using OPMs of regular expressions as resources.
In contrast to \Cref{example:opm-files}, we do not fix a
particular envelope language, but assign each resource the envelope language
according to the regular expression with which it was created.
That is, we can describe resources by arbitrary non-empty regular
languages.

Here is an example program supported by our implementation that
revisits \Cref{example:opm-files}.
\begin{lstlisting}[numbers=left,xleftmargin=2\parindent]
let copy : {r*} ox {w*} -[u 1]-> Unit
    copy (rc, wc) = drop (!{r} rc); drop (!{w} wc)
in
let if0     = new {(r|w)*c} in  let of0     = new {(r|w)*c} in
let b1, if1 = split {r*} if0 in let b2, of1 = split {w*} of0 in
let _ = copy (b1, b2) in
drop (!{c} if1); drop (!{c} of1)
\end{lstlisting}
In the implementation, a resource type is written as $\{r\}$ where $r$
is a regular expression. The expression \lstinline/new {(r|w)*c}/ (in
line~4) creates a new resource with envelope language
$E = (r|w)^*c$. Line~5 splits off a borrow for reading (writing)
operations indicated by the \lstinline/{r*}/ (\lstinline/{w*}/) argument
to split. The implementation computes the best continuation of a split from
the type of the borrow: Its resource corresponds to the \emph{product
  derivative} (a variation of the standard language derivative
operation \cite{suzuki08:_produc_deriv_regul_expres}) of the argument
resource by the borrowed resource. In the example, the product
derivative of $E$ by $r^*$ ($w^*$) is again $E$. The function
\lstinline/copy/, defined in lines~1/2, takes an unordered pair of
borrows for reading and writing and returns the unit value. The
function type is annotated with \lstinline/u/ for unrestricted and
\lstinline/1/ for the latent effect of the function (it performs
resource operations). The function body performs a read and a write
operation and discards the borrows using the \lstinline/drop/
operation. The expression \lstinline/!{r} rc/
symbolizes a resource operation for \lstinline/{r}/, i.e., a single
read. The type would allow further reads and writes. In line~9, we
close the resources using \lstinline/!{c}/ and discard them
afterwards.

The implementation mostly saves the programmer from making ordered
types explicit. It infers a viable ordering mode for a
pair, a function, or a let by working through the rules from least
restrictive to most restrictive: first the unrestricted variant, then
the unordered linear variant, and finally the ordered variants. We
illustrate this approach with the let expressions in the
example.
\begin{itemize}
\item[Line~1] This let binds an unrestricted value, so it creates an
  unrestricted binding.
\item[Line~4.1] adds an ordered binding to the context.
\item[Line~4.2] the second let in this line is more tricky. If we
  interpreted this let as ordered, we would create an artificial
  constraint between \lstinline/if0/ and \lstinline/of0/. But using
  the the unordered linear variant is successful and results in
  independent bindings for  \lstinline/if0/ and \lstinline/of0/.
\item[Line~5] each split creates two ordered pair bindings, but the
  pairs are independent, so the second let must be unordered linear,
  again. Here is a place where context-contexts come into play: we
  have to replace the ordered binding \lstinline/of0/ by
  \lstinline/b2,of1/ without disturbing the surrounding context
  structure. 
\item[Line~6] creates an unrestricted binding. The variables
  \lstinline/b1/ and \lstinline/b2/ are correctly removed.
\item[Line~7] processes the remaining bindings. This line has similar
  issues as a let, because the semicolon can also be interpreted as
  independent or sequential.
\end{itemize}



\section{Related Work}
\label{sec:related-work}

\subsection{Aliasing, Linear Types, and Ordered Types}
\label{sec:ordered-typing}

The issue of aliasing has prompted
extensive research into methods for restricting and analyzing aliases
\cite{DBLP:series/lncs/7850}. However, these approaches tend to be
quite intricate. For instance, Boyland's  fractional permissions
\cite{boyland03:_check_inter_fract_permis} 
treat permissions for memory cells algebraically, allowing permissions
to be divided among multiple references. By doing so, if the
references are combined, the original (complete) permission can be
recovered. However, aside from basic reference counting, general
fractional permissions have not seen widespread adoption in practical
languages.

A lot of investigation has focused on ownership types
\cite{DBLP:conf/oopsla/ClarkePN98}, of which there are many
variations. Ownership types enforce encapsulation by ensuring
that an object's reference remains within its owner's control. Our
work focuses on ensuring sound typestate semantics.

Gordon et al. \cite{DBLP:conf/oopsla/GordonPPBD12} outline a type
system employing permissions to facilitate safe concurrency. While
this work concentrates on concurrency, it does not assist in reasoning
about object protocols, as typestate does.

Substructural logic imposes restrictions on the principles of
weakening, contraction, and exchange of hypotheses that are common in
intuitionistic as well as classical logic. Most famously, linear logic
elides weakening and contraction, so that hypotheses must be used
exactly once. This property makes linear logic a powerful foundation
for reasoning about resources \cite{Girard1987}.

\citet{DBLP:conf/ifip2/Wadler90} brought
linear types in focus of much research,
yet they were not widely adopted. Rust
\cite{DBLP:conf/sigada/MatsakisK14, 
  team24:_rust_progr_languag} stands as an exception, utilizing a form
of linearity to confine aliases to mutable objects. However, this
limited usage of linearity in Rust does not necessitate a
sophisticated permission system; for instance, Rust types cannot
directly express the states of referenced objects. Alms
\cite{DBLP:conf/popl/TovP11} is an ML-like language supporting affine
types, but not linearity. Typestate is not directly supported, but it
could be encoded. Session types
\cite{DBLP:journals/csur/HuttelLVCCDMPRT16} offer another approach to
linear types in programming languages, as seen in Concurrent C0
\cite{DBLP:journals/corr/WillseyPP17}. However, session types are more
directly suited to communicating concurrent processes. Recent versions
of Haskell support linearity
\cite{DBLP:journals/pacmpl/BernardyBNJS18,
  DBLP:journals/pacmpl/SpiwackKBWE22}, so that typestate could also be
encoded with some limitations. 

\citet{DBLP:journals/fuin/AhmedFM07} defined a core functional programming language
$\lambda$\textit{URAL} supporting strong updates, i.e., changing the type of
an object in a reference cell. They employ substructural typing in
several flavors: \textbf{U}nrestricted, \textbf{R}elevant,
\textbf{A}ffine, and \textnormal{L}inear. They introduced two
languages, L3 and extended L3. While L3 permits aliasing, it only
allows exclusive access through a capability, restricting a single
reference's ability to read/write to an object. In contrast, full,
shared, and pure access permissions allow for a broader range of
aliasing patterns. Extended L3 enables the recovery of a capability,
but programmers must prove that no other capabilities to the reference
cell exist. Extended L3 operates within a parametrized framework,
requiring the addition of one's own type system to associate proof
with the capability request. 

\citet{DBLP:journals/corr/abs-2310-18166} apply the theory of graded
types to control ownership and uniqueness. In their calculus,
permissions are described by a grade algebra, i.e., a preordered
semiring. We were expecting such structures to model resources, but
found them inadequate in our examples.

Ordered logic goes one step beyond linear logic in eliding exchange.
The Lambek calculus \cite{lambek58:_mathem_senten_struc,
  lambek60:_calcul_syntac_types} uses this rigid structure to reason
about sentence structure in natural languages. There are further
foundational studies exploring non-commutative logic
\cite{DBLP:journals/mlq/Abrusci90a, ABRUSCI199929}, but they are not
geared towards computer science applications.

\citet{DBLP:conf/tlca/PolakowP99} investigate the proof theory of
intuitionistic non-commutative linear logic (INCLL). Their presentation is based on
natural deduction and comes with a term calculus. They prove subject
reduction and canonicity and thus establish their calculus as a
foundation for a logical framework.
The typing rules of their calculus differ subtly from ours: they are
defined from a linear logic point of view, where every assumption is
linear, and their hypotheses are categorized in intuitionistic,
linear, or ordered contexts. In contrast, types in our system carry substructural
qualifications that determine their possible modes of use, thus our
contexts hold a mixture of intuitionistic, linear, and ordered
hypotheses. 

\citet{DBLP:journals/entcs/PolakowP99} extends the study of INCLL by
considering a sequent calculus and developing some of its proof
theory.
\citet{DBLP:conf/ppdp/Polakow00} further extends this development with
proof search in the style of logic programming.
Polakow's thesis \cite{polakow01:_order_linear_logic_applic} offers a
comprehensive review of the developments arising from INCLL.
\citet{DBLP:conf/flops/PolakovY01} present an application of an
ordered logical framework to obtain a theoretical foundation for a
stack-based compilation technique for exceptions.
DeYoung's thesis \cite{deyoung20:_session_typed_order_logic_specif} provides a more
recent overview of further work on INCLL and related systems.

\citet{Walker2005-attapl} discusses substructural type systems in
general. One part introduces an ordered type system with the goal to
``provide a foundation for managing memory allocated on the
stack''. The idea is to have the ordering on type assumptions
guarantee that stack-allocated values are used in a last-in/first-out
manner. The system restricts ordering to pairs and base types, thus
avoiding the need to introduce a complex environment structure or to
discuss several distinct function spaces.

We conclude this discussion with an observation. 
Most published work on ordered logic and typing relies on the core idea
that resources must not be shuffled. However, it does not
preclude reasoning steps that process resources ``in the middle'' of
a sequence of resources.
This view is quite different from our use of order. We require that
resources are processed in the prescribed sequence and do not permit
operations on resources in the middle of a sequence.

\subsection{Typestate}
\label{sec:typestate}

Strom and Yemini \cite{DBLP:journals/tse/StromY86} introduced the
concept of typestate for programming languages. The main difficulty in
making typestate useful is tracking aliases. Doing so is easy in the
presence of linearly handled resources, but real programs are rarely
written in this way so that any typestate system has to deal with
aliasing, e.g., by supporting a concept of borrowing.
Often, a typestate checker is implemented as an additional analysis on
top of a (typed) program. Techniques rooted in abstract
interpretation, such as those discussed by Fink et al. \cite{DBLP:journals/tosem/FinkYDRG08}, perform
a global alias analysis and typically assume the correct implementation of the
protocol in the resource class while focusing on verifying client
conformance. Naeem and Lhoták \cite{DBLP:conf/oopsla/NaeemL08} devised an analysis for
scrutinizing typestate properties across multiple interacting
objects. Such global analyses are expensive as they execute on the
entire code base of the system. 

Fugue \cite{DBLP:conf/ecoop/DeLineF04} was the first modular typestate
verification system for object-oriented programs. It categorizes
objects as "unaliased" or "maybe aliased", and permits state changes
only for "unaliased" objects.  \citet{DBLP:conf/oopsla/BierhoffA07} extended
this approach by introducing expressive run-time dependent method specifications based on
linear logic \cite{Girard1987}.  They introduced access permissions to
enable state changes in the presence of aliasing. They rely on
fractional permissions \cite{boyland03:_check_inter_fract_permis} to
accommodate patterns like borrowing and adoption
\cite{DBLP:conf/popl/BoylandR05}.  Plural, a tool for modular
typestate checking with 
access permissions for Java, has been studied in several practical
examples \cite{DBLP:conf/ecoop/BierhoffBA09}.  Nanda et al. \cite{DBLP:conf/oopsla/NandaGC05} outlined a
system for deriving typestate information from Java
programs.

Distributed session types \cite{DBLP:journals/corr/abs-1205-5344} offer a similar
expressiveness to Plural, albeit with protocols based on
structural types rather than nominal typestates. They  deal
with communication over distributed channels and object protocols but
do not permit aliasing for objects with protocols. 

Voinea et al \cite{DBLP:conf/forte/VoineaDG20} have developed a static
checker for typestate protocols in Java based on ideas from session types.
Bacchiani et al \cite{DBLP:journals/scp/BacchianiBGMR22} presented a
tool that statically verifies that during the execution of a Java
program: sequences of method calls obey the objects’ protocols; the
objects’ protocols are completed; there are no null-pointer
exceptions; instances of subclasses respect the
protocol of their superclasses.

The aforementioned approaches do not integrate typestate within the
programming model, but rather overlay a static typestate analysis on
top of an existing language. Typestate-oriented programming (TSOP),
proposed by Aldrich et al.  \cite{DBLP:conf/oopsla/AldrichSSS09}, distinguishes itself by supporting
changes to object representations at run time both in the statics as
in the dynamics. The programming language Plaid was the first to
integrate typestates in its programming model. Saini et al. \cite{DBLP:conf/ecoop/SainiSA10}
developed a core calculus for a TSOP language, which is object-based
and relies on structural types. Featherweight Typestate
\cite{DBLP:journals/toplas/GarciaTWA14} builds upon this work but
adapts it to a class-based, nominal approach with shared access
permissions and state guarantees for reasoning about typestate in the
presence of aliasing. Subsequently it has been extended to a
concurrent setting \cite{DBLP:journals/corr/abs-1904-01286}. Earlier work related to TSOP includes the Fickle
system \cite{DBLP:conf/ecoop/DrossopoulouDDG01}, which is capable of changing an object's
class at run time (dynamic reclassification) but with limited ability
to reason about field states.

\citet{DBLP:conf/popl/IgarashiK02,DBLP:conf/vmcai/KobayashiSW06,DBLP:conf/pepm/IwamaIK06}
propose a type system for \emph{resource usage}, which analyzes how a program accesses resources.
Their uses are expressed in terms of sets of sequences of program
points, denoted by a usage algebra with sophisticated operators.
Their system can analyze arbitrary
programs including interleaving usage patterns of aliases, whereas our
system is prescriptive as it sequentializes the processing of aliases.

We conclude this discussion with two observations.

Several typestate systems can reason about programs with multiple
active aliases where operations  can be interleaved. \emph{State
  guarantees} \cite{DBLP:journals/toplas/GarciaTWA14} are an example
for such a mechanism.
We can only implement a sequential state guarantee by using an index
structure with sufficiently high granularity. In fact, the first
example for split in the introduction is such an example. 

All existing typestate systems are state-oriented, in the sense that resources are annotated with pairs
of input and output state (cf.\
\Cref{sec:example-from-garcia}), whereas our system is
transition-oriented as we ask for (an abstraction of)
a sequence of operations. We believe that the latter provides significant flexibility as it can
encompass all transitions that result from applying the operations; a
state-oriented system would have to resort to structures like
intersection types or the logical specifications of \citet{DBLP:conf/oopsla/BierhoffA07}.
That is, our view to state change is dual to the traditional view.



\section{Conclusion}
\label{sec:conclusion}

We presented a novel transition-based foundation for typestate that
admits borrowing and has a range of interesting applications. Our
paper contributes a new perspective on ordered types using inspiration
from BI.
We believe that the shift from a state-based
presentation of typestate to a transition-based one can address a concern
that is often raised in the context of typestate. As
\citet[Sec.~4]{DBLP:journals/toplas/CoblenzOEKBBMSA20} put it: ``[A language]
including typestate could result in a design that was hard to
use, since typical typestate languages require users to understand a
complex permission model''. In our design, users have to
understand the resource API and traces. Future user studies will have
to clarify the usability of the transition-based approach.


\section*{Data Availability Statement}
A full version of this paper with all proofs is available as a
technical report \cite{saffrich2024lawordertypestateborrowing}.
The implementation is available in a GitHub repository
\url{https://github.com/m0rphism/bsession-impl.git}. 

\bibliography{main}

\begin{FULLVERSION}
  \FloatBarrier
  \appendix

\newpage
\section{Algorithmic Typing}
\label{sec:algorithmic-typing-appendix}

\begin{definition}[Context Restriction]
  Let $\HSCtx$ be a context and $X$ be a set of variables.
  The \emph{context restriction} of $\HSCtx$ with respect to $X$ is defined as
  \begin{align*}
    \HSRestrCtx\cdot\cdot &: \mathtt{Context} \times \mathcal{P}(\mathtt{Var}) \to \mathtt{Context} \\
    \HSRestrCtx\HSCtx{X} &=
    \begin{cases}
      \HSCtxEmpty & \HSCtx = \HSCtxEmpty \\
      \HSVarA : \ottnt{T} & \HSCtx = \HSVarA : \ottnt{T} \land \HSVarA \in X \\
      \HSCtxEmpty & \HSCtx = \HSVarA : \ottnt{T} \land \HSVarA \not\in X \\
      (\HSRestrCtx{\HSCtx_1}{X}) \HSSeq (\HSRestrCtx{\HSCtx_2}{X}) & \HSCtx = \HSCtx_1 \HSSeq \HSCtx_2 \\
      (\HSRestrCtx{\HSCtx_1}{X}) \HSPar (\HSRestrCtx{\HSCtx_2}{X}) & \HSCtx = \HSCtx_1 \HSPar \HSCtx_2
    \end{cases}
  \end{align*}
\end{definition}

\begin{definition}[Context Decomposition]
  Let $\HSCtx$ be a context and $X$ be a set of variables.
  The \emph{context decomposition} of $\HSCtx$ with $X$ is defined as
  \begin{align*}
    \HSSplitCtx\cdot\cdot &: \mathtt{Context} \times \mathcal{P}(\mathtt{Var}) \rightharpoonup
                             \mathtt{ContextContext} \times \mathtt{Context} \\
    \HSSplitCtx\HSCtx{X} &=
    \begin{cases}
      (\HSCtxCtxHole, \HSCtxEmpty)
        & \HSCtx = \HSCtxEmpty \\
      ((\HSCtxCtxHole\HSPar\HSVarA), \HSCtxEmpty)
        & \HSCtx = \HSVarA : \ottnt{T} \land \HSVarA \not\in X \\
      (\HSCtxCtxHole, \HSVarA : \ottnt{T})
        & \HSCtx = \HSVarA : \ottnt{T} \land \HSVarA \in X \land \HSOrd {\ottnt{T}}\\
      ((\HSCtxCtxHole\HSPar\HSVarA), \HSVarA : \ottnt{T})
        & \HSCtx = \HSVarA : \ottnt{T} \land \HSVarA \in X \land \HSUnr {\ottnt{T}} \\
      ((\HSCtx_1 \HSSeq \HSCtxCtx_2), \HSCtx_2')
        & \HSCtx = \HSCtx_1 \HSSeq \HSCtx_2 \land
          \HSDom{\HSCtx_1} \uplus X \land
          \HSSplitCtx{\HSCtx_2}{X} = (\HSCtxCtx_2, \HSCtx_2') \\
      ((\HSCtxCtx_1 \HSSeq \HSCtx_2), \HSCtx_1')
        & \HSCtx = \HSCtx_1 \HSSeq \HSCtx_2 \land
          \HSDom{\HSCtx_2} \uplus X \land
          \HSSplitCtx{\HSCtx_1}{X} = (\HSCtxCtx_1, \HSCtx_1') \\
      ((\HSCtx_r \HSSeq \HSCtxCtxHole \HSSeq \HSCtx_l), (\HSCtx_1' \HSSeq \HSCtx_2'))
        & \parbox[t]{.55\textwidth}{$
          \HSCtx = \HSCtx_1 \HSSeq \HSCtx_2 \land
          \HSSplitCtx{\HSCtx_1}{X} = (\HSCtxCtx_1, \HSCtx_1') \land
          \HSSplitCtx{\HSCtx_2}{X} = (\HSCtxCtx_2, \HSCtx_2') \land
          \HSRightCtxCtx{\HSCtxCtx_1} = \HSCtx_r \land
          \HSLeftCtxCtx{\HSCtxCtx_2} = \HSCtx_l
          $} \\
      ((\HSCtx_1 \HSPar \HSCtxCtx_2), \HSCtx_2')
        & \HSCtx = \HSCtx_1 \HSPar \HSCtx_2 \land
          \HSDom{\HSCtx_1} \uplus X \land
          \HSSplitCtx{\HSCtx_2}{X} = (\HSCtxCtx_2, \HSCtx_2') \\
      ((\HSCtxCtx_1 \HSPar \HSCtx_2), \HSCtx_1')
        & \HSCtx = \HSCtx_1 \HSPar \HSCtx_2 \land
          \HSDom{\HSCtx_2} \uplus X \land
          \HSSplitCtx{\HSCtx_1}{X} = (\HSCtxCtx_1, \HSCtx_1') \\
      ((\HSCtx_1'' \HSPar \HSCtx_2'' \HSPar \HSCtxCtxHole), (\HSCtx_1' \HSPar \HSCtx_2'))
        & \parbox[t]{.55\textwidth}{$
          \HSCtx = \HSCtx_1 \HSPar \HSCtx_2 \land
          \HSSplitCtx{\HSCtx_1}{X} = (\HSCtxCtx_1, \HSCtx_1') \land
          \HSSplitCtx{\HSCtx_2}{X} = (\HSCtxCtx_2, \HSCtx_2') \land
          \HSParCtxCtx{\HSCtxCtx_1} = \HSCtx_1'' \land
          \HSParCtxCtx{\HSCtxCtx_2} = \HSCtx_2''
          $} \\
      ((\HSCtxCtxHole \HSSeq (\HSCtx_1'' \HSPar \HSCtx_2'')), (\HSCtx_1' \HSPar \HSCtx_2'))
        & \parbox[t]{.55\textwidth}{$
          \HSCtx = \HSCtx_1 \HSPar \HSCtx_2 \land
          \HSSplitCtx{\HSCtx_1}{X} = (\HSCtxCtx_1, \HSCtx_1') \land
          \HSSplitCtx{\HSCtx_2}{X} = (\HSCtxCtx_2, \HSCtx_2') \land
          \HSLeftCtxCtx{\HSCtxCtx_1} = \HSCtx_1'' \land
          \HSLeftCtxCtx{\HSCtxCtx_2} = \HSCtx_2''
          $} \\
      (((\HSCtx_1'' \HSPar \HSCtx_2'') \HSSeq \HSCtxCtxHole), (\HSCtx_1' \HSPar \HSCtx_2'))
        & \parbox[t]{.55\textwidth}{$
          \HSCtx = \HSCtx_1 \HSPar \HSCtx_2 \land
          \HSSplitCtx{\HSCtx_1}{X} = (\HSCtxCtx_1, \HSCtx_1') \land
          \HSSplitCtx{\HSCtx_2}{X} = (\HSCtxCtx_2, \HSCtx_2') \land
          \HSRightCtxCtx{\HSCtxCtx_1} = \HSCtx_1'' \land
          \HSRightCtxCtx{\HSCtxCtx_2} = \HSCtx_2''
          $} \\
      ((\HSCtx_r \HSSeq \HSCtxCtxHole \HSSeq \HSCtx_l), (\HSCtx_2' \HSSeq \HSCtx_1'))
        & \parbox[t]{.55\textwidth}{$
          \HSCtx = \HSCtx_1 \HSPar \HSCtx_2 \land
          \HSSplitCtx{\HSCtx_1}{X} = (\HSCtxCtx_1, \HSCtx_1') \land
          \HSSplitCtx{\HSCtx_2}{X} = (\HSCtxCtx_2, \HSCtx_2') \land
          \HSRightCtxCtx{\HSCtxCtx_1} = \HSCtx_r \land
          \HSLeftCtxCtx{\HSCtxCtx_2} = \HSCtx_l
          $} \\
      ((\HSCtx_r \HSSeq \HSCtxCtxHole \HSSeq \HSCtx_l), (\HSCtx_1' \HSSeq \HSCtx_2'))
        & \parbox[t]{.55\textwidth}{$
          \HSCtx = \HSCtx_1 \HSPar \HSCtx_2 \land
          \HSSplitCtx{\HSCtx_1}{X} = (\HSCtxCtx_1, \HSCtx_1') \land
          \HSSplitCtx{\HSCtx_2}{X} = (\HSCtxCtx_2, \HSCtx_2') \land
          \HSLeftCtxCtx{\HSCtxCtx_1} = \HSCtx_l \land
          \HSRightCtxCtx{\HSCtxCtx_2} = \HSCtx_r
          $} \\
      (((\HSCtx_{11}'' \HSPar \HSCtx_{21}'') \HSSeq \HSCtxCtxHole \HSSeq (\HSCtx_{12}'' \HSPar \HSCtx_{22}'')), (\HSCtx_1' \HSPar \HSCtx_2'))
        & \parbox[t]{.55\textwidth}{$
          \HSCtx = \HSCtx_1 \HSPar \HSCtx_2 \land
          \HSSplitCtx{\HSCtx_1}{X} = (\HSCtxCtx_1, \HSCtx_1') \land
          \HSSplitCtx{\HSCtx_2}{X} = (\HSCtxCtx_2, \HSCtx_2') \land
          \HSClosedCtxCtx{\HSCtxCtx_1} = (\HSCtx_{11}'', \HSCtx_{12}'') \land
          \HSClosedCtxCtx{\HSCtxCtx_2} = (\HSCtx_{21}'', \HSCtx_{22}'')
          $}
    \end{cases}
  \end{align*}
\end{definition}

\begin{figure}
  \centering
  \begin{mathpar}
    \HSRuleAbsUnr \and
    \HSRuleAbsLin \and
    \HSRuleAbsLeft \and
    \HSRuleAbsRight \and
    \HSRuleChkInf
  \end{mathpar}
  \caption{Algorithmic Typing (Checking Rules)}
  \label{fig:algorithmic-typing-check}
\end{figure}

\begin{figure}
  \centering
  \begin{mathpar}
    \HSRuleUnit \and
    \HSRuleNew \and
    \HSRuleOp \and
    \HSRuleSplit \and
    \HSRuleDrop \and
    \HSRuleVar \and
    \HSRuleAppUnr \and
    \HSRuleAppLin \and
    \HSRuleAppLeft \and
    \HSRuleAppRight
  \end{mathpar}
  \caption{Algorithmic Typing (Inference Rules, Part I)}
  \label{fig:algorithmic-typing-infer-I}
\end{figure}

\begin{figure}
  \centering
  \begin{mathpar}
    \HSRulePairLin \and
    \HSRulePairLeft \and
    \HSRuleLetPairLin \and
    \HSRuleLetPairLeft \and
    \HSRuleAnn
  \end{mathpar}
  \caption{Algorithmic Typing (Inference Rules, Part II)}
  \label{fig:algorithmic-typing-infer-II}
\end{figure}

Figure~\ref{fig:algorithmic-typing-check},
\ref{fig:algorithmic-typing-infer-I}, and
\ref{fig:algorithmic-typing-infer-II}
shows the rules for the algorithmic typing relation.


  \clearpage
  \section{Proof of Proposition~\ref{prop:type-system/graph-cmp}}\label{sec:graph-cmp}
  \begin{proposition}[Proposition 3.13 in the paper]\label{prop:graph-cmp}
    Contexts $\Gamma_{{\mathrm{1}}}$ and $\Gamma_{{\mathrm{2}}}$ are related by applying the monoid laws,
    commutativity of unordered composition, and demotion with
    unrestricted bindings, if and only if  $\Gamma_{{\mathrm{1}}}  \simeq  \Gamma_{{\mathrm{2}}}$.
\end{proposition}

We start by formally defining a syntactic relation mentioned in the property as follows.

\begin{ynrules}
    \ynjudgment{$\Gamma_{{\mathrm{1}}}  \equiv  \Gamma_{{\mathrm{2}}}$}
    \yninfer[CIdL]{}{
        $ \cdot   \ottsym{,}  \Gamma  \equiv  \Gamma$
    }
    \yninfer[CIdR]{}{
        $\Gamma  \ottsym{,}   \cdot   \equiv  \Gamma$
    }
    \yninfer[CAssoc]{}{
        $\Gamma_{{\mathrm{1}}}  \ottsym{,}  \ottsym{(}  \Gamma_{{\mathrm{2}}}  \ottsym{,}  \Gamma_{{\mathrm{3}}}  \ottsym{)}  \equiv  \ottsym{(}  \Gamma_{{\mathrm{1}}}  \ottsym{,}  \Gamma_{{\mathrm{2}}}  \ottsym{)}  \ottsym{,}  \Gamma_{{\mathrm{3}}}$
    }
    \yninfer[PId]{}{
        $ \cdot   \parallel  \Gamma  \equiv  \Gamma$
    }
    \yninfer[PComm]{}{
        $\Gamma_{{\mathrm{1}}}  \parallel  \Gamma_{{\mathrm{2}}}  \equiv  \Gamma_{{\mathrm{2}}}  \parallel  \Gamma_{{\mathrm{1}}}$
    }
    \yninfer[PAssoc]{}{
        $\Gamma_{{\mathrm{1}}}  \parallel  \ottsym{(}  \Gamma_{{\mathrm{2}}}  \parallel  \Gamma_{{\mathrm{3}}}  \ottsym{)}  \equiv  \ottsym{(}  \Gamma_{{\mathrm{1}}}  \parallel  \Gamma_{{\mathrm{2}}}  \ottsym{)}  \parallel  \Gamma_{{\mathrm{3}}}$
    }
    \yninfer[UUnr]{
        $\ottkw{unr} \, \ottnt{T}$
    }{
        $\mathcal{G}  \ottsym{[}  \ottmv{x}  \mathord:  \ottnt{T}  \ottsym{]}  \equiv  \mathcal{G}  \ottsym{[}   \cdot   \ottsym{]}  \parallel  \ottmv{x}  \mathord:  \ottnt{T}$
    }
    \yninfer[UDemo]{
        $\ottkw{unr} \, \ottnt{T}$
    }{
        $\ottmv{x}  \mathord:  \ottnt{T}  \parallel  \ottmv{x}  \mathord:  \ottnt{T}  \equiv  \ottmv{x}  \mathord:  \ottnt{T}$
    }
    \yninfer[Ctx]{
        $\Gamma_{{\mathrm{1}}}  \equiv  \Gamma_{{\mathrm{2}}}$
    }{
        $\mathcal{G}  \ottsym{[}  \Gamma_{{\mathrm{1}}}  \ottsym{]}  \equiv  \mathcal{G}  \ottsym{[}  \Gamma_{{\mathrm{2}}}  \ottsym{]}$
    }
    \yninfer[Refl]{}{
        $\Gamma  \equiv  \Gamma$
    }
    \yninfer[Sym]{
        $\Gamma_{{\mathrm{2}}}  \equiv  \Gamma_{{\mathrm{1}}}$
    }{
        $\Gamma_{{\mathrm{1}}}  \equiv  \Gamma_{{\mathrm{2}}}$
    }
    \yninfer[Trans]{
        $\Gamma_{{\mathrm{1}}}  \equiv  \Gamma_{{\mathrm{2}}}$ \\
        $\Gamma_{{\mathrm{2}}}  \equiv  \Gamma_{{\mathrm{3}}}$
    }{
        $\Gamma_{{\mathrm{1}}}  \equiv  \Gamma_{{\mathrm{3}}}$
    }
\end{ynrules}

The first three rules determine the monoid lows for the ordered composition.
The next three rules determine the monoid lows and commutativity for the unordered composition.
\ruleref{UUnr} and \ruleref{UDemo} are about unrestricted bindings.
\ruleref{UUnr} characterizes embed treatment of unrestricted bindings; that is their position in a context is irrelevant.
\ruleref{UDemo} determines the demotion (contraction) of unrestricted bindings.
The last four rules give the congruence relation as usual.

Now we can re-state the proposition more formally as follows.

\begin{proposition}
    Let $\mathfrak{G}_{{\mathrm{1}}}; S_1 =  \lBrack \Gamma_{{\mathrm{1}}} \rBrack $ and $\mathfrak{G}_{{\mathrm{2}}}; S_2 =  \lBrack \Gamma_{{\mathrm{2}}} \rBrack $.
    Then, $\Gamma_{{\mathrm{1}}}  \equiv  \Gamma_{{\mathrm{2}}}$ iff $\mathfrak{G}_{{\mathrm{1}}}  \simeq  \mathfrak{G}_{{\mathrm{2}}}$ and $S_1 = S_2$.
\end{proposition}

It is rather straightforward to show the only-if direction (left to right) by induction on the given derivation (in the cases \ruleref{UUnr} and \ruleref{Ctx} require other induction proofs though).
So, we concentrate on a proof of the if direction (right to left).

The first observation to show the property is that we can split unrestricted bindings out from a typing context by using \ruleref{UUnr}.
That is, for any $\Gamma$, there is a syntactic equivalent $\Gamma'$, namely $\Gamma  \equiv  \Gamma'$, of the form $\Gamma_{{\mathrm{0}}}  \parallel  \ottmv{x_{{\mathrm{1}}}}  \mathord:  \ottnt{T_{{\mathrm{1}}}}  \parallel  \dots  \parallel  \ottmv{x_{\ottmv{k}}}  \mathord:  \ottnt{T_{\ottmv{k}}}$ where $\Gamma_{{\mathrm{0}}}$  has no unrestricted bindings and all $\ottnt{T_{{\mathrm{1}}}} \dots \ottnt{T_{\ottmv{k}}}$ are unrestricted.
Using this observation, we can examine the restricted bindings and unrestricted bindings separately and see that the unrestricted part ($S_1 = S_2$) of semantic equivalence can be related syntactically via that form.

So, in the following, we discuss how to relate the restricted binding part ($\mathfrak{G}_{{\mathrm{1}}}  \simeq  \mathfrak{G}_{{\mathrm{2}}}$) of semantic equivalence syntactically.
Thus, from here, we only consider typing contexts not containing unrestricted bindings, that is $\Gamma$ such that $  \lBrack \Gamma \rBrack  \ottsym{=} \mathfrak{G} ; \emptyset$; and, for ease of writing, we abuse $ \lBrack \Gamma \rBrack $ to denote the graph part $\mathfrak{G}$.

The proof strategy is as follows.
Firstly, we define \emph{realization} of a graph representation $\mathfrak{G}$, denoted by $ \REP{ \mathfrak{G} } $, which is a partial function from graph representations to typing contexts providing a syntactic presentation of $\mathfrak{G}$.
Note that the realization is surely partial as some graph representations have no corresponding typing context.
For instance, no typing context is interpreted as the following graph.
\begin{center}
    \begin{tikzpicture}[every node/.style={node distance=1em}]
        \node (A) {$\ottmv{x_{{\mathrm{1}}}}  \mathord:  \ottnt{T_{{\mathrm{1}}}}$};
        \node [right=of A] (B) {$\ottmv{x_{{\mathrm{2}}}}  \mathord:  \ottnt{T_{{\mathrm{2}}}}$};
        \node [below=of A] (C) {$\ottmv{x_{{\mathrm{3}}}}  \mathord:  \ottnt{T_{{\mathrm{3}}}}$};
        \node [right=of C] (D) {$\ottmv{x_{{\mathrm{4}}}}  \mathord:  \ottnt{T_{{\mathrm{4}}}}$};
        \draw[->] (A) -- (B);
        \draw[->] (A) -- (D);
        \draw[->] (C) -- (D);
    \end{tikzpicture}
\end{center}
Even so, we will show the following proposition that says given realization suffices for interpretations of typing contexts, meaning that it deserves our purpose.

\begin{proposition}\label{prop:realize/correct}
    $ \REP{  \lBrack \Gamma \rBrack  }   \equiv  \Gamma$.
\end{proposition}

Lastly, we show the following proposition.

\begin{proposition}\label{prop:realize/equiv}
    If $ \lBrack \Gamma_{{\mathrm{1}}} \rBrack   \simeq   \lBrack \Gamma_{{\mathrm{2}}} \rBrack $, then $ \REP{  \lBrack \Gamma_{{\mathrm{1}}} \rBrack  }   \equiv   \REP{  \lBrack \Gamma_{{\mathrm{2}}} \rBrack  } $.
\end{proposition}

Combining the propositions, we can see $\Gamma_{{\mathrm{1}}}  \equiv  \Gamma_{{\mathrm{2}}}$ if $ \lBrack \Gamma_{{\mathrm{1}}} \rBrack   \simeq   \lBrack \Gamma_{{\mathrm{2}}} \rBrack $ as $\Gamma_{{\mathrm{1}}}  \equiv   \REP{  \lBrack \Gamma_{{\mathrm{1}}} \rBrack  }   \equiv   \REP{  \lBrack \Gamma_{{\mathrm{2}}} \rBrack  }   \equiv  \Gamma_{{\mathrm{2}}}$.

We start by defining several \emph{cuts}.
Since the vertices of a graph representation are indexed by natural numbers, we can specify a cut by a natural number.
We formally define the sense by giving the pair of induced sub-graphs from a given natural number as follows.

\begin{definition}[Index cut]
    Given $\mathfrak{G} = (n, v, E)$ and $n'$ such that $0 < n' < n$, we define a pair of graph representations, denoted by $( \mathfrak{G} {}_{<  \ottnt{n'} } ,  \mathfrak{G} {}_{\ge  \ottnt{n'} } )$, as follows.
    \begin{align*}
         \mathfrak{G} {}_{<  \ottnt{n'} }  & = (n', v\vert_{\NAT_{< n'}}, \{(i, j) \in E \mid i, j < n' \})                              \\
         \mathfrak{G} {}_{\ge  \ottnt{n'} }  & = (n - n', i \mapsto v(i + n'), \{(i - n', j - n') \mid (i, j) \in E \wedge n' \le i, j \})
    \end{align*}
\end{definition}

Note that interpretations of typing contexts are closed under index cut.

\begin{lemma}\label{prop:cut-closed-left}
    For any index cut $(  \lBrack \Gamma \rBrack  {}_{<  \ottnt{n'} } ,   \lBrack \Gamma \rBrack  {}_{\ge  \ottnt{n'} } )$, there exist $\Gamma'$ such that $   \lBrack \Gamma \rBrack  {}_{<  \ottnt{n'} }  \ottsym{=}  \lBrack \Gamma' \rBrack  $.

    \proof By structural induction on $\Gamma$.
    \begin{itemize}
        \item ($ \Gamma \ottsym{=}  \cdot  $ or $ \Gamma \ottsym{=} \ottnt{b} $)
              Let $ \lBrack \Gamma \rBrack  = (n, v, E)$.
              In these cases, $n \le 1$.
              Thus, we have no index cut for $ \lBrack \Gamma \rBrack $.
              Consequently, this case is vacuously true.
        \item ($ \Gamma \ottsym{=} \Gamma_{{\mathrm{1}}}  \parallel  \Gamma_{{\mathrm{2}}} )$
              In this case, $  \lBrack \Gamma \rBrack  \ottsym{=}  \lBrack \Gamma_{{\mathrm{1}}} \rBrack   \cup   \lBrack \Gamma_{{\mathrm{2}}} \rBrack  $.
              Let $ \lBrack \Gamma_{{\mathrm{1}}} \rBrack  = (n_1, v_1, E_1)$ and $ \lBrack \Gamma_{{\mathrm{2}}} \rBrack  = (n_2, v_2, E_2)$.
              We consider the order relationship between $n'$ and $n_1$.
              \begin{itemize}
                  \item ($n' < n_1$)
                        In this case, $   \lBrack \Gamma \rBrack  {}_{<  \ottnt{n'} }  \ottsym{=}   \lBrack \Gamma_{{\mathrm{1}}} \rBrack  {}_{<  \ottnt{n'} }  $.
                        So, we have some $\Gamma'_{{\mathrm{1}}}$ such that $   \lBrack \Gamma_{{\mathrm{1}}} \rBrack  {}_{<  \ottnt{n'} }  \ottsym{=}  \lBrack \Gamma'_{{\mathrm{1}}} \rBrack  $.
                        We choose $\Gamma'_{{\mathrm{1}}}$ as $\Gamma'$ for the goal.
                  \item ($n' = n_1$)
                        In this case, $   \lBrack \Gamma \rBrack  {}_{<  \ottnt{n'} }  \ottsym{=}  \lBrack \Gamma_{{\mathrm{1}}} \rBrack  $.
                        So, choosing $\Gamma_{{\mathrm{1}}}$ as $\Gamma'$ suffices the goal.
                  \item ($n' > n_1$)
                        In this case, $   \lBrack \Gamma \rBrack  {}_{<  \ottnt{n'} }  \ottsym{=}  \lBrack \Gamma_{{\mathrm{1}}} \rBrack    \cup   \lBrack \Gamma_{{\mathrm{2}}} \rBrack _{<n' - n_1}$.
                        We have $\Gamma'_{{\mathrm{2}}}$ such that $ \lBrack \Gamma'_{{\mathrm{2}}} \rBrack  =  \lBrack \Gamma_{{\mathrm{2}}} \rBrack _{<n'-n_1}$ by IH.
                        We choose $\Gamma_{{\mathrm{1}}}  \parallel  \Gamma'_{{\mathrm{2}}}$ as $\Gamma'$ to show the goal.
                        In fact, $  \lBrack \Gamma_{{\mathrm{1}}}  \parallel  \Gamma'_{{\mathrm{2}}} \rBrack  \ottsym{=}  \lBrack \Gamma_{{\mathrm{1}}} \rBrack   \cup   \lBrack \Gamma'_{{\mathrm{2}}} \rBrack   =   \lBrack \Gamma \rBrack  {}_{<  \ottnt{n'} } $.
              \end{itemize}
        \item ($ \Gamma \ottsym{=} \Gamma_{{\mathrm{1}}}  \ottsym{,}  \Gamma_{{\mathrm{2}}} $) Similar to the case above.
              \qedhere
    \end{itemize}
\end{lemma}

\begin{lemma}\label{prop:cut-closed-right}
    For any index cut $(  \lBrack \Gamma \rBrack  {}_{<  \ottnt{n'} } ,   \lBrack \Gamma \rBrack  {}_{\ge  \ottnt{n'} } )$, there exist $\Gamma'$ such that $   \lBrack \Gamma \rBrack  {}_{\ge  \ottnt{n'} }  \ottsym{=}  \lBrack \Gamma' \rBrack  $.

    \proof Similar to the proof or \cref{prop:cut-closed-left}. \qedhere
\end{lemma}

We define the following characteristic cuts to define the realization.

\begin{definition}[Union cut]
    Given $\mathfrak{G} = (n, v, E)$, we call a natural number $n'$ \emph{union cut} iff
    \begin{itemize}
        \item $0 < n' < n$; and
        \item $(i, j) \notin E$ for any $0 \le i < n'$ and $n' \le j < n$.
    \end{itemize}
\end{definition}

\begin{definition}[Join cut]
    Given $\mathfrak{G} = (n, v, E)$, we call a natural number $n'$ \emph{join cut} iff
    \begin{itemize}
        \item $0 < n' < n$; and
        \item $(i, j) \in E$ for any $0 \le i < n'$ and $n' \le j < n$.
    \end{itemize}
\end{definition}

Union and join cut correspond to the connecting point of the union and join of graph representations, respectively, as we state as follows.

\begin{lemma}\label{prop:union-pt}
    Suppose $\mathfrak{G}_{{\mathrm{1}}} = (n_1, v_1, E_1)$ and $\mathfrak{G}_{{\mathrm{2}}} = (n_2, v_2, E_2)$ where $n_1, n_2 > 0$.
    Then, $n_1$ is a union cut of $\mathfrak{G}_{{\mathrm{1}}}  \cup  \mathfrak{G}_{{\mathrm{2}}}$, and $  \ottsym{(}  \mathfrak{G}_{{\mathrm{1}}}  \cup  \mathfrak{G}_{{\mathrm{2}}}  \ottsym{)} {}_{<  \ottnt{n_{{\mathrm{1}}}} }  \ottsym{=} \mathfrak{G}_{{\mathrm{1}}} $, $  \ottsym{(}  \mathfrak{G}_{{\mathrm{1}}}  \cup  \mathfrak{G}_{{\mathrm{2}}}  \ottsym{)} {}_{\ge  \ottnt{n_{{\mathrm{1}}}} }  \ottsym{=} \mathfrak{G}_{{\mathrm{2}}} $.
    \proof By definition. \qedhere
\end{lemma}

\begin{lemma}\label{prop:join-pt}
    Suppose $\mathfrak{G}_{{\mathrm{1}}} = (n_1, v_1, E_1)$ and $\mathfrak{G}_{{\mathrm{2}}} = (n_2, v_2, E_2)$ where $n_1, n_2 > 0$.
    Then $n_1$ is a join cut of $ \GJOIN{ \mathfrak{G}_{{\mathrm{1}}} }{ \mathfrak{G}_{{\mathrm{2}}} } $, and $  \ottsym{(}   \GJOIN{ \mathfrak{G}_{{\mathrm{1}}} }{ \mathfrak{G}_{{\mathrm{2}}} }   \ottsym{)} {}_{<  \ottnt{n_{{\mathrm{1}}}} }  \ottsym{=} \mathfrak{G}_{{\mathrm{1}}} $, $  \ottsym{(}   \GJOIN{ \mathfrak{G}_{{\mathrm{1}}} }{ \mathfrak{G}_{{\mathrm{2}}} }   \ottsym{)} {}_{\ge  \ottnt{n_{{\mathrm{1}}}} }  \ottsym{=} \mathfrak{G}_{{\mathrm{2}}} $.
    \proof By definition. \qedhere
\end{lemma}

We put the fact union cuts of a graph representation do not influence each other.
A similar property holds for join cuts.
Such properties are handy for doing mathematical induction proof for the number of union and join cuts.
Note that, however, a union cut and a join cut will influence each other; that is a join cut may spawn in a sub-graph, obtained by a union cut.

\begin{lemma}
    Suppose $\mathfrak{G}' = (n', v', E')$ and $n$ is a union cut of $\mathfrak{G}$.
    Then,
    \begin{itemize}
        \item $n$ is a union cut of $\mathfrak{G}  \cup  \mathfrak{G}'$; and
        \item $n + n'$ is a union cut of $\mathfrak{G}'  \cup  \mathfrak{G}$.
    \end{itemize}
    \proof By definition. \qedhere
\end{lemma}

\begin{lemma}
    Suppose $n_1 < n_2$ are union cuts of $\mathfrak{G} = (n, v, E)$.
    Then,
    \begin{itemize}
        \item $n_1$ is a union cut of $ \mathfrak{G} {}_{<  \ottnt{n_{{\mathrm{2}}}} } $; and
        \item $n_2 - n_1$ is a union cut of $ \mathfrak{G} {}_{\ge  \ottnt{n_{{\mathrm{1}}}} } $.
    \end{itemize}
    \proof By definition. \qedhere
\end{lemma}

\begin{lemma}
    Suppose $\mathfrak{G}' = (n', v', E')$ and $n$ is a join cut of $\mathfrak{G}$.
    Then,
    \begin{itemize}
        \item $n$ is a join cut of $ \GJOIN{ \mathfrak{G} }{ \mathfrak{G}' } $; and
        \item $n + n'$ is a union cut of $ \GJOIN{ \mathfrak{G}' }{ \mathfrak{G} } $.
    \end{itemize}
    \proof By definition. \qedhere
\end{lemma}

\begin{lemma}
    Suppose $n_1 < n_2$ are join cuts of $\mathfrak{G} = (n, v, E)$.
    Then,
    \begin{itemize}
        \item $n_1$ is a join cut of $ \mathfrak{G} {}_{<  \ottnt{n_{{\mathrm{2}}}} } $; and
        \item $n_2 - n_1$ is a join cut of $ \mathfrak{G} {}_{\ge  \ottnt{n_{{\mathrm{1}}}} } $.
    \end{itemize}
    \proof By definition. \qedhere
\end{lemma}

Now we define the realization accompanied by the following lemma that guarantees the well-definedness.

\begin{lemma}\label{prop:union-join-ex}
    For any $\mathfrak{G}$, either union cuts or join cuts can exclusively exist.

    \proof By contradiction.
    Suppose there exist both union and join cuts for some $\mathfrak{G} = (n, v, E)$.
    Then, we have $(0, n - 1) \notin E$ and $(0, n - 1) \in E$ by the condition of the union and join cuts, respectively; but they contradict each other.
    \qedhere
\end{lemma}

\begin{definition}[Realization]
    For $\mathfrak{G} = (n, v, E)$, we define $ \REP{ \mathfrak{G} } $ inductively for $n$ as follows.
    \begin{gather*}
         \REP{ \mathfrak{G} }  = \begin{cases}
             \cdot                             & \text{if $n = 0$}                                             \\
            v(0)                             & \text{if $n = 1$}                                             \\
             \REP{  \mathfrak{G} {}_{<  \ottnt{n'} }  }   \parallel   \REP{  \mathfrak{G} {}_{\ge  \ottnt{n'} }  }  & \text{if there exist the minimum union cut $n'$ for $\mathfrak{G}$} \\
             \REP{  \mathfrak{G} {}_{<  \ottnt{n'} }  }   \ottsym{,}   \REP{  \mathfrak{G} {}_{\ge  \ottnt{n'} }  }   & \text{if there exist the minimum join cut $n'$ for $\mathfrak{G}$}
        \end{cases}
    \end{gather*}
\end{definition}


As we mentioned, realization is a partial function.
Moreover, it does not work correctly in general.
A simple troubling case is shown in the following diagram, where the interpretation of the realization of the graph representation is not isomorphic to the original graph representation.
\begin{gather*}
    \begin{tikzpicture}[baseline=(A.base), every node/.style={node distance=1em}]
        \node (A) {$\ottmv{x_{{\mathrm{1}}}}  \mathord:  \ottnt{T_{{\mathrm{1}}}}$};
        \node [right=of A] (B) {$\ottmv{x_{{\mathrm{2}}}}  \mathord:  \ottnt{T_{{\mathrm{2}}}}$};
        \draw[->] (B) to (A);
    \end{tikzpicture}
    \qquad
    \overset{\REP{-}}{\longmapsto}
    \qquad
    \ottmv{x_{{\mathrm{1}}}}  \mathord:  \ottnt{T_{{\mathrm{1}}}}  \parallel  \ottmv{x_{{\mathrm{2}}}}  \mathord:  \ottnt{T_{{\mathrm{2}}}}
    \qquad
    \overset{\INTRP{-}}{\longmapsto}
    \qquad
    \begin{tikzpicture}[baseline=(A.base), every node/.style={node distance=1em}]
        \node (A) {$\ottmv{x_{{\mathrm{1}}}}  \mathord:  \ottnt{T_{{\mathrm{1}}}}$};
        \node [right=of A] (B) {$\ottmv{x_{{\mathrm{2}}}}  \mathord:  \ottnt{T_{{\mathrm{2}}}}$};
    \end{tikzpicture}
\end{gather*}
One might think that is caused by the definitions of union and join cuts, where we ignore inverse direction edges, that is ones from larger indexed vertices to smaller indexed ones.
However, it is not.
We can find another troubling case even if we fix the definitions taking inverse direction edges into account.

Anyway, the realization works correctly for the range of $\INTRP{-}$, as we will show \cref{prop:realize/correct}.
The key insight to show the proposition is the interpretation of a typing context is always topologically ordered, no inverse direction edges exist.

\begin{definition}[Topologically ordered]
    We say $\mathfrak{G} = (n, v, E)$ \emph{topologically ordered} iff the indentify function for $\NAT_{< n}$ is a topological ordering of $\mathfrak{G}$.
    In other words, $n_1 < n_2$ for any $(n_1, n_2) \in E$.
\end{definition}

\begin{lemma}\label{prop:intrp-ordered}
    $ \lBrack \Gamma \rBrack $ is topologically ordered.

    \proof By structural induction on $\Gamma$. \qedhere
\end{lemma}

For a topologically ordered graph representation, we can recover the original graph from the two sub-graphs obtained by a union or join cut of the graph by the union or join of them, respectively.

\begin{lemma}\label{prop:union-cut}
    Suppose $\mathfrak{G}$ is topologically ordered.
    If $n$ is a union cut of $\mathfrak{G}$, then $ \mathfrak{G} \ottsym{=}   \mathfrak{G} {}_{<  \ottnt{n} }   \cup  \mathfrak{G} {}_{\ge  \ottnt{n} }  $.
    \proof By definition.
    \qedhere
\end{lemma}

\begin{lemma}\label{prop:join-cut}
    Suppose $\mathfrak{G}$ is topologically ordered.
    If $n$ is a join cut of $\mathfrak{G}$, then $ \mathfrak{G} \ottsym{=}   \GJOIN{  \mathfrak{G} {}_{<  \ottnt{n} }  }{ \mathfrak{G} }  {}_{\ge  \ottnt{n} }  $.
    \proof By definition.
    \qedhere
\end{lemma}

Using the lemmas, we can relate the semantic composition (union and join) and syntactic composition (parallel and concatenate) under the realization as follows.

\begin{lemma}\label{prop:union-comm}
    For any topologically ordered $\mathfrak{G}_{{\mathrm{1}}} = (n_1, v_1, E_1)$ and $\mathfrak{G}_{{\mathrm{2}}} = (n_2, v_2, E_2)$, $ \REP{ \mathfrak{G}_{{\mathrm{1}}}  \cup  \mathfrak{G}_{{\mathrm{2}}} } $ is defined iff both $ \REP{ \mathfrak{G}_{{\mathrm{1}}} } $ and $ \REP{ \mathfrak{G}_{{\mathrm{2}}} } $ are defined, and $ \REP{ \mathfrak{G}_{{\mathrm{1}}}  \cup  \mathfrak{G}_{{\mathrm{2}}} }   \equiv   \REP{ \mathfrak{G}_{{\mathrm{1}}} }   \parallel   \REP{ \mathfrak{G}_{{\mathrm{2}}} } $.
    \proof If either $n_1$ or $n_2$ is 0, we can have the goal by \ruleref{PId} and \ruleref{PComm}.
    So, we show the case $n_1, n_2 > 0$ by mathematical induction on the number of union cuts of $\mathfrak{G}_{{\mathrm{1}}}$.
    \begin{itemize}
        \item (Base case)
              In this case, $\mathfrak{G}_{{\mathrm{1}}}$ has no union cut.
              Then, $n_1$ is the minimum union cut of $\mathfrak{G}_{{\mathrm{1}}}  \cup  \mathfrak{G}_{{\mathrm{2}}}$.
              That is because if there exists a smaller union cut of $\mathfrak{G}_{{\mathrm{1}}}  \cup  \mathfrak{G}_{{\mathrm{2}}}$ than $n_1$, it must be a union cut of $\mathfrak{G}_{{\mathrm{1}}}$, too; that contradicts this case's assumption.
              So, we have $  \REP{ \mathfrak{G}_{{\mathrm{1}}}  \cup  \mathfrak{G}_{{\mathrm{2}}} }  \ottsym{=}  \REP{ \mathfrak{G}_{{\mathrm{1}}} }   \parallel   \REP{ \mathfrak{G}_{{\mathrm{2}}} }  $ by definition and \cref{prop:union-pt}.
              As a result the goal can be derived by \ruleref{Refl}.
        \item (Induction step)
              In this case, $\mathfrak{G}_{{\mathrm{1}}}$ has some union cut.
              So, let $n'$ be the minimum union cut of $\mathfrak{G}_{{\mathrm{1}}}$.
              Then, we have
              \begin{gather*}
                    \REP{ \mathfrak{G}_{{\mathrm{1}}} }  \ottsym{=}  \REP{  \mathfrak{G}_{{\mathrm{1}}} {}_{<  \ottnt{n'} }  }   \parallel   \REP{  \mathfrak{G}_{{\mathrm{1}}} {}_{\ge  \ottnt{n'} }  }  
              \end{gather*}
              We can see $n'$ is the minimum union cut of $\mathfrak{G}_{{\mathrm{1}}}  \cup  \mathfrak{G}_{{\mathrm{2}}}$, too.
              So, we have
              \begin{gather*}
                    \REP{ \mathfrak{G}_{{\mathrm{1}}}  \cup  \mathfrak{G}_{{\mathrm{2}}} }  \ottsym{=}  \REP{  \mathfrak{G}_{{\mathrm{1}}} {}_{<  \ottnt{n'} }  }   \parallel   \REP{  \mathfrak{G}_{{\mathrm{1}}} {}_{\ge  \ottnt{n'} }   \cup  \mathfrak{G}_{{\mathrm{2}}} }  
              \end{gather*}
              Here, we can see $ \REP{  \mathfrak{G}_{{\mathrm{1}}} {}_{\ge  \ottnt{n'} }   \cup  \mathfrak{G}_{{\mathrm{2}}} }   \equiv   \REP{  \mathfrak{G}_{{\mathrm{1}}} {}_{\ge  \ottnt{n'} }  }   \parallel   \REP{ \mathfrak{G}_{{\mathrm{2}}} } $ by IH.
              Therefore, we can have the goal by \ruleref{PAssoc}. \qedhere
    \end{itemize}
\end{lemma}

\begin{lemma}\label{prop:join-comm}
    For any topologically ordered $\mathfrak{G}_{{\mathrm{1}}}$ and $\mathfrak{G}_{{\mathrm{2}}}$, $ \REP{  \GJOIN{ \mathfrak{G}_{{\mathrm{1}}} }{ \mathfrak{G}_{{\mathrm{2}}} }  } $ is defined iff both $ \REP{ \mathfrak{G}_{{\mathrm{1}}} } $ and $ \REP{ \mathfrak{G}_{{\mathrm{2}}} } $ are defined, and $ \REP{  \GJOIN{ \mathfrak{G}_{{\mathrm{1}}} }{ \mathfrak{G}_{{\mathrm{2}}} }  }   \equiv   \REP{ \mathfrak{G}_{{\mathrm{1}}} }   \ottsym{,}   \REP{ \mathfrak{G}_{{\mathrm{2}}} } $.
    \proof Similar to the proof of \cref{prop:union-comm}, using \cref{prop:join-pt} and \cref{prop:join-cut}. \qedhere
\end{lemma}

Now we can prove \cref{prop:realize/correct}.

\begin{lemma}[\cref{prop:realize/correct}]\label{prop:realize/correct/lemma}
    $ \REP{  \lBrack \Gamma \rBrack  }   \equiv  \Gamma$.
    \proof
    We prove the proposition by structural induction on $\Gamma$.
    \begin{itemize}
        \item (Cases $ \Gamma \ottsym{=}  \cdot  $ and $ \Gamma \ottsym{=} \ottnt{b} $)
              In these cases, $  \REP{  \lBrack \Gamma \rBrack  }  \ottsym{=} \Gamma $ by definition.
              So, we can have the goal $ \REP{  \lBrack \Gamma \rBrack  }   \equiv  \Gamma$ by \ruleref{Refl}.
        \item (Case $ \Gamma \ottsym{=} \Gamma_{{\mathrm{1}}}  \parallel  \Gamma_{{\mathrm{2}}} $)
              In this case, $  \lBrack \Gamma \rBrack  \ottsym{=}  \lBrack \Gamma_{{\mathrm{1}}} \rBrack   \cup   \lBrack \Gamma_{{\mathrm{2}}} \rBrack  $.
              We have
              \begin{gather*}
                   \REP{  \lBrack \Gamma_{{\mathrm{1}}} \rBrack  }   \equiv  \Gamma_{{\mathrm{1}}} \tag{IH1}\\
                   \REP{  \lBrack \Gamma_{{\mathrm{2}}} \rBrack  }   \equiv  \Gamma_{{\mathrm{2}}} \tag{IH2}
              \end{gather*}
              by IH.
              Now we can have the goal as follows
              \begin{align*}
                   \REP{  \lBrack \Gamma \rBrack  }  & =  \REP{  \lBrack \Gamma_{{\mathrm{1}}} \rBrack   \cup   \lBrack \Gamma_{{\mathrm{2}}} \rBrack  }                                                                   \\
                                     &  \equiv   \REP{  \lBrack \Gamma_{{\mathrm{1}}} \rBrack  }   \parallel   \REP{  \lBrack \Gamma_{{\mathrm{2}}} \rBrack  }  & \text{by \cref{prop:intrp-ordered,prop:union-comm}} \\
                                     &  \equiv  \Gamma_{{\mathrm{1}}}  \parallel  \Gamma_{{\mathrm{2}}}                           & \text{by (IH1), (IH2), and \ruleref{Ctx}}           \\
                                     & = \Gamma
              \end{align*}
        \item (Case $ \Gamma \ottsym{=} \Gamma_{{\mathrm{1}}}  \ottsym{,}  \Gamma_{{\mathrm{2}}} $)
              Similar to the case above, using \cref{prop:join-comm}.
              \qedhere
    \end{itemize}
\end{lemma}

To show \cref{prop:realize/equiv}, we need another observation for the range of $\INTRP{-}$.
That is, any disconnected interpretation is always constructed by the union of other interpretations.
It is not sufficient that the interpretation is topologically ordered, as the following graph is topologically ordered and disconnected but has no union cut.
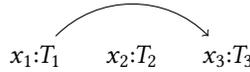
\begin{center}
    \begin{tikzpicture}[every node/.style={node distance=1em}]
        \node (A) at (1,0) {$\ottmv{x_{{\mathrm{1}}}}  \mathord:  \ottnt{T_{{\mathrm{1}}}}$};
        \node [right=of A] (C) {$\ottmv{x_{{\mathrm{2}}}}  \mathord:  \ottnt{T_{{\mathrm{2}}}}$};
        \node [right=of C] (B) {$\ottmv{x_{{\mathrm{3}}}}  \mathord:  \ottnt{T_{{\mathrm{3}}}}$};
        \draw[->] (A) to [bend left=45] (B);
    \end{tikzpicture}
\end{center}
We formalize the observation as follows.

\begin{definition}[Weakly connected]
    Given $\mathfrak{G} = (n, v, E)$, we call two vertices in $\mathfrak{G}$ indexed by $n_x$ and $n_y$ \emph{weakly connected} iff there exists a sequence $\{n_i\}_{0 \le i \le m}$ such that $n_x = n_0$, $n_y = n_m$, and $(n_i, n_{i+1}) \in E$ or $(n_{i+1}, n_i) \in E$ for any $0 \le i < m$.
    We call a graph representation itself weakly connected iff any two vertices in the graph representation are weakly connected.
    We call a graph representation \emph{weakly disconnected} iff it is not weakly connected.
\end{definition}

\begin{definition}[Union decomposition]
    Given $\mathfrak{G} = (n, v, E)$, we call a non-empty sequence of weakly connected non-null graph representations $\{\mathfrak{G}_i\}_{0 \le i \le m}$ \emph{union decomposition} of $\mathfrak{G}$ iff $\mathfrak{G} = \mathfrak{G}_0 \cup \dots \cup \mathfrak{G}_m$.
\end{definition}

\begin{lemma}\label{prop:union-decompose}
    If $ \lBrack \Gamma \rBrack $ is non-null, then there exits a sequence $\{\Gamma_i\}_{0 \le i \le m}$ such that $\{\INTRP{\Gamma_i}\}_{0 \le i \le m}$ is a union decomposition of $ \lBrack \Gamma \rBrack $.

    \proof By structural induction on $\Gamma$.
    \begin{itemize}
        \item ($ \Gamma \ottsym{=}  \cdot  $)
              In this case, $ \lBrack \Gamma \rBrack $ is null.
              So, the property is vacuously true.
        \item ($ \Gamma \ottsym{=} \ottnt{b} $)
              In this case, $ \lBrack \Gamma \rBrack $ is weakly connected.
              So, by choosing the singleton sequence $\{\Gamma\}$, we can show the goal.
              Note that $ \lBrack \Gamma \rBrack $ is non-null by the assumption.
        \item ($ \Gamma \ottsym{=} \Gamma_{{\mathrm{1}}}  \parallel  \Gamma_{{\mathrm{2}}} $)
              In this case, $  \lBrack \Gamma \rBrack  \ottsym{=}  \lBrack \Gamma_{{\mathrm{1}}} \rBrack   \cup   \lBrack \Gamma_{{\mathrm{2}}} \rBrack  $.
              If either $ \lBrack \Gamma_{{\mathrm{1}}} \rBrack $ or $ \lBrack \Gamma_{{\mathrm{2}}} \rBrack $ is null, $  \lBrack \Gamma \rBrack  \ottsym{=}  \lBrack \Gamma_{{\mathrm{2}}} \rBrack  $ or $  \lBrack \Gamma \rBrack  \ottsym{=}  \lBrack \Gamma_{{\mathrm{1}}} \rBrack  $, respectively.
              So, in either case, we have the goal directly by IH.
              Otherwise, both $ \lBrack \Gamma_{{\mathrm{1}}} \rBrack $ and $ \lBrack \Gamma_{{\mathrm{2}}} \rBrack $ is non-null.
              In that case, we have sequences $\{\Gamma^1_i\}_{0 \le i < m_1}$ and $\{\Gamma^2_j\}_{0 \le j < m_2}$ satisfying the condition of the property against $\Gamma_{{\mathrm{1}}}$ and $\Gamma_{{\mathrm{2}}}$, respectively, by IH.
              Now, we can see the concatenation of them satisfies the condition of property against $\Gamma$.
        \item ($ \Gamma \ottsym{=} \Gamma_{{\mathrm{1}}}  \ottsym{,}  \Gamma_{{\mathrm{2}}} $)
              In this case $  \lBrack \Gamma \rBrack  \ottsym{=}  \GJOIN{  \lBrack \Gamma_{{\mathrm{1}}} \rBrack  }{  \lBrack \Gamma_{{\mathrm{2}}} \rBrack  }  $.
              If either $ \lBrack \Gamma_{{\mathrm{1}}} \rBrack $ or $ \lBrack \Gamma_{{\mathrm{2}}} \rBrack $ is null, $  \lBrack \Gamma \rBrack  \ottsym{=}  \lBrack \Gamma_{{\mathrm{2}}} \rBrack  $ or $  \lBrack \Gamma \rBrack  \ottsym{=}  \lBrack \Gamma_{{\mathrm{1}}} \rBrack  $, respecitvely.
              So, in either case, we have the goal directly by IH.
              Otherwise, $ \lBrack \Gamma_{{\mathrm{1}}} \rBrack $ and $ \lBrack \Gamma_{{\mathrm{2}}} \rBrack $ is non-null.
              In that case, $  \lBrack \Gamma \rBrack  \ottsym{=}  \GJOIN{  \lBrack \Gamma_{{\mathrm{1}}} \rBrack  }{  \lBrack \Gamma_{{\mathrm{2}}} \rBrack  }  $ is weakly connected.
              So, by choosing the singleton sequence $\{\Gamma\}$, we can show the goal.
              Note that $ \lBrack \Gamma \rBrack $ is non-null by the assumption.
              \qedhere
    \end{itemize}

\end{lemma}

\begin{lemma}\label{prop:union-decompose-iso}
    Let $\{\mathfrak{G}^1_i\}_{0 \le i \le n_1}$ and $\{\mathfrak{G}^2_j\}_{0 \le j \le n_2}$ be union decompositions.
    If $\mathfrak{G}^1_0 \cup \dots \cup \mathfrak{G}^1_{n_1} \sim \mathfrak{G}^2_0 \cup \dots \cup \mathfrak{G}^2_{n_2}$, we have a one-to-one isomorphic correspondence between the elements of them; that is
    \begin{itemize}
        \item $n_1 = n_2$; and
        \item there exists a bijection $g$ between $\NAT_{<n_1}$ and $\NAT_{<n_1}$ such that $\mathfrak{G}^1_i \sim \mathfrak{G}^2_{g(i)}$ for any $0 \le i \le n_1$.
    \end{itemize}

    \proof For convenience, let $\mathfrak{G}_{{\mathrm{1}}} = \mathfrak{G}^1_0 \cup \dots \cup \mathfrak{G}^1_{n_1}$ and $\mathfrak{G}_{{\mathrm{2}}} = \mathfrak{G}^2_0 \cup \dots \cup \mathfrak{G}^2_{n_2}$.
    We show the property by mathematical induction on $n_1$.
    \begin{itemize}
        \item (Base case $n_1 = 0$)
              We consider the following two cases.
              \begin{itemize}
                  \item ($n_2 = 0$) In this sub-case, the first sub-goal $n_1 = n_2$ is obvious.
                        Regarding the second sub-goal, we know $\mathfrak{G}_{{\mathrm{1}}} = \mathfrak{G}^1_0$ and $\mathfrak{G}_{{\mathrm{2}}} = \mathfrak{G}^2_0$.
                        So, it suffices to choose $g(x) = x$, where $\mathfrak{G}^1_0 \sim \mathfrak{G}^2_0$ by the isomorphism between $\mathfrak{G}_{{\mathrm{1}}}$ and $\mathfrak{G}_{{\mathrm{2}}}$.
                  \item ($n_2 > 0$)
                        In this sub-case, we can pick distinct components $\mathfrak{G}^2_x$ and $\mathfrak{G}^2_y$, that is $x \neq y$, from $\{\mathfrak{G}^2_j\}_{0 \le j \le n_2}$.
                        Since they are non-null, there exist vertices $a$ and $b$ in $\mathfrak{G}^2_x$ and $\mathfrak{G}^2_y$, respectively.
                        We can see $a$ and $b$ are weakly disconnected in $\mathfrak{G}_{{\mathrm{2}}}$, but the isomorphism between $\mathfrak{G}_{{\mathrm{1}}}$ and $\mathfrak{G}_{{\mathrm{2}}}$ maps $a$ and $b$ into $\mathfrak{G}_{{\mathrm{1}}} = \mathfrak{G}^1_0$, which is weakly connected, that is a contradiction.
                        Therefore, this sub-case cannot happen.
              \end{itemize}
        \item (Induction step $n_1 > 0$)
              Let $f$ be the isomorphism between $\mathfrak{G}_{{\mathrm{1}}}$ and $\mathfrak{G}_{{\mathrm{2}}}$.
              We can see $f$ maps all vertices in $\mathfrak{G}^1_0$ into all vertices in some $\mathfrak{G}^2_k$ and this mapping becomes the isomorphism between $\mathfrak{G}^1_0$ and $\mathfrak{G}j^2_k$.
              Furthermore, we can see $\mathfrak{G}^1_1 \cup \dots \cup \mathfrak{G}^1_{n_1} \sim \mathfrak{G}^2_0 \cup \dots \cup \mathfrak{G}^2_{k-1} \cup \mathfrak{G}^2_{k + 1} \cup \dots \cup \mathfrak{G}^2_{n_2}$.
              So, we have $n_1 - 1 = n_2 - 1$ and one-to-one correspondence $g'$ by IH.
              We choose $g$ as follows.
              \begin{gather*}
                  g(i) = \begin{cases}
                      k             & \text{if $i = 0$}         \\
                      g'(i - 1)     & \text{if $g'(i - 1) < k$} \\
                      g'(i - 1) + 1 & \text{if $g'(i - 1) > k$}
                  \end{cases}
              \end{gather*}
              \qedhere
    \end{itemize}
\end{lemma}


We need to consider decomposition by join cuts, too.
However, it is rather well-organized than union cuts, that is no commute happens among components, and so, what we need is the following lemma.

\begin{lemma}\label{prop:join-cut-iso}
    Let $\mathfrak{G}_{{\mathrm{1}}} = (n, v_1, E_1)$ and $\mathfrak{G}_{{\mathrm{2}}} = (n, v_2, E_2)$ be topologically ordered graph representations such that $\mathfrak{G}_{{\mathrm{1}}}  \simeq  \mathfrak{G}_{{\mathrm{2}}}$.
    Then, $\ottnt{n'}$ is a join cut of $\mathfrak{G}_{{\mathrm{1}}}$ iff $\ottnt{c}$ is a join cut of $\mathfrak{G}_{{\mathrm{2}}}$.
    Furthermore, $ \mathfrak{G}_{{\mathrm{1}}} {}_{<  \ottnt{n'} }   \simeq   \mathfrak{G}_{{\mathrm{2}}} {}_{<  \ottnt{n'} } $ and $ \mathfrak{G}_{{\mathrm{1}}} {}_{\ge  \ottnt{n'} }   \simeq   \mathfrak{G}_{{\mathrm{2}}} {}_{\ge  \ottnt{n'} } $.

    \proof Since $ \simeq $ is symmetric, it suffices to show only one direction.
    Let $f$ be the isomorphism between $\mathfrak{G}_{{\mathrm{1}}}$ and $\mathfrak{G}_{{\mathrm{2}}}$, meaning that
    \begin{gather}
        \text{$f$ is bijective} \tag{H1}\\
        v_1 \circ f = v_2 \tag{H2}\\
        E_1 = \{ (f(x), f(y)) \mid (x,y) \in E_2 \} \tag{H3}
    \end{gather}
    Supposing $\ottnt{n'}$ is a join cut of $\mathfrak{G}_{{\mathrm{2}}}$, we have
    \begin{gather*}
        0 < n' < n \tag{H4}\\
        (x, y) \in E_2 \text{ for any } x < n' \text{ and } n' \le y \tag{H5}
    \end{gather*}
    by definition.
    Now we will show $f(x) < n'$ for any $x < n'$ by contradiction, that implies $n'$ is a join cut of $\mathfrak{G}_{{\mathrm{1}}}$ by (H4) and (H5).
    If $f(x) \ge n'$, there exists $y$ such that $n' \le y$ and $f(y) < n'$ since $f$ is bijective.
    However, this contradicts the fact that $\mathfrak{G}_{{\mathrm{1}}}$ is topologically ordered, as $(f(x), f(y)) \in E_1$ by (H3) and (H5), while $f(y) < f(x)$.
    Therefore, $n'$ is a join cut of $\mathfrak{G}_{{\mathrm{1}}}$, too.
    As a side-result, we can see $f\vert_{\NAT_{<n'}}$ becomse a isomorphism between $ \mathfrak{G}_{{\mathrm{1}}} {}_{<  \ottnt{n'} } $ and $ \mathfrak{G}_{{\mathrm{2}}} {}_{<  \ottnt{n'} } $.
    Similarly, we can have a isomorphism between $ \mathfrak{G}_{{\mathrm{1}}} {}_{\ge  \ottnt{n'} } $ and $ \mathfrak{G}_{{\mathrm{2}}} {}_{\ge  \ottnt{n'} } $ from $f$. \qedhere
\end{lemma}



\begin{lemma}[\cref{prop:realize/equiv}]\label{prop:realize/equiv/lemma}
    If $ \lBrack \Gamma_{{\mathrm{1}}} \rBrack   \simeq   \lBrack \Gamma_{{\mathrm{2}}} \rBrack $, then $ \REP{  \lBrack \Gamma_{{\mathrm{1}}} \rBrack  }   \equiv   \REP{  \lBrack \Gamma_{{\mathrm{2}}} \rBrack  } $.
    \proof
    Suppose $ \lBrack \Gamma_{{\mathrm{1}}} \rBrack  = (n_1, v_1, E_1)$, $ \lBrack \Gamma_{{\mathrm{2}}} \rBrack  = (n_2, v_2, E_2)$, and let $f$ be the isomorphism between $ \lBrack \Gamma_{{\mathrm{1}}} \rBrack $ and $ \lBrack \Gamma_{{\mathrm{2}}} \rBrack $.
    Note that $n_1 = n_2$ since $ \lBrack \Gamma_{{\mathrm{1}}} \rBrack $ and $ \lBrack \Gamma_{{\mathrm{2}}} \rBrack $ are isomorphic.
    By \cref{prop:realize/correct/lemma}, $ \REP{  \lBrack \Gamma_{{\mathrm{1}}} \rBrack  } $ and $ \REP{  \lBrack \Gamma_{{\mathrm{2}}} \rBrack  } $ are defined, and $ \REP{  \lBrack \Gamma_{{\mathrm{1}}} \rBrack  }   \equiv  \Gamma_{{\mathrm{1}}}$, $ \REP{  \lBrack \Gamma_{{\mathrm{2}}} \rBrack  }   \equiv  \Gamma_{{\mathrm{2}}}$.
    We prove the proposition by complete induction on $n_1$.
    We consider whether $n_1$ is 0, 1, or more than 1.
    \begin{itemize}
        \item ($n_1 = 0$ or $n_1 = 1$)
              In these cases, $  \lBrack \Gamma_{{\mathrm{1}}} \rBrack  \ottsym{=}  \lBrack \Gamma_{{\mathrm{2}}} \rBrack  $.
              So, the goal is immediate by \ruleref{Refl}.
        \item ($n_1 > 1$)
              We have some sequences $\{\Gamma^1_i\}_{0\le i\le m_1}$ and $\{\Gamma^2_j\}_{0\le j\le m_2}$ satisfying the following condition by \cref{prop:union-decompose}.
              \begin{gather*}
                  \text{Every $\INTRP{\Gamma^1_i}$ is non-null} \\
                  \text{Every $\INTRP{\Gamma^2_j}$ is non-null} \\
                   \lBrack \Gamma_{{\mathrm{1}}} \rBrack  = \INTRP{\Gamma^1_0} \cup \dots \cup \INTRP{\mathfrak{G}^1_{m_1}} \\
                   \lBrack \Gamma_{{\mathrm{2}}} \rBrack  = \INTRP{\Gamma^2_0} \cup \dots \cup \INTRP{\mathfrak{G}^2_{m_2}}
              \end{gather*}
              Furthermore, there is a one-to-one correspondence $g$ between the elements of them by \cref{prop:union-decompose-iso} as follows.
              \begin{gather*}
                  m_1 = m_2 \\
                  \INTRP{\Gamma^1_i} \sim \INTRP{\Gamma^2_{g(i)}} \text{ for any $0 \le i \le m_1$}
              \end{gather*}
              We also have the following relation by using \cref{prop:union-comm} (rigorously we need mathematical induction), where we suppose $ \parallel $ is right-associative, not important though.
              \begin{gather*}
                   \REP{  \lBrack \Gamma_{{\mathrm{1}}} \rBrack  }   \equiv  \REP{\INTRP{\Gamma^1_0}}  \parallel  \dots  \parallel  \REP{\INTRP{\Gamma^1_{m_1}}} \\
                   \REP{  \lBrack \Gamma_{{\mathrm{2}}} \rBrack  }   \equiv  \REP{\INTRP{\Gamma^2_0}}  \parallel  \dots  \parallel  \REP{\INTRP{\Gamma^2_{m_2}}}
              \end{gather*}
              Now we consider the following two sub-cases.
              \begin{itemize}
                  \item ($m_1 = 0$)
                        In this sub-case, $ \lBrack \Gamma_{{\mathrm{1}}} \rBrack $ and $ \lBrack \Gamma_{{\mathrm{2}}} \rBrack $ is weakly connected.
                        So, there exists the minimum join cut $n'_1$ of $ \lBrack \Gamma_{{\mathrm{1}}} \rBrack $ for which
                        \begin{gather*}
                              \REP{  \lBrack \Gamma_{{\mathrm{1}}} \rBrack  }  \ottsym{=}  \REP{   \lBrack \Gamma_{{\mathrm{1}}} \rBrack  {}_{<  \ottnt{n'_{{\mathrm{1}}}} }  }   \ottsym{,}   \REP{   \lBrack \Gamma_{{\mathrm{1}}} \rBrack  {}_{\ge  \ottnt{n'_{{\mathrm{1}}}} }  }  .
                        \end{gather*}
                        Similarly, there exists the minimum join cut $n'_2$ of $ \lBrack \Gamma_{{\mathrm{2}}} \rBrack $ for which
                        \begin{gather*}
                              \REP{  \lBrack \Gamma_{{\mathrm{2}}} \rBrack  }  \ottsym{=}  \REP{   \lBrack \Gamma_{{\mathrm{2}}} \rBrack  {}_{<  \ottnt{n'_{{\mathrm{2}}}} }  }   \ottsym{,}   \REP{   \lBrack \Gamma_{{\mathrm{2}}} \rBrack  {}_{\ge  \ottnt{n'_{{\mathrm{2}}}} }  }  .
                        \end{gather*}
                        Here, we can see
                        \begin{gather*}
                            n'_1 = n'_2 \\
                              \lBrack \Gamma_{{\mathrm{1}}} \rBrack  {}_{<  \ottnt{n'_{{\mathrm{1}}}} }   \simeq    \lBrack \Gamma_{{\mathrm{2}}} \rBrack  {}_{<  \ottnt{n'_{{\mathrm{2}}}} }  \\
                              \lBrack \Gamma_{{\mathrm{1}}} \rBrack  {}_{\ge  \ottnt{n'_{{\mathrm{1}}}} }   \simeq    \lBrack \Gamma_{{\mathrm{2}}} \rBrack  {}_{\ge  \ottnt{n'_{{\mathrm{2}}}} } 
                        \end{gather*}
                        by \cref{prop:join-cut-iso}.
                        That because $n'_1$ and $n'_2$ is a join cut of $ \lBrack \Gamma_{{\mathrm{2}}} \rBrack $ and $ \lBrack \Gamma_{{\mathrm{1}}} \rBrack $, respectively, too, by the lemma, so $n'_1 \neq n'_2$ contradicts minimality of them.
                        Since we know interpretation is closed under index cuts by \cref{prop:cut-closed-left,prop:cut-closed-right}, now we have
                        \begin{gather*}
                             \REP{   \lBrack \Gamma_{{\mathrm{1}}} \rBrack  {}_{<  \ottnt{n'_{{\mathrm{1}}}} }  }   \equiv   \REP{   \lBrack \Gamma_{{\mathrm{2}}} \rBrack  {}_{<  \ottnt{n'_{{\mathrm{2}}}} }  }  \\
                             \REP{   \lBrack \Gamma_{{\mathrm{1}}} \rBrack  {}_{\ge  \ottnt{n'_{{\mathrm{1}}}} }  }   \equiv   \REP{   \lBrack \Gamma_{{\mathrm{2}}} \rBrack  {}_{\ge  \ottnt{n'_{{\mathrm{2}}}} }  } 
                        \end{gather*}
                        by IH.
                        Using them, we can have the goal $ \REP{  \lBrack \Gamma_{{\mathrm{1}}} \rBrack  }   \equiv   \REP{  \lBrack \Gamma_{{\mathrm{2}}} \rBrack  } $ by \ruleref{Ctx}.
                  \item ($m_1 > 0$)
                        In this sub-case, we can see that the order of $\INTRP{\Gamma^1_i}$ is less than $n_1$ for any $0 \le i \le m_1$.
                        So, we have $\REP{\INTRP{\Gamma^1_i}}  \equiv  \REP{\INTRP{\Gamma^2_{g(i)}}}$ for any $0 \le i < m_1$ by IH.
                        So, we can have the goal by \ruleref{PAssoc}, \ruleref{PComm}, and \ruleref{Ctx}. \qedhere
              \end{itemize}
    \end{itemize}
\end{lemma}

  \clearpage
  \section{Detailed Proofs for Lemmas in Section~\ref{sec:language}}
  \hypstricttrue
  \proplabeltrue


\begin{prop}{subst/void}
    If $ \ottmv{x}  \notin   \ottkw{fv} ( \ottnt{M} )  $, then $ \ottnt{M}  \ottsym{[}  \ottnt{V}  \ottsym{/}  \ottmv{x}  \ottsym{]} \ottsym{=} \ottnt{M} $.
\end{prop}

  \subsection{Properties for graph}
  \begin{prop}{graph/iso}
    Let $\mathfrak{G}_{{\mathrm{0}}} = (0, \emptyset, \emptyset)$.
    \begin{statements}
        \item $  \GJOIN{ \mathfrak{G} }{ \mathfrak{G}_{{\mathrm{0}}} }  \ottsym{=} \mathfrak{G} $.
        \item $  \GJOIN{ \mathfrak{G}_{{\mathrm{0}}} }{ \mathfrak{G} }  \ottsym{=} \mathfrak{G} $.
        \item $ \mathfrak{G}  \cup  \mathfrak{G}_{{\mathrm{0}}} \ottsym{=} \mathfrak{G} $.
        \item $ \mathfrak{G}_{{\mathrm{0}}}  \cup  \mathfrak{G} \ottsym{=} \mathfrak{G} $.
        \item $\mathfrak{G}_{{\mathrm{1}}}  \cup  \mathfrak{G}_{{\mathrm{2}}}  \simeq  \mathfrak{G}_{{\mathrm{2}}}  \cup  \mathfrak{G}_{{\mathrm{1}}}$.
        \item $  \GJOIN{ \mathfrak{G}_{{\mathrm{1}}} }{ \ottsym{(}   \GJOIN{ \mathfrak{G}_{{\mathrm{2}}} }{ \mathfrak{G}_{{\mathrm{3}}} }   \ottsym{)} }  \ottsym{=}  \GJOIN{ \ottsym{(}   \GJOIN{ \mathfrak{G}_{{\mathrm{1}}} }{ \mathfrak{G}_{{\mathrm{2}}} }   \ottsym{)} }{ \mathfrak{G}_{{\mathrm{3}}} }  $.
        \item $ \mathfrak{G}_{{\mathrm{1}}}  \cup  \ottsym{(}  \mathfrak{G}_{{\mathrm{2}}}  \cup  \mathfrak{G}_{{\mathrm{3}}}  \ottsym{)} \ottsym{=} \ottsym{(}  \mathfrak{G}_{{\mathrm{1}}}  \cup  \mathfrak{G}_{{\mathrm{2}}}  \ottsym{)}  \cup  \mathfrak{G}_{{\mathrm{3}}} $.
    \end{statements}

    \proof By definition except the fifth one.
    Regarding the fifth one, let $(n_1, v_1, E_1) = \mathfrak{G}_{{\mathrm{1}}}$ and $(n_2, v_2, E_2) = \mathfrak{G}_{{\mathrm{2}}}$.
    We can choose $f:\NAT_{<n_1 + n_2}$ such that $f(i) = i + n_2$ for $0 \le i < n_1$ and $f(i) = i - n_1$ for $n_1 \le i < n_1 + n_2$ as an witness bijection.
\end{prop}

\begin{prop}{graph/span}
    \noindent
    \begin{statements}
        \item $\mathfrak{G}_{{\mathrm{1}}}  \cup  \mathfrak{G}_{{\mathrm{2}}}  \mathrel{<_\rightarrow}   \GJOIN{ \mathfrak{G}_{{\mathrm{1}}} }{ \mathfrak{G}_{{\mathrm{2}}} } $.
        \item $\mathfrak{G}_{{\mathrm{1}}}  \cup  \ottsym{(}   \GJOIN{ \mathfrak{G}_{{\mathrm{2}}} }{ \mathfrak{G}_{{\mathrm{3}}} }   \ottsym{)}  \mathrel{<_\rightarrow}   \GJOIN{ \ottsym{(}  \mathfrak{G}_{{\mathrm{1}}}  \cup  \mathfrak{G}_{{\mathrm{2}}}  \ottsym{)} }{ \mathfrak{G}_{{\mathrm{3}}} } $.
        \item $\ottsym{(}   \GJOIN{ \mathfrak{G}_{{\mathrm{1}}} }{ \mathfrak{G}_{{\mathrm{2}}} }   \ottsym{)}  \cup  \mathfrak{G}_{{\mathrm{3}}}  \mathrel{<_\rightarrow}   \GJOIN{ \mathfrak{G}_{{\mathrm{1}}} }{ \ottsym{(}  \mathfrak{G}_{{\mathrm{2}}}  \cup  \mathfrak{G}_{{\mathrm{3}}}  \ottsym{)} } $.
    \end{statements}

    \proof By definition.
\end{prop}

\begin{prop}{graph/iso/num/bindings}
    If $(n_1, v_1, E_1) \sim (n_2, v_2, E_2)$, then $|\{\,n' \in \NAT_{<n_1} \mid v_1(n') = \ottnt{b}\,\}| = |\{\,n' \in \NAT_{<n_2} \mid v_2(n') = \ottnt{b}\,\}|$ for any $\ottnt{b}$.

    \proof
    By definition, $n_1 = n_2$, and we have some bijection $f$ between $\NAT_{<n_1}$ and $\NAT_{<n_1}$ such that
    \begin{gather}
        v_2 = v_1 \circ f.
    \end{gather}
    So, the right-hand side of the goal becomes to
    \begin{gather*}
        |\{\,n' \in \NAT_{<n_1} \mid v_1(f(n')) = \ottnt{b}\,\}|,
    \end{gather*}
    which equals the left-hand side of the goal since $f$ is a bijection.
\end{prop}

\begin{prop}{graph/span/num/bindings}
    If $(n_1, v_1, E_1)  \mathrel{<_\rightarrow}  (n_2, v_2, E_2)$, then $|\{\,n' \in \NAT_{<n_1} \mid v_1(n') = \ottnt{b}\,\}| = |\{\,n' \in \NAT_{<n_2} \mid v_2(n') = \ottnt{b}\,\}|$ for any $\ottnt{b}$.

    \proof By definition, $n_1 = n_2$ and $v_1 = v_2$. So, the goal is obvious.
\end{prop}

\begin{prop}{graph/iso/refl}
    $\mathfrak{G}  \simeq  \mathfrak{G}$.

    \proof We can choose the identity function as the witness bijection for the isomorphism.
\end{prop}

\begin{prop}{graph/iso/join}
    If $\mathfrak{G}_{{\mathrm{11}}}  \simeq  \mathfrak{G}_{{\mathrm{12}}}$ and $\mathfrak{G}_{{\mathrm{21}}}  \simeq  \mathfrak{G}_{{\mathrm{22}}}$, then $ \GJOIN{ \mathfrak{G}_{{\mathrm{11}}} }{ \mathfrak{G}_{{\mathrm{21}}} }   \simeq   \GJOIN{ \mathfrak{G}_{{\mathrm{12}}} }{ \mathfrak{G}_{{\mathrm{22}}} } $.

    \proof Let $f_1$ and $f_2$ be the first and second isomorphism, respectively; and $n_1 =  | \mathfrak{G}_{{\mathrm{11}}} |_{\bullet} $ and $n_2 =  | \mathfrak{G}_{{\mathrm{21}}} |_{\bullet} $.
    Then, the following $f$ becomes the isomorphism for the goal.
    \begin{align*}
        f(n) = \begin{cases}
            f_1(n) & \text{if $0 \le n < n_1$} \\
            f_2(n - n_1) & \text{if $n_1 \le n < n_1 + n_2$}
        \end{cases}
    \end{align*}
\end{prop}

\begin{prop}{graph/iso/union}
    If $\mathfrak{G}_{{\mathrm{11}}}  \simeq  \mathfrak{G}_{{\mathrm{12}}}$ and $\mathfrak{G}_{{\mathrm{21}}}  \simeq  \mathfrak{G}_{{\mathrm{22}}}$, then $\mathfrak{G}_{{\mathrm{11}}}  \cup  \mathfrak{G}_{{\mathrm{21}}}  \simeq  \mathfrak{G}_{{\mathrm{12}}}  \cup  \mathfrak{G}_{{\mathrm{22}}}$.

    \proof Similar to the proof of \propref{graph/iso/join}.
\end{prop}

\begin{prop}{graph/span/refl}
    $\mathfrak{G}  \mathrel{<_\rightarrow}  \mathfrak{G}$

    \proof By definition.
\end{prop}

\begin{prop}{graph/span/join}
    If $\mathfrak{G}_{{\mathrm{11}}}  \mathrel{<_\rightarrow}  \mathfrak{G}_{{\mathrm{12}}}$ and $\mathfrak{G}_{{\mathrm{21}}}  \mathrel{<_\rightarrow}  \mathfrak{G}_{{\mathrm{22}}}$, then $ \GJOIN{ \mathfrak{G}_{{\mathrm{11}}} }{ \mathfrak{G}_{{\mathrm{21}}} }   \mathrel{<_\rightarrow}   \GJOIN{ \mathfrak{G}_{{\mathrm{12}}} }{ \mathfrak{G}_{{\mathrm{22}}} } $.

    \proof This follows by \propref{graph/iso/join} and the definition of the spanning relation.
\end{prop}

\begin{prop}{graph/span/union}
    If $\mathfrak{G}_{{\mathrm{11}}}  \mathrel{<_\rightarrow}  \mathfrak{G}_{{\mathrm{12}}}$ and $\mathfrak{G}_{{\mathrm{21}}}  \mathrel{<_\rightarrow}  \mathfrak{G}_{{\mathrm{22}}}$, then $\mathfrak{G}_{{\mathrm{11}}}  \cup  \mathfrak{G}_{{\mathrm{21}}}  \mathrel{<_\rightarrow}  \mathfrak{G}_{{\mathrm{12}}}  \cup  \mathfrak{G}_{{\mathrm{22}}}$.

    \proof Similar to the proof of \propref{graph/span/join}.
\end{prop}

\begin{prop}[type=definition]{graph/replace/vert}
    We write $\mathfrak{G}_{{\mathrm{1}}} [ \mathfrak{G}_{{\mathrm{2}}} / i ]$ to express the replacement of vertex $i$ in $\mathfrak{G}_{{\mathrm{1}}}$ by $\mathfrak{G}_{{\mathrm{2}}}$.
    Let $\mathfrak{G}_{{\mathrm{1}}} [ \mathfrak{G}_{{\mathrm{2}}} / i ] = (n, v, E)$, $\mathfrak{G}_{{\mathrm{1}}} = (n_1, v_1, E_1)$, and $\mathfrak{G}_{{\mathrm{2}}} = (n_2, v_2, E_2)$.
    The replacement is defined as follows.
    \begin{align*}
        n    & = n_1 + n_2 - 1                                                           \\
        v(n) & = \begin{cases}
                     v_1(n)           & (0 \le n < i)                   \\
                     v_2(n - i)       & (i \le n < n_2 + i)             \\
                     v_1(n - n_2 + 1) & (n_2 + i \le n < n_1 + n_2 - 1)
                 \end{cases}                      \\
        E    & = \begin{aligned}
                      & \{\,(f(n), f(m)) \mid n \neq i, m \neq i, (n, m) \in E_1\,\} \cup {} \\
                      & \{\,(f(n), m) \mid i \le m < n_2 + i, (n, i) \in E_1\,\} \cup {}     \\
                      & \{\,(n, f(m)) \mid i \le n < n_2 + i, (i, m) \in E_1\,\} \cup {}     \\
                      & \{\,(n + i, m + i) \mid (n, m) \in E_2\,\}
                 \end{aligned} & \text{where } f(n) = \begin{cases}
                                                          n           & (0 \le n < i) \\
                                                          n + n_2 - 1 & (i < n < n_1)
                                                      \end{cases}
    \end{align*}
    Note that $f$ is undefined for $i$, but it makes no problem because $(i, i) \notin E$ for an acyclic graph.
\end{prop}

\begin{prop}[type=definition]{graph/replace}
    We write $\mathfrak{G}_{{\mathrm{1}}}  \ottsym{[}  \mathfrak{G}_{{\mathrm{2}}}  \ottsym{/}  \ottnt{b}  \ottsym{]}$ to express the replacement of all vertices labeled $\ottnt{b}$ in $\mathfrak{G}_{{\mathrm{1}}}$ by $\mathfrak{G}_{{\mathrm{2}}}$.
    Let $\mathfrak{G}_{{\mathrm{1}}} = (n, v, E)$ and $S = \{\, i \mid v(i) = \ottnt{b}\,\}$.
    We define the replacement $\mathfrak{G}_{{\mathrm{1}}}  \ottsym{[}  \mathfrak{G}_{{\mathrm{2}}}  \ottsym{/}  \ottnt{b}  \ottsym{]}$ as $T(\mathfrak{G}_{{\mathrm{1}}}, \mathfrak{G}_{{\mathrm{2}}}, S)$ defined as follows.
    \begin{align}
        T(\mathfrak{G}_{{\mathrm{1}}}, \mathfrak{G}_{{\mathrm{2}}}, S) = \begin{cases}
                                     \mathfrak{G}_{{\mathrm{1}}}                                                & |S| = 0             \\
                                     T(\mathfrak{G}_{{\mathrm{1}}} [ \mathfrak{G}_{{\mathrm{2}}} / i ], \mathfrak{G}_{{\mathrm{2}}}, S \setminus \{i\}) & i = \mathrm{max}(S)
                                 \end{cases}
    \end{align}
\end{prop}

\begin{prop}{graph/replace/vert/join}
    Let $\mathfrak{G}_{{\mathrm{1}}} = (n_1, v_1, E_1)$ and $\mathfrak{G}_{{\mathrm{2}}} = (n_2, v_2, E_2)$.
    Then,
    \begin{gather*}
        \ottsym{(}   \GJOIN{ \mathfrak{G}_{{\mathrm{1}}} }{ \mathfrak{G}_{{\mathrm{2}}} }   \ottsym{)} [ \mathfrak{G} / i ] =
        \begin{cases}
            \mathfrak{G}_{{\mathrm{1}}} [ \mathfrak{G} / i ] + \mathfrak{G}_{{\mathrm{2}}} & (0 \le i < n_1)          \\
            \mathfrak{G}_{{\mathrm{1}}} + \mathfrak{G}_{{\mathrm{2}}} [ \mathfrak{G} / i ] & (n_1 \le i < n_1 + n_2).
        \end{cases}
    \end{gather*}

    \proof We can simply show the proposition by definition, we need to carefully examine intervals, though.
\end{prop}

\begin{prop}{graph/replace/vert/union}
    Let $\mathfrak{G}_{{\mathrm{1}}} = (n_1, v_1, E_1)$ and $\mathfrak{G}_{{\mathrm{2}}} = (n_2, v_2, E_2)$.
    Then,
    \begin{gather*}
        \ottsym{(}  \mathfrak{G}_{{\mathrm{1}}}  \cup  \mathfrak{G}_{{\mathrm{2}}}  \ottsym{)} [ \mathfrak{G} / i ] =
        \begin{cases}
            \mathfrak{G}_{{\mathrm{1}}} [ \mathfrak{G} / i ] \cup \mathfrak{G}_{{\mathrm{2}}} & (0 \le i < n_1)          \\
            \mathfrak{G}_{{\mathrm{1}}} \cup \mathfrak{G}_{{\mathrm{2}}} [ \mathfrak{G} / i ] & (n_1 \le i < n_1 + n_2).
        \end{cases}
    \end{gather*}

    \proof Similar to the proof of \propref{graph/replace/vert/join}.
\end{prop}

\begin{prop}{graph/replace/join}
    $ \ottsym{(}   \GJOIN{ \mathfrak{G}_{{\mathrm{1}}} }{ \mathfrak{G}_{{\mathrm{2}}} }   \ottsym{)}  \ottsym{[}  \mathfrak{G}  \ottsym{/}  \ottnt{b}  \ottsym{]} \ottsym{=}  \GJOIN{ \mathfrak{G}_{{\mathrm{1}}}  \ottsym{[}  \mathfrak{G}  \ottsym{/}  \ottnt{b}  \ottsym{]} }{ \mathfrak{G}_{{\mathrm{2}}} }   \ottsym{[}  \mathfrak{G}  \ottsym{/}  \ottnt{b}  \ottsym{]} $.

    \proof Let $\mathfrak{G}_{{\mathrm{1}}} = (n_1, v_1, E_1)$, $\mathfrak{G}_{{\mathrm{2}}} = (n_2, v_2, E_2)$, $ \GJOIN{ \mathfrak{G}_{{\mathrm{1}}} }{ \mathfrak{G}_{{\mathrm{2}}} }  = (n, v, E)$, $S_1 = \{\, i \mid v_1(i) = \ottnt{b}\,\}$, $S_2 = \{\, i \mid v_2(i) = \ottnt{b}\,\}$, and $S = \{\, i \mid v(i) = \ottnt{b}\,\}$.
    We show the goal by induction on the size of $S$.
    \begin{match}
        \item[Base case]
        In this case, $|S| = 0$.
        Therefore, we can see $|S_1| = |S_2| = 0$ by the definition of $v$.
        So, both sides become $ \GJOIN{ \mathfrak{G}_{{\mathrm{1}}} }{ \mathfrak{G}_{{\mathrm{2}}} } $ by definition.

        \item[Inductive step]
        In this case, we have some $i = \mathrm{max}(S)$.
        So,
        \begin{gather*}
            \ottsym{(}   \GJOIN{ \mathfrak{G}_{{\mathrm{1}}} }{ \mathfrak{G}_{{\mathrm{2}}} }   \ottsym{)}  \ottsym{[}  \mathfrak{G}  \ottsym{/}  \ottnt{b}  \ottsym{]} = T(\ottsym{(}   \GJOIN{ \mathfrak{G}_{{\mathrm{1}}} }{ \mathfrak{G}_{{\mathrm{2}}} }   \ottsym{)} [ \mathfrak{G} / i ], \mathfrak{G}, S \setminus \{i\})
        \end{gather*}
        by definition.
        We consider the following cases by \propref{graph/replace/vert/join}.
        \begin{match}
            \item[$0 \le i < n_1$]
            In this case,
            \begin{gather*}
                \ottsym{(}   \GJOIN{ \mathfrak{G}_{{\mathrm{1}}} }{ \mathfrak{G}_{{\mathrm{2}}} }   \ottsym{)} [ \mathfrak{G} / i ] = \mathfrak{G}_{{\mathrm{1}}} [ \mathfrak{G} / i ] + \mathfrak{G}_{{\mathrm{2}}}.
            \end{gather*}
            So, we have
            \begin{gather*}
                T(\ottsym{(}   \GJOIN{ \mathfrak{G}_{{\mathrm{1}}} }{ \mathfrak{G}_{{\mathrm{2}}} }   \ottsym{)} [ \mathfrak{G} / i ], \mathfrak{G}, S \setminus \{i\}) = T(\mathfrak{G}_{{\mathrm{1}}} [ \mathfrak{G} / i ], \mathfrak{G}, S_1 \setminus \{i\}) + T(\mathfrak{G}_{{\mathrm{2}}}, \mathfrak{G}, S_2)
            \end{gather*}
            by the induction hypohtesis.
            Since we have $S_1 \subseteq S$ and $i = \mathrm{max}(S)$, we can have $i = \mathrm{max}(S_1)$.
            Consequently,
            \begin{gather*}
                T(\mathfrak{G}_{{\mathrm{1}}} [ \mathfrak{G} / i ], \mathfrak{G}, S_1 \setminus \{i\}) = T(\mathfrak{G}_{{\mathrm{1}}}, \mathfrak{G}, S_1)
            \end{gather*}
            follows by definition.

            \item[$n_1 \le i < n_1 + n_2$]
            Similar to the case above.
        \end{match}
    \end{match}
\end{prop}

\begin{prop}{graph/replace/union}
    $ \ottsym{(}  \mathfrak{G}_{{\mathrm{1}}}  \cup  \mathfrak{G}_{{\mathrm{2}}}  \ottsym{)}  \ottsym{[}  \mathfrak{G}  \ottsym{/}  \ottnt{b}  \ottsym{]} \ottsym{=} \mathfrak{G}_{{\mathrm{1}}}  \ottsym{[}  \mathfrak{G}  \ottsym{/}  \ottnt{b}  \ottsym{]}  \cup  \mathfrak{G}_{{\mathrm{2}}}  \ottsym{[}  \mathfrak{G}  \ottsym{/}  \ottnt{b}  \ottsym{]} $.

    \proof Similar to the proof of \propref{graph/replace/join}.
\end{prop}

\begin{prop}{graph/replace/iso}
    If $\mathfrak{G}_{{\mathrm{1}}}  \simeq  \mathfrak{G}_{{\mathrm{2}}}$, then $\mathfrak{G}_{{\mathrm{1}}}  \ottsym{[}  \mathfrak{G}  \ottsym{/}  \ottnt{b}  \ottsym{]}  \simeq  \mathfrak{G}_{{\mathrm{2}}}  \ottsym{[}  \mathfrak{G}  \ottsym{/}  \ottnt{b}  \ottsym{]}$.

    \proof Let $(n_1, v_1, E_1) = \mathfrak{G}_{{\mathrm{1}}}$, $(n_2, v_2, E_2) = \mathfrak{G}_{{\mathrm{2}}}$, $S_1 = \{i \mid v_1(i) = b\}$, and $S_2 = \{i \mid v_2(i) = b\}$.
    Note that $n_1 = n_2$ and $|S_1| = |S_2|$ by the assumption.
    We show the goal by size induction of $S_1$.
\end{prop}

\begin{prop}{graph/replace/span}
    If $\mathfrak{G}_{{\mathrm{1}}}  \mathrel{<_\rightarrow}  \mathfrak{G}_{{\mathrm{2}}}$, then $\mathfrak{G}_{{\mathrm{1}}}  \ottsym{[}  \mathfrak{G}  \ottsym{/}  \ottnt{b}  \ottsym{]}  \mathrel{<_\rightarrow}  \mathfrak{G}_{{\mathrm{2}}}  \ottsym{[}  \mathfrak{G}  \ottsym{/}  \ottnt{b}  \ottsym{]}$.

    \proof This follows \propref{graph/replace/iso} and the definition of spanning.
\end{prop}

\begin{prop}{graph/ordering/exists}
    For any graph representation, a topological ordering exists.

    \proof 
    This is a well-known property for DAGs.
    We could employ this property as the acyclic condition.
\end{prop}

\begin{prop}[type=definition]{graph/ordering/concat}
    Let $f_1$ and $f_2$ be topological orderings of graph representations of order $n_1$ and $n_2$.
    Then, the \emph{join} of $f_1$ and $f_2$, denoted by $f_1 + f_2$, is the function defined as follows.
    \begin{align*}
        (f_1 + f_2)(n) & = \begin{cases}
                               f_1(n)             & (0 \le n < n_1)         \\
                               f_2(n - n_1) + n_1 & (n_1 \le n < n_1 + n_2)
                           \end{cases}
    \end{align*}
\end{prop}

\begin{prop}{graph/ordering/join}
    If $f_1$ and $f_2$ are topological orderings of $\mathfrak{G}_{{\mathrm{1}}}$ and $\mathfrak{G}_{{\mathrm{2}}}$, respectively, then $f_1 + f_2$ is a topological ordering of $ \GJOIN{ \mathfrak{G}_{{\mathrm{1}}} }{ \mathfrak{G}_{{\mathrm{2}}} } $.

    \proof By definition.
\end{prop}

\begin{prop}{graph/ordering/join/inv}
    If $f$ is a topological ordering of $ \GJOIN{ \mathfrak{G}_{{\mathrm{1}}} }{ \mathfrak{G}_{{\mathrm{2}}} } $, then there exist topological orderings $f_1$ and $f_2$ of $\mathfrak{G}_{{\mathrm{1}}}$ and $\mathfrak{G}_{{\mathrm{2}}}$, respectively, satisfying $f = f_1 + f_2$.

    \proof We choose $f_1$ and $f_2$ as follows.
    \begin{align*}
        f_1(n) & = f(n)                                   & (0 \le n <  | \mathfrak{G}_{{\mathrm{1}}} |_{\bullet} ) \\
        f_2(n) & = f(n +  | \mathfrak{G}_{{\mathrm{1}}} |_{\bullet} ) -  | \mathfrak{G}_{{\mathrm{1}}} |_{\bullet}  & (0 \le n <  | \mathfrak{G}_{{\mathrm{2}}} |_{\bullet} )
    \end{align*}
    The conditions follow by the definition.
\end{prop}

\begin{prop}{graph/ordering/union}
    If $f_1$ and $f_2$ are topological orderings of $\mathfrak{G}_{{\mathrm{1}}}$ and $\mathfrak{G}_{{\mathrm{2}}}$, respectively, then both
    \begin{align*}
        f(n)  & = \begin{cases}
                      f_1(n)                                   & (0 \le n <  | \mathfrak{G}_{{\mathrm{1}}} |_{\bullet} )                    \\
                      f_2(n -  | \mathfrak{G}_{{\mathrm{1}}} |_{\bullet} ) +  | \mathfrak{G}_{{\mathrm{1}}} |_{\bullet}  & ( | \mathfrak{G}_{{\mathrm{1}}} |_{\bullet}  \le n <  |  \GJOIN{ \mathfrak{G}_{{\mathrm{1}}} }{ \mathfrak{G}_{{\mathrm{2}}} }  |_{\bullet} )
                  \end{cases} \\
        f'(n) & = \begin{cases}
                      f_2(n) +  | \mathfrak{G}_{{\mathrm{1}}} |_{\bullet}  & (0 \le n <  | \mathfrak{G}_{{\mathrm{2}}} |_{\bullet} )                    \\
                      f_1(n -  | \mathfrak{G}_{{\mathrm{2}}} |_{\bullet} ) & ( | \mathfrak{G}_{{\mathrm{2}}} |_{\bullet}  \le n <  |  \GJOIN{ \mathfrak{G}_{{\mathrm{1}}} }{ \mathfrak{G}_{{\mathrm{2}}} }  |_{\bullet} )
                  \end{cases}
    \end{align*}
    are topological orderings of $\mathfrak{G}_{{\mathrm{1}}}  \cup  \mathfrak{G}_{{\mathrm{2}}}$.

    \proof By definition.
\end{prop}

\begin{prop}{graph/ordering/proj/iso}
    Suppose $\mathfrak{G}_{{\mathrm{1}}} = (n, v_1, E_1)$ and $\mathfrak{G}_{{\mathrm{2}}} = (n, v_2, E_2)$.
    If $\mathfrak{G}_{{\mathrm{1}}}  \simeq  \mathfrak{G}_{{\mathrm{2}}}$ and $f_1$ is a topological ordering of $\mathfrak{G}_{{\mathrm{1}}}$, then there exists a topological ordering $f_2$ of $\mathfrak{G}_{{\mathrm{2}}}$ such that $v_1 \circ f_1 = v_2 \circ f_2$.

    \proof We have some bijection $g$ from $\NAT_{<n}$ to $\NAT_{<n}$ such that $v_2 = v_1 \circ g$ since $\mathfrak{G}_{{\mathrm{1}}}  \simeq  \mathfrak{G}_{{\mathrm{2}}}$.
    To this end, we choose $f_2 = g^{-1} \circ f_1$.
    This is a bijection from $\NAT_{<n}$ to $\NAT_{<n}$, and for any $(n_1, n_2) \in E_2$, $f_2^{-1}(n_1) < f_2^{-1}(n_2)$ since $f_1^{-1}(g(n_1)) < f_1^{-1}(g(n_2))$.
    So, $f_2$ is a topological ordering of $\mathfrak{G}_{{\mathrm{2}}}$.
    Lastly, we have $v_2 \circ g^{-1} = v_1$ from $v_2 = v_1 \circ g$, and then $v_2 \circ f_2 = v_2 \circ g^{-1} \circ f_1 = v_1 \circ f_1$.
\end{prop}

\begin{prop}{graph/ordering/span}
    If $\mathfrak{G}_{{\mathrm{1}}}  \mathrel{<_\rightarrow}  \mathfrak{G}_{{\mathrm{2}}}$ and $f$ is a topological ordering of $\mathfrak{G}_{{\mathrm{2}}}$, then $f$ is also a topological ordering of $\mathfrak{G}_{{\mathrm{1}}}$.

    \proof By the assumption, we have some $(n, v, E_1) = \mathfrak{G}_{{\mathrm{1}}}$ and $(n, v, E_2) = \mathfrak{G}_{{\mathrm{2}}}$ where $E_1 \subseteq E_2$.
    We know $f$ is a bijection between $\NAT_{<n}$ to $\NAT_{<n}$, and $f^{-1}(n_1) < f^{-1}(n_2)$ for any $(n_1, n_2) \in E_2$.
    To this end, it suffices to show that $f^{-1}(n_1) < f^{-1}(n_2)$ for any $(n_1, n_2) \in E_1$, which actually holds since $E_1 \subseteq E_2$.
\end{prop}

\begin{prop}{graph/traceable/join}
    $ \GJOIN{ \mathfrak{G}_{{\mathrm{1}}} }{ \mathfrak{G}_{{\mathrm{2}}} } $ is traceable if and only if both $\mathfrak{G}_{{\mathrm{1}}}$ and $\mathfrak{G}_{{\mathrm{2}}}$ are traceable.

    \proof This is a corrorally of \propref{graph/ordering/join} and \propref{graph/ordering/join/inv}.
\end{prop}

\begin{prop}{graph/ordering/multi/iso}
    If $\mathfrak{G}_{{\mathrm{1}}}  \simeq  \mathfrak{G}_{{\mathrm{2}}}$ and $\mathfrak{G}_{{\mathrm{1}}}$ has more than one topological ordering, then $\mathfrak{G}_{{\mathrm{2}}}$ has more than one topological ordering.

    \proof Let $(n, v_1, E_1) = \mathfrak{G}_{{\mathrm{1}}}$ and $(n, v_2, E_2) = \mathfrak{G}_{{\mathrm{2}}}$.
    Since $\mathfrak{G}_{{\mathrm{1}}}  \simeq  \mathfrak{G}_{{\mathrm{2}}}$, we have some bijection $g$ such that $E1 = \{\,(g(n_1), g(n_2) \mid (n_1, n_2) \in E_2)\,\}$ and $v_1 \circ g = v_2$.
    Assuming $f_1$ and $f_2$ are the given topological orderings of $\mathfrak{G}_{{\mathrm{1}}}$, we can see that $g^{-1} \circ f_1$ and $g^{-1} \circ f_2$ are topological orderings of $\mathfrak{G}_{{\mathrm{2}}}$, and they are distinct since $f_1 \neq f_2$.

    To this end, let's see $g^{-1} \circ f_1$ is a topologial ordering of $\mathfrak{G}_{{\mathrm{2}}}$.
    We can easily see $g^{-1} \circ f_1$ is a bijection between $\NAT_{<n}$ to $\NAT_{<n}$ since $g$ and $f_1$ are bijections between $\NAT_{<n}$ to $\NAT_{<n}$.
    Let $(n_1, n_2) \in E_2$.
    We have $f_1^{-1}(g(n_1)) < f_1^{-1}(g(n_2))$ since $(g(n_1), g(n_2)) \in E_1$ and $f_1$ is topological ordering of $\Gamma_{{\mathrm{1}}}$.
    So, $(g^{-1} \circ f_1)^{-1}(n_1) < (g^{-1} \circ f_1)^{-1}(n_2)$.
\end{prop}

\begin{prop}{graph/traceable/iso}
    If $\mathfrak{G}_{{\mathrm{1}}}  \simeq  \mathfrak{G}_{{\mathrm{2}}}$ and $\mathfrak{G}_{{\mathrm{1}}}$ is traceable, then $\mathfrak{G}_{{\mathrm{2}}}$ is traceable.

    \proof This follows by the contraposition of \propref{graph/ordering/multi/iso}.
\end{prop}

\begin{prop}{graph/ordering/multi/span}
    If $\mathfrak{G}_{{\mathrm{1}}}  \mathrel{<_\rightarrow}  \mathfrak{G}_{{\mathrm{2}}}$ and $\mathfrak{G}_{{\mathrm{2}}}$ has more than one topological ordering, then $\mathfrak{G}_{{\mathrm{1}}}$ has more than one topological ordering.

    \proof This immediately follows by \propref{graph/ordering/span}.
\end{prop}

\begin{prop}{graph/traceable/span}
    If $\mathfrak{G}_{{\mathrm{1}}}  \mathrel{<_\rightarrow}  \mathfrak{G}_{{\mathrm{2}}}$ and $\mathfrak{G}_{{\mathrm{1}}}$ is traceable, then $\mathfrak{G}_{{\mathrm{2}}}$ is traceable.

    \proof This follows by the contraposition of \propref{graph/ordering/multi/span}.
\end{prop}

\begin{prop}{graph/ordering/union/multi}
    If $0 <  | \mathfrak{G}_{{\mathrm{1}}} |_{\bullet} $ and $0 <  | \mathfrak{G}_{{\mathrm{2}}} |_{\bullet} $, then $\mathfrak{G}_{{\mathrm{1}}}  \cup  \mathfrak{G}_{{\mathrm{2}}}$ has more than one topological ordering.

    \proof We have the following topological orderings of $\mathfrak{G}_{{\mathrm{1}}}  \cup  \mathfrak{G}_{{\mathrm{2}}}$ by \propref{graph/ordering/union}.
    \begin{align*}
        f(n)  & = \begin{cases}
                      f_1(n)                                   & (0 \le n <  | \mathfrak{G}_{{\mathrm{1}}} |_{\bullet} )                    \\
                      f_2(n -  | \mathfrak{G}_{{\mathrm{1}}} |_{\bullet} ) +  | \mathfrak{G}_{{\mathrm{1}}} |_{\bullet}  & ( | \mathfrak{G}_{{\mathrm{1}}} |_{\bullet}  \le n <  |  \GJOIN{ \mathfrak{G}_{{\mathrm{1}}} }{ \mathfrak{G}_{{\mathrm{2}}} }  |_{\bullet} )
                  \end{cases} \\
        f'(n) & = \begin{cases}
                      f_2(n) +  | \mathfrak{G}_{{\mathrm{1}}} |_{\bullet}  & (0 \le n <  | \mathfrak{G}_{{\mathrm{2}}} |_{\bullet} )                    \\
                      f_1(n -  | \mathfrak{G}_{{\mathrm{2}}} |_{\bullet} ) & ( | \mathfrak{G}_{{\mathrm{2}}} |_{\bullet}  \le n <  |  \GJOIN{ \mathfrak{G}_{{\mathrm{1}}} }{ \mathfrak{G}_{{\mathrm{2}}} }  |_{\bullet} )
                  \end{cases}
    \end{align*}
    To this end, we show $f \neq f'$.
    We consider the following three cases.
    \begin{match}
        \item[$  | \mathfrak{G}_{{\mathrm{1}}} |_{\bullet}  \ottsym{=}  | \mathfrak{G}_{{\mathrm{2}}} |_{\bullet}  $]
        We have some $n' <  | \mathfrak{G}_{{\mathrm{2}}} |_{\bullet} $ such that $f_2(n') = 0$ since $f_2$ is a bijection between $\NAT_{< | \mathfrak{G}_{{\mathrm{2}}} |_{\bullet} }$ and $\NAT_{< | \mathfrak{G}_{{\mathrm{2}}} |_{\bullet} }$ and $0 <  | \mathfrak{G}_{{\mathrm{2}}} |_{\bullet} $.
        Now we can see $f(n') = f_1(n')$ and $f'(n') =  | \mathfrak{G}_{{\mathrm{1}}} |_{\bullet} $.
        That implies $f(n') \neq f'(n')$ since $f_1(n') <  | \mathfrak{G}_{{\mathrm{1}}} |_{\bullet} $.
        As a result $f \neq f'$.

        \item[$  | \mathfrak{G}_{{\mathrm{1}}} |_{\bullet}  \ottsym{<}  | \mathfrak{G}_{{\mathrm{2}}} |_{\bullet}  $]
        We can see $f( | \mathfrak{G}_{{\mathrm{1}}} |_{\bullet} ) = f_2(0) +  | \mathfrak{G}_{{\mathrm{1}}} |_{\bullet} $ and $f'( | \mathfrak{G}_{{\mathrm{1}}} |_{\bullet} ) = f_2( | \mathfrak{G}_{{\mathrm{1}}} |_{\bullet} ) +  | \mathfrak{G}_{{\mathrm{1}}} |_{\bullet} $.
        This implies $f( | \mathfrak{G}_{{\mathrm{1}}} |_{\bullet} ) \neq f'( | \mathfrak{G}_{{\mathrm{1}}} |_{\bullet} )$ since $f_2$ is a bijection, and $0 <  | \mathfrak{G}_{{\mathrm{1}}} |_{\bullet} $.
        As a result $f \neq f'$.

        \item[$  | \mathfrak{G}_{{\mathrm{1}}} |_{\bullet}  \ottsym{>}  | \mathfrak{G}_{{\mathrm{2}}} |_{\bullet}  $]
        We can show $f( | \mathfrak{G}_{{\mathrm{2}}} |_{\bullet} ) \neq f'( | \mathfrak{G}_{{\mathrm{2}}} |_{\bullet} )$ in a similar way to the case above. \qedhere
    \end{match}
\end{prop}

\begin{prop}{graph/traceable/union}
    $\mathfrak{G}_{{\mathrm{1}}}  \cup  \mathfrak{G}_{{\mathrm{2}}}$ is traceable if and only if
    \begin{enumerate}
        \item $  | \mathfrak{G}_{{\mathrm{1}}} |_{\bullet}  \ottsym{=} \ottsym{0} $ and $\mathfrak{G}_{{\mathrm{2}}}$ is traceable, or
        \item $  | \mathfrak{G}_{{\mathrm{2}}} |_{\bullet}  \ottsym{=} \ottsym{0} $ and $\mathfrak{G}_{{\mathrm{1}}}$ is traceable.
    \end{enumerate}

    \proof We show each direction.
    \begin{itemize}
        \item $\Rightarrow$ By the contraposition of \propref{graph/ordering/union/multi}, $  | \mathfrak{G}_{{\mathrm{1}}} |_{\bullet}  \ottsym{=} \ottsym{0} $ or $  | \mathfrak{G}_{{\mathrm{2}}} |_{\bullet}  \ottsym{=} \ottsym{0} $ hold.
        If $  | \mathfrak{G}_{{\mathrm{1}}} |_{\bullet}  \ottsym{=} \ottsym{0} $, $\mathfrak{G}_{{\mathrm{1}}}  \cup  \mathfrak{G}_{{\mathrm{2}}}  \simeq  \mathfrak{G}_{{\mathrm{2}}}$ holds.
        So, the first case of the goal follows by \propref{graph/traceable/iso}.
        The case for $  | \mathfrak{G}_{{\mathrm{2}}} |_{\bullet}  \ottsym{=} \ottsym{0} $ can be shown similary.
        \item $\Leftarrow$ If the first case holds, $\mathfrak{G}_{{\mathrm{1}}}  \cup  \mathfrak{G}_{{\mathrm{2}}}  \simeq  \mathfrak{G}_{{\mathrm{2}}}$. So the goal follows by \propref{graph/traceable/iso}. Another case is similar. \qedhere
    \end{itemize}
\end{prop}

  \subsection{Properties for environment}
  \begin{prop}{env/equiv}
    \noindent
    \begin{statements}
        \item $\Gamma  \ottsym{,}   \cdot   \simeq  \Gamma$.
        \item $ \cdot   \ottsym{,}  \Gamma  \simeq  \Gamma$.
        \item $\Gamma  \parallel   \cdot   \simeq  \Gamma$.
        \item $ \cdot   \parallel  \Gamma  \simeq  \Gamma$.
        \item $\Gamma_{{\mathrm{1}}}  \parallel  \Gamma_{{\mathrm{2}}}  \simeq  \Gamma_{{\mathrm{2}}}  \parallel  \Gamma_{{\mathrm{1}}}$.
        \item $\Gamma_{{\mathrm{1}}}  \ottsym{,}  \ottsym{(}  \Gamma_{{\mathrm{2}}}  \ottsym{,}  \Gamma_{{\mathrm{3}}}  \ottsym{)}  \simeq  \ottsym{(}  \Gamma_{{\mathrm{1}}}  \ottsym{,}  \Gamma_{{\mathrm{2}}}  \ottsym{)}  \ottsym{,}  \Gamma_{{\mathrm{3}}}$.
        \item $\Gamma_{{\mathrm{1}}}  \parallel  \ottsym{(}  \Gamma_{{\mathrm{2}}}  \parallel  \Gamma_{{\mathrm{3}}}  \ottsym{)}  \simeq  \ottsym{(}  \Gamma_{{\mathrm{1}}}  \parallel  \Gamma_{{\mathrm{2}}}  \ottsym{)}  \parallel  \Gamma_{{\mathrm{3}}}$.
    \end{statements}
\end{prop}

\begin{prop}{env/sub}
    \noindent
    \begin{statements}
        \item[span] $\Gamma_{{\mathrm{1}}}  \parallel  \Gamma_{{\mathrm{2}}}  \lesssim  \Gamma_{{\mathrm{1}}}  \ottsym{,}  \Gamma_{{\mathrm{2}}}$.
        \item[distl] $\Gamma_{{\mathrm{1}}}  \parallel  \ottsym{(}  \Gamma_{{\mathrm{2}}}  \ottsym{,}  \Gamma_{{\mathrm{3}}}  \ottsym{)}  \lesssim  \ottsym{(}  \Gamma_{{\mathrm{1}}}  \parallel  \Gamma_{{\mathrm{2}}}  \ottsym{)}  \ottsym{,}  \Gamma_{{\mathrm{3}}}$.
        \item[distr] $\ottsym{(}  \Gamma_{{\mathrm{1}}}  \ottsym{,}  \Gamma_{{\mathrm{2}}}  \ottsym{)}  \parallel  \Gamma_{{\mathrm{3}}}  \lesssim  \Gamma_{{\mathrm{1}}}  \ottsym{,}  \ottsym{(}  \Gamma_{{\mathrm{2}}}  \parallel  \Gamma_{{\mathrm{3}}}  \ottsym{)}$.
    \end{statements}
\end{prop}

\begin{prop}{env/eq/ctx}
    \noindent
    \begin{statements}
        \item[null] If $  \cdot  \ottsym{=} \mathcal{G}  \ottsym{[}  \Gamma  \ottsym{]} $, then $ \mathcal{G} \ottsym{=} \ottsym{[]} $ and $ \Gamma \ottsym{=}  \cdot  $.
        \item[binding] If $ \ottmv{x}  \mathord:  \ottnt{T} \ottsym{=} \mathcal{G}  \ottsym{[}  \Gamma  \ottsym{]} $, then $ \mathcal{G} \ottsym{=} \ottsym{[]} $ and $ \Gamma \ottsym{=} \ottmv{x}  \mathord:  \ottnt{T} $.
        \item[location] If $ \ottmv{l}  \mathord:  \ottsym{[}  \ottnt{m}  \ottsym{]} \ottsym{=} \mathcal{G}  \ottsym{[}  \Gamma  \ottsym{]} $, then $ \mathcal{G} \ottsym{=} \ottsym{[]} $ and $ \Gamma \ottsym{=} \ottmv{l}  \mathord:  \ottsym{[}  \ottnt{m}  \ottsym{]} $
        \item[concat] If $ \Gamma_{{\mathrm{1}}}  \ottsym{,}  \Gamma_{{\mathrm{2}}} \ottsym{=} \mathcal{G}  \ottsym{[}  \Gamma  \ottsym{]} $, then
        \begin{itemize}
            \item $ \mathcal{G} \ottsym{=} \ottsym{[]} $ and $ \Gamma \ottsym{=} \Gamma_{{\mathrm{1}}}  \ottsym{,}  \Gamma_{{\mathrm{2}}} $,
            \item $ \mathcal{G} \ottsym{=} \mathcal{G}_{{\mathrm{1}}}  \ottsym{,}  \Gamma_{{\mathrm{2}}} $ and $ \mathcal{G}_{{\mathrm{1}}}  \ottsym{[}  \Gamma  \ottsym{]} \ottsym{=} \Gamma_{{\mathrm{1}}} $ for some $\mathcal{G}_{{\mathrm{1}}}$, or
            \item $ \mathcal{G} \ottsym{=} \Gamma_{{\mathrm{1}}}  \ottsym{,}  \mathcal{G}_{{\mathrm{2}}} $ and $ \mathcal{G}_{{\mathrm{2}}}  \ottsym{[}  \Gamma  \ottsym{]} \ottsym{=} \Gamma_{{\mathrm{2}}} $ for some $\mathcal{G}_{{\mathrm{2}}}$.
        \end{itemize}
        \item[parallel] If $ \Gamma_{{\mathrm{1}}}  \parallel  \Gamma_{{\mathrm{2}}} \ottsym{=} \mathcal{G}  \ottsym{[}  \Gamma  \ottsym{]} $, then
        \begin{itemize}
            \item $ \mathcal{G} \ottsym{=} \ottsym{[]} $ and $ \Gamma \ottsym{=} \Gamma_{{\mathrm{1}}}  \parallel  \Gamma_{{\mathrm{2}}} $,
            \item $ \mathcal{G} \ottsym{=} \mathcal{G}_{{\mathrm{1}}}  \parallel  \Gamma_{{\mathrm{2}}} $ and $ \mathcal{G}_{{\mathrm{1}}}  \ottsym{[}  \Gamma  \ottsym{]} \ottsym{=} \Gamma_{{\mathrm{1}}} $ for some $\mathcal{G}_{{\mathrm{1}}}$, or
            \item $ \mathcal{G} \ottsym{=} \Gamma_{{\mathrm{1}}}  \parallel  \mathcal{G}_{{\mathrm{2}}} $ and $ \mathcal{G}_{{\mathrm{2}}}  \ottsym{[}  \Gamma  \ottsym{]} \ottsym{=} \Gamma_{{\mathrm{2}}} $ for some $\mathcal{G}_{{\mathrm{2}}}$.
        \end{itemize}
    \end{statements}

    \proof By case analysis on $\mathcal{G}$.
\end{prop}

\begin{prop}{env/num/bindings/ctx}
    $  \ottsym{\#} _{ \ottnt{b} }( \mathcal{G}  \ottsym{[}  \Gamma  \ottsym{]} )  \ottsym{=}  \ottsym{\#} _{ \ottnt{b} }( \mathcal{G} )   \ottsym{+}   \ottsym{\#} _{ \ottnt{b} }( \Gamma )  $.

    \proof Routine by structural induction on $\mathcal{G}$.
\end{prop}

\begin{prop}{env/in/binding}
    $ \ottnt{b_{{\mathrm{1}}}}  \in  \ottnt{b_{{\mathrm{2}}}} $ if and only if $ \ottnt{b_{{\mathrm{1}}}} \ottsym{=} \ottnt{b_{{\mathrm{2}}}} $.

    \proof We will prove each direction.
    \begin{match}
        \item[$\Rightarrow$]
        We show the contraposition.
        So, suppose $ \ottnt{b_{{\mathrm{1}}}} \neq \ottnt{b_{{\mathrm{2}}}} $.
        We have $ \ottsym{\#} _{ \ottnt{b_{{\mathrm{1}}}} }( \ottnt{b_{{\mathrm{2}}}} )  = 0$ by definition.
        Thus, $ \ottnt{b_{{\mathrm{1}}}}  \notin  \ottnt{b_{{\mathrm{2}}}} $.

        \item[$\Leftarrow$]
        We have $ \ottsym{\#} _{ \ottnt{b_{{\mathrm{1}}}} }( \ottnt{b_{{\mathrm{2}}}} )  = 1$ by definition.
        Thus, $ \ottnt{b_{{\mathrm{1}}}}  \in  \ottnt{b_{{\mathrm{2}}}} $. \qedhere
    \end{match}
\end{prop}







\begin{prop}{env/in/ctx/binding}
    $ \ottnt{b}  \in  \mathcal{G}  \ottsym{[}  \ottnt{b}  \ottsym{]} $.

    \proof We have $ \ottnt{b}  \in  \ottnt{b} $ by \propref{env/in/binding}.
    Then, the goal follows by \propref{env/num/bindings/ctx}.
\end{prop}

\begin{prop}{env/in/ctx/exists}
    If $ \ottnt{b}  \in  \Gamma $, then $ \Gamma \ottsym{=} \mathcal{G}  \ottsym{[}  \ottnt{b}  \ottsym{]} $ for some $\mathcal{G}$.

    \proof Routine by structural induction on $\Gamma$.
\end{prop}

\begin{prop}{env/intrp/num/bindings/unr}
    Let $  \lBrack \Gamma \rBrack  \ottsym{=} \mathfrak{G}  \ottsym{;}  \mathbb{S} $. Then, $ \ottnt{b}  \in  \Gamma $ if and only if $\ottnt{b} \in \mathbb{S}$ for any $\ottkw{unr} \, \ottnt{b}$.

    \proof By structural induction on $\Gamma$.
\end{prop}

\begin{prop}{env/intrp/num/bindings/ord}
    Let $ \lBrack \Gamma \rBrack  = (n, v, E); \mathbb{S}$.
    Then, $ \ottsym{\#} _{ \ottnt{b} }( \Gamma )  = |\{\,n' \in \NAT_{<n} \mid v(n') = \ottnt{b}\,\}|$ for any $\ottkw{ord} \, \ottnt{b}$.

    \proof By structural induction on $\Gamma$.
    \begin{match}
        \item[$ \Gamma \ottsym{=}  \cdot  $]
        In this case, $ \ottsym{\#} _{ \ottnt{b} }( \Gamma )  = 0$.
        On the other hand, $ \ottnt{n} \ottsym{=} \ottsym{0} $.
        So, the right hand of the goal is also $0$.

        \item[$ \Gamma \ottsym{=} \ottnt{b'} $]
        We consider whether $\ottkw{unr} \, \ottnt{b'}$ or $\ottkw{ord} \, \ottnt{b'}$.
        \begin{match}
            \item[$\ottkw{unr} \, \ottnt{b'}$]
            In this case, $ \ottsym{\#} _{ \ottnt{b} }( \Gamma )  = 0$ and $ \lBrack \Gamma \rBrack  = (0, \emptyset, \emptyset); \{b'\}$, which implies $|\{\,n' \in \NAT_{<n} \mid v(n') = \ottnt{b}\,\}| = 0$.
            \item[$\ottkw{ord} \, \ottnt{b'}$]
            In this case,
            \begin{gather}
                 \lBrack \Gamma \rBrack  = (1, \{0 \mapsto b'\}, \emptyset); \emptyset. \hyp{2.1}
            \end{gather}
            We consider wether $ \ottnt{b'} \ottsym{=} \ottnt{b} $.
            \begin{match}
                \item[$ \ottnt{b'} \ottsym{=} \ottnt{b} $]
                $ \ottsym{\#} _{ \ottnt{b} }( \Gamma )  = 1$ by definition, and $|\{\,n' \in \NAT_{<n} \mid v(n') = \ottnt{b}\,\}| = 1$ by \hypref{2.1}.

                \item[$ \ottnt{b'} \neq \ottnt{b} $]
                $ \ottsym{\#} _{ \ottnt{b} }( \Gamma )  = 0$ by definition, and $|\{\,n' \in \NAT_{<n} \mid v(n') = \ottnt{b}\,\}| = 0$ by \hypref{2.1}.
            \end{match}
        \end{match}

        \item[$ \Gamma \ottsym{=} \Gamma_{{\mathrm{1}}}  \ottsym{,}  \Gamma_{{\mathrm{2}}} $]
        Let $ \lBrack \Gamma_{{\mathrm{1}}} \rBrack  = (n_1, v_1, E_1); \mathbb{S}_{{\mathrm{1}}}$ and $ \lBrack \Gamma_{{\mathrm{2}}} \rBrack  = (n_2, v_2, E_2); \mathbb{S}_{{\mathrm{2}}}$.
        We have
        \begin{gather}
             \ottsym{\#} _{ \ottnt{b} }( \Gamma_{{\mathrm{1}}} )  = |\{\,n' \in \NAT_{<n_1} \mid v_1(n') = \ottnt{b}\,\}| \hyp{3.1}\\
             \ottsym{\#} _{ \ottnt{b} }( \Gamma_{{\mathrm{2}}} )  = |\{\,n' \in \NAT_{<n_2} \mid v_2(n') = \ottnt{b}\,\}| \hyp{3.2}
        \end{gather}
        by the induction hypothesis.
        We also have
        \begin{gather}
            n = n_1 + n_2 \hyp{3.3}\\
            v(n) = \begin{cases}
                v_1(n)       & \text{if $0 \le n < n_1$}         \\
                v_2(n - n_1) & \text{if $n_1 \le n < n_1 + n_2$}
            \end{cases} \hyp{3.4}
        \end{gather}
        by definition.
        From those, $|\{\,n' \in \NAT_{<n} \mid v(n') = \ottnt{b}\,\}| =  \ottsym{\#} _{ \ottnt{b} }( \Gamma_{{\mathrm{1}}} )   \ottsym{+}   \ottsym{\#} _{ \ottnt{b} }( \Gamma_{{\mathrm{2}}} )  =  \ottsym{\#} _{ \ottnt{b} }( \Gamma ) $.

        \item[$ \Gamma \ottsym{=} \Gamma_{{\mathrm{1}}}  \parallel  \Gamma_{{\mathrm{2}}} $] Similar to the case $ \Gamma \ottsym{=} \Gamma_{{\mathrm{1}}}  \ottsym{,}  \Gamma_{{\mathrm{2}}} $. \qedhere
    \end{match}
\end{prop}

\begin{prop}{env/in/equiv}
    Suppose $\Gamma_{{\mathrm{1}}}  \simeq  \Gamma_{{\mathrm{2}}}$.
    Then, $ \ottnt{b}  \in  \Gamma_{{\mathrm{1}}} $ if and only if $ \ottnt{b}  \in  \Gamma_{{\mathrm{2}}} $.

    \proof Let $ \mathfrak{G}_{{\mathrm{1}}}  \ottsym{;}  \mathbb{S}_{{\mathrm{1}}} \ottsym{=}  \lBrack \Gamma_{{\mathrm{1}}} \rBrack  $ and $ \mathfrak{G}_{{\mathrm{2}}}  \ottsym{;}  \mathbb{S}_{{\mathrm{2}}} \ottsym{=}  \lBrack \Gamma_{{\mathrm{2}}} \rBrack  $.
    By definition, we have
    \begin{gather}
        \mathfrak{G}_{{\mathrm{1}}}  \simeq  \mathfrak{G}_{{\mathrm{2}}} \hyp{1.1}\\
         \mathbb{S}_{{\mathrm{1}}} \ottsym{=} \mathbb{S}_{{\mathrm{2}}}  \hyp{1.2}.
    \end{gather}
    We consider wether $\ottkw{unr} \, \ottnt{b}$ or $\ottkw{ord} \, \ottnt{b}$.
    \begin{match}
        \item[$\ottkw{unr} \, \ottnt{b}$]
        The goal follows by \propref{env/intrp/num/bindings/unr} and \hypref{1.2}.
        \item[$\ottkw{ord} \, \ottnt{b}$]
        The goal follows by \propref{graph/iso/num/bindings}, \propref{env/intrp/num/bindings/ord}, and \hypref{1.1}.
    \end{match}
\end{prop}

\begin{prop}{env/in/sup}
    If $\Gamma_{{\mathrm{1}}}  \lesssim  \Gamma_{{\mathrm{2}}}$ and $ \ottnt{b}  \in  \Gamma_{{\mathrm{2}}} $, then $ \ottnt{b}  \in  \Gamma_{{\mathrm{1}}} $.

    \proof Let $  \lBrack \Gamma_{{\mathrm{1}}} \rBrack  \ottsym{=} \mathfrak{G}_{{\mathrm{1}}}  \ottsym{;}  \mathbb{S}_{{\mathrm{1}}} $ and $  \lBrack \Gamma_{{\mathrm{2}}} \rBrack  \ottsym{=} \mathfrak{G}_{{\mathrm{2}}}  \ottsym{;}  \mathbb{S}_{{\mathrm{2}}} $.
    By definition, we have
    \begin{gather}
        \mathfrak{G}_{{\mathrm{1}}}  \simeq  \mathfrak{G}'_{{\mathrm{1}}} \mathrel{<_\rightarrow} \mathfrak{G}'_{{\mathrm{2}}}  \simeq  \mathfrak{G}_{{\mathrm{2}}} \hyp{1.1}\\
         \mathbb{S}_{{\mathrm{1}}}  \supseteq  \mathbb{S}_{{\mathrm{2}}}  \hyp{1.2}
    \end{gather}
    for some $\mathfrak{G}'_{{\mathrm{1}}}$ and $\mathfrak{G}'_{{\mathrm{2}}}$.
    We consider wether $\ottkw{unr} \, \ottnt{b}$ or $\ottkw{ord} \, \ottnt{b}$.
    \begin{match}
        \item[$\ottkw{unr} \, \ottnt{b}$]
        We have $b \in \mathbb{S}_{{\mathrm{2}}}$ by \propref{env/intrp/num/bindings/unr}.
        So, $b \in \mathbb{S}_{{\mathrm{1}}}$ by \hypref{1.2}.
        Consequently, $ \ottnt{b}  \in  \Gamma_{{\mathrm{1}}} $ follows by \propref{env/intrp/num/bindings/unr}.
        \item[$\ottkw{ord} \, \ottnt{b}$]
        This case follows by \propref{graph/iso/num/bindings}, \propref{graph/span/num/bindings}, \propref{env/intrp/num/bindings/ord}, and \hypref{1.1}. \qedhere
    \end{match}

\end{prop}

\begin{prop}{env/in/sub}
    If $\Gamma_{{\mathrm{1}}}  \lesssim  \Gamma_{{\mathrm{2}}}$, $ \ottnt{b}  \in  \Gamma_{{\mathrm{1}}} $, and $\ottkw{ord} \, \ottnt{b}$, then $ \ottnt{b}  \in  \Gamma_{{\mathrm{2}}} $.

    \proof Let $  \lBrack \Gamma_{{\mathrm{1}}} \rBrack  \ottsym{=} \mathfrak{G}_{{\mathrm{1}}}  \ottsym{;}  \mathbb{S}_{{\mathrm{1}}} $ and $  \lBrack \Gamma_{{\mathrm{2}}} \rBrack  \ottsym{=} \mathfrak{G}_{{\mathrm{2}}}  \ottsym{;}  \mathbb{S}_{{\mathrm{2}}} $.
    By definition, we have
    \begin{gather}
        \mathfrak{G}_{{\mathrm{1}}}  \simeq  \mathfrak{G}'_{{\mathrm{1}}} \mathrel{<_\rightarrow} \mathfrak{G}'_{{\mathrm{2}}}  \simeq  \mathfrak{G}_{{\mathrm{2}}} \hyp{1.1}\\
         \mathbb{S}_{{\mathrm{1}}}  \supseteq  \mathbb{S}_{{\mathrm{2}}}  \hyp{1.2}
    \end{gather}
    for some $\mathfrak{G}'_{{\mathrm{1}}}$ and $\mathfrak{G}'_{{\mathrm{2}}}$.
    So, the goal follows by \propref{graph/iso/num/bindings}, \propref{graph/span/num/bindings}, \propref{env/intrp/num/bindings/ord}, and \hypref{1.1}.

\end{prop}

\begin{prop}{env/equiv/runtime}
    Suppose $\Gamma_{{\mathrm{1}}}  \simeq  \Gamma_{{\mathrm{2}}}$.
    Then, $\Gamma_{{\mathrm{1}}}$ is a runtime environment if and only if $\Gamma_{{\mathrm{2}}}$ is a runtime environment.

    \proof We show the contraposition of each direction.
    If $\Gamma_{{\mathrm{1}}}$ is not a runtime environment, we have $ \ottmv{x}  \mathord:  \ottnt{S}  \in  \Gamma_{{\mathrm{1}}} $ for some $\ottmv{x}  \mathord:  \ottnt{S}$.
    So, $ \ottmv{x}  \mathord:  \ottnt{S}  \in  \Gamma_{{\mathrm{2}}} $ follows by \propref{env/in/equiv}, which implies $\Gamma_{{\mathrm{2}}}$ is not a runtime environment.
    Another direction is similar.
\end{prop}

\begin{prop}{env/sub/runtime}
    If $\ottnt{C}  \lesssim  \Gamma$, then $\Gamma$ is a runtime environment.

    \proof We prove the goal by contraction.
    Assuming $\Gamma$ is not a runtime environment, we have some $ \ottmv{x}  \mathord:  \ottnt{S}  \in  \Gamma $.
    So, $ \ottmv{x}  \mathord:  \ottnt{S}  \in  \ottnt{C} $ by \propref{env/in/sup}.
    However, this contradicts the fact $\ottnt{C}$ is a runtime environment.
\end{prop}

\begin{prop}{env/equiv/dom}
    If $\Gamma_{{\mathrm{1}}}  \simeq  \Gamma_{{\mathrm{2}}}$, then $  \ottkw{dom} ( \Gamma_{{\mathrm{1}}} )  \ottsym{=}  \ottkw{dom} ( \Gamma_{{\mathrm{2}}} )  $.

    \proof This follows by \propref{env/in/equiv}.
\end{prop}

\begin{prop}{env/sub/runtime/dom}
    If $\ottnt{C}  \lesssim  \Gamma$, then $  \ottkw{dom} ( \ottnt{C} )  \ottsym{=}  \ottkw{dom} ( \Gamma )  $.

    \proof We can see $\Gamma$ is a runtime environment by \propref{env/sub/runtime}.
    Then, the goal follows by \propref{env/in/sup} and \propref{env/in/sub}.
\end{prop}

\begin{prop}{env/equiv/ctx}
    If $\Gamma_{{\mathrm{1}}}  \simeq  \Gamma_{{\mathrm{2}}}$, then $\mathcal{G}  \ottsym{[}  \Gamma_{{\mathrm{1}}}  \ottsym{]}  \simeq  \mathcal{G}  \ottsym{[}  \Gamma_{{\mathrm{2}}}  \ottsym{]}$.

    \proof By structural induction on $\mathcal{G}$.
    \begin{match}
        \item[$ \mathcal{G} \ottsym{=} \ottsym{[]} $]
        Immediately by the assumption.

        \item[$ \mathcal{G} \ottsym{=} \mathcal{G}'  \ottsym{,}  \Gamma' $]
        We have
        \begin{gather}
            \mathcal{G}'  \ottsym{[}  \Gamma_{{\mathrm{1}}}  \ottsym{]}  \simeq  \mathcal{G}'  \ottsym{[}  \Gamma_{{\mathrm{2}}}  \ottsym{]} \hyp{2.1}
        \end{gather}
        by IH.
        Let
        \begin{gather}
              \lBrack \mathcal{G}'  \ottsym{[}  \Gamma_{{\mathrm{1}}}  \ottsym{]} \rBrack  \ottsym{=} \mathfrak{G}'_{{\mathrm{1}}}  \ottsym{;}  \mathbb{S}'_{{\mathrm{1}}}  \hyp{2.2}\\
              \lBrack \mathcal{G}'  \ottsym{[}  \Gamma_{{\mathrm{2}}}  \ottsym{]} \rBrack  \ottsym{=} \mathfrak{G}'_{{\mathrm{2}}}  \ottsym{;}  \mathbb{S}'_{{\mathrm{2}}}  \hyp{2.3}\\
              \lBrack \Gamma' \rBrack  \ottsym{=} \mathfrak{G}'  \ottsym{;}  \mathbb{S}'  \hyp{2.4}.
        \end{gather}
        We have
        \begin{gather}
            \mathfrak{G}'_{{\mathrm{1}}}  \simeq  \mathfrak{G}'_{{\mathrm{2}}} \hyp{2.5}\\
             \mathbb{S}'_{{\mathrm{1}}} \ottsym{=} \mathbb{S}'_{{\mathrm{2}}}  \hyp{2.6}
        \end{gather}
        by \hypref{2.1}, \hypref{2.2}, and \hypref{2.3}.
        Then, we have
        \begin{gather}
             \GJOIN{ \mathfrak{G}'_{{\mathrm{1}}} }{ \mathfrak{G}' }   \simeq   \GJOIN{ \mathfrak{G}'_{{\mathrm{2}}} }{ \mathfrak{G}' }  \hyp{2.7}
        \end{gather}
        by \propref{graph/iso/refl} and \propref{graph/iso/join}.
        Furethermore,
        \begin{gather}
            \mathbb{S}'_{{\mathrm{1}}} \cup \mathbb{S}' = \mathbb{S}'_{{\mathrm{2}}} \cup \mathbb{S}'. \hyp{2.8}
        \end{gather}
        From \hypref{2.7} and \hypref{2.8}, we have the goal $\mathcal{G}'  \ottsym{[}  \Gamma_{{\mathrm{1}}}  \ottsym{]}  \ottsym{,}  \Gamma'  \simeq  \mathcal{G}'  \ottsym{[}  \Gamma_{{\mathrm{2}}}  \ottsym{]}  \ottsym{,}  \Gamma'$ by definition.

        \item[$ \mathcal{G} \ottsym{=} \Gamma'  \ottsym{,}  \mathcal{G}' $] Similar to the case above.
        \item[$ \mathcal{G} \ottsym{=} \mathcal{G}'  \parallel  \Gamma' $ and $ \mathcal{G} \ottsym{=} \Gamma'  \parallel  \mathcal{G}' $] Similar to the case $ \mathcal{G} \ottsym{=} \mathcal{G}'  \ottsym{,}  \Gamma' $ except we use \propref{graph/iso/union}. \qedhere
    \end{match}

\end{prop}

\begin{prop}{env/sub/ctx}
    If $\Gamma_{{\mathrm{1}}}  \lesssim  \Gamma_{{\mathrm{2}}}$, then $\mathcal{G}  \ottsym{[}  \Gamma_{{\mathrm{1}}}  \ottsym{]}  \lesssim  \mathcal{G}  \ottsym{[}  \Gamma_{{\mathrm{2}}}  \ottsym{]}$.

    \proof By structural induction on $\mathcal{G}$.
    \begin{match}
        \item[$ \mathcal{G} \ottsym{=} \ottsym{[]} $]
        Immediately by the assumption.

        \item[$ \mathcal{G} \ottsym{=} \mathcal{G}'  \ottsym{,}  \Gamma' $]
        Let
        \begin{gather*}
              \lBrack \mathcal{G}'  \ottsym{[}  \Gamma_{{\mathrm{1}}}  \ottsym{]} \rBrack  \ottsym{=} \mathfrak{G}_{{\mathrm{1}}}  \ottsym{;}  \mathbb{S}_{{\mathrm{1}}}  \\
              \lBrack \mathcal{G}'  \ottsym{[}  \Gamma_{{\mathrm{2}}}  \ottsym{]} \rBrack  \ottsym{=} \mathfrak{G}_{{\mathrm{2}}}  \ottsym{;}  \mathbb{S}_{{\mathrm{2}}}  \\
              \lBrack \Gamma' \rBrack  \ottsym{=} \mathfrak{G}  \ottsym{;}  \mathbb{S} .
        \end{gather*}
        Then, we have
        \begin{gather}
            \mathfrak{G}_{{\mathrm{1}}}  \simeq  \mathfrak{G}'_{{\mathrm{1}}} \mathrel{<_\rightarrow} \mathfrak{G}'_{{\mathrm{2}}}  \simeq  \mathfrak{G}_{{\mathrm{2}}} \hyp{2.1}\\
             \mathbb{S}_{{\mathrm{1}}}  \supseteq  \mathbb{S}_{{\mathrm{2}}}  \hyp{2.2}
        \end{gather}
        for some $\mathfrak{G}'_{{\mathrm{1}}}$ and $\mathfrak{G}'_{{\mathrm{2}}}$ by IH.
        What we need to show is
        \begin{gather*}
             \GJOIN{ \mathfrak{G}_{{\mathrm{1}}} }{ \mathfrak{G} }   \simeq  \mathfrak{G}''_{{\mathrm{1}}} \mathrel{<_\rightarrow} \mathfrak{G}''_{{\mathrm{2}}}  \simeq   \GJOIN{ \mathfrak{G}_{{\mathrm{2}}} }{ \mathfrak{G} }  \\
              \mathbb{S}_{{\mathrm{1}}} \cup \mathbb{S}   \supseteq   \mathbb{S}_{{\mathrm{2}}} \cup \mathbb{S}  
        \end{gather*}
        for some $\mathfrak{G}''_{{\mathrm{1}}}$ and $\mathfrak{G}''_{{\mathrm{2}}}$.
        We choose $ \mathfrak{G}''_{{\mathrm{1}}} \ottsym{=}  \GJOIN{ \mathfrak{G}'_{{\mathrm{1}}} }{ \mathfrak{G}' }  $ and $ \mathfrak{G}''_{{\mathrm{2}}} \ottsym{=}  \GJOIN{ \mathfrak{G}'_{{\mathrm{2}}} }{ \mathfrak{G}' }  $.
        The first goal follows by \propref{graph/iso/join}, \propref{graph/span/join}, and \hypref{2.1}.
        The second goal follows by \hypref{2.2}.

        \item[$ \mathcal{G} \ottsym{=} \Gamma'  \ottsym{,}  \mathcal{G}' $] Similar to the case above.
        \item[$ \mathcal{G} \ottsym{=} \mathcal{G}'  \parallel  \Gamma' $ and $ \mathcal{G} \ottsym{=} \Gamma'  \parallel  \mathcal{G}' $] Similar to the case $ \mathcal{G} \ottsym{=} \mathcal{G}'  \ottsym{,}  \Gamma' $ except we use \propref{graph/iso/union} and \propref{graph/span/union}. \qedhere
    \end{match}




\end{prop}

\begin{prop}{env/remove/unr}
    Let $  \lBrack \Gamma \rBrack  \ottsym{=} \mathfrak{G}  \ottsym{;}  \mathbb{S} $ and $\ottkw{unr} \, \ottnt{b}$, then $  \lBrack  \Gamma ^{- \ottnt{b} }  \rBrack  \ottsym{=} \mathfrak{G}  \ottsym{;}   \mathbb{S} \setminus \ottsym{\{}  \ottnt{b}  \ottsym{\}}  $.

    \proof Routine by structural induction on $\Gamma$.
\end{prop}

\begin{prop}{env/replace/ord}
    Let $ \ottnt{b}  \in  \Gamma $, $  \lBrack \Gamma \rBrack  \ottsym{=} \mathfrak{G}  \ottsym{;}  \mathbb{S} $, $  \lBrack \Gamma' \rBrack  \ottsym{=} \mathfrak{G}'  \ottsym{;}  \mathbb{S}' $, and $\ottkw{ord} \, \ottnt{b}$, then $  \lBrack \Gamma  \ottsym{[}  \Gamma'  \ottsym{/}  \ottnt{b}  \ottsym{]} \rBrack  \ottsym{=} \mathfrak{G}  \ottsym{[}  \mathfrak{G}'  \ottsym{/}  \ottnt{b}  \ottsym{]}  \ottsym{;}   \mathbb{S} \cup \mathbb{S}'  $.

    \proof By structural induction on $\Gamma$.
    \begin{match}
        \item[$ \Gamma \ottsym{=}  \cdot  $] Vacuously true since $ \ottnt{b}  \in  \Gamma $ cannot hold.
        \item[$ \Gamma \ottsym{=} \ottnt{b'} $] In this case $ \ottnt{b'} \ottsym{=} \ottnt{b} $ since $ \ottnt{b}  \in  \Gamma $ is supposed.
        Then, $\mathfrak{G} = (1, {0 \mapsto b}, \emptyset)$ and $\mathbb{S} = \emptyset$ by definition.
        That implies $ \mathfrak{G}  \ottsym{[}  \mathfrak{G}'  \ottsym{/}  \ottnt{b}  \ottsym{]} \ottsym{=} \mathfrak{G}' $ and $  \mathbb{S} \cup \mathbb{S}'  \ottsym{=} \mathbb{S}' $.
        Consequetly, the goal $  \lBrack \Gamma  \ottsym{[}  \Gamma'  \ottsym{/}  \ottnt{b}  \ottsym{]} \rBrack  \ottsym{=}  \lBrack \Gamma' \rBrack   = \mathfrak{G}  \ottsym{[}  \mathfrak{G}'  \ottsym{/}  \ottnt{b}  \ottsym{]}  \ottsym{;}   \mathbb{S} \cup \mathbb{S}' $ holds.
        \item[$ \Gamma \ottsym{=} \Gamma_{{\mathrm{1}}}  \ottsym{,}  \Gamma_{{\mathrm{2}}} $]
        Let $ \mathfrak{G}_{{\mathrm{1}}}  \ottsym{;}  \mathbb{S}_{{\mathrm{1}}} \ottsym{=}  \lBrack \Gamma_{{\mathrm{1}}} \rBrack  $ and $ \mathfrak{G}_{{\mathrm{2}}}  \ottsym{;}  \mathbb{S}_{{\mathrm{2}}} \ottsym{=}  \lBrack \Gamma_{{\mathrm{2}}} \rBrack  $.
        We have $ \mathfrak{G} \ottsym{=}  \GJOIN{ \mathfrak{G}_{{\mathrm{1}}} }{ \mathfrak{G}_{{\mathrm{2}}} }  $ and $ \mathbb{S} \ottsym{=}  \mathbb{S}_{{\mathrm{1}}} \cup \mathbb{S}_{{\mathrm{2}}}  $ by definition.
        We consider the following cases.
        \begin{match}
            \item[$ \ottnt{b}  \in  \Gamma_{{\mathrm{1}}} $ and $ \ottnt{b}  \in  \Gamma_{{\mathrm{2}}} $]
            We have $  \lBrack \Gamma_{{\mathrm{1}}}  \ottsym{[}  \Gamma'  \ottsym{/}  \ottnt{b}  \ottsym{]} \rBrack  \ottsym{=} \mathfrak{G}_{{\mathrm{1}}}  \ottsym{[}  \mathfrak{G}'  \ottsym{/}  \ottnt{b}  \ottsym{]}  \ottsym{;}   \mathbb{S}_{{\mathrm{1}}} \cup \mathbb{S}'  $ and $  \lBrack \Gamma_{{\mathrm{2}}}  \ottsym{[}  \Gamma'  \ottsym{/}  \ottnt{b}  \ottsym{]} \rBrack  \ottsym{=} \mathfrak{G}_{{\mathrm{2}}}  \ottsym{[}  \mathfrak{G}'  \ottsym{/}  \ottnt{b}  \ottsym{]}  \ottsym{;}   \mathbb{S}_{{\mathrm{2}}} \cup \mathbb{S}'  $ by the induction hypohtesis.
            Then, $  \lBrack \Gamma  \ottsym{[}  \Gamma'  \ottsym{/}  \ottnt{b}  \ottsym{]} \rBrack  \ottsym{=}  \GJOIN{ \mathfrak{G}_{{\mathrm{1}}}  \ottsym{[}  \mathfrak{G}'  \ottsym{/}  \ottnt{b}  \ottsym{]} }{ \mathfrak{G}_{{\mathrm{2}}} }   \ottsym{[}  \mathfrak{G}'  \ottsym{/}  \ottnt{b}  \ottsym{]}  \ottsym{;}     \mathbb{S}_{{\mathrm{1}}} \cup \mathbb{S}'  \cup \mathbb{S}_{{\mathrm{2}}}  \cup \mathbb{S}'  $ by definition.
            To this end, we have $  \GJOIN{ \mathfrak{G}_{{\mathrm{1}}}  \ottsym{[}  \mathfrak{G}'  \ottsym{/}  \ottnt{b}  \ottsym{]} }{ \mathfrak{G}_{{\mathrm{2}}} }   \ottsym{[}  \mathfrak{G}'  \ottsym{/}  \ottnt{b}  \ottsym{]} \ottsym{=} \ottsym{(}   \GJOIN{ \mathfrak{G}_{{\mathrm{1}}} }{ \mathfrak{G}_{{\mathrm{2}}} }   \ottsym{)}  \ottsym{[}  \mathfrak{G}'  \ottsym{/}  \ottnt{b}  \ottsym{]} $ by \propref{graph/replace/join}.

            \item[$ \ottnt{b}  \in  \Gamma_{{\mathrm{1}}} $ and $ \ottnt{b}  \notin  \Gamma_{{\mathrm{2}}} $]
            We have $  \lBrack \Gamma_{{\mathrm{1}}}  \ottsym{[}  \Gamma'  \ottsym{/}  \ottnt{b}  \ottsym{]} \rBrack  \ottsym{=} \mathfrak{G}_{{\mathrm{1}}}  \ottsym{[}  \mathfrak{G}'  \ottsym{/}  \ottnt{b}  \ottsym{]}  \ottsym{;}   \mathbb{S}_{{\mathrm{1}}} \cup \mathbb{S}'  $ by the induction hypothesis.
            On the other hand, $  \lBrack \Gamma_{{\mathrm{2}}}  \ottsym{[}  \Gamma'  \ottsym{/}  \ottnt{b}  \ottsym{]} \rBrack  \ottsym{=}  \lBrack \Gamma_{{\mathrm{2}}} \rBrack   =  \mathfrak{G}_{{\mathrm{2}}}  \ottsym{;}  \mathbb{S}_{{\mathrm{2}}} \ottsym{=} \mathfrak{G}_{{\mathrm{2}}}  \ottsym{[}  \mathfrak{G}'  \ottsym{/}  \ottnt{b}  \ottsym{]}  \ottsym{;}  \mathbb{S}_{{\mathrm{2}}} $.
            Then, $  \lBrack \Gamma  \ottsym{[}  \Gamma'  \ottsym{/}  \ottnt{b}  \ottsym{]} \rBrack  \ottsym{=}  \GJOIN{ \mathfrak{G}_{{\mathrm{1}}}  \ottsym{[}  \mathfrak{G}'  \ottsym{/}  \ottnt{b}  \ottsym{]} }{ \mathfrak{G}_{{\mathrm{2}}} }   \ottsym{[}  \mathfrak{G}'  \ottsym{/}  \ottnt{b}  \ottsym{]}  \ottsym{;}    \mathbb{S}_{{\mathrm{1}}} \cup \mathbb{S}'  \cup \mathbb{S}_{{\mathrm{2}}}  $ by definition.
            To this end, we have $  \GJOIN{ \mathfrak{G}_{{\mathrm{1}}}  \ottsym{[}  \mathfrak{G}'  \ottsym{/}  \ottnt{b}  \ottsym{]} }{ \mathfrak{G}_{{\mathrm{2}}} }   \ottsym{[}  \mathfrak{G}'  \ottsym{/}  \ottnt{b}  \ottsym{]} \ottsym{=} \ottsym{(}   \GJOIN{ \mathfrak{G}_{{\mathrm{1}}} }{ \mathfrak{G}_{{\mathrm{2}}} }   \ottsym{)}  \ottsym{[}  \mathfrak{G}'  \ottsym{/}  \ottnt{b}  \ottsym{]} $ by \propref{graph/replace/join}.

            \item[$ \ottnt{b}  \notin  \Gamma_{{\mathrm{1}}} $ and $ \ottnt{b}  \in  \Gamma_{{\mathrm{2}}} $] Similar to the case above.
            \item[$ \ottnt{b}  \notin  \Gamma_{{\mathrm{1}}} $ and $ \ottnt{b}  \notin  \Gamma_{{\mathrm{2}}} $] Vacuously true since $ \ottnt{b}  \in  \Gamma $.
        \end{match}

        \item[$ \Gamma \ottsym{=} \Gamma_{{\mathrm{1}}}  \parallel  \Gamma_{{\mathrm{2}}} $] Similar to the case above except we use \propref{graph/replace/union}.
    \end{match}
\end{prop}

\begin{prop}{env/remove/equiv}
    If $\Gamma_{{\mathrm{1}}}  \simeq  \Gamma_{{\mathrm{2}}}$, then $ \Gamma_{{\mathrm{1}}} ^{- \ottnt{b} }   \simeq   \Gamma_{{\mathrm{2}}} ^{- \ottnt{b} } $.

    \proof
    We consider wether $ \ottnt{b}  \in  \Gamma_{{\mathrm{1}}} $
    \begin{match}
        \item[$ \ottnt{b}  \notin  \Gamma_{{\mathrm{1}}} $]
        We have $ \ottnt{b}  \notin  \Gamma_{{\mathrm{2}}} $ by \propref{env/in/equiv}.
        So, $  \Gamma_{{\mathrm{1}}} ^{- \ottnt{b} }  \ottsym{=} \Gamma_{{\mathrm{1}}} $ and $  \Gamma_{{\mathrm{2}}} ^{- \ottnt{b} }  \ottsym{=} \Gamma_{{\mathrm{2}}} $.
        Consequently, the goal directly follows by the assumption.
        \item[$ \ottnt{b}  \in  \Gamma_{{\mathrm{1}}} $]
        We have $ \ottnt{b}  \in  \Gamma_{{\mathrm{2}}} $ by \propref{env/in/equiv}.
        Let $ \mathfrak{G}_{{\mathrm{1}}}  \ottsym{;}  \mathbb{S}_{{\mathrm{1}}} \ottsym{=}  \lBrack \Gamma_{{\mathrm{1}}} \rBrack  $ and $ \mathfrak{G}_{{\mathrm{2}}}  \ottsym{;}  \mathbb{S}_{{\mathrm{2}}} \ottsym{=}  \lBrack \Gamma_{{\mathrm{2}}} \rBrack  $.
        We have
        \begin{gather}
            \mathfrak{G}_{{\mathrm{1}}}  \simeq  \mathfrak{G}_{{\mathrm{2}}} \hyp{1.1}\\
             \mathbb{S}_{{\mathrm{1}}} \ottsym{=} \mathbb{S}_{{\mathrm{2}}}  \hyp{1.2}
        \end{gather}
        by the assumption.
        Now we consider whether $\ottkw{unr} \, \ottnt{b}$ or not.
        \begin{match}
            \item[$\ottkw{unr} \, \ottnt{b}$]
            We have $  \lBrack  \Gamma_{{\mathrm{1}}} ^{- \ottnt{b} }  \rBrack  \ottsym{=} \mathfrak{G}_{{\mathrm{1}}}  \ottsym{;}   \mathbb{S}_{{\mathrm{1}}} \setminus \ottsym{\{}  \ottnt{b}  \ottsym{\}}  $ and $  \lBrack  \Gamma_{{\mathrm{2}}} ^{- \ottnt{b} }  \rBrack  \ottsym{=} \mathfrak{G}_{{\mathrm{2}}}  \ottsym{;}   \mathbb{S}_{{\mathrm{2}}} \setminus \ottsym{\{}  \ottnt{b}  \ottsym{\}}  $ by \propref{env/remove/unr}.
            So, the goal follows by \hypref{1.1} and \hypref{1.2}.

            \item[$\ottkw{ord} \, \ottnt{b}$]
            We have $  \lBrack  \Gamma_{{\mathrm{1}}} ^{- \ottnt{b} }  \rBrack  \ottsym{=} \mathfrak{G}_{{\mathrm{1}}}  \ottsym{[}   \lBrack  \cdot  \rBrack   \ottsym{/}  \ottnt{b}  \ottsym{]}  \ottsym{;}  \mathbb{S}_{{\mathrm{1}}} $ and $  \lBrack  \Gamma_{{\mathrm{2}}} ^{- \ottnt{b} }  \rBrack  \ottsym{=} \mathfrak{G}_{{\mathrm{2}}}  \ottsym{[}   \lBrack  \cdot  \rBrack   \ottsym{/}  \ottnt{b}  \ottsym{]}  \ottsym{;}  \mathbb{S}_{{\mathrm{2}}} $ by \propref{env/replace/ord}.
            So, the goal follows by \propref{graph/replace/iso}. \qedhere
        \end{match}
    \end{match}

\end{prop}

\begin{prop}{env/remove/sub}
    If $\Gamma_{{\mathrm{1}}}  \lesssim  \Gamma_{{\mathrm{2}}}$, then $ \Gamma_{{\mathrm{1}}} ^{- \ottnt{b} }   \lesssim   \Gamma_{{\mathrm{2}}} ^{- \ottnt{b} } $.

    \proof Let $ \mathfrak{G}_{{\mathrm{1}}}  \ottsym{;}  \mathbb{S}_{{\mathrm{1}}} \ottsym{=}  \lBrack \Gamma_{{\mathrm{1}}} \rBrack  $ and $ \mathfrak{G}_{{\mathrm{2}}}  \ottsym{;}  \mathbb{S}_{{\mathrm{2}}} \ottsym{=}  \lBrack \Gamma_{{\mathrm{2}}} \rBrack  $.
    We have
    \begin{gather}
        \mathfrak{G}_{{\mathrm{1}}}  \simeq  \mathfrak{G}'_{{\mathrm{1}}} \mathrel{<_\rightarrow} \mathfrak{G}'_{{\mathrm{2}}}  \simeq  \mathfrak{G}_{{\mathrm{2}}} \hyp{0.1}\\
         \mathbb{S}_{{\mathrm{1}}}  \supseteq  \mathbb{S}_{{\mathrm{2}}}  \hyp{0.2}
    \end{gather}
    for some $\mathfrak{G}'_{{\mathrm{1}}}$ and $\mathfrak{G}'_{{\mathrm{2}}}$ by the assumption.
    Now, we consider wether $\ottkw{unr} \, \ottnt{b}$ or not.
    \begin{match}
        \item[$\ottkw{unr} \, \ottnt{b}$] We have
        \begin{gather*}
              \lBrack  \Gamma_{{\mathrm{1}}} ^{- \ottnt{b} }  \rBrack  \ottsym{=} \mathfrak{G}_{{\mathrm{1}}}  \ottsym{;}   \mathbb{S}_{{\mathrm{1}}} \setminus \ottsym{\{}  \ottnt{b}  \ottsym{\}}   \\
              \lBrack  \Gamma_{{\mathrm{2}}} ^{- \ottnt{b} }  \rBrack  \ottsym{=} \mathfrak{G}_{{\mathrm{2}}}  \ottsym{;}   \mathbb{S}_{{\mathrm{2}}} \setminus \ottsym{\{}  \ottnt{b}  \ottsym{\}}  
        \end{gather*}
        by \propref{env/remove/unr}.
        So, the goal follows by \hypref{0.1} and \hypref{0.2}.

        \item[$\ottkw{ord} \, \ottnt{b}$] We consider whether $ \ottnt{b}  \in  \Gamma_{{\mathrm{1}}} $ or not.
        \begin{match}
            \item[$ \ottnt{b}  \in  \Gamma_{{\mathrm{1}}} $] We have $ \ottnt{b}  \in  \Gamma_{{\mathrm{2}}} $ by \propref{env/in/sub}.
            We have
            \begin{gather*}
                  \lBrack  \Gamma_{{\mathrm{1}}} ^{- \ottnt{b} }  \rBrack  \ottsym{=} \mathfrak{G}_{{\mathrm{1}}}  \ottsym{[}   \lBrack  \cdot  \rBrack   \ottsym{/}  \ottnt{b}  \ottsym{]}  \ottsym{;}  \mathbb{S}_{{\mathrm{1}}}  \\
                  \lBrack  \Gamma_{{\mathrm{2}}} ^{- \ottnt{b} }  \rBrack  \ottsym{=} \mathfrak{G}_{{\mathrm{2}}}  \ottsym{[}   \lBrack  \cdot  \rBrack   \ottsym{/}  \ottnt{b}  \ottsym{]}  \ottsym{;}  \mathbb{S}_{{\mathrm{2}}} 
            \end{gather*}
            by \propref{env/replace/ord}.
            So, the goal follows by \propref{graph/replace/iso}, \propref{graph/replace/span}, \hypref{0.1}, and \hypref{0.2}.

            \item[$ \ottnt{b}  \notin  \Gamma_{{\mathrm{1}}} $] We have $ \ottnt{b}  \notin  \Gamma_{{\mathrm{2}}} $ by \propref{env/in/sup}.
            So, $  \Gamma_{{\mathrm{1}}} ^{- \ottnt{b} }  \ottsym{=} \Gamma_{{\mathrm{1}}} $ and $  \Gamma_{{\mathrm{2}}} ^{- \ottnt{b} }  \ottsym{=} \Gamma_{{\mathrm{2}}} $.
            Consequently, the goal directly follows by the assumption.
        \end{match}
    \end{match}
\end{prop}

\begin{prop}{env/replace/equiv}
    If $\Gamma_{{\mathrm{1}}}  \simeq  \Gamma_{{\mathrm{2}}}$ and $\ottkw{ord} \, \ottnt{T}$, then $\Gamma_{{\mathrm{1}}}  \ottsym{[}  \Gamma  \ottsym{/}  \ottmv{x}  \mathord:  \ottnt{T}  \ottsym{]}  \simeq  \Gamma_{{\mathrm{2}}}  \ottsym{[}  \Gamma  \ottsym{/}  \ottmv{x}  \mathord:  \ottnt{T}  \ottsym{]}$.

    \proof Let $ \mathfrak{G}_{{\mathrm{1}}}  \ottsym{;}  \mathbb{S}_{{\mathrm{1}}} \ottsym{=}  \lBrack \Gamma_{{\mathrm{1}}} \rBrack  $, $ \mathfrak{G}_{{\mathrm{2}}}  \ottsym{;}  \mathbb{S}_{{\mathrm{2}}} \ottsym{=}  \lBrack \Gamma_{{\mathrm{2}}} \rBrack  $, and $ \mathfrak{G}  \ottsym{;}  \mathbb{S} \ottsym{=}  \lBrack \Gamma \rBrack  $.
    Then, we have $\mathfrak{G}_{{\mathrm{1}}}  \simeq  \mathfrak{G}_{{\mathrm{2}}}$ and $ \mathbb{S}_{{\mathrm{1}}} \ottsym{=} \mathbb{S}_{{\mathrm{2}}} $ by definition.
    We consider if $ \ottmv{x}  \mathord:  \ottnt{T}  \in  \Gamma_{{\mathrm{1}}} $
    \begin{match}
        \item[$ \ottmv{x}  \mathord:  \ottnt{T}  \in  \Gamma_{{\mathrm{1}}} $]
        We have $ \ottmv{x}  \mathord:  \ottnt{T}  \in  \Gamma_{{\mathrm{2}}} $ by \propref{env/in/equiv}.
        So, $  \lBrack \Gamma_{{\mathrm{1}}}  \ottsym{[}  \Gamma  \ottsym{/}  \ottmv{x}  \mathord:  \ottnt{T}  \ottsym{]} \rBrack  \ottsym{=} \mathfrak{G}_{{\mathrm{1}}}  \ottsym{[}  \mathfrak{G}  \ottsym{/}  \ottmv{x}  \mathord:  \ottnt{T}  \ottsym{]}  \ottsym{;}   \mathbb{S}_{{\mathrm{1}}} \cup \mathbb{S}  $ and $  \lBrack \Gamma_{{\mathrm{2}}}  \ottsym{[}  \Gamma  \ottsym{/}  \ottmv{x}  \mathord:  \ottnt{T}  \ottsym{]} \rBrack  \ottsym{=} \mathfrak{G}_{{\mathrm{2}}}  \ottsym{[}  \mathfrak{G}  \ottsym{/}  \ottmv{x}  \mathord:  \ottnt{T}  \ottsym{]}  \ottsym{;}   \mathbb{S}_{{\mathrm{2}}} \cup \mathbb{S}  $ by \propref{env/replace/ord}.
        To this end, it suffices to show $\mathfrak{G}_{{\mathrm{1}}}  \ottsym{[}  \mathfrak{G}  \ottsym{/}  \ottmv{x}  \mathord:  \ottnt{T}  \ottsym{]}  \simeq  \mathfrak{G}_{{\mathrm{2}}}  \ottsym{[}  \mathfrak{G}  \ottsym{/}  \ottmv{x}  \mathord:  \ottnt{T}  \ottsym{]}$, which follows by \propref{graph/replace/iso}.

        \item[$ \ottmv{x}  \mathord:  \ottnt{T}  \notin  \Gamma_{{\mathrm{1}}} $]
        We have $ \ottmv{x}  \mathord:  \ottnt{T}  \notin  \Gamma_{{\mathrm{2}}} $ by \propref{env/in/equiv}.
        So, $ \Gamma_{{\mathrm{1}}}  \ottsym{[}  \Gamma  \ottsym{/}  \ottmv{x}  \mathord:  \ottnt{T}  \ottsym{]} \ottsym{=} \Gamma_{{\mathrm{1}}} $ and $ \Gamma_{{\mathrm{2}}}  \ottsym{[}  \Gamma  \ottsym{/}  \ottmv{x}  \mathord:  \ottnt{T}  \ottsym{]} \ottsym{=} \Gamma_{{\mathrm{2}}} $.
        Now the goal directly follows by the assumption.
    \end{match}



\end{prop}

\begin{prop}{env/replace/runtime/sub}
    If $\Gamma_{{\mathrm{1}}}  \lesssim  \Gamma_{{\mathrm{2}}}$ and $\ottkw{ord} \, \ottnt{T}$, then $\Gamma_{{\mathrm{1}}}  \ottsym{[}  \ottnt{C}  \ottsym{/}  \ottmv{x}  \mathord:  \ottnt{T}  \ottsym{]}  \lesssim  \Gamma_{{\mathrm{2}}}  \ottsym{[}  \ottnt{C}  \ottsym{/}  \ottmv{x}  \mathord:  \ottnt{T}  \ottsym{]}$.

    \proof Let $ \mathfrak{G}_{{\mathrm{1}}}  \ottsym{;}  \mathbb{S}_{{\mathrm{1}}} \ottsym{=}  \lBrack \Gamma_{{\mathrm{1}}} \rBrack  $ and $ \mathfrak{G}_{{\mathrm{2}}}  \ottsym{;}  \mathbb{S}_{{\mathrm{2}}} \ottsym{=}  \lBrack \Gamma_{{\mathrm{2}}} \rBrack  $.
    By definition we have $\mathfrak{G}_{{\mathrm{1}}}  \simeq  \mathfrak{G}'_{{\mathrm{1}}} \mathrel{<_\rightarrow} \mathfrak{G}'_{{\mathrm{2}}}  \simeq  \mathfrak{G}_{{\mathrm{2}}}$ for some $\mathfrak{G}'_{{\mathrm{1}}}$ and $\mathfrak{G}'_{{\mathrm{2}}}$ and $ \mathbb{S}_{{\mathrm{1}}}  \supseteq  \mathbb{S}_{{\mathrm{2}}} $.
    We consider wether $ \ottmv{x}  \mathord:  \ottnt{T}  \in  \Gamma_{{\mathrm{1}}} $.
    \begin{match}
        \item[$ \ottmv{x}  \mathord:  \ottnt{T}  \in  \Gamma_{{\mathrm{1}}} $]
        We have $ \ottmv{x}  \mathord:  \ottnt{T}  \in  \Gamma_{{\mathrm{2}}} $ by \propref{env/in/sub}.
        So, we have
        \begin{gather*}
              \lBrack \Gamma_{{\mathrm{1}}}  \ottsym{[}  \ottnt{C}  \ottsym{/}  \ottmv{x}  \mathord:  \ottnt{T}  \ottsym{]} \rBrack  \ottsym{=} \mathfrak{G}_{{\mathrm{1}}}  \ottsym{[}   \lBrack \ottnt{C} \rBrack   \ottsym{/}  \ottmv{x}  \mathord:  \ottnt{T}  \ottsym{]}  \ottsym{;}  \mathbb{S}_{{\mathrm{1}}}  \\
              \lBrack \Gamma_{{\mathrm{2}}}  \ottsym{[}  \ottnt{C}  \ottsym{/}  \ottmv{x}  \mathord:  \ottnt{T}  \ottsym{]} \rBrack  \ottsym{=} \mathfrak{G}_{{\mathrm{2}}}  \ottsym{[}   \lBrack \ottnt{C} \rBrack   \ottsym{/}  \ottmv{x}  \mathord:  \ottnt{T}  \ottsym{]}  \ottsym{;}  \mathbb{S}_{{\mathrm{2}}} 
        \end{gather*}
        by \propref{env/replace/ord}.
        So, to this end, we get some $\mathfrak{G}''_{{\mathrm{1}}}$ and $\mathfrak{G}''_{{\mathrm{2}}}$ such that $\mathfrak{G}_{{\mathrm{1}}}  \ottsym{[}   \lBrack \ottnt{C} \rBrack   \ottsym{/}  \ottmv{x}  \mathord:  \ottnt{T}  \ottsym{]}  \simeq  \mathfrak{G}''_{{\mathrm{1}}} \mathrel{<_\rightarrow} \mathfrak{G}''_{{\mathrm{2}}}  \simeq  \mathfrak{G}_{{\mathrm{2}}}  \ottsym{[}   \lBrack \ottnt{C} \rBrack   \ottsym{/}  \ottmv{x}  \mathord:  \ottnt{T}  \ottsym{]}$.
        We choose $ \mathfrak{G}''_{{\mathrm{1}}} \ottsym{=} \mathfrak{G}'_{{\mathrm{1}}}  \ottsym{[}   \lBrack \ottnt{C} \rBrack   \ottsym{/}  \ottmv{x}  \mathord:  \ottnt{T}  \ottsym{]} $ and $ \mathfrak{G}''_{{\mathrm{2}}} \ottsym{=} \mathfrak{G}'_{{\mathrm{2}}}  \ottsym{[}   \lBrack \ottnt{C} \rBrack   \ottsym{/}  \ottmv{x}  \mathord:  \ottnt{T}  \ottsym{]} $.
        The condition follows by \propref{graph/replace/iso} and \propref{graph/replace/span}.

        \item[$ \ottmv{x}  \mathord:  \ottnt{T}  \notin  \Gamma_{{\mathrm{1}}} $]
        We have $ \ottmv{x}  \mathord:  \ottnt{T}  \notin  \Gamma_{{\mathrm{2}}} $ by \propref{env/in/sup}.
        So, the goal directly follows by the assumption.
    \end{match}

\end{prop}

\subsection{Properties for runtime environment}

\begin{prop}{env/runtime/focus/null}
    If $ \ottmv{l}  \notin   \ottkw{dom} ( \ottnt{C} )  $, then $ \langle \ottnt{C} \rangle_{ \ottmv{l} }   \simeq   \cdot $.

    \proof Routine by structural induction on $\ottnt{C}$.
    We use \propref{env/equiv}.
\end{prop}

\begin{prop}{env/runtime/focus/equiv}
    If $\ottnt{C_{{\mathrm{1}}}}  \simeq  \ottnt{C_{{\mathrm{2}}}}$, then $ \langle \ottnt{C_{{\mathrm{1}}}} \rangle_{ \ottmv{l} }   \simeq   \langle \ottnt{C_{{\mathrm{2}}}} \rangle_{ \ottmv{l} } $.

    \proof By mathematical induction on the number of elements in $ \ottkw{dom} ( \ottnt{C_{{\mathrm{1}}}} ) $.
    We use \propref{env/remove/equiv} in the induction step.
\end{prop}

\begin{prop}{env/runtime/focus/sub}
    If $\ottnt{C_{{\mathrm{1}}}}  \lesssim  \ottnt{C_{{\mathrm{2}}}}$, then $ \langle \ottnt{C_{{\mathrm{1}}}} \rangle_{ \ottmv{l} }   \lesssim   \langle \ottnt{C_{{\mathrm{2}}}} \rangle_{ \ottmv{l} } $.

    \proof By mathematical induction on the number of elements in $ \ottkw{dom} ( \ottnt{C_{{\mathrm{1}}}} ) $.
    We use \propref{env/remove/sub} in the induction step.
\end{prop}

\begin{prop}{env/runtime/intrp/order/ctx}
    $  |  \lBrack \mathcal{C}  \ottsym{[}  \ottnt{C}  \ottsym{]} \rBrack  |_{\bullet}  \ottsym{=}  |  \lBrack \mathcal{C}  \ottsym{[}   \cdot   \ottsym{]} \rBrack  |_{\bullet}   \ottsym{+}   |  \lBrack \ottnt{C} \rBrack  |_{\bullet}  $.

    \proof By structural induction on $\mathcal{C}$.
    \begin{match}
        \item[$ \mathcal{C} \ottsym{=} \ottsym{[]} $]
        In this case,
        \begin{align*}
             |  \lBrack \mathcal{C}  \ottsym{[}  \ottnt{C}  \ottsym{]} \rBrack  |_{\bullet}                    & =  |  \lBrack \ottnt{C} \rBrack  |_{\bullet}  \\
             |  \lBrack \mathcal{C}  \ottsym{[}   \cdot   \ottsym{]} \rBrack  |_{\bullet}   \ottsym{+}   |  \lBrack \ottnt{C} \rBrack  |_{\bullet}  & =  |  \lBrack \ottnt{C} \rBrack  |_{\bullet} 
        \end{align*}
        by definition.
        So, we can see the goal holds.

        \item[$ \mathcal{C} \ottsym{=} \mathcal{C}_{{\mathrm{1}}}  \ottsym{,}  \ottnt{C_{{\mathrm{2}}}} $]
        In this case,
        \begin{align*}
             |  \lBrack \mathcal{C}  \ottsym{[}  \ottnt{C}  \ottsym{]} \rBrack  |_{\bullet}                    & =  |  \lBrack \mathcal{C}_{{\mathrm{1}}}  \ottsym{[}  \ottnt{C}  \ottsym{]} \rBrack  |_{\bullet}   \ottsym{+}   |  \lBrack \ottnt{C_{{\mathrm{2}}}} \rBrack  |_{\bullet}                    \\
             |  \lBrack \mathcal{C}  \ottsym{[}   \cdot   \ottsym{]} \rBrack  |_{\bullet}   \ottsym{+}   |  \lBrack \ottnt{C} \rBrack  |_{\bullet}  & =  |  \lBrack \mathcal{C}_{{\mathrm{1}}}  \ottsym{[}   \cdot   \ottsym{]} \rBrack  |_{\bullet}   \ottsym{+}   |  \lBrack \ottnt{C_{{\mathrm{2}}}} \rBrack  |_{\bullet}   \ottsym{+}   |  \lBrack \ottnt{C} \rBrack  |_{\bullet} 
        \end{align*}
        by definition.
        Now, we have $  |  \lBrack \mathcal{C}_{{\mathrm{1}}}  \ottsym{[}  \ottnt{C}  \ottsym{]} \rBrack  |_{\bullet}  \ottsym{=}  |  \lBrack \mathcal{C}_{{\mathrm{1}}}  \ottsym{[}   \cdot   \ottsym{]} \rBrack  |_{\bullet}   \ottsym{+}   |  \lBrack \ottnt{C} \rBrack  |_{\bullet}  $ by the induction hypothesis, by which we can have the goal.

        \item[$ \mathcal{C} \ottsym{=} \ottnt{C_{{\mathrm{1}}}}  \ottsym{,}  \mathcal{C}_{{\mathrm{2}}} $, $ \mathcal{C} \ottsym{=} \mathcal{C}_{{\mathrm{1}}}  \parallel  \ottnt{C_{{\mathrm{2}}}} $, and $ \mathcal{C} \ottsym{=} \ottnt{C_{{\mathrm{1}}}}  \parallel  \mathcal{C}_{{\mathrm{2}}} $]
        These cases are similar to the case $ \mathcal{C} \ottsym{=} \mathcal{C}_{{\mathrm{1}}}  \ottsym{,}  \ottnt{C_{{\mathrm{2}}}} $. \qedhere
    \end{match}
\end{prop}

\begin{prop}{env/runtime/intrp/traceable/ctx}
    If $ \lBrack \mathcal{C}  \ottsym{[}  \ottnt{C}  \ottsym{]} \rBrack $ is traceable, then $ \lBrack \ottnt{C} \rBrack $ is traceable.

    \proof By structural induction on $\mathcal{C}$.
    \begin{match}
        \item[$ \mathcal{C} \ottsym{=} \ottsym{[]} $]
        This case follows by the assumption.

        \item[$ \mathcal{C} \ottsym{=} \mathcal{C}_{{\mathrm{1}}}  \ottsym{,}  \ottnt{C_{{\mathrm{2}}}} $]
        In this case
        \begin{gather*}
             \GJOIN{  \lBrack \mathcal{C}_{{\mathrm{1}}}  \ottsym{[}  \ottnt{C}  \ottsym{]} \rBrack  }{  \lBrack \ottnt{C_{{\mathrm{2}}}} \rBrack  } 
        \end{gather*}
        is traceable.
        Then, we know $ \lBrack \mathcal{C}_{{\mathrm{1}}}  \ottsym{[}  \ottnt{C}  \ottsym{]} \rBrack $ is traceable by \propref{graph/traceable/join}.
        Now, we can have the goal by the induction hypothesis.

        \item[$ \mathcal{C} \ottsym{=} \ottnt{C_{{\mathrm{1}}}}  \ottsym{,}  \mathcal{C}_{{\mathrm{2}}} $] Similar to the case $ \mathcal{C} \ottsym{=} \mathcal{C}_{{\mathrm{1}}}  \ottsym{,}  \ottnt{C_{{\mathrm{2}}}} $.
        \item[$ \mathcal{C} \ottsym{=} \mathcal{C}_{{\mathrm{1}}}  \parallel  \ottnt{C_{{\mathrm{2}}}} $]
        In this case
        \begin{gather*}
             \lBrack \mathcal{C}_{{\mathrm{1}}}  \ottsym{[}  \ottnt{C}  \ottsym{]} \rBrack   \cup   \lBrack \ottnt{C_{{\mathrm{2}}}} \rBrack 
        \end{gather*}
        is traceable.
        So, we have the following cases by \propref{graph/traceable/union}.
        \begin{match}
            \item[$  |  \lBrack \mathcal{C}_{{\mathrm{1}}}  \ottsym{[}  \ottnt{C}  \ottsym{]} \rBrack  |_{\bullet}  \ottsym{=} \ottsym{0} $]
            We have $  |  \lBrack \ottnt{C} \rBrack  |_{\bullet}  \ottsym{=} \ottsym{0} $ by \propref{env/runtime/intrp/order/ctx}.
            So, $  \lBrack \ottnt{C} \rBrack  \ottsym{=}  \lBrack  \cdot  \rBrack  $ by definition, which is trivially traceable.

            \item[$ \lBrack \mathcal{C}_{{\mathrm{1}}}  \ottsym{[}  \ottnt{C}  \ottsym{]} \rBrack $ is traceable]
            We can have the goal by the induction hypothsis.
        \end{match}

        \item[$ \mathcal{C} \ottsym{=} \ottnt{C_{{\mathrm{1}}}}  \parallel  \mathcal{C}_{{\mathrm{2}}} $] Similar to the case $ \mathcal{C} \ottsym{=} \mathcal{C}_{{\mathrm{1}}}  \parallel  \ottnt{C_{{\mathrm{2}}}} $. \qedhere
    \end{match}
\end{prop}

\begin{prop}{env/runtime/intrp/traceable/equiv}
    Suppose $\ottnt{C_{{\mathrm{1}}}}  \simeq  \ottnt{C_{{\mathrm{2}}}}$.
    Then, $ \lBrack \ottnt{C_{{\mathrm{1}}}} \rBrack $ is traceable if and only if $ \lBrack \ottnt{C_{{\mathrm{2}}}} \rBrack $ is traceable.

    \proof This is a corollary from \propref{graph/traceable/iso}.


\end{prop}

\begin{prop}{env/runtime/order-defined/equiv}
    Suppose $\ottnt{C_{{\mathrm{1}}}}  \simeq  \ottnt{C_{{\mathrm{2}}}}$.
    Then, $\ottnt{C_{{\mathrm{1}}}}$ is order-defined if and only if $\ottnt{C_{{\mathrm{2}}}}$ is order-defined.

    \proof
    It suffices to show
    \begin{quote}
        $ \lBrack  \langle \ottnt{C_{{\mathrm{1}}}} \rangle_{ \ottmv{l} }  \rBrack $ is traceable if and only if $ \lBrack  \langle \ottnt{C_{{\mathrm{2}}}} \rangle_{ \ottmv{l} }  \rBrack $ is traceable
    \end{quote}
    for any $\ottmv{l}$.
    We have $ \langle \ottnt{C_{{\mathrm{1}}}} \rangle_{ \ottmv{l} }   \simeq   \langle \ottnt{C_{{\mathrm{2}}}} \rangle_{ \ottmv{l} } $ by \propref{env/runtime/focus/equiv}.
    Then, the goal follows by \propref{graph/traceable/iso}.
\end{prop}

\begin{prop}{env/runtime/order-defined/sub}
    If $\ottnt{C_{{\mathrm{1}}}}  \lesssim  \ottnt{C_{{\mathrm{2}}}}$ and $\ottnt{C_{{\mathrm{1}}}}$ is order-defined, then $\ottnt{C_{{\mathrm{2}}}}$ is order-defined.

    \proof
    It suffices to show
    \begin{quote}
        $ \lBrack  \langle \ottnt{C_{{\mathrm{2}}}} \rangle_{ \ottmv{l} }  \rBrack $ is traceable if $ \lBrack  \langle \ottnt{C_{{\mathrm{1}}}} \rangle_{ \ottmv{l} }  \rBrack $ is traceable
    \end{quote}
    for any $\ottmv{l}$.
    We have $ \langle \ottnt{C_{{\mathrm{1}}}} \rangle_{ \ottmv{l} }   \lesssim   \langle \ottnt{C_{{\mathrm{2}}}} \rangle_{ \ottmv{l} } $ by \propref{env/runtime/focus/sub}.
    Then, the goal follows by \propref{graph/traceable/iso} and \propref{graph/traceable/span}.
\end{prop}

\begin{prop}{env/runtime/usage/equiv}
    Suppose $ \lBrack \ottnt{C_{{\mathrm{1}}}} \rBrack $ and $ \lBrack \ottnt{C_{{\mathrm{2}}}} \rBrack $ are traceable.
    If $\ottnt{C_{{\mathrm{1}}}}  \simeq  \ottnt{C_{{\mathrm{2}}}}$, then $  \overline{ \ottnt{C_{{\mathrm{1}}}} }  \ottsym{=}  \overline{ \ottnt{C_{{\mathrm{2}}}} }  $.

    \proof Let $f_1$ and $f_2$ be the unique topological orderings of $ \lBrack \ottnt{C_{{\mathrm{1}}}} \rBrack $ and $ \lBrack \ottnt{C_{{\mathrm{2}}}} \rBrack $, respectively.
    Let $(n_1, v_1, E_1) =  \lBrack \ottnt{C_{{\mathrm{1}}}} \rBrack $ and $(n_2, v_2, E_2) =  \lBrack \ottnt{C_{{\mathrm{2}}}} \rBrack $.
    By definition, we have $n_1 = n_2$.
    Moreover, we have some topological ordering $f_2'$ of $ \lBrack \ottnt{C_{{\mathrm{2}}}} \rBrack $ by \propref{graph/ordering/proj/iso} such that $v_1 \circ f_1 = v_2 \circ f_2'$.
    Since $\ottnt{C_{{\mathrm{2}}}}$ is order-defined, $f_2' = f_2$, and thus, $v_1 \circ f_1 = v_2 \circ f_2$.
    This means $  \overline{ \ottnt{C_{{\mathrm{1}}}} }  \ottsym{=}  \overline{ \ottnt{C_{{\mathrm{2}}}} }  $.

\end{prop}

\begin{prop}{env/runtime/usage/sub}
    Suppose $ \lBrack \ottnt{C_{{\mathrm{1}}}} \rBrack $ and $ \lBrack \ottnt{C_{{\mathrm{2}}}} \rBrack $ are traceable.
    If $\ottnt{C_{{\mathrm{1}}}}  \lesssim  \ottnt{C_{{\mathrm{2}}}}$, then $  \overline{ \ottnt{C_{{\mathrm{1}}}} }  \ottsym{=}  \overline{ \ottnt{C_{{\mathrm{2}}}} }  $.

    \proof Let $(n_1, v_1, E_1) =  \lBrack \ottnt{C_{{\mathrm{1}}}} \rBrack $ and $(n_2, v_2, E_2) =  \lBrack \ottnt{C_{{\mathrm{2}}}} \rBrack $.
    By definition, we have $(n_1, v_1, E_1) \simeq \mathfrak{G}'_{{\mathrm{1}}} \mathrel{<_\rightarrow} \mathfrak{G}'_{{\mathrm{2}}} \simeq (n_2, v_2, E_2)$ for some $\mathfrak{G}'_{{\mathrm{1}}} = (n_1', v_1', E_1')$ and $\mathfrak{G}'_{{\mathrm{2}}} = (n_2', v_2', E_2')$, where $n_1 = n_1' = n_2' = n_2$  and $v_1' = v_2'$.
    Here $\mathfrak{G}'_{{\mathrm{1}}}$ and $\mathfrak{G}'_{{\mathrm{2}}}$ are also traceable by \propref{graph/traceable/iso}.
    Let $f_1$ be the unique topological ordering of $(n_1, v_1, E_1)$.
    By \propref{graph/ordering/proj/iso}, we have some topological ordering $f_1'$ of $(n_1', v_1', E_1')$ such that $v_1 \circ f_1 = v_1' \circ f_1'$.
    By \propref{graph/ordering/span}, $f_1'$ is also a topological ordering of $(n_2', v_2', E_2')$.
    By \propref{graph/ordering/proj/iso} again, we have some topological ordering $f_2'$ of $(n_2, v_2, E_2)$ such that $v_2' \circ f_1' = v_2 \circ f_2'$.
    Knowing $v_1' = v_2'$, we have
    \begin{gather}
        v_1 \circ f_1 = v_2 \circ f_2' \hyp{1.1}
    \end{gather}
    Since, a topological ordering of $(n_2, v_2, E_2)$ is unique by definition.
    \hypref{1.1} implies $  \overline{ \ottnt{C_{{\mathrm{1}}}} }  \ottsym{=}  \overline{ \ottnt{C_{{\mathrm{2}}}} }  $.

\end{prop}




\begin{prop}{env/runtime/equiv/null/ctx}
    If $\mathcal{C}  \ottsym{[}  \ottnt{C}  \ottsym{]}  \simeq   \cdot $, then $\ottnt{C}  \simeq   \cdot $.

    \proof We have $ \lBrack \mathcal{C}  \ottsym{[}  \ottnt{C}  \ottsym{]} \rBrack   \simeq   \lBrack  \cdot  \rBrack $ by definition.
    So, $ |  \lBrack \mathcal{C}  \ottsym{[}   \cdot   \ottsym{]} \rBrack  |_{\bullet}  +  |  \lBrack \ottnt{C} \rBrack  |_{\bullet}  = 0$ by \propref{env/runtime/intrp/order/ctx}.
    From that, we can see $  |  \lBrack \ottnt{C} \rBrack  |_{\bullet}  \ottsym{=} \ottsym{0} $, which implies $ \lBrack \ottnt{C} \rBrack   \simeq   \lBrack  \cdot  \rBrack $, and thus, $\ottnt{C}  \simeq   \cdot $.
\end{prop}

\begin{prop}{env/runtime/intrp/traceable/ctx/replace}
    Suppose $\ottnt{C}  \not\simeq   \cdot $.
    If $ \lBrack \mathcal{C}  \ottsym{[}  \ottnt{C}  \ottsym{]} \rBrack $ and $ \lBrack \ottnt{C'} \rBrack $ are traceable, then $ \lBrack \mathcal{C}  \ottsym{[}  \ottnt{C'}  \ottsym{]} \rBrack $ is traceable.

    \proof By structural induction on $\mathcal{C}$.
    \begin{match}
        \item[$ \mathcal{C} \ottsym{=} \ottsym{[]} $]
        This case immediately follows by the assumption $ \lBrack \ottnt{C'} \rBrack $ is traceable.

        \item[$ \mathcal{C} \ottsym{=} \mathcal{C}_{{\mathrm{1}}}  \ottsym{,}  \ottnt{C_{{\mathrm{2}}}} $]
        In this case, $ \GJOIN{  \lBrack \mathcal{C}_{{\mathrm{1}}}  \ottsym{[}  \ottnt{C}  \ottsym{]} \rBrack  }{  \lBrack \ottnt{C_{{\mathrm{2}}}} \rBrack  } $ is traceable by definition.
        So, $ \lBrack \mathcal{C}_{{\mathrm{1}}}  \ottsym{[}  \ottnt{C}  \ottsym{]} \rBrack $ and $ \lBrack \ottnt{C_{{\mathrm{2}}}} \rBrack $ are tracebale by \propref{graph/traceable/join}.
        We have $ \lBrack \mathcal{C}_{{\mathrm{1}}}  \ottsym{[}  \ottnt{C'}  \ottsym{]} \rBrack $ is traceable by the induction hypothesis.
        Therefore, $  \GJOIN{  \lBrack \mathcal{C}_{{\mathrm{1}}}  \ottsym{[}  \ottnt{C'}  \ottsym{]} \rBrack  }{  \lBrack \ottnt{C_{{\mathrm{2}}}} \rBrack  }  \ottsym{=}  \lBrack \mathcal{C}  \ottsym{[}  \ottnt{C'}  \ottsym{]} \rBrack  $ is traceable by \propref{graph/traceable/join}.

        \item[$ \mathcal{C} \ottsym{=} \ottnt{C_{{\mathrm{1}}}}  \ottsym{,}  \mathcal{C}_{{\mathrm{2}}} $] Similar to the case $ \mathcal{C} \ottsym{=} \mathcal{C}_{{\mathrm{1}}}  \ottsym{,}  \ottnt{C_{{\mathrm{2}}}} $.

        \item[$ \mathcal{C} \ottsym{=} \mathcal{C}_{{\mathrm{1}}}  \parallel  \ottnt{C_{{\mathrm{2}}}} $]
        In this case, $ \lBrack \mathcal{C}_{{\mathrm{1}}}  \ottsym{[}  \ottnt{C}  \ottsym{]} \rBrack   \cup   \lBrack \ottnt{C_{{\mathrm{2}}}} \rBrack $ is traceable by definition.
        So, we have the following cases by \propref{graph/traceable/union}.
        \begin{match}
            \item[$  |  \lBrack \mathcal{C}_{{\mathrm{1}}}  \ottsym{[}  \ottnt{C}  \ottsym{]} \rBrack  |_{\bullet}  \ottsym{=} \ottsym{0} $ and $ \lBrack \ottnt{C_{{\mathrm{2}}}} \rBrack $ is traceable]
            We have $ |  \lBrack \mathcal{C}_{{\mathrm{1}}}  \ottsym{[}   \cdot   \ottsym{]} \rBrack  |_{\bullet}  +  |  \lBrack \ottnt{C} \rBrack  |_{\bullet}  = 0$ by \propref{env/runtime/intrp/order/ctx}.
            That implies $  |  \lBrack \ottnt{C} \rBrack  |_{\bullet}  \ottsym{=} \ottsym{0} $, and thus $  \lBrack \ottnt{C} \rBrack  \ottsym{=}  \lBrack  \cdot  \rBrack  $.
            Therefore, $\ottnt{C}  \simeq   \cdot $, but, this contradicts to the assumption $\ottnt{C}  \not\simeq   \cdot $.
            So, this case is vacuously true.

            \item[$  |  \lBrack \ottnt{C_{{\mathrm{2}}}} \rBrack  |_{\bullet}  \ottsym{=} \ottsym{0} $ and $ \lBrack \mathcal{C}_{{\mathrm{1}}}  \ottsym{[}  \ottnt{C}  \ottsym{]} \rBrack $ is traceable]
            We have $ \lBrack \mathcal{C}_{{\mathrm{1}}}  \ottsym{[}  \ottnt{C'}  \ottsym{]} \rBrack $ is traceable by the induction hypothesis.
            So, $  \lBrack \mathcal{C}  \ottsym{[}  \ottnt{C'}  \ottsym{]} \rBrack  \ottsym{=}  \lBrack \mathcal{C}_{{\mathrm{1}}}  \ottsym{[}  \ottnt{C'}  \ottsym{]} \rBrack   \cup   \lBrack \ottnt{C} \rBrack  $ is traceable by \propref{graph/traceable/union}.
        \end{match}

        \item[$ \mathcal{G} \ottsym{=} \Gamma_{{\mathrm{1}}}  \parallel  \mathcal{G}_{{\mathrm{2}}} $] Similar to the case $ \mathcal{G} \ottsym{=} \Gamma_{{\mathrm{1}}}  \ottsym{,}  \mathcal{G}_{{\mathrm{2}}} $.
    \end{match}
\end{prop}

\begin{prop}{env/ordering/pop}
    If $\ottmv{l}  \mathord:  \ottsym{[}  \ottnt{m}  \ottsym{]}  \ottsym{,}  \ottnt{C}$ is order-defined, then $\ottnt{C}$ is order-defined.

    \proof We show
    \begin{quote}
        if $ \lBrack  \langle \ottsym{(}  \ottmv{l}  \mathord:  \ottsym{[}  \ottnt{m}  \ottsym{]}  \ottsym{,}  \ottnt{C}  \ottsym{)} \rangle_{ \ottmv{l'} }  \rBrack $ is traceable, then $ \lBrack  \langle \ottnt{C} \rangle_{ \ottmv{l'} }  \rBrack $ is traceable
    \end{quote}
    for any $\ottmv{l'}$.
    We consider wether $ \ottmv{l'} \ottsym{=} \ottmv{l} $ or not.
    \begin{match}
        \item[$ \ottmv{l'} \ottsym{=} \ottmv{l} $]
        In this case, $  \lBrack  \langle \ottsym{(}  \ottmv{l}  \mathord:  \ottsym{[}  \ottnt{m}  \ottsym{]}  \ottsym{,}  \ottnt{C}  \ottsym{)} \rangle_{ \ottmv{l'} }  \rBrack  \ottsym{=}  \GJOIN{  \lBrack \ottmv{l}  \mathord:  \ottsym{[}  \ottnt{m}  \ottsym{]} \rBrack  }{  \lBrack  \langle \ottnt{C} \rangle_{ \ottmv{l'} }  \rBrack  }  $.
        So, we can have $ \lBrack  \langle \ottnt{C} \rangle_{ \ottmv{l'} }  \rBrack $ is traceable by \propref{graph/traceable/join}.

        \item[$ \ottmv{l'} \neq \ottmv{l} $]
        In this case, $  \lBrack  \langle \ottsym{(}  \ottmv{l}  \mathord:  \ottsym{[}  \ottnt{m}  \ottsym{]}  \ottsym{,}  \ottnt{C}  \ottsym{)} \rangle_{ \ottmv{l'} }  \rBrack  \ottsym{=}  \lBrack  \langle \ottnt{C} \rangle_{ \ottmv{l'} }  \rBrack  $.
        So, the goal follows by the assumption.
    \end{match}
\end{prop}

\begin{prop}{env/ordering/push}
    If $\ottnt{C}$ is order-defined, then $\ottmv{l}  \mathord:  \ottsym{[}  \ottnt{m}  \ottsym{]}  \ottsym{,}  \ottnt{C}$ is order-defined.

    \proof We show
    \begin{quote}
        if $ \lBrack  \langle \ottnt{C} \rangle_{ \ottmv{l'} }  \rBrack $ is traceable, then $ \lBrack  \langle \ottsym{(}  \ottmv{l}  \mathord:  \ottsym{[}  \ottnt{m}  \ottsym{]}  \ottsym{,}  \ottnt{C}  \ottsym{)} \rangle_{ \ottmv{l'} }  \rBrack $ is traceable
    \end{quote}
    for any $\ottmv{l'}$.
    We consider wether $ \ottmv{l'} \ottsym{=} \ottmv{l} $ or not.
    \begin{match}
        \item[$ \ottmv{l'} \ottsym{=} \ottmv{l} $]
        In this case, $  \lBrack  \langle \ottsym{(}  \ottmv{l}  \mathord:  \ottsym{[}  \ottnt{m}  \ottsym{]}  \ottsym{,}  \ottnt{C}  \ottsym{)} \rangle_{ \ottmv{l'} }  \rBrack  \ottsym{=}  \GJOIN{  \lBrack \ottmv{l}  \mathord:  \ottsym{[}  \ottnt{m}  \ottsym{]} \rBrack  }{  \lBrack  \langle \ottnt{C} \rangle_{ \ottmv{l'} }  \rBrack  }  $.
        So, the goal follows by \propref{graph/traceable/join}.

        \item[$ \ottmv{l'} \neq \ottmv{l} $]
        In this case, $  \lBrack  \langle \ottsym{(}  \ottmv{l}  \mathord:  \ottsym{[}  \ottnt{m}  \ottsym{]}  \ottsym{,}  \ottnt{C}  \ottsym{)} \rangle_{ \ottmv{l'} }  \rBrack  \ottsym{=}  \lBrack  \langle \ottnt{C} \rangle_{ \ottmv{l'} }  \rBrack  $.
        So, the goal follows by the assumption.
    \end{match}
\end{prop}

\begin{prop}{env/usage/concat}
    $  \overline{ \ottnt{C_{{\mathrm{1}}}}  \ottsym{,}  \ottnt{C_{{\mathrm{2}}}} }  \ottsym{=}   \overline{ \ottnt{C_{{\mathrm{1}}}} }  \odot  \overline{ \ottnt{C_{{\mathrm{2}}}} }   $.

    \proof By definition, we have some topological ordering $f$ of $ \lBrack \ottnt{C_{{\mathrm{1}}}}  \ottsym{,}  \ottnt{C_{{\mathrm{2}}}} \rBrack $.
    From that, we have topological orderings $f_1$ and $f_2$ of $ \lBrack \ottnt{C_{{\mathrm{1}}}} \rBrack $ and $ \lBrack \ottnt{C_{{\mathrm{2}}}} \rBrack $ such that $f = f_1 + f_2$ by \propref{graph/ordering/join/inv}.
    By the way, each topological ordering of $ \lBrack \ottnt{C_{{\mathrm{1}}}} \rBrack $ and $ \lBrack \ottnt{C_{{\mathrm{2}}}} \rBrack $ is unique.
    So, the goal follows by definition.
\end{prop}

\begin{prop}{env/usage/replace/ctx}
    Suppose $\ottnt{C}  \not\simeq   \cdot $.
    If
    \begin{gather}
          \overline{ \mathcal{C}  \ottsym{[}  \ottnt{C}  \ottsym{]} }  \ottsym{=} \ottnt{m}  \hyp{0.3}\\
          \overline{ \ottnt{C} }  \ottsym{=} \ottnt{m_{{\mathrm{2}}}}  \hyp{0.1}\\
          \overline{ \ottnt{C'} }  \ottsym{=} \ottnt{m'_{{\mathrm{2}}}}  \hyp{0.2},
    \end{gather}
    then there exist $\ottnt{m_{{\mathrm{1}}}}$ and $\ottnt{m_{{\mathrm{3}}}}$ such that
    \begin{enumerate}
        \item $ \ottnt{m} \ottsym{=}   \ottnt{m_{{\mathrm{1}}}} \odot \ottnt{m_{{\mathrm{2}}}}  \odot \ottnt{m_{{\mathrm{3}}}}  $ and
        \item $  \overline{ \mathcal{C}  \ottsym{[}  \ottnt{C'}  \ottsym{]} }  \ottsym{=}   \ottnt{m_{{\mathrm{1}}}} \odot \ottnt{m'_{{\mathrm{2}}}}  \odot \ottnt{m_{{\mathrm{3}}}}  $.
    \end{enumerate}

    \proof By structural induction on $\mathcal{C}$.
    \begin{match}
        \item[$ \mathcal{C} \ottsym{=} \ottsym{[]} $]
        In this case, $ \ottnt{m} \ottsym{=} \ottnt{m_{{\mathrm{2}}}} $.
        So, we choose $ \ottnt{m_{{\mathrm{1}}}} \ottsym{=} \varepsilon $ and $ \ottnt{m_{{\mathrm{3}}}} \ottsym{=} \varepsilon $, which clearly
        satisfy the condition in the goal.

        \item[$ \mathcal{C} \ottsym{=} \mathcal{C}_{{\mathrm{1}}}  \ottsym{,}  \ottnt{C_{{\mathrm{2}}}} $]
        In this case, $ \ottnt{m} \ottsym{=}   \overline{ \mathcal{C}_{{\mathrm{1}}}  \ottsym{[}  \ottnt{C}  \ottsym{]} }  \odot  \overline{ \ottnt{C_{{\mathrm{2}}}} }   $ by \propref{env/usage/concat}.
        We have some $\ottnt{m'_{{\mathrm{1}}}}$ and $\ottnt{m'_{{\mathrm{3}}}}$ such that
        \begin{gather}
              \overline{ \mathcal{C}_{{\mathrm{1}}}  \ottsym{[}  \ottnt{C}  \ottsym{]} }  \ottsym{=}   \ottnt{m'_{{\mathrm{1}}}} \odot \ottnt{m_{{\mathrm{2}}}}  \odot \ottnt{m'_{{\mathrm{3}}}}   \hyp{2.1}\\
              \overline{ \mathcal{C}_{{\mathrm{1}}}  \ottsym{[}  \ottnt{C'}  \ottsym{]} }  \ottsym{=}   \ottnt{m'_{{\mathrm{1}}}} \odot \ottnt{m'_{{\mathrm{2}}}}  \odot \ottnt{m'_{{\mathrm{3}}}}   \hyp{2.2}
        \end{gather}
        by the induction hypothesis.
        To this end, we choose $ \ottnt{m_{{\mathrm{1}}}} \ottsym{=} \ottnt{m'_{{\mathrm{1}}}} $ and $ \ottnt{m_{{\mathrm{3}}}} \ottsym{=}  \ottnt{m'_{{\mathrm{3}}}} \odot  \overline{ \ottnt{C_{{\mathrm{2}}}} }   $.
        We can see
        \begin{gather*}
               \ottnt{m_{{\mathrm{1}}}} \odot \ottnt{m_{{\mathrm{2}}}}  \odot \ottnt{m_{{\mathrm{3}}}}  \ottsym{=}    \ottnt{m'_{{\mathrm{1}}}} \odot \ottnt{m_{{\mathrm{2}}}}  \odot \ottnt{m'_{{\mathrm{3}}}}  \odot  \overline{ \ottnt{C_{{\mathrm{2}}}} }    =    \overline{ \mathcal{C}_{{\mathrm{1}}}  \ottsym{[}  \ottnt{C}  \ottsym{]} }  \odot  \overline{ \ottnt{C_{{\mathrm{2}}}} }   \ottsym{=} \ottnt{m} ,
        \end{gather*}
        and
        \begin{gather*}
               \ottnt{m_{{\mathrm{1}}}} \odot \ottnt{m'_{{\mathrm{2}}}}  \odot \ottnt{m_{{\mathrm{3}}}}  \ottsym{=}    \ottnt{m'_{{\mathrm{1}}}} \odot \ottnt{m'_{{\mathrm{2}}}}  \odot \ottnt{m'_{{\mathrm{3}}}}  \odot  \overline{ \ottnt{C_{{\mathrm{2}}}} }    =    \overline{ \mathcal{C}_{{\mathrm{1}}}  \ottsym{[}  \ottnt{C'}  \ottsym{]} }  \odot  \overline{ \ottnt{C_{{\mathrm{2}}}} }   \ottsym{=}  \overline{ \mathcal{C}  \ottsym{[}  \ottnt{C'}  \ottsym{]} }  
        \end{gather*}
        by \propref{env/usage/concat}.

        \item[$ \mathcal{C} \ottsym{=} \ottnt{C_{{\mathrm{1}}}}  \ottsym{,}  \mathcal{C}_{{\mathrm{2}}} $] Similar to the case $ \mathcal{C} \ottsym{=} \mathcal{C}_{{\mathrm{1}}}  \ottsym{,}  \ottnt{C_{{\mathrm{2}}}} $.

        \item[$ \mathcal{C} \ottsym{=} \mathcal{C}_{{\mathrm{1}}}  \parallel  \ottnt{C_{{\mathrm{2}}}} $]
        In this case, $ |  \lBrack \mathcal{C}_{{\mathrm{1}}}  \ottsym{[}  \ottnt{C}  \ottsym{]} \rBrack  |_{\bullet}  = 0$ or $ |  \lBrack \ottnt{C_{{\mathrm{2}}}} \rBrack  |_{\bullet}  = 0$ because if the orders of both graphs are greater than zero $ \overline{ \mathcal{C}_{{\mathrm{1}}}  \ottsym{[}  \ottnt{C}  \ottsym{]}  \parallel  \ottnt{C_{{\mathrm{2}}}} } $ is undefined by \propref{graph/ordering/union/multi}.
        By the way, for any $\ottnt{C}$, we have $ \lBrack \ottnt{C} \rBrack   \simeq   \lBrack  \cdot  \rBrack $ by $  |  \lBrack \ottnt{C} \rBrack  |_{\bullet}  \ottsym{=} \ottsym{0} $, which implies $\ottnt{C}  \simeq   \cdot $ by definition.
        So, it suffices to consider the following cases.
        \begin{match}
            \item[$\mathcal{C}_{{\mathrm{1}}}  \ottsym{[}  \ottnt{C}  \ottsym{]}  \simeq   \cdot $]
            We have $\ottnt{C}  \simeq   \cdot $ by \propref{env/runtime/equiv/null/ctx}.
            This contradicts to the assumption $\ottnt{C}  \not\simeq   \cdot $.
            So, this case is vacuously true.

            \item[$\ottnt{C_{{\mathrm{2}}}}  \simeq   \cdot $]
            We have $\mathcal{C}  \ottsym{[}  \ottnt{C}  \ottsym{]}  \simeq  \mathcal{C}_{{\mathrm{1}}}  \ottsym{[}  \ottnt{C}  \ottsym{]}$ and $\mathcal{C}  \ottsym{[}  \ottnt{C'}  \ottsym{]}  \simeq  \mathcal{C}_{{\mathrm{1}}}  \ottsym{[}  \ottnt{C'}  \ottsym{]}$.
            We know $ \lBrack \mathcal{C}  \ottsym{[}  \ottnt{C}  \ottsym{]} \rBrack $ and $ \lBrack \mathcal{C}  \ottsym{[}  \ottnt{C'}  \ottsym{]} \rBrack $ are traceable.
            So, $ \lBrack \mathcal{C}_{{\mathrm{1}}}  \ottsym{[}  \ottnt{C}  \ottsym{]} \rBrack $ and $ \lBrack \mathcal{C}_{{\mathrm{1}}}  \ottsym{[}  \ottnt{C'}  \ottsym{]} \rBrack $ are traceable by \propref{graph/traceable/iso}.
            So, $  \overline{ \mathcal{C}  \ottsym{[}  \ottnt{C}  \ottsym{]} }  \ottsym{=}  \overline{ \mathcal{C}_{{\mathrm{1}}}  \ottsym{[}  \ottnt{C}  \ottsym{]} }  $ and $  \overline{ \mathcal{C}  \ottsym{[}  \ottnt{C'}  \ottsym{]} }  \ottsym{=}  \overline{ \mathcal{C}_{{\mathrm{1}}}  \ottsym{[}  \ottnt{C'}  \ottsym{]} }  $. by \propref{env/runtime/usage/equiv}.
            Now, the goal follows by the induction hypothesis.
        \end{match}

        \item[$ \mathcal{C} \ottsym{=} \ottnt{C_{{\mathrm{1}}}}  \parallel  \mathcal{C}_{{\mathrm{2}}} $] Similar to the case $ \mathcal{C} \ottsym{=} \mathcal{C}_{{\mathrm{1}}}  \parallel  \ottnt{C_{{\mathrm{2}}}} $.
    \end{match}
\end{prop}

  \subsection{Properties for typing}

\begin{prop}{typing/env/remove/well-formed}
    If $\Gamma$ is well-formed, then $ \Gamma ^{- \ottnt{b} } $ is well-formed.

    \proof We can see $ \ottsym{\#} _{ \ottnt{b'} }(  \Gamma ^{- \ottnt{b} }  )  \le  \ottsym{\#} _{ \ottnt{b'} }( \Gamma ) $ for any $\ottnt{b'}$.
    So, the goal follows by the aassumption.
\end{prop}

\begin{prop}{typing/env/hole/runtime/well-formed}
    If $\mathcal{G}  \ottsym{[}  \ottmv{x}  \mathord:  \ottnt{T}  \ottsym{]}$ is well-formed, then $\mathcal{G}  \ottsym{[}  \ottnt{C}  \ottsym{]}$ is well-formed.

    \proof We can see $ \ottsym{\#} _{ \ottmv{y}  \mathord:  \ottnt{S} }( \mathcal{G}  \ottsym{[}  \ottnt{C}  \ottsym{]} )  \le  \ottsym{\#} _{ \ottmv{y}  \mathord:  \ottnt{S} }( \mathcal{G}  \ottsym{[}  \ottmv{x}  \mathord:  \ottnt{T}  \ottsym{]} ) $ for any $\ottmv{y}  \mathord:  \ottnt{S}$ since $\ottnt{C}$ has no variable bindings.
    So, the goal follows by the assumption.
\end{prop}

\begin{prop}{typing/env/in/ord/unique}
    If $\mathcal{G}  \ottsym{[}  \ottmv{x}  \mathord:  \ottnt{S}  \ottsym{]}$ is well-formed and $\ottkw{ord} \, \ottnt{S}$, then $ \ottmv{x}  \mathord:  \ottnt{S}  \notin  \mathcal{G} $.

    \proof By definition $ \ottsym{\#} _{ \ottmv{x}  \mathord:  \ottnt{S} }( \mathcal{G}  \ottsym{[}  \ottmv{x}  \mathord:  \ottnt{S}  \ottsym{]} )  =  \ottsym{\#} _{ \ottmv{x}  \mathord:  \ottnt{S} }( \mathcal{G} )  + 1$.
    So, $  \ottsym{\#} _{ \ottmv{x}  \mathord:  \ottnt{S} }( \mathcal{G} )  \ottsym{=} \ottsym{0} $ since $\mathcal{G}  \ottsym{[}  \ottmv{x}  \mathord:  \ottnt{S}  \ottsym{]}$ is well-formed and $\ottkw{ord} \, \ottnt{S}$.
    That implies the goal $ \ottmv{x}  \mathord:  \ottnt{S}  \notin  \mathcal{G} $.
\end{prop}

\begin{prop}{typing/env/replace/ord}
    If $\mathcal{G}  \ottsym{[}  \ottmv{x}  \mathord:  \ottnt{S}  \ottsym{]}$ is well-formed and $\ottkw{ord} \, \ottnt{S}$, then $ \mathcal{G}  \ottsym{[}  \ottmv{x}  \mathord:  \ottnt{S}  \ottsym{]}  \ottsym{[}  \Gamma  \ottsym{/}  \ottmv{x}  \mathord:  \ottnt{S}  \ottsym{]} \ottsym{=} \mathcal{G}  \ottsym{[}  \Gamma  \ottsym{]} $.

    \proof We know $ \ottmv{x}  \mathord:  \ottnt{S}  \notin  \mathcal{G} $ by \propref{typing/env/in/ord/unique}.
    So, $ \mathcal{G}  \ottsym{[}  \Gamma  \ottsym{/}  \ottmv{x}  \mathord:  \ottnt{S}  \ottsym{]} \ottsym{=} \mathcal{G} $, by which the goal follows.
\end{prop}

\begin{prop}{typing/closed}
    If $\Gamma  \vdash  \ottnt{M}  :  \ottnt{T}  \mid  \ottnt{e}$, then $  \ottkw{fv} ( \ottnt{M} )   \subseteq   \ottkw{dom} ( \Gamma )  $.

    \proof Routine by induction on the given derivation.
    We use \propref{env/in/sup} in the case \ruleref{T-Weaken}.
\end{prop}

\begin{prop}{typing/inv}
    \noindent
    \begin{statements}
        \item[unit] If $\Gamma  \vdash  \ottkw{unit}  :  \ottnt{T}  \mid  \ottnt{e}$, then $\Gamma  \lesssim   \cdot $ and $ \ottnt{T} \ottsym{=}  \mathtt{Unit}  $.
        \item[new] If $\Gamma  \vdash   \ottkw{new} _{ \ottnt{m} }   :  \ottnt{T}  \mid  \ottnt{e}$, then $\Gamma  \lesssim   \cdot $ and $ \ottnt{T} \ottsym{=}   \mathtt{Unit}  \rightarrow _{ \ottsym{0} } \ottsym{[}  \ottnt{m}  \ottsym{]}  $.
        \item[op] If $\Gamma  \vdash   \ottkw{op} _{ \ottnt{m_{{\mathrm{1}}}} }   :  \ottnt{T}  \mid  \ottnt{e}$, then $\Gamma  \lesssim   \cdot $, $ \ottnt{T} \ottsym{=}  \ottsym{[}  \ottnt{m_{{\mathrm{0}}}}  \ottsym{]} \rightarrow _{ \ottsym{1} } \ottsym{[}  \ottnt{m_{{\mathrm{2}}}}  \ottsym{]}  $, and $  \ottnt{m_{{\mathrm{1}}}} \odot \ottnt{m_{{\mathrm{2}}}}  \le \ottnt{m_{{\mathrm{0}}}} $ for some $\ottnt{m_{{\mathrm{0}}}}$ and $\ottnt{m_{{\mathrm{2}}}}$.
        \item[split] If $\Gamma  \vdash   \ottkw{split} _{ \ottnt{m_{{\mathrm{1}}}} , \ottnt{m_{{\mathrm{2}}}} }   :  \ottnt{T}  \mid  \ottnt{e}$, then $\Gamma  \lesssim   \cdot $, $ \ottnt{T} \ottsym{=}  \ottsym{[}  \ottnt{m_{{\mathrm{0}}}}  \ottsym{]} \rightarrow _{ \ottsym{0} } \ottsym{(}  \ottsym{[}  \ottnt{m_{{\mathrm{1}}}}  \ottsym{]}  \odot  \ottsym{[}  \ottnt{m_{{\mathrm{2}}}}  \ottsym{]}  \ottsym{)}  $, and $  \ottnt{m_{{\mathrm{1}}}} \odot \ottnt{m_{{\mathrm{2}}}}  \le \ottnt{m_{{\mathrm{0}}}} $ for some $\ottnt{m_{{\mathrm{0}}}}$.
        \item[drop] If $\Gamma  \vdash  \ottkw{drop}  :  \ottnt{T}  \mid  \ottnt{e}$, then $\Gamma  \lesssim   \cdot $, $ \ottnt{T} \ottsym{=}  \ottsym{[}  \ottnt{m}  \ottsym{]} \rightarrow _{ \ottsym{0} }  \mathtt{Unit}   $, and $ \varepsilon \le \ottnt{m} $ for some $\ottnt{m}$.
        \item[var] If $\Gamma  \vdash  \ottmv{x}  :  \ottnt{T}  \mid  \ottnt{e}$, then $\Gamma  \lesssim  \ottmv{x}  \mathord:  \ottnt{T}$.
        \item[loc] If $\Gamma  \vdash  \ottmv{l}  :  \ottnt{T}  \mid  \ottnt{e}$, then $\Gamma  \lesssim  \ottmv{l}  \mathord:  \ottsym{[}  \ottnt{m}  \ottsym{]}$ and $ \ottnt{T} \ottsym{=} \ottsym{[}  \ottnt{m}  \ottsym{]} $ for some $\ottnt{m}$.
        \item[abs] If $\Gamma  \vdash  \lambda  \ottmv{x}  \ottsym{.}  \ottnt{M}  :  \ottnt{T}  \mid  \ottnt{e}$, then $\Gamma  \lesssim  \Gamma_{{\mathrm{1}}}  \ottsym{,}  \ottmv{x}  \mathord:  \ottnt{T_{{\mathrm{1}}}}$, $ \ottnt{T} \ottsym{=}  \ottnt{T_{{\mathrm{1}}}} \rightarrow _{ \ottnt{e_{{\mathrm{1}}}} } \ottnt{T_{{\mathrm{2}}}}  $, $ \ottmv{x}  \notin   \ottkw{dom} ( \Gamma_{{\mathrm{1}}} )  $, $\ottkw{unr} \, \Gamma_{{\mathrm{1}}}$, and $\Gamma_{{\mathrm{1}}}  \ottsym{,}  \ottmv{x}  \mathord:  \ottnt{T_{{\mathrm{1}}}}  \vdash  \ottnt{M}  :  \ottnt{T_{{\mathrm{2}}}}  \mid  \ottnt{e_{{\mathrm{1}}}}$ for some $\Gamma_{{\mathrm{1}}}$, $\ottnt{T_{{\mathrm{1}}}}$, $\ottnt{T_{{\mathrm{2}}}}$, and $\ottnt{e_{{\mathrm{1}}}}$.
        \item[uabs] If $\Gamma  \vdash  \lambda^\circ  \ottmv{x}  \ottsym{.}  \ottnt{M}  :  \ottnt{T}  \mid  \ottnt{e}$, then $\Gamma  \lesssim  \Gamma_{{\mathrm{1}}}  \parallel  \ottmv{x}  \mathord:  \ottnt{T_{{\mathrm{1}}}}$, $ \ottnt{T} \ottsym{=}  \ottnt{T_{{\mathrm{1}}}} \rightarrowtriangle _{ \ottnt{e_{{\mathrm{1}}}} } \ottnt{T_{{\mathrm{2}}}}  $, $ \ottmv{x}  \notin   \ottkw{dom} ( \Gamma_{{\mathrm{1}}} )  $ and $\Gamma_{{\mathrm{1}}}  \parallel  \ottmv{x}  \mathord:  \ottnt{T_{{\mathrm{1}}}}  \vdash  \ottnt{M}  :  \ottnt{T_{{\mathrm{2}}}}  \mid  \ottnt{e_{{\mathrm{1}}}}$ for some $\Gamma_{{\mathrm{1}}}$, $\ottnt{T_{{\mathrm{1}}}}$, $\ottnt{T_{{\mathrm{2}}}}$, and $\ottnt{e_{{\mathrm{1}}}}$.
        \item[rabs] If $\Gamma  \vdash  \lambda^>  \ottmv{x}  \ottsym{.}  \ottnt{M}  :  \ottnt{T}  \mid  \ottnt{e}$, then $\Gamma  \lesssim  \Gamma_{{\mathrm{1}}}  \ottsym{,}  \ottmv{x}  \mathord:  \ottnt{T_{{\mathrm{1}}}}$, $ \ottnt{T} \ottsym{=}  \ottnt{T_{{\mathrm{1}}}} \twoheadrightarrow _{ \ottnt{e_{{\mathrm{1}}}} } \ottnt{T_{{\mathrm{2}}}}  $, $ \ottmv{x}  \notin   \ottkw{dom} ( \Gamma_{{\mathrm{1}}} )  $ and $\Gamma_{{\mathrm{1}}}  \ottsym{,}  \ottmv{x}  \mathord:  \ottnt{T_{{\mathrm{1}}}}  \vdash  \ottnt{M}  :  \ottnt{T_{{\mathrm{2}}}}  \mid  \ottnt{e_{{\mathrm{1}}}}$ for some $\Gamma_{{\mathrm{1}}}$, $\ottnt{T_{{\mathrm{1}}}}$, $\ottnt{T_{{\mathrm{2}}}}$, and $\ottnt{e_{{\mathrm{1}}}}$.
        \item[labs] If $\Gamma  \vdash  \lambda^<  \ottmv{x}  \ottsym{.}  \ottnt{M}  :  \ottnt{T}  \mid  \ottnt{e}$, then $\Gamma  \lesssim  \ottmv{x}  \mathord:  \ottnt{T_{{\mathrm{1}}}}  \ottsym{,}  \Gamma_{{\mathrm{1}}}$, $ \ottnt{T} \ottsym{=}  \ottnt{T_{{\mathrm{1}}}} \rightarrowtail _{ \ottnt{e_{{\mathrm{1}}}} } \ottnt{T_{{\mathrm{2}}}}  $, $ \ottmv{x}  \notin   \ottkw{dom} ( \Gamma_{{\mathrm{1}}} )  $ and $\ottmv{x}  \mathord:  \ottnt{T_{{\mathrm{1}}}}  \ottsym{,}  \Gamma_{{\mathrm{1}}}  \vdash  \ottnt{M}  :  \ottnt{T_{{\mathrm{2}}}}  \mid  \ottnt{e_{{\mathrm{1}}}}$ for some $\Gamma_{{\mathrm{1}}}$, $\ottnt{T_{{\mathrm{1}}}}$, $\ottnt{T_{{\mathrm{2}}}}$, and $\ottnt{e_{{\mathrm{1}}}}$.
        \item[app] If $\Gamma  \vdash  \ottnt{M_{{\mathrm{1}}}} \, \ottnt{M_{{\mathrm{2}}}}  :  \ottnt{T}  \mid  \ottnt{e}$, then $\Gamma  \lesssim  \Gamma_{{\mathrm{1}}}  \ottsym{,}  \Gamma_{{\mathrm{2}}}$, $   \ottnt{e_{{\mathrm{0}}}}  \sqcup  \ottnt{e_{{\mathrm{1}}}}   \sqcup  \ottnt{e_{{\mathrm{2}}}}  \le \ottnt{e} $, $\Gamma_{{\mathrm{1}}}  \vdash  \ottnt{M_{{\mathrm{1}}}}  :   \ottnt{S} \rightarrow _{ \ottnt{e_{{\mathrm{0}}}} } \ottnt{T}   \mid  \ottnt{e_{{\mathrm{1}}}}$, and $\Gamma_{{\mathrm{2}}}  \vdash  \ottnt{M_{{\mathrm{2}}}}  :  \ottnt{S}  \mid  \ottnt{e_{{\mathrm{2}}}}$ for some $\Gamma_{{\mathrm{1}}}$, $\Gamma_{{\mathrm{2}}}$, $\ottnt{S}$, $\ottnt{e_{{\mathrm{0}}}}$, $\ottnt{e_{{\mathrm{1}}}}$, and $\ottnt{e_{{\mathrm{2}}}}$.
        \item[uapp]  If $\Gamma  \vdash  \ottnt{M_{{\mathrm{1}}}}  {}^\circ  \ottnt{M_{{\mathrm{2}}}}  :  \ottnt{T}  \mid  \ottnt{e}$, then $\Gamma  \lesssim  \Gamma_{{\mathrm{1}}}  \parallel  \Gamma_{{\mathrm{2}}}$, $   \ottnt{e_{{\mathrm{0}}}}  \sqcup  \ottnt{e_{{\mathrm{1}}}}   \sqcup  \ottnt{e_{{\mathrm{2}}}}  \le \ottnt{e} $, $\Gamma_{{\mathrm{1}}}  \vdash  \ottnt{M_{{\mathrm{1}}}}  :   \ottnt{S} \rightarrowtriangle _{ \ottnt{e_{{\mathrm{0}}}} } \ottnt{T}   \mid  \ottnt{e_{{\mathrm{1}}}}$, and $\Gamma_{{\mathrm{2}}}  \vdash  \ottnt{M_{{\mathrm{2}}}}  :  \ottnt{S}  \mid  \ottnt{e_{{\mathrm{2}}}}$ for some $\Gamma_{{\mathrm{1}}}$, $\Gamma_{{\mathrm{2}}}$, $\ottnt{S}$, $\ottnt{e_{{\mathrm{0}}}}$, $\ottnt{e_{{\mathrm{1}}}}$, and $\ottnt{e_{{\mathrm{2}}}}$.
        \item[rapp]  If $\Gamma  \vdash  \ottnt{M_{{\mathrm{1}}}}  {}^>  \ottnt{M_{{\mathrm{2}}}}  :  \ottnt{T}  \mid  \ottnt{e}$, then $\Gamma  \lesssim  \Gamma_{{\mathrm{1}}}  \ottsym{,}  \Gamma_{{\mathrm{2}}}$, $  \ottnt{e_{{\mathrm{0}}}}  \sqcup  \ottnt{e_{{\mathrm{1}}}}  \le \ottnt{e} $, $\Gamma_{{\mathrm{1}}}  \vdash  \ottnt{M_{{\mathrm{1}}}}  :   \ottnt{S} \twoheadrightarrow _{ \ottnt{e_{{\mathrm{0}}}} } \ottnt{T}   \mid  \ottnt{e_{{\mathrm{1}}}}$, and $\Gamma_{{\mathrm{2}}}  \vdash  \ottnt{M_{{\mathrm{2}}}}  :  \ottnt{S}  \mid  \ottsym{0}$ for some $\Gamma_{{\mathrm{1}}}$, $\Gamma_{{\mathrm{2}}}$, $\ottnt{S}$, $\ottnt{e_{{\mathrm{0}}}}$, and $\ottnt{e_{{\mathrm{1}}}}$.
        \item[lapp]  If $\Gamma  \vdash  \ottnt{M_{{\mathrm{1}}}}  {}^<  \ottnt{M_{{\mathrm{2}}}}  :  \ottnt{T}  \mid  \ottnt{e}$, then $\Gamma  \lesssim  \Gamma_{{\mathrm{1}}}  \ottsym{,}  \Gamma_{{\mathrm{2}}}$, $  \ottnt{e_{{\mathrm{0}}}}  \sqcup  \ottnt{e_{{\mathrm{2}}}}  \le \ottnt{e} $, $\Gamma_{{\mathrm{1}}}  \vdash  \ottnt{M_{{\mathrm{1}}}}  :   \ottnt{S} \rightarrowtail _{ \ottnt{e_{{\mathrm{0}}}} } \ottnt{T}   \mid  \ottsym{0}$, and $\Gamma_{{\mathrm{2}}}  \vdash  \ottnt{M_{{\mathrm{2}}}}  :  \ottnt{S}  \mid  \ottnt{e_{{\mathrm{2}}}}$ for some $\Gamma_{{\mathrm{1}}}$, $\Gamma_{{\mathrm{2}}}$, $\ottnt{S}$, $\ottnt{e_{{\mathrm{0}}}}$, and $\ottnt{e_{{\mathrm{2}}}}$.
        \item[upair] If $\Gamma  \vdash  \ottnt{M_{{\mathrm{1}}}}  \otimes  \ottnt{M_{{\mathrm{2}}}}  :  \ottnt{T}  \mid  \ottnt{e}$, then $\Gamma  \lesssim  \Gamma_{{\mathrm{1}}}  \parallel  \Gamma_{{\mathrm{2}}}$, $ \ottnt{T} \ottsym{=} \ottnt{S_{{\mathrm{1}}}}  \otimes  \ottnt{S_{{\mathrm{2}}}} $, $  \ottnt{e_{{\mathrm{1}}}}  \sqcup  \ottnt{e_{{\mathrm{2}}}}  \le \ottnt{e} $, $\Gamma_{{\mathrm{1}}}  \vdash  \ottnt{M_{{\mathrm{1}}}}  :  \ottnt{S_{{\mathrm{1}}}}  \mid  \ottnt{e_{{\mathrm{1}}}}$, and $\Gamma_{{\mathrm{2}}}  \vdash  \ottnt{M_{{\mathrm{2}}}}  :  \ottnt{S_{{\mathrm{2}}}}  \mid  \ottnt{e_{{\mathrm{2}}}}$ for some $\Gamma_{{\mathrm{1}}}$, $\Gamma_{{\mathrm{2}}}$, $\ottnt{S_{{\mathrm{1}}}}$, $\ottnt{S_{{\mathrm{2}}}}$, $\ottnt{e_{{\mathrm{1}}}}$ and $\ottnt{e_{{\mathrm{2}}}}$.
        \item[opair] If $\Gamma  \vdash  \ottnt{M_{{\mathrm{1}}}}  \odot  \ottnt{M_{{\mathrm{2}}}}  :  \ottnt{T}  \mid  \ottnt{e}$, then $\Gamma  \lesssim  \Gamma_{{\mathrm{1}}}  \ottsym{,}  \Gamma_{{\mathrm{2}}}$, $ \ottnt{T} \ottsym{=} \ottnt{S_{{\mathrm{1}}}}  \odot  \ottnt{S_{{\mathrm{2}}}} $, $  \ottnt{e_{{\mathrm{1}}}}  \sqcup  \ottnt{e_{{\mathrm{2}}}}  \le \ottnt{e} $, $\ottkw{ord} \, \ottnt{S_{{\mathrm{1}}}} \text{ implies }  \ottnt{e_{{\mathrm{2}}}} \ottsym{=} \ottsym{0} $, $\Gamma_{{\mathrm{1}}}  \vdash  \ottnt{M_{{\mathrm{1}}}}  :  \ottnt{S_{{\mathrm{1}}}}  \mid  \ottnt{e_{{\mathrm{1}}}}$, and $\Gamma_{{\mathrm{2}}}  \vdash  \ottnt{M_{{\mathrm{2}}}}  :  \ottnt{S_{{\mathrm{2}}}}  \mid  \ottnt{e_{{\mathrm{2}}}}$ for some $\Gamma_{{\mathrm{1}}}$, $\Gamma_{{\mathrm{2}}}$, $\ottnt{S_{{\mathrm{1}}}}$, $\ottnt{S_{{\mathrm{2}}}}$, $\ottnt{e_{{\mathrm{1}}}}$ and $\ottnt{e_{{\mathrm{2}}}}$.
        \item[ulet] If $\Gamma  \vdash  \ottkw{let} \, \ottmv{x}  \otimes  \ottmv{y}  \ottsym{=}  \ottnt{M} \, \ottkw{in} \, \ottnt{N}  :  \ottnt{T}  \mid  \ottnt{e}$, then $\Gamma  \lesssim  \mathcal{G}  \ottsym{[}  \Gamma_{{\mathrm{1}}}  \ottsym{]}$, $\ottsym{\{}  \ottmv{x}  \ottsym{\}}  \uplus  \ottsym{\{}  \ottmv{y}  \ottsym{\}}  \uplus   \ottkw{dom} ( \mathcal{G}  \ottsym{[}  \Gamma_{{\mathrm{1}}}  \ottsym{]} ) $, $\Gamma_{{\mathrm{1}}}  \vdash  \ottnt{M}  :  \ottnt{S_{{\mathrm{1}}}}  \otimes  \ottnt{S_{{\mathrm{2}}}}  \mid  \ottsym{0}$, and $\mathcal{G}  \ottsym{[}  \ottmv{x}  \mathord:  \ottnt{S_{{\mathrm{1}}}}  \parallel  \ottmv{y}  \mathord:  \ottnt{S_{{\mathrm{2}}}}  \ottsym{]}  \vdash  \ottnt{N}  :  \ottnt{T}  \mid  \ottnt{e}$ for some $\mathcal{G}$, $\Gamma_{{\mathrm{1}}}$, $\ottnt{S_{{\mathrm{1}}}}$ and $\ottnt{S_{{\mathrm{2}}}}$.
        \item[olet] If $\Gamma  \vdash  \ottkw{let} \, \ottmv{x}  \odot  \ottmv{y}  \ottsym{=}  \ottnt{M} \, \ottkw{in} \, \ottnt{N}  :  \ottnt{T}  \mid  \ottnt{e}$, then $\Gamma  \lesssim  \mathcal{G}  \ottsym{[}  \Gamma_{{\mathrm{1}}}  \ottsym{]}$, $\ottsym{\{}  \ottmv{x}  \ottsym{\}}  \uplus  \ottsym{\{}  \ottmv{y}  \ottsym{\}}  \uplus   \ottkw{dom} ( \mathcal{G}  \ottsym{[}  \Gamma_{{\mathrm{1}}}  \ottsym{]} ) $, $\Gamma_{{\mathrm{1}}}  \vdash  \ottnt{M}  :  \ottnt{S_{{\mathrm{1}}}}  \odot  \ottnt{S_{{\mathrm{2}}}}  \mid  \ottsym{0}$, and $\mathcal{G}  \ottsym{[}  \ottmv{x}  \mathord:  \ottnt{S_{{\mathrm{1}}}}  \ottsym{,}  \ottmv{y}  \mathord:  \ottnt{S_{{\mathrm{2}}}}  \ottsym{]}  \vdash  \ottnt{N}  :  \ottnt{T}  \mid  \ottnt{e}$ for some $\mathcal{G}$, $\Gamma_{{\mathrm{1}}}$, $\ottnt{S_{{\mathrm{1}}}}$ and $\ottnt{S_{{\mathrm{2}}}}$.
    \end{statements}

    \proof Routine by induction on the given derivation.
\end{prop}

\begin{prop}{typing/canonical}
    \noindent
    \begin{statements}
        \item[unit] If $\ottnt{C}  \vdash  \ottnt{V}  :   \mathtt{Unit}   \mid  \ottnt{e}$, then $ \ottnt{V} \ottsym{=} \ottkw{unit} $.
        \item[resource] If $\ottnt{C}  \vdash  \ottnt{V}  :  \ottsym{[}  \ottnt{m}  \ottsym{]}  \mid  \ottnt{e}$, then $ \ottnt{V} \ottsym{=} \ottmv{l} $ for some $\ottmv{l}$.
        \item[arrow] If $\ottnt{C}  \vdash  \ottnt{V}  :   \ottnt{S} \rightarrow _{ \ottnt{e'} } \ottnt{T}   \mid  \ottnt{e}$, then
        \begin{itemize}
            \item $ \ottnt{V} \ottsym{=}  \ottkw{new} _{ \ottnt{m} }  $ for some $\ottnt{m}$,
            \item $ \ottnt{V} \ottsym{=}  \ottkw{op} _{ \ottnt{m_{{\mathrm{1}}}} }  $ for some $\ottnt{m_{{\mathrm{1}}}}$,
            \item $ \ottnt{V} \ottsym{=}  \ottkw{split} _{ \ottnt{m_{{\mathrm{1}}}} , \ottnt{m_{{\mathrm{2}}}} }  $ for some $\ottnt{m_{{\mathrm{1}}}}$ and $\ottnt{m_{{\mathrm{2}}}}$,
            \item $ \ottnt{V} \ottsym{=} \ottkw{drop} $, or
            \item $ \ottnt{V} \ottsym{=} \lambda  \ottmv{x}  \ottsym{.}  \ottnt{M} $ for some $\ottmv{x}$ and $\ottnt{M}$.
        \end{itemize}
        \item[uarrow] If $\ottnt{C}  \vdash  \ottnt{V}  :   \ottnt{S} \rightarrowtriangle _{ \ottnt{e'} } \ottnt{T}   \mid  \ottnt{e}$, then $ \ottnt{V} \ottsym{=} \lambda^\circ  \ottmv{x}  \ottsym{.}  \ottnt{M} $ for some $\ottmv{x}$ and $\ottnt{M}$.
        \item[rarrow] If $\ottnt{C}  \vdash  \ottnt{V}  :   \ottnt{S} \twoheadrightarrow _{ \ottnt{e'} } \ottnt{T}   \mid  \ottnt{e}$, then $ \ottnt{V} \ottsym{=} \lambda^>  \ottmv{x}  \ottsym{.}  \ottnt{M} $ for some $\ottmv{x}$ and $\ottnt{M}$.
        \item[larrow] If $\ottnt{C}  \vdash  \ottnt{V}  :   \ottnt{S} \rightarrowtail _{ \ottnt{e'} } \ottnt{T}   \mid  \ottnt{e}$, then $ \ottnt{V} \ottsym{=} \lambda^<  \ottmv{x}  \ottsym{.}  \ottnt{M} $ for some $\ottmv{x}$ and $\ottnt{M}$.
        \item[uprod] If $\ottnt{C}  \vdash  \ottnt{V}  :  \ottnt{S}  \otimes  \ottnt{T}  \mid  \ottnt{e}$, then $ \ottnt{V} \ottsym{=} \ottnt{V_{{\mathrm{1}}}}  \otimes  \ottnt{V_{{\mathrm{2}}}} $ for some $\ottnt{V_{{\mathrm{1}}}}$ and $\ottnt{V_{{\mathrm{2}}}}$.
        \item[oprod] If $\ottnt{C}  \vdash  \ottnt{V}  :  \ottnt{S}  \odot  \ottnt{T}  \mid  \ottnt{e}$, then $ \ottnt{V} \ottsym{=} \ottnt{V_{{\mathrm{1}}}}  \odot  \ottnt{V_{{\mathrm{2}}}} $ for some $\ottnt{V_{{\mathrm{1}}}}$ and $\ottnt{V_{{\mathrm{2}}}}$.
    \end{statements}

    \proof By case analysis on $\ottnt{V}$.
    We drop unsuitable cases by \propref{typing/inv}.
\end{prop}

\begin{prop}{typing/value/unr}
    If $\ottnt{C}  \vdash  \ottnt{V}  :  \ottnt{T}  \mid  \ottnt{e}$ and $\ottkw{unr} \, \ottnt{T}$, then $\ottnt{C}  \lesssim   \cdot $ and $ \cdot   \vdash  \ottnt{V}  :  \ottnt{T}  \mid  \ottsym{0}$.

    \proof Routine by induction on the given typing derivation.
\end{prop}

\begin{prop}{typing/value/ord}
    If $\ottnt{C}  \vdash  \ottnt{V}  :  \ottnt{T}  \mid  \ottnt{e}$ and $\ottkw{ord} \, \ottnt{T}$, then $\ottnt{C}  \vdash  \ottnt{V}  :  \ottnt{T}  \mid  \ottsym{0}$.

    \proof Routine by induction on the given typing derivation.
\end{prop}

\begin{prop}{typing/subst/void}
    If $\Gamma  \vdash  \ottnt{M}  :  \ottnt{T}  \mid  \ottnt{e}$ and $ \ottmv{x}  \notin   \ottkw{dom} ( \Gamma )  $, then $\Gamma  \vdash  \ottnt{M}  \ottsym{[}  \ottnt{V}  \ottsym{/}  \ottmv{x}  \ottsym{]}  :  \ottnt{T}  \mid  \ottnt{e}$ for any $\ottnt{V}$.

    \proof
    It suffices to show $ \ottnt{M} \ottsym{=} \ottnt{M}  \ottsym{[}  \ottnt{V}  \ottsym{/}  \ottmv{x}  \ottsym{]} $.
    We have $ \ottmv{x}  \notin   \ottkw{fv} ( \ottnt{M} )  $ by \propref{typing/closed} and the assumptions.
    Therefore, $ \ottnt{M} \ottsym{=} \ottnt{M}  \ottsym{[}  \ottnt{V}  \ottsym{/}  \ottmv{x}  \ottsym{]} $ follows by \propref{subst/void}.
\end{prop}

\begin{prop}{typing/subst/unr}
    If
    \begin{gather}
        \Gamma  \vdash  \ottnt{M}  :  \ottnt{T}  \mid  \ottnt{e} \hyp{0.1}\\
         \ottmv{x}  \mathord:  \ottnt{S}  \in  \Gamma  \hyp{0.2}\\
         \cdot   \vdash  \ottnt{V}  :  \ottnt{S}  \mid  \ottsym{0} \hyp{0.3},
    \end{gather}
    then $ \Gamma ^{- \ottmv{x}  \mathord:  \ottnt{S} }   \vdash  \ottnt{M}  \ottsym{[}  \ottnt{V}  \ottsym{/}  \ottmv{x}  \ottsym{]}  :  \ottnt{T}  \mid  \ottnt{e}$.

    \proof By induction on the given derivation of \hypref{0.1}.
    \begin{match}
        \item[\ruleref{T-Unit}, \ruleref{T-New}, \ruleref{T-Op}, \ruleref{T-Split}, and \ruleref{T-Close}]
        In these cases $ \Gamma \ottsym{=}  \cdot  $.
        So, \hypref{0.2} cannot hold.

        \item[\ruleref{T-Loc}]
        In this case, $ \Gamma \ottsym{=} \ottmv{l}  \mathord:  \ottsym{[}  \ottnt{m}  \ottsym{]} $.
        So, \hypref{0.2} cannot hold.

        \item[\ruleref{T-Var}]
        In this case,
        \begin{gather*}
             \Gamma \ottsym{=} \ottmv{x'}  \mathord:  \ottnt{T}  \hyp{1.1}\\
             \ottnt{M} \ottsym{=} \ottmv{x'}  \hyp{1.2}\\
             \ottnt{e} \ottsym{=} \ottsym{0}  \hyp{1.3}.
        \end{gather*}
        We have $ \ottmv{x} \ottsym{=} \ottmv{x'} $ and $ \ottnt{S} \ottsym{=} \ottnt{T} $ by \propref{env/in/binding}, \hypref{0.2}, and \hypref{1.1}.
        Now the goal $ \cdot   \vdash  \ottnt{V}  :  \ottnt{S}  \mid  \ottsym{0}$ is identical to \hypref{0.3}.

        \item[\ruleref{T-Abs}]
        In this case,
        \begin{gather}
             \ottnt{M} \ottsym{=} \lambda  \ottmv{x_{{\mathrm{1}}}}  \ottsym{.}  \ottnt{M_{{\mathrm{1}}}}  \hyp{3.1}\\
             \ottnt{T} \ottsym{=}  \ottnt{S_{{\mathrm{1}}}} \rightarrow _{ \ottnt{e_{{\mathrm{1}}}} } \ottnt{T_{{\mathrm{1}}}}   \hyp{3.2}\\
             \ottnt{e} \ottsym{=} \ottsym{0}  \hyp{3.3}\\
             \ottmv{x_{{\mathrm{1}}}}  \notin   \ottkw{dom} ( \Gamma )   \hyp{3.4}\\
            \ottkw{unr} \, \Gamma \hyp{3.3.1}\\
            \Gamma  \ottsym{,}  \ottmv{x_{{\mathrm{1}}}}  \mathord:  \ottnt{S_{{\mathrm{1}}}}  \vdash  \ottnt{M_{{\mathrm{1}}}}  :  \ottnt{T_{{\mathrm{1}}}}  \mid  \ottnt{e_{{\mathrm{1}}}} \hyp{3.5}.
        \end{gather}
        We have $ \ottmv{x_{{\mathrm{1}}}} \neq \ottmv{x} $ by \hypref{0.2} and \hypref{3.4}.
        So, we have
        \begin{gather*}
             \Gamma ^{- \ottmv{x}  \mathord:  \ottnt{S} }   \ottsym{,}  \ottmv{x_{{\mathrm{1}}}}  \mathord:  \ottnt{S_{{\mathrm{1}}}}  \vdash  \ottnt{M_{{\mathrm{1}}}}  \ottsym{[}  \ottnt{V}  \ottsym{/}  \ottmv{x}  \ottsym{]}  :  \ottnt{T_{{\mathrm{1}}}}  \mid  \ottnt{e_{{\mathrm{1}}}} \hyp{3.6}
        \end{gather*}
        by the induction hypothesis, and
        \begin{gather*}
             \ottsym{(}  \lambda  \ottmv{x_{{\mathrm{1}}}}  \ottsym{.}  \ottnt{M_{{\mathrm{1}}}}  \ottsym{)}  \ottsym{[}  \ottnt{V}  \ottsym{/}  \ottmv{x}  \ottsym{]} \ottsym{=} \lambda  \ottmv{x_{{\mathrm{1}}}}  \ottsym{.}  \ottnt{M_{{\mathrm{1}}}}  \ottsym{[}  \ottnt{V}  \ottsym{/}  \ottmv{x}  \ottsym{]} .
        \end{gather*}
        It also clearly holds that
        \begin{gather*}
            \ottkw{unr} \,  \Gamma ^{- \ottmv{x}  \mathord:  \ottnt{S} }  \\
             \ottmv{x_{{\mathrm{1}}}}  \notin   \ottkw{dom} (  \Gamma ^{- \ottmv{x}  \mathord:  \ottnt{S} }  )  
        \end{gather*}
        by \hypref{3.3.1} and \hypref{3.4}, respectively.
        So we can derive the goal
        \begin{gather*}
             \Gamma ^{- \ottmv{x}  \mathord:  \ottnt{S} }   \vdash  \lambda  \ottmv{x_{{\mathrm{1}}}}  \ottsym{.}  \ottnt{M_{{\mathrm{1}}}}  \ottsym{[}  \ottnt{V}  \ottsym{/}  \ottmv{x}  \ottsym{]}  :   \ottnt{S_{{\mathrm{1}}}} \rightarrow _{ \ottnt{e_{{\mathrm{1}}}} } \ottnt{T_{{\mathrm{1}}}}   \mid  \ottsym{0}
        \end{gather*}
        by \ruleref{T-Abs}, where $ \Gamma ^{- \ottmv{x}  \mathord:  \ottnt{S} } $ is well-formed because we know $ \Gamma ^{- \ottmv{x}  \mathord:  \ottnt{S} }   \ottsym{,}  \ottmv{x_{{\mathrm{1}}}}  \mathord:  \ottnt{S_{{\mathrm{1}}}}$ is well-formed by \hypref{3.6}.

        \item[\ruleref{T-UAbs}, \ruleref{T-RAbs}, \ruleref{T-LAbs}]
        These cases are similar to the case \ruleref{T-Abs}.

        \item[\ruleref{T-App}]
        In this case,
        \begin{gather}
             \Gamma \ottsym{=} \Gamma_{{\mathrm{1}}}  \ottsym{,}  \Gamma_{{\mathrm{2}}}  \hyp{4.1}\\
             \ottnt{M} \ottsym{=} \ottnt{M_{{\mathrm{1}}}} \, \ottnt{M_{{\mathrm{2}}}}  \hyp{4.2}\\
             \ottnt{e} \ottsym{=}   \ottnt{e'}  \sqcup  \ottnt{e_{{\mathrm{1}}}}   \sqcup  \ottnt{e_{{\mathrm{2}}}}   \hyp{4.3}\\
            \Gamma_{{\mathrm{1}}}  \vdash  \ottnt{M_{{\mathrm{1}}}}  :   \ottnt{S'} \rightarrow _{ \ottnt{e'} } \ottnt{T}   \mid  \ottnt{e_{{\mathrm{1}}}} \hyp{4.4}\\
            \Gamma_{{\mathrm{2}}}  \vdash  \ottnt{M_{{\mathrm{2}}}}  :  \ottnt{S'}  \mid  \ottnt{e_{{\mathrm{2}}}} \hyp{4.5}.
        \end{gather}
        We have $ \ottmv{x}  \mathord:  \ottnt{S}  \in  \Gamma_{{\mathrm{1}}} $ or $ \ottmv{x}  \mathord:  \ottnt{S}  \in  \Gamma_{{\mathrm{2}}} $ by \hypref{0.2} and \hypref{4.1}.
        So, it suffices to consider the following three cases.
        \begin{match}
            \item[$ \ottmv{x}  \mathord:  \ottnt{S}  \in  \Gamma_{{\mathrm{1}}} $ and $ \ottmv{x}  \mathord:  \ottnt{S}  \in  \Gamma_{{\mathrm{2}}} $]
            We have
            \begin{gather*}
                 \Gamma_{{\mathrm{1}}} ^{- \ottmv{x}  \mathord:  \ottnt{S} }   \vdash  \ottnt{M_{{\mathrm{1}}}}  \ottsym{[}  \ottnt{V}  \ottsym{/}  \ottmv{x}  \ottsym{]}  :   \ottnt{S'} \rightarrow _{ \ottnt{e'} } \ottnt{T}   \mid  \ottnt{e_{{\mathrm{1}}}} \\
                 \Gamma_{{\mathrm{2}}} ^{- \ottmv{x}  \mathord:  \ottnt{S} }   \vdash  \ottnt{M_{{\mathrm{2}}}}  \ottsym{[}  \ottnt{V}  \ottsym{/}  \ottmv{x}  \ottsym{]}  :  \ottnt{S'}  \mid  \ottnt{e_{{\mathrm{2}}}}
            \end{gather*}
            by the induction hypothesis.
            So, we can derive the goal $  \Gamma_{{\mathrm{1}}} ^{- \ottmv{x}  \mathord:  \ottnt{S} }   \ottsym{,}  \Gamma_{{\mathrm{2}}} ^{- \ottmv{x}  \mathord:  \ottnt{S} }   \vdash  \ottnt{M_{{\mathrm{1}}}}  \ottsym{[}  \ottnt{V}  \ottsym{/}  \ottmv{x}  \ottsym{]} \, \ottnt{M_{{\mathrm{2}}}}  \ottsym{[}  \ottnt{V}  \ottsym{/}  \ottmv{x}  \ottsym{]}  :  \ottnt{T}  \mid    \ottnt{e'}  \sqcup  \ottnt{e_{{\mathrm{1}}}}   \sqcup  \ottnt{e_{{\mathrm{2}}}} $ by \ruleref{T-App}, where $   \Gamma_{{\mathrm{1}}} ^{- \ottmv{x}  \mathord:  \ottnt{S} }   \ottsym{,}  \Gamma_{{\mathrm{2}}} ^{- \ottmv{x}  \mathord:  \ottnt{S} }  \ottsym{=}  \ottsym{(}  \Gamma_{{\mathrm{1}}}  \ottsym{,}  \Gamma_{{\mathrm{2}}}  \ottsym{)} ^{- \ottmv{x}  \mathord:  \ottnt{S} }  $ is well-formed by \propref{typing/env/remove/well-formed} and \hypref{0.1}.

            \item[$ \ottmv{x}  \mathord:  \ottnt{S}  \in  \Gamma_{{\mathrm{1}}} $ and $ \ottmv{x}  \mathord:  \ottnt{S}  \notin  \Gamma_{{\mathrm{2}}} $]
            We have
            \begin{gather*}
                 \Gamma_{{\mathrm{1}}} ^{- \ottmv{x}  \mathord:  \ottnt{S} }   \vdash  \ottnt{M_{{\mathrm{1}}}}  \ottsym{[}  \ottnt{V}  \ottsym{/}  \ottmv{x}  \ottsym{]}  :   \ottnt{S'} \rightarrow _{ \ottnt{e'} } \ottnt{T}   \mid  \ottnt{e_{{\mathrm{1}}}}
            \end{gather*}
            by the induction hypothesis for \hypref{4.4}.
            We have $ \ottmv{x}  \notin   \ottkw{dom} ( \Gamma_{{\mathrm{2}}} )  $ because if $ \ottmv{x}  \in   \ottkw{dom} ( \Gamma_{{\mathrm{2}}} )  $, we have $ \ottmv{x}  \mathord:  \ottnt{S'}  \in  \Gamma_{{\mathrm{2}}} $ for some $ \ottnt{S'} \neq \ottnt{S} $, but this contradicts to the well-formedness of $\Gamma (=\Gamma_{{\mathrm{1}}}  \ottsym{,}  \Gamma_{{\mathrm{2}}})$.
            So, we have
            \begin{gather*}
                \Gamma_{{\mathrm{2}}}  \vdash  \ottnt{M_{{\mathrm{2}}}}  \ottsym{[}  \ottnt{V}  \ottsym{/}  \ottmv{x}  \ottsym{]}  :  \ottnt{S'}  \mid  \ottnt{e_{{\mathrm{2}}}}
            \end{gather*}
            by applying \propref{typing/subst/void} to \hypref{4.5}.
            We also know $  \Gamma_{{\mathrm{2}}} ^{- \ottmv{x}  \mathord:  \ottnt{S} }  \ottsym{=} \Gamma_{{\mathrm{2}}} $.
            So, we can derive the goal by \ruleref{T-App}, where the environment is well-formed by \propref{typing/env/remove/well-formed} and \hypref{0.1}.

            \item[$ \ottmv{x}  \mathord:  \ottnt{S}  \notin  \Gamma_{{\mathrm{1}}} $ and $ \ottmv{x}  \mathord:  \ottnt{S}  \in  \Gamma_{{\mathrm{2}}} $]
            Similar to the sub case above.
        \end{match}

        \item[\ruleref{T-UApp}, \ruleref{T-RApp}, \ruleref{T-LApp}, \ruleref{T-UPair}, and \ruleref{T-OPair}]
        These cases are similar to the case \ruleref{T-App}.

        \item[\ruleref{T-ULet}]
        In this case,
        \begin{gather}
             \Gamma \ottsym{=} \mathcal{G}  \ottsym{[}  \Gamma_{{\mathrm{1}}}  \ottsym{]}  \hyp{5.1}\\
             \ottnt{M} \ottsym{=} \ottkw{let} \, \ottmv{x_{{\mathrm{1}}}}  \otimes  \ottmv{x_{{\mathrm{2}}}}  \ottsym{=}  \ottnt{M_{{\mathrm{1}}}} \, \ottkw{in} \, \ottnt{N}  \hyp{5.2}\\
            \ottsym{\{}  \ottmv{x_{{\mathrm{1}}}}  \ottsym{\}}  \uplus  \ottsym{\{}  \ottmv{x_{{\mathrm{2}}}}  \ottsym{\}}  \uplus   \ottkw{dom} ( \mathcal{G}  \ottsym{[}  \Gamma_{{\mathrm{1}}}  \ottsym{]} )  \hyp{5.5}\\
            \Gamma_{{\mathrm{1}}}  \vdash  \ottnt{M_{{\mathrm{1}}}}  :  \ottnt{S_{{\mathrm{1}}}}  \otimes  \ottnt{S_{{\mathrm{2}}}}  \mid  \ottsym{0} \hyp{5.3}\\
            \mathcal{G}  \ottsym{[}  \ottmv{x_{{\mathrm{1}}}}  \mathord:  \ottnt{S_{{\mathrm{1}}}}  \parallel  \ottmv{x_{{\mathrm{2}}}}  \mathord:  \ottnt{S_{{\mathrm{2}}}}  \ottsym{]}  \vdash  \ottnt{N}  :  \ottnt{T}  \mid  \ottnt{e} \hyp{5.4}.
        \end{gather}
        First of all, we can see $ \ottmv{x} \neq \ottmv{x_{{\mathrm{1}}}} $ and $ \ottmv{x} \neq \ottmv{x_{{\mathrm{2}}}} $ by \hypref{0.2} and \hypref{5.5}.
        So,
        \begin{gather}
              \mathcal{G}  \ottsym{[}  \ottmv{x_{{\mathrm{1}}}}  \mathord:  \ottnt{S_{{\mathrm{1}}}}  \parallel  \ottmv{x_{{\mathrm{2}}}}  \mathord:  \ottnt{S_{{\mathrm{2}}}}  \ottsym{]} ^{- \ottmv{x}  \mathord:  \ottnt{S} }  \ottsym{=}  \mathcal{G} ^{- \ottmv{x}  \mathord:  \ottnt{S} }   \ottsym{[}  \ottmv{x_{{\mathrm{1}}}}  \mathord:  \ottnt{S_{{\mathrm{1}}}}  \parallel  \ottmv{x_{{\mathrm{2}}}}  \mathord:  \ottnt{S_{{\mathrm{2}}}}  \ottsym{]}  \hyp{5.5.3}.
        \end{gather}
        We also note that
        \begin{gather}
            \ottsym{\{}  \ottmv{x_{{\mathrm{1}}}}  \ottsym{\}}  \uplus  \ottsym{\{}  \ottmv{x_{{\mathrm{2}}}}  \ottsym{\}}  \uplus   \ottkw{dom} (  \mathcal{G} ^{- \ottmv{x}  \mathord:  \ottnt{S} }   \ottsym{[}   \Gamma_{{\mathrm{1}}} ^{- \ottmv{x}  \mathord:  \ottnt{S} }   \ottsym{]} )  \hyp{5.8}
        \end{gather}
        by \hypref{5.5}.
        We have $ \ottmv{x}  \mathord:  \ottnt{S}  \in  \Gamma_{{\mathrm{1}}} $ or $ \ottmv{x}  \mathord:  \ottnt{S}  \in  \mathcal{G} $ by \propref{env/num/bindings/ctx}, \hypref{0.2}, and \hypref{5.1}.
        So, it suffices to consider the following three cases.
        \begin{match}
            \item[$ \ottmv{x}  \mathord:  \ottnt{S}  \in  \Gamma_{{\mathrm{1}}} $ and $ \ottmv{x}  \mathord:  \ottnt{S}  \in  \mathcal{G} $]
            We have $ \ottmv{x}  \mathord:  \ottnt{S}  \in  \mathcal{G}  \ottsym{[}  \ottmv{x_{{\mathrm{1}}}}  \mathord:  \ottnt{S_{{\mathrm{1}}}}  \parallel  \ottmv{x_{{\mathrm{2}}}}  \mathord:  \ottnt{S_{{\mathrm{2}}}}  \ottsym{]} $ by \propref{env/num/bindings/ctx}.
            So, we have
            \begin{gather}
                 \Gamma_{{\mathrm{1}}} ^{- \ottmv{x}  \mathord:  \ottnt{S} }   \vdash  \ottnt{M_{{\mathrm{1}}}}  \ottsym{[}  \ottnt{V}  \ottsym{/}  \ottmv{x}  \ottsym{]}  :  \ottnt{S_{{\mathrm{1}}}}  \otimes  \ottnt{S_{{\mathrm{2}}}}  \mid  \ottsym{0} \hyp{5.6}\\
                 \mathcal{G} ^{- \ottmv{x}  \mathord:  \ottnt{S} }   \ottsym{[}  \ottmv{x_{{\mathrm{1}}}}  \mathord:  \ottnt{S_{{\mathrm{1}}}}  \parallel  \ottmv{x_{{\mathrm{2}}}}  \mathord:  \ottnt{S_{{\mathrm{2}}}}  \ottsym{]}  \vdash  \ottnt{N}  \ottsym{[}  \ottnt{V}  \ottsym{/}  \ottmv{x}  \ottsym{]}  :  \ottnt{T}  \mid  \ottnt{e}, \hyp{5.7}
            \end{gather}
            by the induction hypothesis and \hypref{5.5.3}.
            As a result, we can derive the goal by \ruleref{T-ULet}, \hypref{5.8}, \hypref{5.6}, and \hypref{5.7}, where the environment is well-formed by \propref{typing/env/remove/well-formed}.

            \item[$ \ottmv{x}  \mathord:  \ottnt{S}  \in  \Gamma_{{\mathrm{1}}} $ and $ \ottmv{x}  \mathord:  \ottnt{S}  \notin  \mathcal{G} $]
            We have
            \begin{gather}
                 \Gamma_{{\mathrm{1}}} ^{- \ottmv{x}  \mathord:  \ottnt{S} }   \vdash  \ottnt{M_{{\mathrm{1}}}}  \ottsym{[}  \ottnt{V}  \ottsym{/}  \ottmv{x}  \ottsym{]}  :  \ottnt{S_{{\mathrm{1}}}}  \otimes  \ottnt{S_{{\mathrm{2}}}}  \mid  \ottsym{0} \hyp{5.9}
            \end{gather}
            by the induction hypothesis.
            We have $ \ottmv{x}  \notin   \ottkw{dom} ( \mathcal{G} )  $ because if $ \ottmv{x}  \in   \ottkw{dom} ( \mathcal{G} )  $, we have $ \ottmv{x}  \mathord:  \ottnt{S'}  \in  \mathcal{G} $ for some $ \ottnt{S'} \neq \ottnt{S} $, but this contradicts to the well-formedness of $\Gamma (=\mathcal{G}  \ottsym{[}  \Gamma_{{\mathrm{1}}}  \ottsym{]})$.
            Moreover, since $ \ottmv{x} \neq \ottmv{x_{{\mathrm{1}}}} $ and $ \ottmv{x} \neq \ottmv{x_{{\mathrm{2}}}} $, we have $ \ottmv{x}  \notin   \ottkw{dom} ( \mathcal{G}  \ottsym{[}  \ottmv{x_{{\mathrm{1}}}}  \mathord:  \ottnt{S_{{\mathrm{1}}}}  \parallel  \ottmv{x_{{\mathrm{2}}}}  \mathord:  \ottnt{S_{{\mathrm{2}}}}  \ottsym{]} )  $.
            So, we have
            \begin{gather}
                \mathcal{G}  \ottsym{[}  \ottmv{x_{{\mathrm{1}}}}  \mathord:  \ottnt{S_{{\mathrm{1}}}}  \parallel  \ottmv{x_{{\mathrm{2}}}}  \mathord:  \ottnt{S_{{\mathrm{2}}}}  \ottsym{]}  \vdash  \ottnt{N}  \ottsym{[}  \ottnt{V}  \ottsym{/}  \ottmv{x}  \ottsym{]}  :  \ottnt{T}  \mid  \ottnt{e} \hyp{5.10}
            \end{gather}
            by \propref{typing/subst/void} and \hypref{5.4}.
            We also know $  \mathcal{G} ^{- \ottmv{x}  \mathord:  \ottnt{S} }  \ottsym{=} \mathcal{G} $.
            As a result, we can derive the goal by \ruleref{T-ULet}, \hypref{5.8}, \hypref{5.9}, and \hypref{5.10}, where the environment is well-formed by \propref{typing/env/remove/well-formed}.

            \item[$ \ottmv{x}  \mathord:  \ottnt{S}  \notin  \Gamma_{{\mathrm{1}}} $ and $ \ottmv{x}  \mathord:  \ottnt{S}  \in  \mathcal{G} $]
            Similar to the sub cases above.
        \end{match}

        \item[\ruleref{T-OLet}]
        Similar to the case \ruleref{T-ULet}.

        \item[\ruleref{T-Weaken}]
        In this case,
        \begin{gather}
            \Gamma_{{\mathrm{1}}}  \vdash  \ottnt{M}  :  \ottnt{T}  \mid  \ottnt{e'} \hyp{6.1}\\
            \Gamma  \lesssim  \Gamma_{{\mathrm{1}}} \hyp{6.2}\\
             \ottnt{e'} \le \ottnt{e}  \hyp{6.3}.
        \end{gather}
        We consider wether $ \ottmv{x}  \mathord:  \ottnt{S}  \in  \Gamma_{{\mathrm{1}}} $ or not.
        \begin{match}
            \item[$ \ottmv{x}  \mathord:  \ottnt{S}  \in  \Gamma_{{\mathrm{1}}} $]
            We have
            \begin{gather*}
                 \Gamma_{{\mathrm{1}}} ^{- \ottmv{x}  \mathord:  \ottnt{S} }   \vdash  \ottnt{M}  \ottsym{[}  \ottnt{V}  \ottsym{/}  \ottmv{x}  \ottsym{]}  :  \ottnt{T}  \mid  \ottnt{e'}
            \end{gather*}
            by the induction hypothesis.

            \item[$ \ottmv{x}  \mathord:  \ottnt{S}  \notin  \Gamma_{{\mathrm{1}}} $]
            We have $ \ottmv{x}  \notin   \ottkw{dom} ( \Gamma_{{\mathrm{1}}} )  $.
            That is because if $ \ottmv{x}  \in   \ottkw{dom} ( \Gamma_{{\mathrm{1}}} )  $, we have $ \ottmv{x}  \mathord:  \ottnt{S'}  \in  \Gamma_{{\mathrm{1}}} $ for some $ \ottnt{S'} \neq \ottnt{S} $.
            Under the hypothesis, we have $ \ottmv{x}  \mathord:  \ottnt{S'}  \in  \Gamma $ by \propref{env/in/sup} and \hypref{6.2}.
            However, this contradicts the well-formedness of $\Gamma$ since $ \ottmv{x}  \mathord:  \ottnt{S}  \in  \Gamma $.
            Consequently, we have
            \begin{gather*}
                \Gamma_{{\mathrm{1}}}  \vdash  \ottnt{M}  \ottsym{[}  \ottnt{V}  \ottsym{/}  \ottmv{x}  \ottsym{]}  :  \ottnt{T}  \mid  \ottnt{e'}
            \end{gather*}
            by \propref{typing/subst/void} and \hypref{6.1}, and, in this case, $ \Gamma_{{\mathrm{1}}} \ottsym{=}  \Gamma_{{\mathrm{1}}} ^{- \ottmv{x}  \mathord:  \ottnt{S} }  $.
        \end{match}
        So, in either case, we can have
        \begin{gather}
             \Gamma_{{\mathrm{1}}} ^{- \ottmv{x}  \mathord:  \ottnt{S} }   \vdash  \ottnt{M}  \ottsym{[}  \ottnt{V}  \ottsym{/}  \ottmv{x}  \ottsym{]}  :  \ottnt{T}  \mid  \ottnt{e'} \hyp{6.4}.
        \end{gather}
        We also have
        \begin{gather}
             \Gamma ^{- \ottmv{x}  \mathord:  \ottnt{S} }   \lesssim   \Gamma_{{\mathrm{1}}} ^{- \ottmv{x}  \mathord:  \ottnt{S} }  \hyp{6.5}
        \end{gather}
        by \propref{env/remove/sub} and \hypref{6.2}.
        Now we can derive the goal by \ruleref{T-Weaken}, \hypref{6.3}, \hypref{6.4}, and \hypref{6.5}; and $ \Gamma ^{- \ottmv{x}  \mathord:  \ottnt{S} } $ is well-formed by \propref{typing/env/remove/well-formed}.
    \end{match}
\end{prop}

\begin{prop}{typing/subst/ord}
    If
    \begin{gather}
        \mathcal{G}  \ottsym{[}  \ottmv{x}  \mathord:  \ottnt{S}  \ottsym{]}  \vdash  \ottnt{M}  :  \ottnt{T}  \mid  \ottnt{e} \hyp{0.1}\\
        \ottkw{ord} \, \ottnt{S} \hyp{0.2}\\
        \ottnt{C}  \vdash  \ottnt{V}  :  \ottnt{S}  \mid  \ottsym{0} \hyp{0.3}
    \end{gather}
    then $\mathcal{G}  \ottsym{[}  \ottnt{C}  \ottsym{]}  \vdash  \ottnt{M}  \ottsym{[}  \ottnt{V}  \ottsym{/}  \ottmv{x}  \ottsym{]}  :  \ottnt{T}  \mid  \ottnt{e}$.

    \proof By induction on the given derivation of \hypref{0.1}.
    \begin{match}
        \item[\ruleref{T-Unit}, \ruleref{T-New}, \ruleref{T-Op}, \ruleref{T-Split}, and \ruleref{T-Close}]
        In these cases, we have $ \mathcal{G}  \ottsym{[}  \ottmv{x}  \mathord:  \ottnt{S}  \ottsym{]} \ottsym{=}  \cdot  $.
        So, $ \ottmv{x}  \mathord:  \ottnt{S} \ottsym{=}  \cdot  $ by \propref{env/eq/ctx(null)}, but this cannot happen.

        \item[\ruleref{T-Loc}]
        In this case, we have $ \mathcal{G}  \ottsym{[}  \ottmv{x}  \mathord:  \ottnt{S}  \ottsym{]} \ottsym{=} \ottmv{l}  \mathord:  \ottsym{[}  \ottnt{m}  \ottsym{]} $.
        So, $ \ottmv{x}  \mathord:  \ottnt{S} \ottsym{=} \ottmv{l}  \mathord:  \ottsym{[}  \ottnt{m}  \ottsym{]} $ by \propref{env/eq/ctx(location)}, but this cannot happen.

        \item[\ruleref{T-Var}]
        In this case,
        \begin{gather}
             \mathcal{G}  \ottsym{[}  \ottmv{x}  \mathord:  \ottnt{S}  \ottsym{]} \ottsym{=} \ottmv{y}  \mathord:  \ottnt{T}  \hyp{1.1}\\
             \ottnt{M} \ottsym{=} \ottmv{y}  \hyp{1.2}\\
             \ottnt{e} \ottsym{=} \ottsym{0}  \hyp{1.3}.
        \end{gather}
        We have
        \begin{gather}
             \mathcal{G} \ottsym{=} \ottsym{[]}  \hyp{1.4}\\
             \ottmv{x}  \mathord:  \ottnt{S} \ottsym{=} \ottmv{y}  \mathord:  \ottnt{T}  \hyp{1.5}
        \end{gather}
        by \propref{env/eq/ctx(binding)} and \hypref{1.1}.
        So, $ \ottmv{x} \ottsym{=} \ottmv{y} $ and $ \ottnt{S} \ottsym{=} \ottnt{T} $ by \hypref{1.5}.
        Now, the goal is identical to \hypref{0.3}.

        \item[\ruleref{T-Abs}]
        In this case, we know $\ottkw{unr} \, \mathcal{G}  \ottsym{[}  \ottmv{x}  \mathord:  \ottnt{S}  \ottsym{]}$.
        However, this contradicts to $\ottkw{ord} \, \ottnt{S}$.
        So, this case cannot happen.

        \item[\ruleref{T-UAbs}]
        In this case,
        \begin{gather}
             \ottnt{M} \ottsym{=} \lambda^\circ  \ottmv{x_{{\mathrm{1}}}}  \ottsym{.}  \ottnt{M_{{\mathrm{1}}}}  \hyp{2.1}\\
             \ottnt{T} \ottsym{=}  \ottnt{T_{{\mathrm{1}}}} \rightarrowtriangle _{ \ottnt{e_{{\mathrm{1}}}} } \ottnt{T_{{\mathrm{2}}}}   \hyp{2.2}\\
             \ottnt{e} \ottsym{=} \ottsym{0}  \hyp{2.3}\\
             \ottmv{x_{{\mathrm{1}}}}  \notin   \ottkw{dom} ( \mathcal{G}  \ottsym{[}  \ottmv{x}  \mathord:  \ottnt{S}  \ottsym{]} )   \hyp{2.4}\\
            \mathcal{G}  \ottsym{[}  \ottmv{x}  \mathord:  \ottnt{S}  \ottsym{]}  \parallel  \ottmv{x_{{\mathrm{1}}}}  \mathord:  \ottnt{T_{{\mathrm{1}}}}  \vdash  \ottnt{M_{{\mathrm{1}}}}  :  \ottnt{T_{{\mathrm{2}}}}  \mid  \ottnt{e_{{\mathrm{1}}}} \hyp{2.5}.
        \end{gather}
        We have
        \begin{gather}
            \mathcal{G}  \ottsym{[}  \ottnt{C}  \ottsym{]}  \parallel  \ottmv{x_{{\mathrm{1}}}}  \mathord:  \ottnt{T_{{\mathrm{1}}}}  \vdash  \ottnt{M_{{\mathrm{1}}}}  \ottsym{[}  \ottnt{V}  \ottsym{/}  \ottmv{x}  \ottsym{]}  :  \ottnt{T_{{\mathrm{2}}}}  \mid  \ottnt{e_{{\mathrm{1}}}} \hyp{2.6}
        \end{gather}
        by the induction hypothesis.
        We have
        \begin{gather}
             \ottmv{x_{{\mathrm{1}}}}  \notin   \ottkw{dom} ( \mathcal{G}  \ottsym{[}  \ottnt{C}  \ottsym{]} )   \hyp{2.7}
        \end{gather}
        by \hypref{2.4} since $ \ottkw{dom} ( \ottnt{C} ) $ has no variable.
        Therefore, we can derive
        \begin{gather*}
            \mathcal{G}  \ottsym{[}  \ottnt{C}  \ottsym{]}  \vdash  \lambda^\circ  \ottmv{x_{{\mathrm{1}}}}  \ottsym{.}  \ottnt{M_{{\mathrm{1}}}}  \ottsym{[}  \ottnt{V}  \ottsym{/}  \ottmv{x}  \ottsym{]}  :   \ottnt{T_{{\mathrm{1}}}} \rightarrowtriangle _{ \ottnt{e_{{\mathrm{1}}}} } \ottnt{T_{{\mathrm{2}}}}   \mid  \ottsym{0}
        \end{gather*}
        by \ruleref{T-UAbs}, \hypref{2.6}, and \hypref{2.7}.
        We can see $\mathcal{G}  \ottsym{[}  \ottnt{C}  \ottsym{]}$ is well-formed because $\mathcal{G}  \ottsym{[}  \ottnt{C}  \ottsym{]}  \parallel  \ottmv{x_{{\mathrm{1}}}}  \mathord:  \ottnt{T_{{\mathrm{1}}}}$ is well-formed by \hypref{2.6}, and $ \lambda^\circ  \ottmv{x_{{\mathrm{1}}}}  \ottsym{.}  \ottnt{M_{{\mathrm{1}}}}  \ottsym{[}  \ottnt{V}  \ottsym{/}  \ottmv{x}  \ottsym{]} \ottsym{=} \ottsym{(}  \lambda^\circ  \ottmv{x_{{\mathrm{1}}}}  \ottsym{.}  \ottnt{M_{{\mathrm{1}}}}  \ottsym{)}  \ottsym{[}  \ottnt{V}  \ottsym{/}  \ottmv{x}  \ottsym{]} $ because we have $ \ottmv{x} \neq \ottmv{x_{{\mathrm{1}}}} $ by \hypref{2.4}.

        \item[\ruleref{T-RAbs} and \ruleref{T-LAbs}]
        Similar to the case \ruleref{T-UAbs}.

        \item[\ruleref{T-App}]
        In this case,
        \begin{gather}
             \mathcal{G}  \ottsym{[}  \ottmv{x}  \mathord:  \ottnt{S}  \ottsym{]} \ottsym{=} \Gamma_{{\mathrm{1}}}  \ottsym{,}  \Gamma_{{\mathrm{2}}}  \hyp{3.1}\\
             \ottnt{M} \ottsym{=} \ottnt{M_{{\mathrm{1}}}} \, \ottnt{M_{{\mathrm{2}}}}  \hyp{3.2}\\
             \ottnt{e} \ottsym{=}   \ottnt{e_{{\mathrm{0}}}}  \sqcup  \ottnt{e_{{\mathrm{1}}}}   \sqcup  \ottnt{e_{{\mathrm{2}}}}   \hyp{3.3}\\
            \Gamma_{{\mathrm{1}}}  \vdash  \ottnt{M_{{\mathrm{1}}}}  :   \ottnt{S} \rightarrow _{ \ottnt{e_{{\mathrm{0}}}} } \ottnt{T}   \mid  \ottnt{e_{{\mathrm{1}}}} \hyp{3.4}\\
            \Gamma_{{\mathrm{2}}}  \vdash  \ottnt{M_{{\mathrm{2}}}}  :  \ottnt{S}  \mid  \ottnt{e_{{\mathrm{2}}}} \hyp{3.5}.
        \end{gather}
        There are three cases by \propref{env/eq/ctx(concat)} and \hypref{3.1}.
        \begin{match}
            \item[$ \mathcal{G} \ottsym{=} \ottsym{[]} $ and $ \ottmv{x}  \mathord:  \ottnt{S} \ottsym{=} \Gamma_{{\mathrm{1}}}  \ottsym{,}  \Gamma_{{\mathrm{2}}} $]
            This case cannot happen since $ \ottmv{x}  \mathord:  \ottnt{S} \ottsym{=} \Gamma_{{\mathrm{1}}}  \ottsym{,}  \Gamma_{{\mathrm{2}}} $ is impossible.

            \item[$ \mathcal{G} \ottsym{=} \mathcal{G}_{{\mathrm{1}}}  \ottsym{,}  \Gamma_{{\mathrm{2}}} $ and $ \mathcal{G}_{{\mathrm{1}}}  \ottsym{[}  \ottmv{x}  \mathord:  \ottnt{S}  \ottsym{]} \ottsym{=} \Gamma_{{\mathrm{1}}} $]
            We have
            \begin{gather}
                \mathcal{G}_{{\mathrm{1}}}  \ottsym{[}  \ottnt{C}  \ottsym{]}  \vdash  \ottnt{M_{{\mathrm{1}}}}  \ottsym{[}  \ottnt{V}  \ottsym{/}  \ottmv{x}  \ottsym{]}  :   \ottnt{S} \rightarrow _{ \ottnt{e_{{\mathrm{0}}}} } \ottnt{T}   \mid  \ottnt{e_{{\mathrm{1}}}} \hyp{3.6}
            \end{gather}
            by the induction hypothesis for \hypref{3.4}.
            Since we know $\mathcal{G}_{{\mathrm{1}}}  \ottsym{[}  \ottmv{x}  \mathord:  \ottnt{S}  \ottsym{]}  \ottsym{,}  \Gamma_{{\mathrm{2}}}$ is well-formed by \hypref{0.1}, $ \ottmv{x}  \mathord:  \ottnt{S}  \notin  \Gamma_{{\mathrm{2}}} $ by \propref{typing/env/in/ord/unique} and \hypref{0.2}.
            From that, we have $ \ottmv{x}  \notin   \ottkw{dom} ( \Gamma_{{\mathrm{2}}} )  $ because if $ \ottmv{x}  \in   \ottkw{dom} ( \Gamma_{{\mathrm{2}}} )  $, we have $ \ottmv{x}  \mathord:  \ottnt{S'}  \in  \Gamma_{{\mathrm{2}}} $ for some $ \ottnt{S'} \neq \ottnt{S} $, but this contradicts to the well-formedness of $\mathcal{G}_{{\mathrm{1}}}  \ottsym{[}  \ottmv{x}  \mathord:  \ottnt{S}  \ottsym{]}  \ottsym{,}  \Gamma_{{\mathrm{2}}}$.
            So, we have
            \begin{gather}
                \Gamma_{{\mathrm{2}}}  \vdash  \ottnt{M_{{\mathrm{2}}}}  \ottsym{[}  \ottnt{V}  \ottsym{/}  \ottmv{x}  \ottsym{]}  :  \ottnt{S}  \mid  \ottnt{e_{{\mathrm{2}}}} \hyp{3.7}
            \end{gather}
            by \propref{typing/subst/void} and \hypref{3.5}.
            Now we can derive the goal
            \begin{gather*}
                \mathcal{G}_{{\mathrm{1}}}  \ottsym{[}  \ottnt{C}  \ottsym{]}  \ottsym{,}  \Gamma_{{\mathrm{2}}}  \vdash  \ottsym{(}  \ottnt{M_{{\mathrm{1}}}} \, \ottnt{M_{{\mathrm{2}}}}  \ottsym{)}  \ottsym{[}  \ottnt{V}  \ottsym{/}  \ottmv{x}  \ottsym{]}  :  \ottnt{T}  \mid    \ottnt{e_{{\mathrm{0}}}}  \sqcup  \ottnt{e_{{\mathrm{1}}}}   \sqcup  \ottnt{e_{{\mathrm{2}}}} 
            \end{gather*}
            by \ruleref{T-App}, \hypref{3.6}, and \hypref{3.7}, where the environment is well-formed by \propref{typing/env/hole/runtime/well-formed}.

            \item[$ \mathcal{G} \ottsym{=} \Gamma_{{\mathrm{1}}}  \ottsym{,}  \mathcal{G}_{{\mathrm{2}}} $ and $ \mathcal{G}_{{\mathrm{2}}}  \ottsym{[}  \ottmv{x}  \mathord:  \ottnt{S}  \ottsym{]} \ottsym{=} \Gamma_{{\mathrm{2}}} $]
            Similar to the sub case above.
        \end{match}

        \item[\ruleref{T-UApp}, \ruleref{T-RApp}, \ruleref{T-LApp}, \ruleref{T-UPair}, and \ruleref{T-OPair}]
        Similar to the case \ruleref{T-App}.

        \item[\ruleref{T-ULet}]
        In this case,
        \begin{gather}
             \mathcal{G}  \ottsym{[}  \ottmv{x}  \mathord:  \ottnt{S}  \ottsym{]} \ottsym{=} \mathcal{G}_{{\mathrm{1}}}  \ottsym{[}  \Gamma_{{\mathrm{1}}}  \ottsym{]}  \hyp{4.1}\\
             \ottnt{M} \ottsym{=} \ottkw{let} \, \ottmv{x_{{\mathrm{1}}}}  \otimes  \ottmv{x_{{\mathrm{2}}}}  \ottsym{=}  \ottnt{M_{{\mathrm{1}}}} \, \ottkw{in} \, \ottnt{M_{{\mathrm{2}}}}  \hyp{4.2}\\
            \ottsym{\{}  \ottmv{x_{{\mathrm{1}}}}  \ottsym{\}}  \uplus  \ottsym{\{}  \ottmv{x_{{\mathrm{2}}}}  \ottsym{\}}  \uplus   \ottkw{dom} ( \mathcal{G}_{{\mathrm{1}}}  \ottsym{[}  \Gamma_{{\mathrm{1}}}  \ottsym{]} )  \hyp{4.3}\\
            \Gamma_{{\mathrm{1}}}  \vdash  \ottnt{M_{{\mathrm{1}}}}  :  \ottnt{S_{{\mathrm{1}}}}  \otimes  \ottnt{S_{{\mathrm{2}}}}  \mid  \ottsym{0} \hyp{4.4}\\
            \mathcal{G}_{{\mathrm{1}}}  \ottsym{[}  \ottmv{x_{{\mathrm{1}}}}  \mathord:  \ottnt{S_{{\mathrm{1}}}}  \parallel  \ottmv{x_{{\mathrm{2}}}}  \mathord:  \ottnt{S_{{\mathrm{2}}}}  \ottsym{]}  \vdash  \ottnt{M_{{\mathrm{2}}}}  :  \ottnt{T}  \mid  \ottnt{e} \hyp{4.5}.
        \end{gather}
        We have
        \begin{gather*}
             \ottmv{x} \neq \ottmv{x_{{\mathrm{1}}}}  \\
             \ottmv{x} \neq \ottmv{x_{{\mathrm{2}}}} 
        \end{gather*}
        by \hypref{4.3}.
        We have $ \ottmv{x}  \mathord:  \ottnt{S}  \in  \mathcal{G}  \ottsym{[}  \ottmv{x}  \mathord:  \ottnt{S}  \ottsym{]} $ by \propref{env/in/ctx/binding}.
        We know $\mathcal{G}_{{\mathrm{1}}}  \ottsym{[}  \Gamma_{{\mathrm{1}}}  \ottsym{]}$ is well-formed by \hypref{0.1}.
        From those, we have $ \ottsym{\#} _{ \ottmv{x}  \mathord:  \ottnt{S} }( \mathcal{G}_{{\mathrm{1}}}  \ottsym{[}  \Gamma_{{\mathrm{1}}}  \ottsym{]} )  = 1$ since $\ottkw{ord} \, \ottnt{S}$.
        So, we have the following cases by \propref{env/num/bindings/ctx}.
        \begin{match}
            \item[$ \ottmv{x}  \mathord:  \ottnt{S}  \in  \Gamma_{{\mathrm{1}}} $ and $ \ottmv{x}  \mathord:  \ottnt{S}  \notin  \mathcal{G}_{{\mathrm{1}}} $]
            We have some $\mathcal{G}'_{{\mathrm{1}}}$ such that $ \mathcal{G}'_{{\mathrm{1}}}  \ottsym{[}  \ottmv{x}  \mathord:  \ottnt{S}  \ottsym{]} \ottsym{=} \Gamma_{{\mathrm{1}}} $ by \propref{env/in/ctx/exists}.
            So, we have
            \begin{gather}
                \mathcal{G}'_{{\mathrm{1}}}  \ottsym{[}  \ottnt{C}  \ottsym{]}  \vdash  \ottnt{M_{{\mathrm{1}}}}  \ottsym{[}  \ottnt{V}  \ottsym{/}  \ottmv{x}  \ottsym{]}  :  \ottnt{S_{{\mathrm{1}}}}  \otimes  \ottnt{S_{{\mathrm{2}}}}  \mid  \ottsym{0} \hyp{4.6}
            \end{gather}
            by the induction hypothesis for \hypref{4.4}.
            We can see $ \ottmv{x}  \notin   \ottkw{dom} ( \mathcal{G}_{{\mathrm{1}}}  \ottsym{[}  \ottmv{x_{{\mathrm{1}}}}  \mathord:  \ottnt{S_{{\mathrm{1}}}}  \parallel  \ottmv{x_{{\mathrm{2}}}}  \mathord:  \ottnt{S_{{\mathrm{2}}}}  \ottsym{]} )  $ because if $ \ottmv{x}  \in   \ottkw{dom} ( \mathcal{G}_{{\mathrm{1}}}  \ottsym{[}  \ottmv{x_{{\mathrm{1}}}}  \mathord:  \ottnt{S_{{\mathrm{1}}}}  \parallel  \ottmv{x_{{\mathrm{2}}}}  \mathord:  \ottnt{S_{{\mathrm{2}}}}  \ottsym{]} )  $, we have $ \ottmv{x}  \mathord:  \ottnt{S'}  \in  \mathcal{G}_{{\mathrm{1}}} $ for some $ \ottnt{S'} \neq \ottnt{S} $ since $ \ottmv{x} \neq \ottmv{x_{{\mathrm{1}}}} $, $ \ottmv{x} \neq \ottmv{x_{{\mathrm{2}}}} $, and $ \ottmv{x}  \mathord:  \ottnt{S}  \notin  \mathcal{G}_{{\mathrm{1}}} $, but this contradicts to the well-formedness of $\mathcal{G}_{{\mathrm{1}}}  \ottsym{[}  \Gamma_{{\mathrm{1}}}  \ottsym{]}$ since $ \ottmv{x}  \mathord:  \ottnt{S}  \in  \Gamma_{{\mathrm{1}}} $.
            So, we have
            \begin{gather}
                \mathcal{G}_{{\mathrm{1}}}  \ottsym{[}  \ottmv{x_{{\mathrm{1}}}}  \mathord:  \ottnt{S_{{\mathrm{1}}}}  \parallel  \ottmv{x_{{\mathrm{2}}}}  \mathord:  \ottnt{S_{{\mathrm{2}}}}  \ottsym{]}  \vdash  \ottnt{M_{{\mathrm{2}}}}  \ottsym{[}  \ottnt{V}  \ottsym{/}  \ottmv{x}  \ottsym{]}  :  \ottnt{T}  \mid  \ottnt{e} \hyp{4.7}
            \end{gather}
            by \propref{typing/subst/void} and \hypref{4.5}.
            We can see
            \begin{gather}
                \ottsym{\{}  \ottmv{x_{{\mathrm{1}}}}  \ottsym{\}}  \uplus  \ottsym{\{}  \ottmv{x_{{\mathrm{2}}}}  \ottsym{\}}  \uplus   \ottkw{dom} ( \mathcal{G}_{{\mathrm{1}}}  \ottsym{[}  \mathcal{G}'_{{\mathrm{1}}}  \ottsym{[}  \ottnt{C}  \ottsym{]}  \ottsym{]} )  \hyp{4.8}
            \end{gather}
            by \hypref{4.3} since $ \mathcal{G}_{{\mathrm{1}}}  \ottsym{[}  \Gamma_{{\mathrm{1}}}  \ottsym{]} \ottsym{=} \mathcal{G}_{{\mathrm{1}}}  \ottsym{[}  \mathcal{G}'_{{\mathrm{1}}}  \ottsym{[}  \ottmv{x}  \mathord:  \ottnt{S}  \ottsym{]}  \ottsym{]} $ and $\ottnt{C}$ has no variable bindings.
            Now, we can derive
            \begin{gather}
                \mathcal{G}_{{\mathrm{1}}}  \ottsym{[}  \mathcal{G}'_{{\mathrm{1}}}  \ottsym{[}  \ottnt{C}  \ottsym{]}  \ottsym{]}  \vdash  \ottkw{let} \, \ottmv{x_{{\mathrm{1}}}}  \otimes  \ottmv{x_{{\mathrm{2}}}}  \ottsym{=}  \ottnt{M_{{\mathrm{1}}}}  \ottsym{[}  \ottnt{V}  \ottsym{/}  \ottmv{x}  \ottsym{]} \, \ottkw{in} \, \ottnt{M_{{\mathrm{2}}}}  \ottsym{[}  \ottnt{V}  \ottsym{/}  \ottmv{x}  \ottsym{]}  :  \ottnt{T}  \mid  \ottnt{e}
            \end{gather}
            by \ruleref{T-ULet}, \hypref{4.6}, \hypref{4.7}, and \hypref{4.8}.
            To this end, we will show
            \begin{gather*}
                 \mathcal{G}_{{\mathrm{1}}}  \ottsym{[}  \mathcal{G}'_{{\mathrm{1}}}  \ottsym{[}  \ottnt{C}  \ottsym{]}  \ottsym{]} \ottsym{=} \mathcal{G}  \ottsym{[}  \ottnt{C}  \ottsym{]} .
            \end{gather*}
            Then, $\mathcal{G}_{{\mathrm{1}}}  \ottsym{[}  \mathcal{G}'_{{\mathrm{1}}}  \ottsym{[}  \ottnt{C}  \ottsym{]}  \ottsym{]}$ is well-formed by \propref{typing/env/hole/runtime/well-formed} and \hypref{0.1}, and the derived judgment matches to the goal.
            We have
            \begin{gather*}
                g[x:S][C/x:S] = g1[g1'[x:S] ][C/x:S]
            \end{gather*}
            by \hypref{4.1}.
            Since, $\ottkw{ord} \, \ottnt{S}$, we have $ \mathcal{G}  \ottsym{[}  \ottnt{C}  \ottsym{]} \ottsym{=} \mathcal{G}_{{\mathrm{1}}}  \ottsym{[}  \mathcal{G}'_{{\mathrm{1}}}  \ottsym{[}  \ottnt{C}  \ottsym{]}  \ottsym{]} $ by applying \propref{typing/env/replace/ord} to both sides.

            \item[$ \ottmv{x}  \mathord:  \ottnt{S}  \in  \mathcal{G}_{{\mathrm{1}}} $ and $ \ottmv{x}  \mathord:  \ottnt{S}  \notin  \Gamma_{{\mathrm{1}}} $]
            We have $ \ottmv{x}  \mathord:  \ottnt{S}  \in  \mathcal{G}_{{\mathrm{1}}}  \ottsym{[}  \ottmv{x_{{\mathrm{1}}}}  \mathord:  \ottnt{S_{{\mathrm{1}}}}  \parallel  \ottmv{x_{{\mathrm{2}}}}  \mathord:  \ottnt{S_{{\mathrm{2}}}}  \ottsym{]} $ by \propref{env/num/bindings/ctx}.
            So, we have some $\mathcal{G}'_{{\mathrm{1}}}$ such that
            \begin{gather}
                 \mathcal{G}'_{{\mathrm{1}}}  \ottsym{[}  \ottmv{x}  \mathord:  \ottnt{S}  \ottsym{]} \ottsym{=} \mathcal{G}_{{\mathrm{1}}}  \ottsym{[}  \ottmv{x_{{\mathrm{1}}}}  \mathord:  \ottnt{S_{{\mathrm{1}}}}  \parallel  \ottmv{x_{{\mathrm{2}}}}  \mathord:  \ottnt{S_{{\mathrm{2}}}}  \ottsym{]}  \hyp{4.9}
            \end{gather}
            by \propref{env/in/ctx/exists}.
            Now, we have
            \begin{gather}
                \mathcal{G}'_{{\mathrm{1}}}  \ottsym{[}  \ottnt{C}  \ottsym{]}  \vdash  \ottnt{M_{{\mathrm{2}}}}  :  \ottnt{S_{{\mathrm{1}}}}  \otimes  \ottnt{S_{{\mathrm{2}}}}  \mid  \ottnt{e} \hyp{4.10}
            \end{gather}
            by the induction hypothesis for \hypref{4.5}.
            Since, $\mathcal{G}'_{{\mathrm{1}}}  \ottsym{[}  \ottmv{x}  \mathord:  \ottnt{S}  \ottsym{]}$ is well-formed and $\ottkw{ord} \, \ottnt{S}$, we have
            \begin{gather*}
                 \mathcal{G}'_{{\mathrm{1}}}  \ottsym{[}  \ottmv{x}  \mathord:  \ottnt{S}  \ottsym{]}  \ottsym{[}  \ottnt{C}  \ottsym{/}  \ottmv{x}  \mathord:  \ottnt{S}  \ottsym{]} \ottsym{=} \mathcal{G}'_{{\mathrm{1}}}  \ottsym{[}  \ottnt{C}  \ottsym{]} 
            \end{gather*}
            by \propref{typing/env/replace/ord}.
            From that and \hypref{4.9}, \hypref{4.10} results in
            \begin{gather}
                \mathcal{G}_{{\mathrm{1}}}  \ottsym{[}  \ottnt{C}  \ottsym{/}  \ottmv{x}  \mathord:  \ottnt{S}  \ottsym{]}  \ottsym{[}  \ottmv{x_{{\mathrm{1}}}}  \mathord:  \ottnt{S_{{\mathrm{1}}}}  \parallel  \ottmv{x_{{\mathrm{2}}}}  \mathord:  \ottnt{S_{{\mathrm{2}}}}  \ottsym{]}  \vdash  \ottnt{M_{{\mathrm{2}}}}  :  \ottnt{S_{{\mathrm{1}}}}  \otimes  \ottnt{S_{{\mathrm{2}}}}  \mid  \ottnt{e} \hyp{4.11}
            \end{gather}
            since $ \ottmv{x} \neq \ottmv{x_{{\mathrm{1}}}} $ and $ \ottmv{x} \neq \ottmv{x_{{\mathrm{2}}}} $.
            We have $ \ottmv{x}  \notin   \ottkw{dom} ( \Gamma_{{\mathrm{1}}} )  $ because if $ \ottmv{x}  \in   \ottkw{dom} ( \Gamma_{{\mathrm{1}}} )  $, we have some $ \ottmv{x}  \mathord:  \ottnt{S'}  \in  \Gamma_{{\mathrm{1}}} $ for some $ \ottnt{S'} \neq \ottnt{S} $ since $ \ottmv{x}  \mathord:  \ottnt{S}  \notin  \Gamma_{{\mathrm{1}}} $ in this case, but this contradicts the well-formedness of $\mathcal{G}_{{\mathrm{1}}}  \ottsym{[}  \Gamma_{{\mathrm{1}}}  \ottsym{]}$ since $ \ottmv{x}  \mathord:  \ottnt{S}  \in  \mathcal{G}_{{\mathrm{1}}} $.
            So, we have
            \begin{gather}
                \Gamma_{{\mathrm{1}}}  \vdash  \ottnt{M_{{\mathrm{1}}}}  \ottsym{[}  \ottnt{V}  \ottsym{/}  \ottmv{x}  \ottsym{]}  :  \ottnt{S_{{\mathrm{1}}}}  \otimes  \ottnt{S_{{\mathrm{2}}}}  \mid  \ottsym{0} \hyp{4.12}
            \end{gather}
            by \propref{typing/subst/void} and \hypref{4.4}.
            We have
            \begin{gather}
                \ottsym{\{}  \ottmv{x_{{\mathrm{1}}}}  \ottsym{\}}  \uplus  \ottsym{\{}  \ottmv{x_{{\mathrm{2}}}}  \ottsym{\}}  \uplus   \ottkw{dom} ( \mathcal{G}_{{\mathrm{1}}}  \ottsym{[}  \ottnt{C}  \ottsym{/}  \ottmv{x}  \mathord:  \ottnt{S}  \ottsym{]}  \ottsym{[}  \Gamma_{{\mathrm{1}}}  \ottsym{]} )  \hyp{4.13}
            \end{gather}
            by \hypref{4.3} since $\ottnt{C}$ has novariable bindings.
            Now, we can derive
            \begin{gather}
                \mathcal{G}_{{\mathrm{1}}}  \ottsym{[}  \ottnt{C}  \ottsym{/}  \ottmv{x}  \mathord:  \ottnt{S}  \ottsym{]}  \ottsym{[}  \Gamma_{{\mathrm{1}}}  \ottsym{]}  \vdash  \ottkw{let} \, \ottmv{x_{{\mathrm{1}}}}  \otimes  \ottmv{x_{{\mathrm{2}}}}  \ottsym{=}  \ottnt{M_{{\mathrm{1}}}}  \ottsym{[}  \ottnt{V}  \ottsym{/}  \ottmv{x}  \ottsym{]} \, \ottkw{in} \, \ottnt{M_{{\mathrm{2}}}}  \ottsym{[}  \ottnt{V}  \ottsym{/}  \ottmv{x}  \ottsym{]}  :  \ottnt{T}  \mid  \ottnt{e}
            \end{gather}
            by \ruleref{T-ULet}, \hypref{4.11}, \hypref{4.12}, and \hypref{4.13}.
            To this end, we will show
            \begin{gather*}
                 \mathcal{G}_{{\mathrm{1}}}  \ottsym{[}  \ottnt{C}  \ottsym{/}  \ottmv{x}  \mathord:  \ottnt{S}  \ottsym{]}  \ottsym{[}  \Gamma_{{\mathrm{1}}}  \ottsym{]} \ottsym{=} \mathcal{G}  \ottsym{[}  \ottnt{C}  \ottsym{]} .
            \end{gather*}
            Then, $\mathcal{G}_{{\mathrm{1}}}  \ottsym{[}  \ottnt{C}  \ottsym{/}  \ottmv{x}  \mathord:  \ottnt{S}  \ottsym{]}  \ottsym{[}  \Gamma_{{\mathrm{1}}}  \ottsym{]}$ is well-formed by \propref{typing/env/hole/runtime/well-formed}, and the derived judgment matches the goal.
            We have
            \begin{gather*}
                 \mathcal{G}  \ottsym{[}  \ottmv{x}  \mathord:  \ottnt{S}  \ottsym{]}  \ottsym{[}  \ottnt{C}  \ottsym{/}  \ottmv{x}  \mathord:  \ottnt{S}  \ottsym{]} \ottsym{=} \mathcal{G}_{{\mathrm{1}}}  \ottsym{[}  \Gamma_{{\mathrm{1}}}  \ottsym{]}  \ottsym{[}  \ottnt{C}  \ottsym{/}  \ottmv{x}  \mathord:  \ottnt{S}  \ottsym{]} 
            \end{gather*}
            by \hypref{4.1}.
            Since $ \ottmv{x}  \mathord:  \ottnt{S}  \notin  \Gamma_{{\mathrm{1}}} $, the right hand side equals to $\mathcal{G}_{{\mathrm{1}}}  \ottsym{[}  \ottnt{C}  \ottsym{/}  \ottmv{x}  \mathord:  \ottnt{S}  \ottsym{]}  \ottsym{[}  \Gamma_{{\mathrm{1}}}  \ottsym{]}$.
            Since $\mathcal{G}  \ottsym{[}  \ottmv{x}  \mathord:  \ottnt{S}  \ottsym{]}$ is well-formed and $\ottkw{ord} \, \ottnt{S}$, the left hand side equals to $\mathcal{G}  \ottsym{[}  \ottnt{C}  \ottsym{]}$ by \propref{typing/env/replace/ord}.
        \end{match}

        \item[\ruleref{T-OLet}]
        Similar to the case \ruleref{T-ULet}.

        \item[\ruleref{T-Weaken}]
        In this case,
        \begin{gather}
            \Gamma_{{\mathrm{1}}}  \vdash  \ottnt{M}  :  \ottnt{T}  \mid  \ottnt{e_{{\mathrm{1}}}} \hyp{6.1}\\
            \mathcal{G}  \ottsym{[}  \ottmv{x}  \mathord:  \ottnt{S}  \ottsym{]}  \lesssim  \Gamma_{{\mathrm{1}}} \hyp{6.2}\\
             \ottnt{e_{{\mathrm{1}}}} \le \ottnt{e}  \hyp{6.3}.
        \end{gather}
        We have $ \ottmv{x}  \mathord:  \ottnt{S}  \in  \mathcal{G}  \ottsym{[}  \ottmv{x}  \mathord:  \ottnt{S}  \ottsym{]} $ by \propref{env/in/ctx/binding}.
        So, we have $ \ottmv{x}  \mathord:  \ottnt{S}  \in  \Gamma_{{\mathrm{1}}} $ by \propref{env/in/sub}, \hypref{0.2}, and \hypref{6.2}.
        From that, we get some $\mathcal{G}'$ such that $ \mathcal{G}'  \ottsym{[}  \ottmv{x}  \mathord:  \ottnt{S}  \ottsym{]} \ottsym{=} \Gamma_{{\mathrm{1}}} $ by \propref{env/in/ctx/exists}.
        Now, we can have
        \begin{gather}
            \mathcal{G}'  \ottsym{[}  \ottnt{C}  \ottsym{]}  \vdash  \ottnt{M}  \ottsym{[}  \ottnt{V}  \ottsym{/}  \ottmv{x}  \ottsym{]}  :  \ottnt{T}  \mid  \ottnt{e_{{\mathrm{1}}}} \hyp{6.4}
        \end{gather}
        by the induction hypothesis.
        To this end, we will show $\mathcal{G}  \ottsym{[}  \ottnt{C}  \ottsym{]}  \lesssim  \mathcal{G}'  \ottsym{[}  \ottnt{C}  \ottsym{]}$.
        Using that, we can show the goal by \ruleref{T-Weaken}, \hypref{6.4}, and \hypref{6.3}.

        Remember that we have had $\mathcal{G}  \ottsym{[}  \ottmv{x}  \mathord:  \ottnt{S}  \ottsym{]}  \lesssim  \mathcal{G}'  \ottsym{[}  \ottmv{x}  \mathord:  \ottnt{S}  \ottsym{]}$.
        So, $\mathcal{G}  \ottsym{[}  \ottmv{x}  \mathord:  \ottnt{S}  \ottsym{]}  \ottsym{[}  \ottnt{C}  \ottsym{/}  \ottmv{x}  \mathord:  \ottnt{S}  \ottsym{]}  \lesssim  \mathcal{G}'  \ottsym{[}  \ottmv{x}  \mathord:  \ottnt{S}  \ottsym{]}  \ottsym{[}  \ottnt{C}  \ottsym{/}  \ottmv{x}  \mathord:  \ottnt{S}  \ottsym{]}$ by \propref{env/replace/runtime/sub} and \hypref{0.2}.
        Here, we know $\mathcal{G}  \ottsym{[}  \ottmv{x}  \mathord:  \ottnt{S}  \ottsym{]}$ and $\mathcal{G}'  \ottsym{[}  \ottmv{x}  \mathord:  \ottnt{S}  \ottsym{]}$ are well-formed by \hypref{0.1} and \hypref{6.1}.
        Consequently, we have $\mathcal{G}  \ottsym{[}  \ottnt{C}  \ottsym{]}  \lesssim  \mathcal{G}'  \ottsym{[}  \ottnt{C}  \ottsym{]}$ by applying \propref{typing/env/replace/ord} to both sides.
    \end{match}
\end{prop}

\begin{prop}{typing/subst}
    If
    \begin{gather}
        \mathcal{G}  \ottsym{[}  \ottmv{x}  \mathord:  \ottnt{S}  \ottsym{]}  \vdash  \ottnt{M}  :  \ottnt{T}  \mid  \ottnt{e_{{\mathrm{1}}}} \hyp{0.1}\\
         \ottmv{x}  \notin   \ottkw{dom} ( \mathcal{G} )   \hyp{0.2}\\
        \ottnt{C}  \vdash  \ottnt{V}  :  \ottnt{S}  \mid  \ottnt{e_{{\mathrm{2}}}}, \hyp{0.3}
    \end{gather}
    then $\mathcal{G}  \ottsym{[}  \ottnt{C}  \ottsym{]}  \vdash  \ottnt{M}  \ottsym{[}  \ottnt{V}  \ottsym{/}  \ottmv{x}  \ottsym{]}  :  \ottnt{T}  \mid   \ottnt{e_{{\mathrm{1}}}}  \sqcup  \ottnt{e_{{\mathrm{2}}}} $.

    \proof By case analysis on whether $\ottkw{unr} \, \ottnt{S}$ or not (namely $\ottkw{ord} \, \ottnt{S}$).
    \begin{match}
        \item[$\ottkw{unr} \, \ottnt{S}$]
        We have $\ottnt{C}  \lesssim   \cdot $ and $ \cdot   \vdash  \ottnt{V}  :  \ottnt{S}  \mid  \ottsym{0}$ by \propref{typing/value/unr} and \hypref{0.3}.
        We have $ \ottmv{x}  \mathord:  \ottnt{S}  \in  \mathcal{G}  \ottsym{[}  \ottmv{x}  \mathord:  \ottnt{S}  \ottsym{]} $ by \propref{env/in/ctx/binding}.
        So, we have $ \ottsym{(}  \mathcal{G}  \ottsym{[}  \ottmv{x}  \mathord:  \ottnt{S}  \ottsym{]}  \ottsym{)} ^{- \ottmv{x}  \mathord:  \ottnt{S} }   \vdash  \ottnt{M}  \ottsym{[}  \ottnt{V}  \ottsym{/}  \ottmv{x}  \ottsym{]}  :  \ottnt{T}  \mid  \ottnt{e_{{\mathrm{1}}}}$ by \propref{typing/subst/unr} and \hypref{0.1}.
        We can see $  \ottsym{(}  \mathcal{G}  \ottsym{[}  \ottmv{x}  \mathord:  \ottnt{S}  \ottsym{]}  \ottsym{)} ^{- \ottmv{x}  \mathord:  \ottnt{S} }  \ottsym{=} \mathcal{G}  \ottsym{[}   \cdot   \ottsym{]} $ from \hypref{0.2}.
        We have $\mathcal{G}  \ottsym{[}  \ottnt{C}  \ottsym{]}  \lesssim  \mathcal{G}  \ottsym{[}   \cdot   \ottsym{]}$ by \propref{env/sub/ctx}.
        Since $ \ottnt{e_{{\mathrm{1}}}} \ottsym{<}  \ottnt{e_{{\mathrm{1}}}}  \sqcup  \ottnt{e_{{\mathrm{2}}}}  $, we can derive the goal by \ruleref{T-Weaken}.

        \item[$\ottkw{ord} \, \ottnt{S}$]
        We have $\ottnt{C}  \vdash  \ottnt{V}  :  \ottnt{S}  \mid  \ottsym{0}$ by \propref{typing/value/ord} and \hypref{0.3}.
        Now we have $\mathcal{G}  \ottsym{[}  \ottnt{C}  \ottsym{]}  \vdash  \ottnt{M}  \ottsym{[}  \ottnt{V}  \ottsym{/}  \ottmv{x}  \ottsym{]}  :  \ottnt{T}  \mid  \ottnt{e_{{\mathrm{1}}}}$ by \propref{typing/subst/ord}.
        Since $ \ottnt{e_{{\mathrm{1}}}} \ottsym{<}  \ottnt{e_{{\mathrm{1}}}}  \sqcup  \ottnt{e_{{\mathrm{2}}}}  $, we can derive the goal by \ruleref{T-Weaken}.
    \end{match}
\end{prop}

  \subsection{Properties for heap typing}
  \begin{prop}{heap/weaken}
    If
    \begin{gather}
        \ottnt{C_{{\mathrm{1}}}}  \vdash  \mathcal{H} \hyp{0.1}\\
        \ottnt{C_{{\mathrm{1}}}}  \lesssim  \ottnt{C_{{\mathrm{2}}}} \hyp{0.2},
    \end{gather}
    then $\ottnt{C_{{\mathrm{2}}}}  \vdash  \mathcal{H}$.

    \proof
    We will show the following facts, by which $\ottnt{C_{{\mathrm{2}}}}  \vdash  \mathcal{H}$ follows since $\ottnt{C_{{\mathrm{1}}}}  \vdash  \mathcal{H}$.
    \begin{enumerate}
        \item $\ottnt{C_{{\mathrm{2}}}}$ is order-defined.
              We have $\ottnt{C_{{\mathrm{1}}}}$ is order-defined by \hypref{0.1}.
              So, this follows by \propref{env/runtime/order-defined/sub}.

        \item $  \ottkw{dom} ( \ottnt{C_{{\mathrm{1}}}} )  \ottsym{=}  \ottkw{dom} ( \ottnt{C_{{\mathrm{2}}}} )  $.
              This follows by \propref{env/sub/runtime/dom}.

        \item $  \overline{  \langle \ottnt{C_{{\mathrm{1}}}} \rangle_{ \ottmv{l} }  }  \ottsym{=}  \overline{  \langle \ottnt{C_{{\mathrm{2}}}} \rangle_{ \ottmv{l} }  }  $ for any $\ottmv{l}$.
              We have shown $ \lBrack  \langle \ottnt{C_{{\mathrm{1}}}} \rangle_{ \ottmv{l} }  \rBrack $ and $ \lBrack  \langle \ottnt{C_{{\mathrm{2}}}} \rangle_{ \ottmv{l} }  \rBrack $ are traceable.
              So, the goal follows by \propref{env/runtime/usage/sub}. \qedhere
    \end{enumerate}
\end{prop}

\begin{prop}{heap/track/new}
    If
    \begin{gather}
        \mathcal{C}  \ottsym{[}   \cdot   \ottsym{]}  \vdash  \mathcal{H} \hyp{0.1}\\
         \ottmv{l}  \notin   \ottkw{dom} ( \mathcal{H} )   \hyp{0.2},
    \end{gather}
    then $\mathcal{C}  \ottsym{[}  \ottmv{l}  \mathord:  \ottsym{[}  \ottnt{m}  \ottsym{]}  \ottsym{]}  \vdash   \mathcal{H} \cup \ottsym{\{} \ottmv{l}  \mapsto  \ottsym{(}  \ottsym{0}  \ottsym{,}  \ottnt{m}  \ottsym{,}  \varepsilon  \ottsym{)} \ottsym{\}} $.

    \proof We show each condition of the heap typing for the goal.
    \begin{enumerate}
        \item $\mathcal{C}  \ottsym{[}  \ottmv{l}  \mathord:  \ottsym{[}  \ottnt{m}  \ottsym{]}  \ottsym{]}$ is order-defined.
              We will show $ \lBrack  \langle \mathcal{C}  \ottsym{[}  \ottmv{l}  \mathord:  \ottsym{[}  \ottnt{m}  \ottsym{]}  \ottsym{]} \rangle_{ \ottmv{l'} }  \rBrack $ is traceable for any $\ottmv{l'}$.
              By \propref{graph/traceable/iso}, it suffices to show the property for some isomorphic graph to $ \lBrack  \langle \mathcal{C}  \ottsym{[}  \ottmv{l}  \mathord:  \ottsym{[}  \ottnt{m}  \ottsym{]}  \ottsym{]} \rangle_{ \ottmv{l'} }  \rBrack $.
              We consider wether $ \ottmv{l'} \ottsym{=} \ottmv{l} $ or not.
              \begin{match}
                  \item[$ \ottmv{l'} \ottsym{=} \ottmv{l} $]
                  In this case, $ \ottmv{l'}  \notin   \ottkw{dom} ( \mathcal{C} )  $ since $  \ottkw{dom} ( \mathcal{C}  \ottsym{[}   \cdot   \ottsym{]} )  \ottsym{=}  \ottkw{dom} ( \mathcal{H} )  $ by \hypref{0.1} and we have \hypref{0.2}.
                  So, we have $ \lBrack  \langle \mathcal{C}  \ottsym{[}  \ottmv{l}  \mathord:  \ottsym{[}  \ottnt{m}  \ottsym{]}  \ottsym{]} \rangle_{ \ottmv{l'} }  \rBrack  \simeq  \lBrack \ottmv{l}  \mathord:  \ottsym{[}  \ottnt{m}  \ottsym{]} \rBrack $.
                  A topological ordering of $ \lBrack \ottmv{l}  \mathord:  \ottsym{[}  \ottnt{m}  \ottsym{]} \rBrack $ is a bijection from $\NAT_{<1}$ to $\NAT_{<1}$, which is unique.

                  \item[$ \ottmv{l'} \neq \ottmv{l} $]
                  In this case, $ \lBrack  \langle \mathcal{C}  \ottsym{[}  \ottmv{l}  \mathord:  \ottsym{[}  \ottnt{m}  \ottsym{]}  \ottsym{]} \rangle_{ \ottmv{l'} }  \rBrack  \simeq  \lBrack  \langle \mathcal{C}  \ottsym{[}   \cdot   \ottsym{]} \rangle_{ \ottmv{l'} }  \rBrack $.
                  Here, a topological ordering of $ \langle \mathcal{C}  \ottsym{[}   \cdot   \ottsym{]} \rangle_{ \ottmv{l'} } $ is unique since $\mathcal{C}  \ottsym{[}   \cdot   \ottsym{]}$ is order-defined by \hypref{0.1}.
              \end{match}

        \item $  \ottkw{dom} ( \mathcal{C}  \ottsym{[}  \ottmv{l}  \mathord:  \ottsym{[}  \ottnt{m}  \ottsym{]}  \ottsym{]} )  \ottsym{=}  \ottkw{dom} (  \mathcal{H} \cup \ottsym{\{} \ottmv{l}  \mapsto  \ottsym{(}  \ottsym{0}  \ottsym{,}  \ottnt{m}  \ottsym{,}  \varepsilon  \ottsym{)} \ottsym{\}}  )  $
              By definition, $  \ottkw{dom} ( \mathcal{C}  \ottsym{[}  \ottmv{l}  \mathord:  \ottsym{[}  \ottnt{m}  \ottsym{]}  \ottsym{]} )  \ottsym{=}   \ottkw{dom} ( \mathcal{C} )  \cup \ottsym{\{}  \ottmv{l}  \ottsym{\}}  $ and $  \ottkw{dom} (  \mathcal{H} \cup \ottsym{\{} \ottmv{l}  \mapsto  \ottsym{(}  \ottsym{0}  \ottsym{,}  \ottnt{m}  \ottsym{,}  \varepsilon  \ottsym{)} \ottsym{\}}  )  \ottsym{=}   \ottkw{dom} ( \mathcal{H} )  \cup \ottsym{\{}  \ottmv{l}  \ottsym{\}}  $.
              So, the subgoal holds.

        \item For any $\ottmv{l'}$, $\ottnt{n}$, $\ottnt{m_{{\mathrm{0}}}}$, and $\ottnt{m''}$ such that $ \ottsym{(}   \mathcal{H} \cup \ottsym{\{} \ottmv{l}  \mapsto  \ottsym{(}  \ottsym{0}  \ottsym{,}  \ottnt{m}  \ottsym{,}  \varepsilon  \ottsym{)} \ottsym{\}}   \ottsym{)}  \ottsym{(}  \ottmv{l'}  \ottsym{)} \ottsym{=} \ottsym{(}  \ottnt{n}  \ottsym{,}  \ottnt{m_{{\mathrm{0}}}}  \ottsym{,}  \ottnt{m''}  \ottsym{)} $, $  |  \lBrack  \langle \mathcal{C}  \ottsym{[}  \ottmv{l}  \mathord:  \ottsym{[}  \ottnt{m}  \ottsym{]}  \ottsym{]} \rangle_{ \ottmv{l'} }  \rBrack  |_{\bullet}  \ottsym{=} \ottnt{n}  \ottsym{+}  \ottsym{1} $ and $  \ottnt{m''} \odot  \overline{  \langle \mathcal{C}  \ottsym{[}  \ottmv{l}  \mathord:  \ottsym{[}  \ottnt{m}  \ottsym{]}  \ottsym{]} \rangle_{ \ottmv{l'} }  }   \le \ottnt{m_{{\mathrm{0}}}} $.
              We consider wether $ \ottmv{l'} \ottsym{=} \ottmv{l} $ or not.
              \begin{match}
                  \item[$ \ottmv{l'} \ottsym{=} \ottmv{l} $]
                  In this case, $ \ottnt{n} \ottsym{=} \ottsym{0} $, $ \ottnt{m_{{\mathrm{0}}}} \ottsym{=} \ottnt{m} $, and $ \ottnt{m''} \ottsym{=} \varepsilon $.
                  Furthermore, $ \ottmv{l'}  \notin   \ottkw{dom} ( \mathcal{C} )  $.
                  So, we have $ \langle \mathcal{C}  \ottsym{[}  \ottmv{l}  \mathord:  \ottsym{[}  \ottnt{m}  \ottsym{]}  \ottsym{]} \rangle_{ \ottmv{l'} }   \simeq  \ottmv{l}  \mathord:  \ottsym{[}  \ottnt{m}  \ottsym{]}$.
                  Consequently, what we need show are
                  \begin{gather*}
                       |  \lBrack \ottmv{l}  \mathord:  \ottsym{[}  \ottnt{m}  \ottsym{]} \rBrack  |_{\bullet}  = 1 \\
                        \varepsilon \odot \ottnt{m}  \le \ottnt{m} .
                  \end{gather*}
                  The first subgoal follows by the definition.
                  The second subgoal follows by the OPM properties.

                  \item[$ \ottmv{l'} \neq \ottmv{l} $]
                  In this case, $ \ottsym{(}   \mathcal{H} \cup \ottsym{\{} \ottmv{l}  \mapsto  \ottsym{(}  \ottsym{0}  \ottsym{,}  \ottnt{m}  \ottsym{,}  \varepsilon  \ottsym{)} \ottsym{\}}   \ottsym{)}  \ottsym{(}  \ottmv{l'}  \ottsym{)} \ottsym{=} \mathcal{H}  \ottsym{(}  \ottmv{l'}  \ottsym{)} $, and $  \langle \mathcal{C}  \ottsym{[}  \ottmv{l}  \mathord:  \ottsym{[}  \ottnt{m}  \ottsym{]}  \ottsym{]} \rangle_{ \ottmv{l'} }  \ottsym{=}  \langle \mathcal{C}  \ottsym{[}   \cdot   \ottsym{]} \rangle_{ \ottmv{l'} }  $.
                  So the subgoal follows by the assumption $\mathcal{C}  \ottsym{[}   \cdot   \ottsym{]}  \vdash  \mathcal{H}$.
              \end{match}
    \end{enumerate}
\end{prop}

\begin{prop}{heap/track/split}
    If
    \begin{gather}
        \mathcal{C}  \ottsym{[}  \ottmv{l}  \mathord:  \ottsym{[}  \ottnt{m'}  \ottsym{]}  \ottsym{]}  \vdash   \mathcal{H} \cup \ottsym{\{} \ottmv{l}  \mapsto  \ottsym{(}  \ottnt{n}  \ottsym{,}  \ottnt{m_{{\mathrm{0}}}}  \ottsym{,}  \ottnt{m}  \ottsym{)} \ottsym{\}}  \hyp{0.1}\\
          \ottnt{m_{{\mathrm{1}}}} \odot \ottnt{m_{{\mathrm{2}}}}  \le \ottnt{m'}  \hyp{0.2},
    \end{gather}
    then $\mathcal{C}  \ottsym{[}  \ottmv{l}  \mathord:  \ottsym{[}  \ottnt{m_{{\mathrm{1}}}}  \ottsym{]}  \ottsym{,}  \ottmv{l}  \mathord:  \ottsym{[}  \ottnt{m_{{\mathrm{2}}}}  \ottsym{]}  \ottsym{]}  \vdash   \mathcal{H} \cup \ottsym{\{} \ottmv{l}  \mapsto  \ottsym{(}  \ottnt{n}  \ottsym{+}  \ottsym{1}  \ottsym{,}  \ottnt{m_{{\mathrm{0}}}}  \ottsym{,}  \ottnt{m}  \ottsym{)} \ottsym{\}} $.

    \proof We show each condition of the heap typing for the goal.
    \begin{enumerate}
        \item $\mathcal{C}  \ottsym{[}  \ottmv{l}  \mathord:  \ottsym{[}  \ottnt{m_{{\mathrm{1}}}}  \ottsym{]}  \ottsym{,}  \ottmv{l}  \mathord:  \ottsym{[}  \ottnt{m_{{\mathrm{2}}}}  \ottsym{]}  \ottsym{]}$ is order-defined.
              We have
              \begin{gather}
                  \mathcal{C}  \ottsym{[}  \ottmv{l}  \mathord:  \ottsym{[}  \ottnt{m'}  \ottsym{]}  \ottsym{]} \text{ is order-defined} \hyp{1.1}
              \end{gather}
              by \hypref{0.1}.
              We will show $ \lBrack  \langle \mathcal{C}  \ottsym{[}  \ottmv{l}  \mathord:  \ottsym{[}  \ottnt{m_{{\mathrm{1}}}}  \ottsym{]}  \ottsym{,}  \ottmv{l}  \mathord:  \ottsym{[}  \ottnt{m_{{\mathrm{2}}}}  \ottsym{]}  \ottsym{]} \rangle_{ \ottmv{l'} }  \rBrack $ is traceable for any $\ottmv{l'}$.
              We consider wether $ \ottmv{l'} \ottsym{=} \ottmv{l} $ or not.
              \begin{match}
                  \item[$ \ottmv{l'} \ottsym{=} \ottmv{l} $]
                  In this case, $  \langle \mathcal{C}  \ottsym{[}  \ottmv{l}  \mathord:  \ottsym{[}  \ottnt{m_{{\mathrm{1}}}}  \ottsym{]}  \ottsym{,}  \ottmv{l}  \mathord:  \ottsym{[}  \ottnt{m_{{\mathrm{2}}}}  \ottsym{]}  \ottsym{]} \rangle_{ \ottmv{l'} }  \ottsym{=}  \langle \mathcal{C} \rangle_{ \ottmv{l'} }   \ottsym{[}  \ottmv{l}  \mathord:  \ottsym{[}  \ottnt{m_{{\mathrm{1}}}}  \ottsym{]}  \ottsym{,}  \ottmv{l}  \mathord:  \ottsym{[}  \ottnt{m_{{\mathrm{2}}}}  \ottsym{]}  \ottsym{]} $.
                  We can see $ \lBrack \ottmv{l}  \mathord:  \ottsym{[}  \ottnt{m_{{\mathrm{1}}}}  \ottsym{]}  \ottsym{,}  \ottmv{l}  \mathord:  \ottsym{[}  \ottnt{m_{{\mathrm{2}}}}  \ottsym{]} \rBrack $ is traceable.
                  So, the goal follows by \propref{env/runtime/intrp/traceable/ctx/replace} because $ \langle \mathcal{C} \rangle_{ \ottmv{l'} }   \ottsym{[}  \ottmv{l}  \mathord:  \ottsym{[}  \ottnt{m'}  \ottsym{]}  \ottsym{]}$ is traceable by \hypref{1.1}.

                  \item[$ \ottmv{l'} \neq \ottmv{l} $]
                  In this case, $  \langle \mathcal{C}  \ottsym{[}  \ottmv{l}  \mathord:  \ottsym{[}  \ottnt{m_{{\mathrm{1}}}}  \ottsym{]}  \ottsym{,}  \ottmv{l}  \mathord:  \ottsym{[}  \ottnt{m_{{\mathrm{2}}}}  \ottsym{]}  \ottsym{]} \rangle_{ \ottmv{l'} }  \ottsym{=}  \langle \mathcal{C} \rangle_{ \ottmv{l'} }   \ottsym{[}   \cdot   \ottsym{]} $.
                  Here, $ \langle \mathcal{C} \rangle_{ \ottmv{l'} }   \ottsym{[}   \cdot   \ottsym{]}$ is traceable because $  \langle \mathcal{C}  \ottsym{[}  \ottmv{l}  \mathord:  \ottsym{[}  \ottnt{m'}  \ottsym{]}  \ottsym{]} \rangle_{ \ottmv{l'} }  \ottsym{=}  \langle \mathcal{C} \rangle_{ \ottmv{l'} }   \ottsym{[}   \cdot   \ottsym{]} $ is traceable by \hypref{1.1}.
              \end{match}

        \item $  \ottkw{dom} ( \mathcal{C}  \ottsym{[}  \ottmv{l}  \mathord:  \ottsym{[}  \ottnt{m_{{\mathrm{1}}}}  \ottsym{]}  \ottsym{,}  \ottmv{l}  \mathord:  \ottsym{[}  \ottnt{m_{{\mathrm{2}}}}  \ottsym{]}  \ottsym{]} )  \ottsym{=}  \ottkw{dom} (  \mathcal{H} \cup \ottsym{\{} \ottmv{l}  \mapsto  \ottsym{(}  \ottnt{n}  \ottsym{+}  \ottsym{1}  \ottsym{,}  \ottnt{m_{{\mathrm{0}}}}  \ottsym{,}  \ottnt{m}  \ottsym{)} \ottsym{\}}  )  $.
              We have $  \ottkw{dom} ( \mathcal{C}  \ottsym{[}  \ottmv{l}  \mathord:  \ottsym{[}  \ottnt{m'}  \ottsym{]}  \ottsym{]} )  \ottsym{=}  \ottkw{dom} (  \mathcal{H} \cup \ottsym{\{} \ottmv{l}  \mapsto  \ottsym{(}  \ottnt{n}  \ottsym{,}  \ottnt{m_{{\mathrm{0}}}}  \ottsym{,}  \ottnt{m}  \ottsym{)} \ottsym{\}}  )  $ by \hypref{0.1}.
              So, $   \ottkw{dom} ( \mathcal{C} )  \cup \ottsym{\{}  \ottmv{l}  \ottsym{\}}  \ottsym{=}   \ottkw{dom} ( \mathcal{H} )  \cup \ottsym{\{}  \ottmv{l}  \ottsym{\}}  $ by definition.
              From that, the subgoal holds.

        \item For any $\ottmv{l'}$, $\ottnt{n'}$, $\ottnt{m'_{{\mathrm{0}}}}$, and $\ottnt{m''}$ such that $ \ottsym{(}   \mathcal{H} \cup \ottsym{\{} \ottmv{l}  \mapsto  \ottsym{(}  \ottnt{n}  \ottsym{+}  \ottsym{1}  \ottsym{,}  \ottnt{m_{{\mathrm{0}}}}  \ottsym{,}  \ottnt{m}  \ottsym{)} \ottsym{\}}   \ottsym{)}  \ottsym{(}  \ottmv{l'}  \ottsym{)} \ottsym{=} \ottsym{(}  \ottnt{n'}  \ottsym{,}  \ottnt{m'_{{\mathrm{0}}}}  \ottsym{,}  \ottnt{m''}  \ottsym{)} $, $  |  \lBrack  \langle \mathcal{C}  \ottsym{[}  \ottmv{l}  \mathord:  \ottsym{[}  \ottnt{m_{{\mathrm{1}}}}  \ottsym{]}  \ottsym{,}  \ottmv{l}  \mathord:  \ottsym{[}  \ottnt{m_{{\mathrm{2}}}}  \ottsym{]}  \ottsym{]} \rangle_{ \ottmv{l'} }  \rBrack  |_{\bullet}  \ottsym{=} \ottnt{n'}  \ottsym{+}  \ottsym{1} $ and $  \ottnt{m''} \odot  \overline{  \langle \mathcal{C}  \ottsym{[}  \ottmv{l}  \mathord:  \ottsym{[}  \ottnt{m_{{\mathrm{1}}}}  \ottsym{]}  \ottsym{,}  \ottmv{l}  \mathord:  \ottsym{[}  \ottnt{m_{{\mathrm{2}}}}  \ottsym{]}  \ottsym{]} \rangle_{ \ottmv{l'} }  }   \le \ottnt{m'_{{\mathrm{0}}}} $.
              We consider wether $ \ottmv{l} \neq \ottmv{l} $ or not.
              \begin{match}
                  \item[$ \ottmv{l'} \ottsym{=} \ottmv{l} $]
                  In this case, $ \ottnt{n'} \ottsym{=} \ottnt{n}  \ottsym{+}  \ottsym{1} $, $ \ottnt{m'_{{\mathrm{0}}}} \ottsym{=} \ottnt{m_{{\mathrm{0}}}} $, and $ \ottnt{m''} \ottsym{=} \ottnt{m} $.
                  So, what we need to show are
                  \begin{gather*}
                        |  \lBrack  \langle \mathcal{C}  \ottsym{[}  \ottmv{l}  \mathord:  \ottsym{[}  \ottnt{m_{{\mathrm{1}}}}  \ottsym{]}  \ottsym{,}  \ottmv{l}  \mathord:  \ottsym{[}  \ottnt{m_{{\mathrm{2}}}}  \ottsym{]}  \ottsym{]} \rangle_{ \ottmv{l} }  \rBrack  |_{\bullet}  \ottsym{=} \ottnt{n}  \ottsym{+}  \ottsym{1}  \ottsym{+}  \ottsym{1}  \\
                        \ottnt{m} \odot  \overline{  \langle \mathcal{C}  \ottsym{[}  \ottmv{l}  \mathord:  \ottsym{[}  \ottnt{m_{{\mathrm{1}}}}  \ottsym{]}  \ottsym{,}  \ottmv{l}  \mathord:  \ottsym{[}  \ottnt{m_{{\mathrm{2}}}}  \ottsym{]}  \ottsym{]} \rangle_{ \ottmv{l} }  }   \le \ottnt{m_{{\mathrm{0}}}} .
                  \end{gather*}
                  We have
                  \begin{gather}
                        |  \lBrack  \langle \mathcal{C}  \ottsym{[}  \ottmv{l}  \mathord:  \ottsym{[}  \ottnt{m'}  \ottsym{]}  \ottsym{]} \rangle_{ \ottmv{l} }  \rBrack  |_{\bullet}  \ottsym{=} \ottnt{n}  \ottsym{+}  \ottsym{1}  \hyp{3.1}\\
                        \ottnt{m} \odot  \overline{  \langle \mathcal{C}  \ottsym{[}  \ottmv{l}  \mathord:  \ottsym{[}  \ottnt{m'}  \ottsym{]}  \ottsym{]} \rangle_{ \ottmv{l} }  }   \le \ottnt{m_{{\mathrm{0}}}}  \hyp{3.3}
                  \end{gather}
                  by the assumption.
                  We have $ |  \lBrack  \langle \mathcal{C}  \ottsym{[}  \ottmv{l}  \mathord:  \ottsym{[}  \ottnt{m_{{\mathrm{1}}}}  \ottsym{]}  \ottsym{,}  \ottmv{l}  \mathord:  \ottsym{[}  \ottnt{m_{{\mathrm{2}}}}  \ottsym{]}  \ottsym{]} \rangle_{ \ottmv{l} }  \rBrack  |_{\bullet}  =  |  \lBrack  \langle \mathcal{C} \rangle_{ \ottmv{l} }   \ottsym{[}   \cdot   \ottsym{]} \rBrack  |_{\bullet}  + 2$ and $ |  \lBrack  \langle \mathcal{C}  \ottsym{[}  \ottmv{l}  \mathord:  \ottsym{[}  \ottnt{m'}  \ottsym{]}  \ottsym{]} \rangle_{ \ottmv{l} }  \rBrack  |_{\bullet}  =  |  \lBrack  \langle \mathcal{C} \rangle_{ \ottmv{l} }   \ottsym{[}   \cdot   \ottsym{]} \rBrack  |_{\bullet}  + 1$ by \propref{env/runtime/intrp/order/ctx}.
                  So, the first subgoal follows by \hypref{3.1}.

                  We have some $\ottnt{m_{{\mathrm{01}}}}$ and $\ottnt{m_{{\mathrm{03}}}}$ such that
                  \begin{gather}
                        \overline{  \langle \mathcal{C}  \ottsym{[}  \ottmv{l}  \mathord:  \ottsym{[}  \ottnt{m'}  \ottsym{]}  \ottsym{]} \rangle_{ \ottmv{l} }  }  \ottsym{=}   \ottnt{m_{{\mathrm{01}}}} \odot \ottnt{m'}  \odot \ottnt{m_{{\mathrm{03}}}}   \hyp{3.4}\\
                        \overline{  \langle \mathcal{C}  \ottsym{[}  \ottmv{l}  \mathord:  \ottsym{[}  \ottnt{m_{{\mathrm{1}}}}  \ottsym{]}  \ottsym{,}  \ottmv{l}  \mathord:  \ottsym{[}  \ottnt{m_{{\mathrm{2}}}}  \ottsym{]}  \ottsym{]} \rangle_{ \ottmv{l} }  }  \ottsym{=}    \ottnt{m_{{\mathrm{01}}}} \odot \ottnt{m_{{\mathrm{1}}}}  \odot \ottnt{m_{{\mathrm{2}}}}  \odot \ottnt{m_{{\mathrm{03}}}}   \hyp{3.5}
                  \end{gather}
                  by \propref{env/usage/replace/ctx}.
                  Now the second subgoal follows by \hypref{0.2}, \hypref{3.3}, \hypref{3.4}, \hypref{3.5}, and the properties of {OPM}.

                  \item[$ \ottmv{l'} \neq \ottmv{l} $]
                  In this case, $ \ottsym{(}   \mathcal{H} \cup \ottsym{\{} \ottmv{l}  \mapsto  \ottsym{(}  \ottnt{n}  \ottsym{+}  \ottsym{1}  \ottsym{,}  \ottnt{m_{{\mathrm{0}}}}  \ottsym{,}  \ottnt{m}  \ottsym{)} \ottsym{\}}   \ottsym{)}  \ottsym{(}  \ottmv{l'}  \ottsym{)} \ottsym{=} \ottsym{(}   \mathcal{H} \cup \ottsym{\{} \ottmv{l}  \mapsto  \ottsym{(}  \ottnt{n}  \ottsym{,}  \ottnt{m_{{\mathrm{0}}}}  \ottsym{,}  \ottnt{m}  \ottsym{)} \ottsym{\}}   \ottsym{)}  \ottsym{(}  \ottmv{l'}  \ottsym{)} $, and $ \langle \mathcal{C}  \ottsym{[}  \ottmv{l}  \mathord:  \ottsym{[}  \ottnt{m_{{\mathrm{1}}}}  \ottsym{]}  \ottsym{,}  \ottmv{l}  \mathord:  \ottsym{[}  \ottnt{m_{{\mathrm{2}}}}  \ottsym{]}  \ottsym{]} \rangle_{ \ottmv{l'} }   \simeq   \langle \mathcal{C}  \ottsym{[}  \ottmv{l}  \mathord:  \ottsym{[}  \ottnt{m'}  \ottsym{]}  \ottsym{]} \rangle_{ \ottmv{l'} } $.
                  So, the subgoal follows by \hypref{0.1}.
              \end{match}
    \end{enumerate}
\end{prop}

\begin{prop}{heap/track/op}
    If
    \begin{gather}
        \ottmv{l}  \mathord:  \ottsym{[}  \ottnt{m'}  \ottsym{]}  \ottsym{,}  \ottnt{C}  \vdash   \mathcal{H} \cup \ottsym{\{} \ottmv{l}  \mapsto  \ottsym{(}  \ottnt{n}  \ottsym{,}  \ottnt{m_{{\mathrm{0}}}}  \ottsym{,}  \ottnt{m}  \ottsym{)} \ottsym{\}}  \hyp{0.1}\\
          \ottnt{m_{{\mathrm{1}}}} \odot \ottnt{m_{{\mathrm{2}}}}  \le \ottnt{m'}  \hyp{0.2},
    \end{gather}
    then $\ottmv{l}  \mathord:  \ottsym{[}  \ottnt{m_{{\mathrm{2}}}}  \ottsym{]}  \ottsym{,}  \ottnt{C}  \vdash   \mathcal{H} \cup \ottsym{\{} \ottmv{l}  \mapsto  \ottsym{(}  \ottnt{n}  \ottsym{,}  \ottnt{m_{{\mathrm{0}}}}  \ottsym{,}   \ottnt{m} \odot \ottnt{m_{{\mathrm{1}}}}   \ottsym{)} \ottsym{\}} $.

    \proof We show each condition of the heap typing for the goal.
    \begin{enumerate}
        \item $\ottmv{l}  \mathord:  \ottsym{[}  \ottnt{m_{{\mathrm{2}}}}  \ottsym{]}  \ottsym{,}  \ottnt{C}$ is order-defined.
              We know $\ottmv{l}  \mathord:  \ottsym{[}  \ottnt{m'}  \ottsym{]}  \ottsym{,}  \ottnt{C}$ is order-defined by \hypref{0.1}.
              So, this case follows by \propref{env/ordering/pop} and \propref{env/ordering/push}.

        \item $  \ottkw{dom} ( \ottmv{l}  \mathord:  \ottsym{[}  \ottnt{m_{{\mathrm{2}}}}  \ottsym{]}  \ottsym{,}  \ottnt{C} )  \ottsym{=}  \ottkw{dom} (  \mathcal{H} \cup \ottsym{\{} \ottmv{l}  \mapsto  \ottsym{(}  \ottnt{n}  \ottsym{,}  \ottnt{m_{{\mathrm{0}}}}  \ottsym{,}   \ottnt{m} \odot \ottnt{m_{{\mathrm{1}}}}   \ottsym{)} \ottsym{\}}  )  $.
              By definition, $  \ottkw{dom} ( \ottmv{l}  \mathord:  \ottsym{[}  \ottnt{m_{{\mathrm{2}}}}  \ottsym{]}  \ottsym{,}  \ottnt{C} )  \ottsym{=}  \ottsym{\{}  \ottmv{l}  \ottsym{\}} \cup  \ottkw{dom} ( \ottnt{C} )   $ and $  \ottkw{dom} (  \mathcal{H} \cup \ottsym{\{} \ottmv{l}  \mapsto  \ottsym{(}  \ottnt{n}  \ottsym{,}  \ottnt{m_{{\mathrm{0}}}}  \ottsym{,}   \ottnt{m} \odot \ottnt{m_{{\mathrm{1}}}}   \ottsym{)} \ottsym{\}}  )  \ottsym{=}   \ottkw{dom} ( \mathcal{H} )  \cup \ottsym{\{}  \ottmv{l}  \ottsym{\}}  $.
              We have $  \ottsym{\{}  \ottmv{l}  \ottsym{\}} \cup  \ottkw{dom} ( \ottnt{C} )   \ottsym{=}   \ottkw{dom} ( \mathcal{H} )  \cup \ottsym{\{}  \ottmv{l}  \ottsym{\}}  $ by \hypref{0.1}.
              So, the subgoal holds.

        \item For any $\ottmv{l'}$, $\ottnt{n'}$, $\ottnt{m'_{{\mathrm{0}}}}$, and $\ottnt{m''}$ such that $ \ottsym{(}   \mathcal{H} \cup \ottsym{\{} \ottmv{l}  \mapsto  \ottsym{(}  \ottnt{n}  \ottsym{,}  \ottnt{m_{{\mathrm{0}}}}  \ottsym{,}   \ottnt{m} \odot \ottnt{m_{{\mathrm{1}}}}   \ottsym{)} \ottsym{\}}   \ottsym{)}  \ottsym{(}  \ottmv{l'}  \ottsym{)} \ottsym{=} \ottsym{(}  \ottnt{n'}  \ottsym{,}  \ottnt{m'_{{\mathrm{0}}}}  \ottsym{,}  \ottnt{m''}  \ottsym{)} $, $  |  \lBrack  \langle \ottsym{(}  \ottmv{l}  \mathord:  \ottsym{[}  \ottnt{m_{{\mathrm{2}}}}  \ottsym{]}  \ottsym{,}  \ottnt{C}  \ottsym{)} \rangle_{ \ottmv{l'} }  \rBrack  |_{\bullet}  \ottsym{=} \ottnt{n'}  \ottsym{+}  \ottsym{1} $ and $  \ottnt{m''} \odot  \overline{  \langle \ottsym{(}  \ottmv{l}  \mathord:  \ottsym{[}  \ottnt{m_{{\mathrm{2}}}}  \ottsym{]}  \ottsym{,}  \ottnt{C}  \ottsym{)} \rangle_{ \ottmv{l'} }  }   \le \ottnt{m'_{{\mathrm{0}}}} $.
              We consider wether $ \ottmv{l'} \ottsym{=} \ottmv{l} $ or not.
              \begin{match}
                  \item[$ \ottmv{l'} \ottsym{=} \ottmv{l} $]
                  In this case, $ \ottnt{n'} \ottsym{=} \ottnt{n} $, $ \ottnt{m'_{{\mathrm{0}}}} \ottsym{=} \ottnt{m_{{\mathrm{0}}}} $, $ \ottnt{m''} \ottsym{=}  \ottnt{m} \odot \ottnt{m_{{\mathrm{1}}}}  $, $ |  \lBrack  \langle \ottsym{(}  \ottmv{l}  \mathord:  \ottsym{[}  \ottnt{m_{{\mathrm{2}}}}  \ottsym{]}  \ottsym{,}  \ottnt{C}  \ottsym{)} \rangle_{ \ottmv{l'} }  \rBrack  |_{\bullet}  = 1 +  |  \lBrack  \langle \ottnt{C} \rangle_{ \ottmv{l} }  \rBrack  |_{\bullet} $, and $  \overline{  \langle \ottsym{(}  \ottmv{l}  \mathord:  \ottsym{[}  \ottnt{m_{{\mathrm{2}}}}  \ottsym{]}  \ottsym{,}  \ottnt{C}  \ottsym{)} \rangle_{ \ottmv{l'} }  }  \ottsym{=}  \ottnt{m_{{\mathrm{2}}}} \odot  \overline{  \langle \ottnt{C} \rangle_{ \ottmv{l} }  }   $.
                  So, what we need to show are
                  \begin{gather*}
                        |  \lBrack  \langle \ottnt{C} \rangle_{ \ottmv{l} }  \rBrack  |_{\bullet}  \ottsym{=} \ottnt{n}  \\
                          \ottnt{m} \odot \ottnt{m_{{\mathrm{1}}}}  \odot \ottnt{m_{{\mathrm{2}}}}  \odot  \overline{  \langle \ottnt{C} \rangle_{ \ottmv{l} }  }   \le \ottnt{m_{{\mathrm{0}}}} .
                  \end{gather*}
                  We have $  |  \lBrack  \langle \ottsym{(}  \ottmv{l}  \mathord:  \ottsym{[}  \ottnt{m'}  \ottsym{]}  \ottsym{,}  \ottnt{C}  \ottsym{)} \rangle_{ \ottmv{l} }  \rBrack  |_{\bullet}  \ottsym{=} \ottnt{n}  \ottsym{+}  \ottsym{1} $ by \hypref{0.1}.
                  So, the first subgoal follows by definition.
                  We have $  \ottnt{m} \odot  \overline{  \langle \ottsym{(}  \ottmv{l}  \mathord:  \ottsym{[}  \ottnt{m'}  \ottsym{]}  \ottsym{,}  \ottnt{C}  \ottsym{)} \rangle_{ \ottmv{l} }  }   \le \ottnt{m_{{\mathrm{0}}}} $ by \hypref{0.1} too.
                  So,
                  \begin{gather}
                         \ottnt{m} \odot \ottnt{m'}  \odot  \overline{  \langle \ottnt{C} \rangle_{ \ottmv{l} }  }   \le \ottnt{m_{{\mathrm{0}}}}  \hyp{3.2}
                  \end{gather}
                  So, the second subgoal follows by \hypref{0.2}, \hypref{3.2}, and the properties of {OPM}.

                  \item[$ \ottmv{l'} \neq \ottmv{l} $]
                  In this case, $ \ottsym{(}   \mathcal{H} \cup \ottsym{\{} \ottmv{l}  \mapsto  \ottsym{(}  \ottnt{n}  \ottsym{,}  \ottnt{m_{{\mathrm{0}}}}  \ottsym{,}   \ottnt{m} \odot \ottnt{m_{{\mathrm{1}}}}   \ottsym{)} \ottsym{\}}   \ottsym{)}  \ottsym{(}  \ottmv{l'}  \ottsym{)} \ottsym{=} \ottsym{(}   \mathcal{H} \cup \ottsym{\{} \ottmv{l}  \mapsto  \ottsym{(}  \ottnt{n}  \ottsym{,}  \ottnt{m_{{\mathrm{0}}}}  \ottsym{,}  \ottnt{m}  \ottsym{)} \ottsym{\}}   \ottsym{)}  \ottsym{(}  \ottmv{l'}  \ottsym{)} $, and $  \langle \ottsym{(}  \ottmv{l}  \mathord:  \ottsym{[}  \ottnt{m_{{\mathrm{2}}}}  \ottsym{]}  \ottsym{,}  \ottnt{C}  \ottsym{)} \rangle_{ \ottmv{l'} }  \ottsym{=}  \langle \ottsym{(}  \ottmv{l}  \mathord:  \ottsym{[}  \ottnt{m'}  \ottsym{]}  \ottsym{,}  \ottnt{C}  \ottsym{)} \rangle_{ \ottmv{l'} }  $.
                  So, the subgoal follows by \hypref{0.1}.
              \end{match}
    \end{enumerate}
\end{prop}

\begin{prop}{heap/track/cl1}
    If
    \begin{gather}
        \mathcal{C}  \ottsym{[}  \ottmv{l}  \mathord:  \ottsym{[}  \ottnt{m'}  \ottsym{]}  \ottsym{]}  \vdash   \mathcal{H} \cup \ottsym{\{} \ottmv{l}  \mapsto  \ottsym{(}  \ottnt{n}  \ottsym{+}  \ottsym{1}  \ottsym{,}  \ottnt{m_{{\mathrm{0}}}}  \ottsym{,}  \ottnt{m}  \ottsym{)} \ottsym{\}}  \hyp{0.1}\\
         \varepsilon \le \ottnt{m'}  \hyp{0.2},
    \end{gather}
    then $\mathcal{C}  \ottsym{[}   \cdot   \ottsym{]}  \vdash   \mathcal{H} \cup \ottsym{\{} \ottmv{l}  \mapsto  \ottsym{(}  \ottnt{n}  \ottsym{,}  \ottnt{m_{{\mathrm{0}}}}  \ottsym{,}  \ottnt{m}  \ottsym{)} \ottsym{\}} $.

    \proof We show each condition of the heap typing for the goal.
    \begin{enumerate}
        \item $\mathcal{C}  \ottsym{[}   \cdot   \ottsym{]}$ is order-defined.
              We know
              \begin{gather}
                  \text{$\mathcal{C}  \ottsym{[}  \ottmv{l}  \mathord:  \ottsym{[}  \ottnt{m'}  \ottsym{]}  \ottsym{]}$ is order-defined} \hyp{1.1}
              \end{gather}
              by \hypref{0.1}
              We will show $ \lBrack  \langle \mathcal{C}  \ottsym{[}   \cdot   \ottsym{]} \rangle_{ \ottmv{l'} }  \rBrack $, which equals to $ \lBrack  \langle \mathcal{C} \rangle_{ \ottmv{l'} }   \ottsym{[}   \cdot   \ottsym{]} \rBrack $ by definition, is traceable for any $\ottmv{l'}$.
              We consider wether $ \ottmv{l'} \ottsym{=} \ottmv{l} $ or not.
              \begin{match}
                  \item[$ \ottmv{l'} \ottsym{=} \ottmv{l} $]
                  For this $\ottmv{l'}$, we know $ \lBrack  \langle \mathcal{C} \rangle_{ \ottmv{l'} }   \ottsym{[}  \ottmv{l}  \mathord:  \ottsym{[}  \ottnt{m'}  \ottsym{]}  \ottsym{]} \rBrack $ is traceable by \hypref{1.1}.
                  So, the goal follows by \propref{env/runtime/intrp/traceable/ctx/replace}.

                  \item[$ \ottmv{l'} \neq \ottmv{l} $]
                  For this $\ottmv{l'}$, we know $ \lBrack  \langle \mathcal{C} \rangle_{ \ottmv{l'} }   \ottsym{[}   \cdot   \ottsym{]} \rBrack $ is traceable by \hypref{1.1}.
                  That is equivalent to the goal.
              \end{match}

        \item $  \ottkw{dom} ( \mathcal{C}  \ottsym{[}   \cdot   \ottsym{]} )  \ottsym{=}  \ottkw{dom} (  \mathcal{H} \cup \ottsym{\{} \ottmv{l}  \mapsto  \ottsym{(}  \ottnt{n}  \ottsym{,}  \ottnt{m_{{\mathrm{0}}}}  \ottsym{,}  \ottnt{m}  \ottsym{)} \ottsym{\}}  )  $.
              We have $   \ottkw{dom} ( \mathcal{C} )  \cup \ottsym{\{}  \ottmv{l}  \ottsym{\}}  \ottsym{=}   \ottkw{dom} ( \mathcal{H} )  \cup \ottsym{\{}  \ottmv{l}  \ottsym{\}}  $ by \hypref{0.1}.
              So, it suffices to show $ \ottmv{l}  \in   \ottkw{dom} ( \mathcal{C} )  $.
              We have $  |  \lBrack  \langle \mathcal{C}  \ottsym{[}  \ottmv{l}  \mathord:  \ottsym{[}  \ottnt{m'}  \ottsym{]}  \ottsym{]} \rangle_{ \ottmv{l} }  \rBrack  |_{\bullet}  \ottsym{=} \ottnt{n}  \ottsym{+}  \ottsym{1}  \ottsym{+}  \ottsym{1} $ by \hypref{0.1}.
              So, $  |  \lBrack  \langle \mathcal{C}  \ottsym{[}   \cdot   \ottsym{]} \rangle_{ \ottmv{l} }  \rBrack  |_{\bullet}  \ottsym{=} \ottnt{n}  \ottsym{+}  \ottsym{1} $ by \propref{env/runtime/intrp/order/ctx}.
              That implies $ \ottmv{l}  \in   \ottkw{dom} ( \mathcal{C} )  $.

        \item For any $\ottmv{l'}$, $\ottnt{n'}$, $\ottnt{m'_{{\mathrm{0}}}}$, and $\ottnt{m''}$ such that $ \ottsym{(}   \mathcal{H} \cup \ottsym{\{} \ottmv{l}  \mapsto  \ottsym{(}  \ottnt{n}  \ottsym{,}  \ottnt{m_{{\mathrm{0}}}}  \ottsym{,}  \ottnt{m}  \ottsym{)} \ottsym{\}}   \ottsym{)}  \ottsym{(}  \ottmv{l'}  \ottsym{)} \ottsym{=} \ottsym{(}  \ottnt{n'}  \ottsym{,}  \ottnt{m'_{{\mathrm{0}}}}  \ottsym{,}  \ottnt{m''}  \ottsym{)} $, $  |  \lBrack  \langle \mathcal{C}  \ottsym{[}   \cdot   \ottsym{]} \rangle_{ \ottmv{l'} }  \rBrack  |_{\bullet}  \ottsym{=} \ottnt{n'}  \ottsym{+}  \ottsym{1} $ and $  \ottnt{m''} \odot  \overline{  \langle \mathcal{C}  \ottsym{[}   \cdot   \ottsym{]} \rangle_{ \ottmv{l'} }  }   \le \ottnt{m'_{{\mathrm{0}}}} $.
              We consider whether $ \ottmv{l'} \ottsym{=} \ottmv{l} $ or not.
              \begin{match}
                  \item[$ \ottmv{l'} \ottsym{=} \ottmv{l} $]
                  In this case, $ \ottnt{n'} \ottsym{=} \ottnt{n} $, $ \ottnt{m'_{{\mathrm{0}}}} \ottsym{=} \ottnt{m_{{\mathrm{0}}}} $, and $ \ottnt{m''} \ottsym{=} \ottnt{m} $.
                  So, what we need show are
                  \begin{gather*}
                        |  \lBrack  \langle \mathcal{C}  \ottsym{[}   \cdot   \ottsym{]} \rangle_{ \ottmv{l} }  \rBrack  |_{\bullet}  \ottsym{=} \ottnt{n}  \ottsym{+}  \ottsym{1}  \\
                        \ottnt{m} \odot  \overline{  \langle \mathcal{C}  \ottsym{[}   \cdot   \ottsym{]} \rangle_{ \ottmv{l} }  }   \le \ottnt{m_{{\mathrm{0}}}} .
                  \end{gather*}
                  We have
                  \begin{gather}
                        |  \lBrack  \langle \ottsym{(}  \mathcal{C}  \ottsym{[}  \ottmv{l}  \mathord:  \ottsym{[}  \ottnt{m'}  \ottsym{]}  \ottsym{]}  \ottsym{)} \rangle_{ \ottmv{l} }  \rBrack  |_{\bullet}  \ottsym{=} \ottnt{n}  \ottsym{+}  \ottsym{1}  \ottsym{+}  \ottsym{1}  \hyp{3.1}\\
                        \ottnt{m} \odot  \overline{  \langle \mathcal{C} \rangle_{ \ottmv{l} }   \ottsym{[}  \ottmv{l}  \mathord:  \ottsym{[}  \ottnt{m'}  \ottsym{]}  \ottsym{]} }   \le \ottnt{m_{{\mathrm{0}}}}  \hyp{3.3}
                  \end{gather}
                  by \hypref{0.1}.
                  So, the first subgoal follows by \propref{env/runtime/intrp/order/ctx} and \hypref{3.1}.
                  We have some $\ottnt{m_{{\mathrm{01}}}}$ and $\ottnt{m_{{\mathrm{02}}}}$ such that
                  \begin{gather}
                        \overline{  \langle \mathcal{C} \rangle_{ \ottmv{l} }   \ottsym{[}  \ottmv{l}  \mathord:  \ottsym{[}  \ottnt{m'}  \ottsym{]}  \ottsym{]} }  \ottsym{=}   \ottnt{m_{{\mathrm{01}}}} \odot \ottnt{m'}  \odot \ottnt{m_{{\mathrm{02}}}}   \hyp{3.4}\\
                        \overline{  \langle \mathcal{C} \rangle_{ \ottmv{l} }   \ottsym{[}   \cdot   \ottsym{]} }  \ottsym{=}  \ottnt{m_{{\mathrm{01}}}} \odot \ottnt{m_{{\mathrm{02}}}}   \hyp{3.5}
                  \end{gather}
                  by \propref{env/usage/replace/ctx}.
                  So, the second subgoal follows by \hypref{0.2}, \hypref{3.3}, \hypref{3.4}, and \hypref{3.5}.

                  \item[$ \ottmv{l'} \neq \ottmv{l} $]
                  In this case, $ \ottsym{(}   \mathcal{H} \cup \ottsym{\{} \ottmv{l}  \mapsto  \ottsym{(}  \ottnt{n}  \ottsym{,}  \ottnt{m_{{\mathrm{0}}}}  \ottsym{,}  \ottnt{m}  \ottsym{)} \ottsym{\}}   \ottsym{)}  \ottsym{(}  \ottmv{l'}  \ottsym{)} \ottsym{=} \ottsym{(}   \mathcal{H} \cup \ottsym{\{} \ottmv{l}  \mapsto  \ottsym{(}  \ottnt{n}  \ottsym{+}  \ottsym{1}  \ottsym{,}  \ottnt{m_{{\mathrm{0}}}}  \ottsym{,}  \ottnt{m}  \ottsym{)} \ottsym{\}}   \ottsym{)}  \ottsym{(}  \ottmv{l'}  \ottsym{)} $, and $  \langle \mathcal{C}  \ottsym{[}   \cdot   \ottsym{]} \rangle_{ \ottmv{l'} }  \ottsym{=}  \langle \mathcal{C}  \ottsym{[}  \ottmv{l}  \mathord:  \ottsym{[}  \ottnt{m'}  \ottsym{]}  \ottsym{]} \rangle_{ \ottmv{l'} }  $.
                  So, the subgoal follows by \hypref{0.1} \qedhere
              \end{match}
    \end{enumerate}
\end{prop}

\begin{prop}{heap/track/cl2}
    If
    \begin{gather}
        \mathcal{C}  \ottsym{[}  \ottmv{l}  \mathord:  \ottsym{[}  \ottnt{m'}  \ottsym{]}  \ottsym{]}  \vdash   \mathcal{H} \cup \ottsym{\{} \ottmv{l}  \mapsto  \ottsym{(}  \ottsym{0}  \ottsym{,}  \ottnt{m_{{\mathrm{0}}}}  \ottsym{,}  \ottnt{m}  \ottsym{)} \ottsym{\}}  \hyp{0.1}\\
         \ottmv{l}  \notin   \ottkw{dom} ( \mathcal{H} )   \hyp{0.2},
    \end{gather}
    then $\mathcal{C}  \ottsym{[}   \cdot   \ottsym{]}  \vdash  \mathcal{H}$.

    \proof We show each condition of the heap typing for the goal.
    \begin{enumerate}
        \item $\mathcal{C}  \ottsym{[}   \cdot   \ottsym{]}$ is order-defined.
              We know
              \begin{gather}
                  \text{$\mathcal{C}  \ottsym{[}  \ottmv{l}  \mathord:  \ottsym{[}  \ottnt{m'}  \ottsym{]}  \ottsym{]}$ is order-defined} \hyp{1.1}
              \end{gather}
              by \hypref{0.1}
              We will show $ \lBrack  \langle \mathcal{C}  \ottsym{[}   \cdot   \ottsym{]} \rangle_{ \ottmv{l'} }  \rBrack $, which equals to $ \lBrack  \langle \mathcal{C} \rangle_{ \ottmv{l'} }   \ottsym{[}   \cdot   \ottsym{]} \rBrack $ by definition, is traceable for any $\ottmv{l'}$.
              We consider wether $ \ottmv{l'} \ottsym{=} \ottmv{l} $ or not.
              \begin{match}
                  \item[$ \ottmv{l'} \ottsym{=} \ottmv{l} $]
                  For this $\ottmv{l'}$, we know $ \lBrack  \langle \mathcal{C} \rangle_{ \ottmv{l'} }   \ottsym{[}  \ottmv{l}  \mathord:  \ottsym{[}  \ottnt{m'}  \ottsym{]}  \ottsym{]} \rBrack $ is traceable by \hypref{1.1}.
                  So, the goal follows by \propref{env/runtime/intrp/traceable/ctx/replace}.

                  \item[$ \ottmv{l'} \neq \ottmv{l} $]
                  For this $\ottmv{l'}$, we know $ \lBrack  \langle \mathcal{C} \rangle_{ \ottmv{l'} }   \ottsym{[}   \cdot   \ottsym{]} \rBrack $ is traceable by \hypref{1.1}.
                  That is equivalent to the goal.
              \end{match}

        \item $  \ottkw{dom} ( \mathcal{C}  \ottsym{[}   \cdot   \ottsym{]} )  \ottsym{=}  \ottkw{dom} ( \mathcal{H} )  $.
              We have $   \ottkw{dom} ( \mathcal{C} )  \cup \ottsym{\{}  \ottmv{l}  \ottsym{\}}  \ottsym{=}   \ottkw{dom} ( \mathcal{H} )  \cup \ottsym{\{}  \ottmv{l}  \ottsym{\}}  $ by \hypref{0.1}, and $ \ottmv{l}  \notin   \ottkw{dom} ( \mathcal{H} )  $ by \hypref{0.2}.
              So, it suffices to show $ \ottmv{l}  \notin   \ottkw{dom} ( \mathcal{C} )  $.
              We have $  |  \lBrack  \langle \mathcal{C}  \ottsym{[}  \ottmv{l}  \mathord:  \ottsym{[}  \ottnt{m'}  \ottsym{]}  \ottsym{]} \rangle_{ \ottmv{l} }  \rBrack  |_{\bullet}  \ottsym{=} \ottsym{1} $ by \hypref{0.1}.
              So, $  |  \lBrack  \langle \mathcal{C}  \ottsym{[}   \cdot   \ottsym{]} \rangle_{ \ottmv{l} }  \rBrack  |_{\bullet}  \ottsym{=} \ottsym{0} $ by \propref{env/runtime/intrp/order/ctx}.
              That implies $ \ottmv{l}  \notin   \ottkw{dom} ( \mathcal{C} )  $.

        \item For any $\ottmv{l'}$, $\ottnt{n'}$, $\ottnt{m'_{{\mathrm{0}}}}$, and $\ottnt{m''}$ such that $ \mathcal{H}  \ottsym{(}  \ottmv{l'}  \ottsym{)} \ottsym{=} \ottsym{(}  \ottnt{n'}  \ottsym{,}  \ottnt{m'_{{\mathrm{0}}}}  \ottsym{,}  \ottnt{m''}  \ottsym{)} $, $  |  \lBrack  \langle \mathcal{C}  \ottsym{[}   \cdot   \ottsym{]} \rangle_{ \ottmv{l'} }  \rBrack  |_{\bullet}  \ottsym{=} \ottnt{n'}  \ottsym{+}  \ottsym{1} $ and $  \ottnt{m''} \odot  \overline{  \langle \mathcal{C}  \ottsym{[}   \cdot   \ottsym{]} \rangle_{ \ottmv{l'} }  }   \le \ottnt{m'_{{\mathrm{0}}}} $.
              We consider whether $ \ottmv{l'} \ottsym{=} \ottmv{l} $ or not.
              \begin{match}
                  \item[$ \ottmv{l'} \ottsym{=} \ottmv{l} $]
                  This case cannot happen since $ \ottmv{l}  \notin   \ottkw{dom} ( \mathcal{H} )  $.

                  \item[$ \ottmv{l'} \neq \ottmv{l} $]
                  In this case, $ \mathcal{H}  \ottsym{(}  \ottmv{l'}  \ottsym{)} \ottsym{=} \ottsym{(}   \mathcal{H} \cup \ottsym{\{} \ottmv{l}  \mapsto  \ottsym{(}  \ottsym{0}  \ottsym{,}  \ottnt{m_{{\mathrm{0}}}}  \ottsym{,}  \ottnt{m}  \ottsym{)} \ottsym{\}}   \ottsym{)}  \ottsym{(}  \ottmv{l'}  \ottsym{)} $, and $  \langle \mathcal{C}  \ottsym{[}   \cdot   \ottsym{]} \rangle_{ \ottmv{l'} }  \ottsym{=}  \langle \mathcal{C}  \ottsym{[}  \ottmv{l}  \mathord:  \ottsym{[}  \ottnt{m'}  \ottsym{]}  \ottsym{]} \rangle_{ \ottmv{l'} }  $.
                  So, the subgoal follows by \hypref{0.1} \qedhere
              \end{match}
    \end{enumerate}
\end{prop}

\begin{prop}{heap/par}
    If $\mathcal{C}  \ottsym{[}  \ottnt{C_{{\mathrm{1}}}}  \parallel  \ottnt{C_{{\mathrm{2}}}}  \ottsym{]}  \vdash  \mathcal{H}$, then $   \ottkw{dom} ( \ottnt{C_{{\mathrm{1}}}} )  \cap  \ottkw{dom} ( \ottnt{C_{{\mathrm{2}}}} )   \ottsym{=} \emptyset $.

    \proof
    We show the contraposition.
    So, suppose $   \ottkw{dom} ( \ottnt{C_{{\mathrm{1}}}} )  \cap  \ottkw{dom} ( \ottnt{C_{{\mathrm{2}}}} )   \neq \emptyset $.
    Then, we have some $\ottmv{l}$, $\ottnt{m}$, and $\ottnt{m'}$ such that $ \ottmv{l}  \mathord:  \ottsym{[}  \ottnt{m}  \ottsym{]}  \in  \ottnt{C_{{\mathrm{1}}}} $ and $ \ottmv{l}  \mathord:  \ottsym{[}  \ottnt{m'}  \ottsym{]}  \in  \ottnt{C_{{\mathrm{2}}}} $.
    By definition, we can see $0 <  |  \lBrack  \langle \ottnt{C_{{\mathrm{1}}}} \rangle_{ \ottmv{l} }  \rBrack  |_{\bullet} $ and $0 <  |  \lBrack  \langle \ottnt{C_{{\mathrm{2}}}} \rangle_{ \ottmv{l} }  \rBrack  |_{\bullet} $.
    So, $ \lBrack  \langle \ottnt{C_{{\mathrm{1}}}} \rangle_{ \ottmv{l} }  \rBrack   \cup   \lBrack  \langle \ottnt{C_{{\mathrm{2}}}} \rangle_{ \ottmv{l} }  \rBrack $ has more than one topological ordering by \propref{graph/ordering/union/multi}, which results in $ \lBrack  \langle \ottsym{(}  \ottnt{C_{{\mathrm{1}}}}  \parallel  \ottnt{C_{{\mathrm{2}}}}  \ottsym{)} \rangle_{ \ottmv{l} }  \rBrack $ has more than one topological ordering.
    Now $ \lBrack  \langle \mathcal{C}  \ottsym{[}  \ottnt{C_{{\mathrm{1}}}}  \parallel  \ottnt{C_{{\mathrm{2}}}}  \ottsym{]} \rangle_{ \ottmv{l} }  \rBrack $ has more than one topological ordering by the contraposition of \propref{env/runtime/intrp/traceable/ctx}, which implies $\mathcal{C}  \ottsym{[}  \ottnt{C_{{\mathrm{1}}}}  \parallel  \ottnt{C_{{\mathrm{2}}}}  \ottsym{]}$ is not order-defined.
\end{prop}

\begin{prop}{heap/par2}
    If $\mathcal{C}  \ottsym{[}  \ottnt{C_{{\mathrm{1}}}}  \ottsym{,}  \ottnt{C_{{\mathrm{2}}}}  \ottsym{]}  \vdash  \mathcal{H}$ and $   \ottkw{dom} ( \ottnt{C_{{\mathrm{1}}}} )  \cap  \ottkw{dom} ( \ottnt{C_{{\mathrm{2}}}} )   \ottsym{=} \emptyset $, then $\mathcal{C}  \ottsym{[}  \ottnt{C_{{\mathrm{1}}}}  \parallel  \ottnt{C_{{\mathrm{2}}}}  \ottsym{]}  \vdash  \mathcal{H}$.

    \proof
    We will show the following facts, by which $\mathcal{C}  \ottsym{[}  \ottnt{C_{{\mathrm{1}}}}  \parallel  \ottnt{C_{{\mathrm{2}}}}  \ottsym{]}  \vdash  \mathcal{H}$ follows since $\mathcal{C}  \ottsym{[}  \ottnt{C_{{\mathrm{1}}}}  \ottsym{,}  \ottnt{C_{{\mathrm{2}}}}  \ottsym{]}  \vdash  \mathcal{H}$.
    \begin{itemize}
        \item $\mathcal{C}  \ottsym{[}  \ottnt{C_{{\mathrm{1}}}}  \parallel  \ottnt{C_{{\mathrm{2}}}}  \ottsym{]}$ is order-defined.
              Since $   \ottkw{dom} ( \ottnt{C_{{\mathrm{1}}}} )  \cap  \ottkw{dom} ( \ottnt{C_{{\mathrm{2}}}} )   \ottsym{=} \emptyset $, we can observe $ \langle \mathcal{C}  \ottsym{[}  \ottnt{C_{{\mathrm{1}}}}  \ottsym{,}  \ottnt{C_{{\mathrm{2}}}}  \ottsym{]} \rangle_{ \ottmv{l} }   \simeq   \langle \mathcal{C}  \ottsym{[}  \ottnt{C_{{\mathrm{1}}}}  \parallel  \ottnt{C_{{\mathrm{2}}}}  \ottsym{]} \rangle_{ \ottmv{l} } $ for any $\ottmv{l}$ by \propref{env/runtime/focus/null}, \propref{env/equiv}, and \propref{env/equiv/ctx}.
              So, a topological ordering of $ \lBrack  \langle \mathcal{C}  \ottsym{[}  \ottnt{C_{{\mathrm{1}}}}  \parallel  \ottnt{C_{{\mathrm{2}}}}  \ottsym{]} \rangle_{ \ottmv{l} }  \rBrack $ is unique by \propref{graph/ordering/multi/iso}.
              That implies the subgoal.

        \item $  \ottkw{dom} ( \mathcal{C}  \ottsym{[}  \ottnt{C_{{\mathrm{1}}}}  \parallel  \ottnt{C_{{\mathrm{2}}}}  \ottsym{]} )  \ottsym{=}  \ottkw{dom} ( \mathcal{C}  \ottsym{[}  \ottnt{C_{{\mathrm{1}}}}  \ottsym{,}  \ottnt{C_{{\mathrm{2}}}}  \ottsym{]} )  $.
              We have $\mathcal{C}  \ottsym{[}  \ottnt{C_{{\mathrm{1}}}}  \parallel  \ottnt{C_{{\mathrm{2}}}}  \ottsym{]}  \lesssim  \mathcal{C}  \ottsym{[}  \ottnt{C_{{\mathrm{1}}}}  \ottsym{,}  \ottnt{C_{{\mathrm{2}}}}  \ottsym{]}$ by \propref{env/sub(span)} and \propref{env/sub/ctx}.
              So, the subgoal follows by \propref{env/sub/runtime/dom}.

        \item $  \overline{  \langle \mathcal{C}  \ottsym{[}  \ottnt{C_{{\mathrm{1}}}}  \parallel  \ottnt{C_{{\mathrm{2}}}}  \ottsym{]} \rangle_{ \ottmv{l} }  }  \ottsym{=}  \overline{  \langle \mathcal{C}  \ottsym{[}  \ottnt{C_{{\mathrm{1}}}}  \ottsym{,}  \ottnt{C_{{\mathrm{2}}}}  \ottsym{]} \rangle_{ \ottmv{l} }  }  $ for any $\ottmv{l}$.
              Since $   \ottkw{dom} ( \ottnt{C_{{\mathrm{1}}}} )  \cap  \ottkw{dom} ( \ottnt{C_{{\mathrm{2}}}} )   \ottsym{=} \emptyset $, we can observe $ \langle \mathcal{C}  \ottsym{[}  \ottnt{C_{{\mathrm{1}}}}  \ottsym{,}  \ottnt{C_{{\mathrm{2}}}}  \ottsym{]} \rangle_{ \ottmv{l} }   \simeq   \langle \mathcal{C}  \ottsym{[}  \ottnt{C_{{\mathrm{1}}}}  \parallel  \ottnt{C_{{\mathrm{2}}}}  \ottsym{]} \rangle_{ \ottmv{l} } $ for any $\ottmv{l}$ by \propref{env/runtime/focus/null}, \propref{env/equiv}, and \propref{env/equiv/ctx}.
              We know $ \lBrack  \langle \mathcal{C}  \ottsym{[}  \ottnt{C_{{\mathrm{1}}}}  \ottsym{,}  \ottnt{C_{{\mathrm{2}}}}  \ottsym{]} \rangle_{ \ottmv{l} }  \rBrack $ is traceable since $\mathcal{C}  \ottsym{[}  \ottnt{C_{{\mathrm{1}}}}  \ottsym{,}  \ottnt{C_{{\mathrm{2}}}}  \ottsym{]}$ is order-defined, and so $ \lBrack  \langle \mathcal{C}  \ottsym{[}  \ottnt{C_{{\mathrm{1}}}}  \parallel  \ottnt{C_{{\mathrm{2}}}}  \ottsym{]} \rangle_{ \ottmv{l} }  \rBrack $ is traceable by \propref{env/runtime/intrp/traceable/equiv}.
              Now we can have the subgoal by \propref{env/runtime/usage/equiv}.
    \end{itemize}

\end{prop}

  \subsection{Properties for preservation}
  \begin{prop}{subj/beta}
    If
    \begin{gather}
        \ottnt{M}  \longrightarrow_\beta  \ottnt{M'} \hyp{0.1}\\
        \ottnt{C}  \vdash  \ottnt{M}  :  \ottnt{T}  \mid  \ottnt{e}, \hyp{0.2}
    \end{gather}
    then $\ottnt{C}  \vdash  \ottnt{M'}  :  \ottnt{T}  \mid  \ottnt{e}$.

    \proof By case analysis on the last rule of the given derivation of \hypref{0.1}.
    \begin{match}
        \item[\ruleref{RE-Beta}]
        In this case,
        \begin{gather}
             \ottnt{M} \ottsym{=} \ottsym{(}  \lambda  \ottmv{x}  \ottsym{.}  \ottnt{M_{{\mathrm{1}}}}  \ottsym{)} \, \ottnt{V}  \hyp{1.1}\\
             \ottnt{M'} \ottsym{=} \ottnt{M_{{\mathrm{1}}}}  \ottsym{[}  \ottnt{V}  \ottsym{/}  \ottmv{x}  \ottsym{]}  \hyp{1.2}.
        \end{gather}
        So, we have
        \begin{gather}
            \ottnt{C}  \lesssim  \ottnt{C_{{\mathrm{1}}}}  \ottsym{,}  \ottnt{C_{{\mathrm{2}}}} \hyp{1.3}\\
              \ottnt{e_{{\mathrm{1}}}}  \sqcup  \ottnt{e_{{\mathrm{2}}}}  \le \ottnt{e}  \hyp{1.4}\\
            \ottkw{unr} \, \ottnt{C_{{\mathrm{1}}}} \hyp{1.5}\\
            \ottnt{C_{{\mathrm{1}}}}  \ottsym{,}  \ottmv{x}  \mathord:  \ottnt{S}  \vdash  \ottnt{M_{{\mathrm{1}}}}  :  \ottnt{T}  \mid  \ottnt{e_{{\mathrm{1}}}} \hyp{1.6}\\
            \ottnt{C_{{\mathrm{2}}}}  \vdash  \ottnt{V}  :  \ottnt{S}  \mid  \ottnt{e_{{\mathrm{2}}}} \hyp{1.7}
        \end{gather}
        by \propref{typing/inv}, \propref{env/sub/runtime}, and \propref{env/sub/ctx}.
        We know $ \ottmv{x}  \notin   \ottkw{dom} ( \ottnt{C_{{\mathrm{1}}}}  \ottsym{,}  \ottsym{[]} )  $.
        So, we have
        \begin{gather*}
            \ottnt{C_{{\mathrm{1}}}}  \ottsym{,}  \ottnt{C_{{\mathrm{2}}}}  \vdash  \ottnt{M_{{\mathrm{1}}}}  \ottsym{[}  \ottnt{V}  \ottsym{/}  \ottmv{x}  \ottsym{]}  :  \ottnt{T}  \mid   \ottnt{e_{{\mathrm{1}}}}  \sqcup  \ottnt{e_{{\mathrm{2}}}}  \hyp{1.9}
        \end{gather*}
        by \propref{typing/subst} given $ \mathcal{G} \ottsym{=} \ottnt{C_{{\mathrm{1}}}}  \ottsym{,}  \ottsym{[]} $.
        We can derive the goal by \ruleref{T-Weaken}, \hypref{1.3}, \hypref{1.4}, and \hypref{1.9}.

        \item[\ruleref{RE-UBeta}, \ruleref{RE-RBeta}, and \ruleref{RE-LBeta}]
        Similar to the case \ruleref{RE-UBeta}.

        \item[\ruleref{RE-ULet}]
        In this case,
        \begin{gather}
             \ottnt{M} \ottsym{=} \ottkw{let} \, \ottmv{x}  \otimes  \ottmv{y}  \ottsym{=}  \ottnt{V_{{\mathrm{1}}}}  \otimes  \ottnt{V_{{\mathrm{2}}}} \, \ottkw{in} \, \ottnt{M_{{\mathrm{1}}}}  \hyp{5.1}\\
             \ottnt{M'} \ottsym{=} \ottnt{M_{{\mathrm{1}}}}  \ottsym{[}  \ottnt{V_{{\mathrm{1}}}}  \ottsym{/}  \ottmv{x}  \ottsym{]}  \ottsym{[}  \ottnt{V_{{\mathrm{2}}}}  \ottsym{/}  \ottmv{y}  \ottsym{]}  \hyp{5.2}.
        \end{gather}
        So, we have
        \begin{gather}
            \ottnt{C}  \lesssim  \mathcal{C}  \ottsym{[}  \ottnt{C_{{\mathrm{1}}}}  \parallel  \ottnt{C_{{\mathrm{2}}}}  \ottsym{]} \hyp{5.3}\\
             \ottnt{e'} \le \ottnt{e}  \hyp{5.4}\\
            \ottnt{C_{{\mathrm{1}}}}  \vdash  \ottnt{V_{{\mathrm{1}}}}  :  \ottnt{S_{{\mathrm{1}}}}  \mid  \ottsym{0} \hyp{5.5}\\
            \ottnt{C_{{\mathrm{2}}}}  \vdash  \ottnt{V_{{\mathrm{2}}}}  :  \ottnt{S_{{\mathrm{2}}}}  \mid  \ottsym{0} \hyp{5.6}\\
            \mathcal{C}  \ottsym{[}  \ottmv{x}  \mathord:  \ottnt{S_{{\mathrm{1}}}}  \parallel  \ottmv{y}  \mathord:  \ottnt{S_{{\mathrm{2}}}}  \ottsym{]}  \vdash  \ottnt{M_{{\mathrm{1}}}}  :  \ottnt{T}  \mid  \ottnt{e'} \hyp{5.7}\\
            \ottsym{\{}  \ottmv{x}  \ottsym{\}}  \uplus  \ottsym{\{}  \ottmv{y}  \ottsym{\}}  \uplus   \ottkw{dom} ( \mathcal{C}  \ottsym{[}  \ottnt{C_{{\mathrm{1}}}}  \parallel  \ottnt{C_{{\mathrm{2}}}}  \ottsym{]} )  \hyp{5.8}
        \end{gather}
        by \propref{typing/inv}, \propref{env/sub/runtime}, and \propref{env/sub/ctx}.
        Let $ \mathcal{G}_{{\mathrm{1}}} \ottsym{=} \mathcal{C}  \ottsym{[}  \ottsym{[]}  \parallel  \ottmv{y}  \mathord:  \ottnt{S_{{\mathrm{2}}}}  \ottsym{]} $ and $ \mathcal{G}_{{\mathrm{2}}} \ottsym{=} \mathcal{C}  \ottsym{[}  \ottnt{C_{{\mathrm{1}}}}  \parallel  \ottsym{[]}  \ottsym{]} $.
        We have $ \ottmv{x}  \notin   \ottkw{dom} ( \mathcal{G}_{{\mathrm{1}}}  \ottsym{[}   \cdot   \ottsym{]} )  $ and $ \ottmv{y}  \notin   \ottkw{dom} ( \mathcal{G}_{{\mathrm{2}}}  \ottsym{[}   \cdot   \ottsym{]} )  $ by \hypref{5.8}.
        So, we have
        \begin{gather}
            \mathcal{C}  \ottsym{[}  \ottnt{C_{{\mathrm{1}}}}  \parallel  \ottnt{C_{{\mathrm{2}}}}  \ottsym{]}  \vdash  \ottnt{M_{{\mathrm{1}}}}  \ottsym{[}  \ottnt{V_{{\mathrm{1}}}}  \ottsym{/}  \ottmv{x}  \ottsym{]}  \ottsym{[}  \ottnt{V_{{\mathrm{2}}}}  \ottsym{/}  \ottmv{y}  \ottsym{]}  :  \ottnt{T}  \mid  \ottnt{e'} \hyp{5.9}
        \end{gather}
        by using \propref{typing/subst} twice given $ \mathcal{G} \ottsym{=} \mathcal{G}_{{\mathrm{1}}} $ and $ \mathcal{G} \ottsym{=} \mathcal{G}_{{\mathrm{2}}} $ in the order.
        Now the gaol is derived by \ruleref{T-Weaken}, \hypref{5.3}, \hypref{5.4}, and \hypref{5.9}.

        \item[\ruleref{RE-OLet}]
        Similar to the case \ruleref{RE-ULet}.
    \end{match}
\end{prop}

\begin{prop}{subj/beta/ctx}
    If
    \begin{gather}
        \ottnt{M}  \longrightarrow_\beta  \ottnt{M'} \hyp{0.1}\\
        \ottnt{C}  \vdash  \mathcal{E}  \ottsym{[}  \ottnt{M}  \ottsym{]}  :  \ottnt{T}  \mid  \ottnt{e}, \hyp{0.2}
    \end{gather}
    then $\ottnt{C}  \vdash  \mathcal{E}  \ottsym{[}  \ottnt{M'}  \ottsym{]}  :  \ottnt{T}  \mid  \ottnt{e}$.

    \proof By using \propref{subj/beta}, \propref{typing/inv}, and \propref{env/sub/runtime}, the proof is routine by structural induction on $\mathcal{E}$.
\end{prop}

\begin{prop}{subj/cmp/noeffect}
    If
    \begin{gather}
        \ottnt{M}  \mid  \mathcal{H}  \longrightarrow_\gamma  \ottnt{M'}  \mid  \mathcal{H}' \hyp{0.1}\\
        \ottnt{C}  \vdash  \ottnt{M}  :  \ottnt{T}  \mid  \ottsym{0} \hyp{0.2}\\
        \mathcal{C}  \ottsym{[}  \ottnt{C}  \ottsym{]}  \vdash  \mathcal{H} \hyp{0.3},
    \end{gather}
    then there exists $\ottnt{C'}$ such that $  \ottsym{(}    \ottkw{dom} ( \ottnt{C'} )  \setminus  \ottkw{dom} ( \ottnt{C} )    \ottsym{)} \cap  \ottkw{dom} ( \mathcal{C} )   \ottsym{=} \emptyset $, $\ottnt{C'}  \vdash  \ottnt{M'}  :  \ottnt{T}  \mid  \ottsym{0}$, and $\mathcal{C}  \ottsym{[}  \ottnt{C'}  \ottsym{]}  \vdash  \mathcal{H}'$.

    \proof
    By the case analysis on the last rule of the given derivation of \hypref{0.1}.
    \begin{match}
        \item[\ruleref{RC-Ne}]
        In this case,
        \begin{gather}
             \ottnt{M} \ottsym{=}  \ottkw{new} _{ \ottnt{m} }  \, \ottkw{unit}  \hyp{1.1}\\
             \ottnt{M'} \ottsym{=} \ottmv{l}  \hyp{1.2}\\
             \mathcal{H}' \ottsym{=}  \mathcal{H} \cup \ottsym{\{} \ottmv{l}  \mapsto  \ottsym{(}  \ottsym{0}  \ottsym{,}  \ottnt{m}  \ottsym{,}  \varepsilon  \ottsym{)} \ottsym{\}}   \hyp{1.3}\\
             \ottmv{l}  \notin   \ottkw{dom} ( \mathcal{H} )   \hyp{1.4}.
        \end{gather}
        So, we have
        \begin{gather}
            \ottnt{C}  \lesssim   \cdot  \hyp{1.5}\\
             \ottnt{T} \ottsym{=} \ottsym{[}  \ottnt{m}  \ottsym{]}  \hyp{1.6}
        \end{gather}
        by \propref{typing/inv}, \propref{env/sub/runtime}, \propref{env/sub/ctx}, and \hypref{0.2}.
        To this end, we choose $ \ottnt{C'} \ottsym{=} \ottmv{l}  \mathord:  \ottsym{[}  \ottnt{m}  \ottsym{]} $.
        The subgoals we need to show are
        \begin{gather*}
              \ottsym{(}   \ottsym{\{}  \ottmv{l}  \ottsym{\}} \setminus  \ottkw{dom} ( \ottnt{C} )    \ottsym{)} \cap  \ottkw{dom} ( \mathcal{C} )   \ottsym{=} \emptyset  \\
            \ottmv{l}  \mathord:  \ottsym{[}  \ottnt{m}  \ottsym{]}  \vdash  \ottmv{l}  :  \ottsym{[}  \ottnt{m}  \ottsym{]}  \mid  \ottsym{0} \\
            \mathcal{C}  \ottsym{[}  \ottmv{l}  \mathord:  \ottsym{[}  \ottnt{m}  \ottsym{]}  \ottsym{]}  \vdash   \mathcal{H} \cup \ottsym{\{} \ottmv{l}  \mapsto  \ottsym{(}  \ottsym{0}  \ottsym{,}  \ottnt{m}  \ottsym{,}  \varepsilon  \ottsym{)} \ottsym{\}} .
        \end{gather*}
        We have $  \ottkw{dom} ( \mathcal{C}  \ottsym{[}  \ottnt{C}  \ottsym{]} )  \ottsym{=}  \ottkw{dom} ( \mathcal{H} )  $ by \hypref{0.3}.
        So, $ \ottmv{l}  \notin   \ottkw{dom} ( \mathcal{C}  \ottsym{[}  \ottnt{C}  \ottsym{]} )  $ by \hypref{1.4}.
        That results in $ \ottmv{l}  \notin   \ottkw{dom} ( \mathcal{C} )  $, which implies the first subgoal.
        We can derive the second subgoal by \ruleref{T-Loc}.
        As for the third subgoal, firstly we have
        \begin{gather*}
            \mathcal{C}  \ottsym{[}  \ottnt{C}  \ottsym{]}  \lesssim  \mathcal{C}  \ottsym{[}   \cdot   \ottsym{]}
        \end{gather*}
        by \propref{env/sub/ctx} and \hypref{1.5}.
        Then, we get
        \begin{gather*}
            \mathcal{C}  \ottsym{[}   \cdot   \ottsym{]}  \vdash  \mathcal{H}
        \end{gather*}
        by \propref{heap/weaken} and \hypref{0.3}.
        Now the subgoal follows by \propref{heap/track/new} and \hypref{1.4}.

        \item[\ruleref{RC-Cl1}]
        In this case,
        \begin{gather}
             \ottnt{M} \ottsym{=} \ottkw{drop} \, \ottmv{l}  \hyp{3.1}\\
             \ottnt{M'} \ottsym{=} \ottkw{unit}  \hyp{3.2}\\
             \mathcal{H} \ottsym{=}  \mathcal{H}_{{\mathrm{1}}} \cup \ottsym{\{} \ottmv{l}  \mapsto  \ottsym{(}  \ottnt{n}  \ottsym{+}  \ottsym{1}  \ottsym{,}  \ottnt{m_{{\mathrm{0}}}}  \ottsym{,}  \ottnt{m}  \ottsym{)} \ottsym{\}}   \hyp{3.3}\\
             \mathcal{H}' \ottsym{=}  \mathcal{H}_{{\mathrm{1}}} \cup \ottsym{\{} \ottmv{l}  \mapsto  \ottsym{(}  \ottnt{n}  \ottsym{,}  \ottnt{m_{{\mathrm{0}}}}  \ottsym{,}  \ottnt{m}  \ottsym{)} \ottsym{\}}   \hyp{3.4}.
        \end{gather}
        So, we have
        \begin{gather}
            \ottnt{C}  \lesssim  \ottmv{l}  \mathord:  \ottsym{[}  \ottnt{m'}  \ottsym{]} \hyp{3.5}\\
             \ottnt{T} \ottsym{=}  \mathtt{Unit}   \hyp{3.6}\\
             \varepsilon \le \ottnt{m'}  \hyp{3.7}
        \end{gather}
        by \propref{typing/inv}, \propref{env/sub/ctx}, and \hypref{0.2}.
        To this end, we choose $ \ottnt{C'} \ottsym{=}  \cdot  $.
        The subgoals we need to show are
        \begin{gather*}
              \ottsym{(}   \emptyset \setminus  \ottkw{dom} ( \ottnt{C} )    \ottsym{)} \cap  \ottkw{dom} ( \mathcal{C} )   \ottsym{=} \emptyset  \\
             \cdot   \vdash  \ottkw{unit}  :   \mathtt{Unit}   \mid  \ottsym{0} \\
            \mathcal{C}  \ottsym{[}   \cdot   \ottsym{]}  \vdash   \mathcal{H}_{{\mathrm{1}}} \cup \ottsym{\{} \ottmv{l}  \mapsto  \ottsym{(}  \ottnt{n}  \ottsym{,}  \ottnt{m_{{\mathrm{0}}}}  \ottsym{,}  \ottnt{m}  \ottsym{)} \ottsym{\}} .
        \end{gather*}
        The first subgoal trivially holds.
        We can derive the second subgoal by \ruleref{T-Unit}.
        As for the thrid subgoal, we have $\mathcal{C}  \ottsym{[}  \ottmv{l}  \mathord:  \ottsym{[}  \ottnt{m'}  \ottsym{]}  \ottsym{]}  \vdash   \mathcal{H}_{{\mathrm{1}}} \cup \ottsym{\{} \ottmv{l}  \mapsto  \ottsym{(}  \ottnt{n}  \ottsym{+}  \ottsym{1}  \ottsym{,}  \ottnt{m_{{\mathrm{0}}}}  \ottsym{,}  \ottnt{m}  \ottsym{)} \ottsym{\}} $ by \propref{env/sub/ctx}, \propref{heap/weaken}, \hypref{0.3}, and \hypref{3.5}.
        So, the goal follows by \propref{heap/track/cl1} and \hypref{3.7}.

        \item[\ruleref{RC-Cl2}]
        Similar to the case \ruleref{RC-Cl1} except we use \propref{heap/track/cl2}.

        \item[\ruleref{RC-Sp}]
        In this case,
        \begin{gather}
             \ottnt{M} \ottsym{=}  \ottkw{split} _{ \ottnt{m_{{\mathrm{1}}}} , \ottnt{m_{{\mathrm{2}}}} }  \, \ottmv{l}  \hyp{2.1}\\
             \ottnt{M'} \ottsym{=} \ottmv{l}  \odot  \ottmv{l}  \hyp{2.2}\\
             \mathcal{H} \ottsym{=}  \mathcal{H}_{{\mathrm{1}}} \cup \ottsym{\{} \ottmv{l}  \mapsto  \ottsym{(}  \ottnt{n}  \ottsym{,}  \ottnt{m_{{\mathrm{0}}}}  \ottsym{,}  \ottnt{m}  \ottsym{)} \ottsym{\}}   \hyp{2.3}\\
             \mathcal{H}' \ottsym{=}  \mathcal{H}_{{\mathrm{1}}} \cup \ottsym{\{} \ottmv{l}  \mapsto  \ottsym{(}  \ottnt{n}  \ottsym{+}  \ottsym{1}  \ottsym{,}  \ottnt{m_{{\mathrm{0}}}}  \ottsym{,}  \ottnt{m}  \ottsym{)} \ottsym{\}}   \hyp{2.4}.
        \end{gather}
        So, we have
        \begin{gather}
            \ottnt{C}  \lesssim  \ottmv{l}  \mathord:  \ottsym{[}  \ottnt{m'}  \ottsym{]} \hyp{2.5}\\
             \ottnt{T} \ottsym{=} \ottsym{[}  \ottnt{m_{{\mathrm{1}}}}  \ottsym{]}  \odot  \ottsym{[}  \ottnt{m_{{\mathrm{2}}}}  \ottsym{]}  \hyp{2.6}\\
              \ottnt{m_{{\mathrm{1}}}} \odot \ottnt{m_{{\mathrm{2}}}}  \le \ottnt{m'}  \hyp{2.6.1.1}
        \end{gather}
        by \propref{typing/inv}, \propref{env/sub/runtime}, \propref{env/sub/ctx} and \hypref{0.2}.
        To this end, we choose $ \ottnt{C'} \ottsym{=} \ottmv{l}  \mathord:  \ottsym{[}  \ottnt{m_{{\mathrm{1}}}}  \ottsym{]}  \ottsym{,}  \ottmv{l}  \mathord:  \ottsym{[}  \ottnt{m_{{\mathrm{2}}}}  \ottsym{]} $.
        The subgoals we need to show are
        \begin{gather*}
              \ottsym{(}   \ottsym{\{}  \ottmv{l}  \ottsym{\}} \setminus  \ottkw{dom} ( \ottnt{C} )    \ottsym{)} \cap  \ottkw{dom} ( \mathcal{C} )   \ottsym{=} \emptyset  \\
            \ottmv{l}  \mathord:  \ottsym{[}  \ottnt{m_{{\mathrm{1}}}}  \ottsym{]}  \ottsym{,}  \ottmv{l}  \mathord:  \ottsym{[}  \ottnt{m_{{\mathrm{2}}}}  \ottsym{]}  \vdash  \ottmv{l}  \odot  \ottmv{l}  :  \ottsym{[}  \ottnt{m_{{\mathrm{1}}}}  \ottsym{]}  \odot  \ottsym{[}  \ottnt{m_{{\mathrm{2}}}}  \ottsym{]}  \mid  \ottsym{0}\\
            \mathcal{C}  \ottsym{[}  \ottmv{l}  \mathord:  \ottsym{[}  \ottnt{m_{{\mathrm{1}}}}  \ottsym{]}  \ottsym{,}  \ottmv{l}  \mathord:  \ottsym{[}  \ottnt{m_{{\mathrm{2}}}}  \ottsym{]}  \ottsym{]}  \vdash   \mathcal{H}_{{\mathrm{1}}} \cup \ottsym{\{} \ottmv{l}  \mapsto  \ottsym{(}  \ottnt{n}  \ottsym{+}  \ottsym{1}  \ottsym{,}  \ottnt{m_{{\mathrm{0}}}}  \ottsym{,}  \ottnt{m}  \ottsym{)} \ottsym{\}} .
        \end{gather*}
        Having $  \ottkw{dom} ( \ottnt{C} )  \ottsym{=} \ottsym{\{}  \ottmv{l}  \ottsym{\}} $ by \propref{env/sub/runtime} and \hypref{2.5}, we can easily see the first subgoal.
        We can derive the second subgoal by \ruleref{T-Loc} and \ruleref{T-OPair}.
        As for the third subgoal, firstly, we have
        \begin{gather*}
            \mathcal{C}  \ottsym{[}  \ottnt{C}  \ottsym{]}  \lesssim  \mathcal{C}  \ottsym{[}  \ottmv{l}  \mathord:  \ottsym{[}  \ottnt{m'}  \ottsym{]}  \ottsym{]}
        \end{gather*}
        by \propref{env/sub/ctx} and \hypref{2.5}.
        Then, we have
        \begin{gather*}
            \mathcal{C}  \ottsym{[}  \ottmv{l}  \mathord:  \ottsym{[}  \ottnt{m'}  \ottsym{]}  \ottsym{]}  \vdash   \mathcal{H}_{{\mathrm{1}}} \cup \ottsym{\{} \ottmv{l}  \mapsto  \ottsym{(}  \ottnt{n}  \ottsym{,}  \ottnt{m_{{\mathrm{0}}}}  \ottsym{,}  \ottnt{m}  \ottsym{)} \ottsym{\}} 
        \end{gather*}
        by \propref{heap/weaken}, \hypref{0.3}, and \hypref{2.3}.
        Therefore, the subgoal follows by \propref{heap/track/split} and \hypref{2.6.1.1}.

        \item[\ruleref{RC-Op}]
        In this case, $ \ottnt{M} \ottsym{=}  \ottkw{op} _{ \ottnt{m} }  \, \ottmv{l} $, which must have an effect.
        However, it contradicts \hypref{0.2}.
        So, this case is vacuously true. \qedhere
    \end{match}
\end{prop}

\begin{prop}{subj/cmp/effect}
    If
    \begin{gather}
        \ottnt{M}  \mid  \mathcal{H}  \longrightarrow_\gamma  \ottnt{M'}  \mid  \mathcal{H}' \hyp{0.1}\\
        \ottnt{C}  \vdash  \ottnt{M}  :  \ottnt{T}  \mid  \ottnt{e} \hyp{0.2}\\
        \ottnt{C}  \ottsym{,}  \ottnt{C''}  \vdash  \mathcal{H} \hyp{0.3},
    \end{gather}
    then there exists $\ottnt{C'}$ such that $  \ottsym{(}    \ottkw{dom} ( \ottnt{C'} )  \setminus  \ottkw{dom} ( \ottnt{C} )    \ottsym{)} \cap  \ottkw{dom} ( \ottnt{C''} )   \ottsym{=} \emptyset $, $\ottnt{C'}  \vdash  \ottnt{M'}  :  \ottnt{T}  \mid  \ottnt{e}$, and $\ottnt{C'}  \ottsym{,}  \ottnt{C''}  \vdash  \mathcal{H}'$.

    \proof By case analysis on the last rule of the given derivation of \hypref{0.1}.
    \begin{match}
        \item[\ruleref{RC-Op}]
        In this case,
        \begin{gather}
             \ottnt{M} \ottsym{=}  \ottkw{op} _{ \ottnt{m_{{\mathrm{1}}}} }  \, \ottmv{l}  \hyp{1.1}\\
             \ottnt{M'} \ottsym{=} \ottmv{l}  \hyp{1.2}\\
             \mathcal{H} \ottsym{=}  \mathcal{H}_{{\mathrm{1}}} \cup \ottsym{\{} \ottmv{l}  \mapsto  \ottsym{(}  \ottnt{n}  \ottsym{,}  \ottnt{m_{{\mathrm{0}}}}  \ottsym{,}  \ottnt{m}  \ottsym{)} \ottsym{\}}   \hyp{1.3}\\
             \mathcal{H}' \ottsym{=}  \mathcal{H}_{{\mathrm{1}}} \cup \ottsym{\{} \ottmv{l}  \mapsto  \ottsym{(}  \ottnt{n}  \ottsym{,}  \ottnt{m_{{\mathrm{0}}}}  \ottsym{,}   \ottnt{m} \odot \ottnt{m_{{\mathrm{1}}}}   \ottsym{)} \ottsym{\}}   \hyp{1.4}.
        \end{gather}
        So, we have
        \begin{gather}
            \ottnt{C}  \lesssim  \ottmv{l}  \mathord:  \ottsym{[}  \ottnt{m'}  \ottsym{]} \hyp{1.5}\\
             \ottnt{T} \ottsym{=} \ottsym{[}  \ottnt{m_{{\mathrm{2}}}}  \ottsym{]}  \hyp{1.6}\\
              \ottnt{m_{{\mathrm{1}}}} \odot \ottnt{m_{{\mathrm{2}}}}  \le \ottnt{m'}  \hyp{1.6.1}\\
             \ottsym{1} \le \ottnt{e}  \hyp{1.8}
        \end{gather}
        by \propref{typing/inv}, \propref{env/sub/runtime}, \propref{env/sub/ctx}, and \hypref{0.2}.
        To this end, we choose $ \ottnt{C'} \ottsym{=} \ottmv{l}  \mathord:  \ottsym{[}  \ottnt{m_{{\mathrm{2}}}}  \ottsym{]} $.
        The subgoals we need to show are
        \begin{gather*}
              \ottsym{(}   \ottsym{\{}  \ottmv{l}  \ottsym{\}} \setminus  \ottkw{dom} ( \ottnt{C} )    \ottsym{)} \cap  \ottkw{dom} ( \ottnt{C''} )   \ottsym{=} \emptyset  \\
            \ottmv{l}  \mathord:  \ottsym{[}  \ottnt{m_{{\mathrm{2}}}}  \ottsym{]}  \vdash  \ottmv{l}  :  \ottsym{[}  \ottnt{m_{{\mathrm{2}}}}  \ottsym{]}  \mid  \ottnt{e} \\
            \ottmv{l}  \mathord:  \ottsym{[}  \ottnt{m_{{\mathrm{2}}}}  \ottsym{]}  \ottsym{,}  \ottnt{C''}  \vdash   \mathcal{H}_{{\mathrm{1}}} \cup \ottsym{\{} \ottmv{l}  \mapsto  \ottsym{(}  \ottnt{n}  \ottsym{,}  \ottnt{m_{{\mathrm{0}}}}  \ottsym{,}   \ottnt{m} \odot \ottnt{m_{{\mathrm{1}}}}   \ottsym{)} \ottsym{\}} .
        \end{gather*}
        Having $  \ottkw{dom} ( \ottnt{C} )  \ottsym{=} \ottsym{\{}  \ottmv{l}  \ottsym{\}} $ by \propref{env/sub/runtime/dom} and \hypref{1.5}, we can easily see the first subgoal.
        We can derive the second subgoal by \ruleref{T-Loc}, \ruleref{T-Weaken}, and \hypref{1.8}.
        As for the third subgoal, we have
        \begin{gather*}
            \ottnt{C}  \ottsym{,}  \ottnt{C''}  \lesssim  \ottmv{l}  \mathord:  \ottsym{[}  \ottnt{m'}  \ottsym{]}  \ottsym{,}  \ottnt{C''}
        \end{gather*}
        by \propref{env/sub/ctx} and \hypref{1.5}.
        So,
        \begin{gather*}
            \ottmv{l}  \mathord:  \ottsym{[}  \ottnt{m'}  \ottsym{]}  \ottsym{,}  \ottnt{C''}  \vdash   \mathcal{H}_{{\mathrm{1}}} \cup \ottsym{\{} \ottmv{l}  \mapsto  \ottsym{(}  \ottnt{n}  \ottsym{,}  \ottnt{m_{{\mathrm{0}}}}  \ottsym{,}  \ottnt{m}  \ottsym{)} \ottsym{\}} 
        \end{gather*}
        by \propref{heap/weaken}, \hypref{0.3}, and \hypref{1.3}.
        Then, the subgoal follows by \propref{heap/track/op} and \hypref{1.6.1}.

        \item[\ruleref{RC-Ne}, \ruleref{RC-Cl1}, \ruleref{RC-Cl2}, and \ruleref{RC-Sp}]
        In these cases, we can derive $\ottnt{C}  \vdash  \ottnt{M}  :  \ottnt{T}  \mid  \ottsym{0}$.
        So, the goal follows by \propref{subj/cmp/noeffect} given $ \mathcal{C} \ottsym{=} \ottsym{[]}  \ottsym{,}  \ottnt{C''} $.
    \end{match}
\end{prop}

\begin{prop}{subj/cmp/noeffect/ctx}
    If
    \begin{gather}
        \ottnt{M}  \mid  \mathcal{H}  \longrightarrow_\gamma  \ottnt{M'}  \mid  \mathcal{H}' \hyp{0.1}\\
        \ottnt{C}  \vdash  \mathcal{E}  \ottsym{[}  \ottnt{M}  \ottsym{]}  :  \ottnt{T}  \mid  \ottsym{0} \hyp{0.2}\\
        \mathcal{C}  \ottsym{[}  \ottnt{C}  \ottsym{]}  \vdash  \mathcal{H} \hyp{0.3},
    \end{gather}
    then there exists $\ottnt{C'}$ such that $  \ottsym{(}    \ottkw{dom} ( \ottnt{C'} )  \setminus  \ottkw{dom} ( \ottnt{C} )    \ottsym{)} \cap  \ottkw{dom} ( \mathcal{C} )   \ottsym{=} \emptyset $, $\ottnt{C'}  \vdash  \mathcal{E}  \ottsym{[}  \ottnt{M'}  \ottsym{]}  :  \ottnt{T}  \mid  \ottsym{0}$, and $\mathcal{C}  \ottsym{[}  \ottnt{C'}  \ottsym{]}  \vdash  \mathcal{H}'$.

    \proof
    By structural induction on $\mathcal{E}$.
    \begin{match}
        \item[$ \mathcal{E} \ottsym{=} \ottsym{[]} $]
        This case directly follows by \propref{subj/cmp/noeffect}.

        \item[$ \mathcal{E} \ottsym{=} \mathcal{E}_{{\mathrm{1}}} \, \ottnt{M_{{\mathrm{2}}}} $]
        In this case, we have
        \begin{gather}
            \ottnt{C}  \lesssim  \ottnt{C_{{\mathrm{1}}}}  \ottsym{,}  \ottnt{C_{{\mathrm{2}}}} \hyp{2.1}\\
            \ottnt{C_{{\mathrm{1}}}}  \vdash  \mathcal{E}_{{\mathrm{1}}}  \ottsym{[}  \ottnt{M}  \ottsym{]}  :   \ottnt{S} \rightarrow _{ \ottnt{e_{{\mathrm{0}}}} } \ottnt{T}   \mid  \ottnt{e_{{\mathrm{1}}}} \hyp{2.2}\\
            \ottnt{C_{{\mathrm{2}}}}  \vdash  \ottnt{M_{{\mathrm{2}}}}  :  \ottnt{S}  \mid  \ottnt{e_{{\mathrm{2}}}} \hyp{2.3}\\
               \ottnt{e_{{\mathrm{0}}}}  \sqcup  \ottnt{e_{{\mathrm{1}}}}   \sqcup  \ottnt{e_{{\mathrm{2}}}}  \le \ottsym{0}  \hyp{2.4}
        \end{gather}
        by \propref{typing/inv}, \propref{env/sub/runtime}, and \hypref{0.2}.
        We have $\mathcal{C}  \ottsym{[}  \ottnt{C}  \ottsym{]}  \lesssim  \mathcal{C}  \ottsym{[}  \ottnt{C_{{\mathrm{1}}}}  \ottsym{,}  \ottnt{C_{{\mathrm{2}}}}  \ottsym{]}$ by \propref{env/sub/ctx} and \hypref{2.1}.
        So,
        \begin{gather}
            \mathcal{C}  \ottsym{[}  \ottnt{C_{{\mathrm{1}}}}  \ottsym{,}  \ottnt{C_{{\mathrm{2}}}}  \ottsym{]}  \vdash  \mathcal{H} \hyp{2.4.1}
        \end{gather}
        by \propref{heap/weaken} and \hypref{0.3}.
        Here, we notice $ \mathcal{C}  \ottsym{[}  \ottnt{C_{{\mathrm{1}}}}  \ottsym{,}  \ottnt{C_{{\mathrm{2}}}}  \ottsym{]} \ottsym{=} \ottsym{(}  \mathcal{C}  \ottsym{[}  \ottsym{[]}  \ottsym{,}  \ottnt{C_{{\mathrm{2}}}}  \ottsym{]}  \ottsym{)}  \ottsym{[}  \ottnt{C_{{\mathrm{1}}}}  \ottsym{]} $.
        Therefore, we can have some $\ottnt{C'_{{\mathrm{1}}}}$ such that
        \begin{gather}
              \ottsym{(}    \ottkw{dom} ( \ottnt{C'_{{\mathrm{1}}}} )  \setminus  \ottkw{dom} ( \ottnt{C_{{\mathrm{1}}}} )    \ottsym{)} \cap  \ottkw{dom} ( \mathcal{C}  \ottsym{[}   \cdot   \ottsym{,}  \ottnt{C_{{\mathrm{2}}}}  \ottsym{]} )   \ottsym{=} \emptyset  \hyp{2.4.2}\\
            \ottnt{C'_{{\mathrm{1}}}}  \vdash  \mathcal{E}_{{\mathrm{1}}}  \ottsym{[}  \ottnt{M'}  \ottsym{]}  :   \ottnt{S} \rightarrow _{ \ottnt{e_{{\mathrm{0}}}} } \ottnt{T}   \mid  \ottnt{e_{{\mathrm{1}}}} \hyp{2.5}\\
            \mathcal{C}  \ottsym{[}  \ottnt{C'_{{\mathrm{1}}}}  \ottsym{,}  \ottnt{C_{{\mathrm{2}}}}  \ottsym{]}  \vdash  \mathcal{H}' \hyp{2.6}
        \end{gather}
        by applying the induction hypothesis to \hypref{2.2} and \hypref{2.4.1}.
        To this end, we choose $ \ottnt{C'} \ottsym{=} \ottnt{C'_{{\mathrm{1}}}}  \ottsym{,}  \ottnt{C_{{\mathrm{2}}}} $.
        The subgoals we need to show are
        \begin{gather*}
              \ottsym{(}    \ottkw{dom} ( \ottnt{C'_{{\mathrm{1}}}}  \ottsym{,}  \ottnt{C_{{\mathrm{2}}}} )  \setminus  \ottkw{dom} ( \ottnt{C} )    \ottsym{)} \cap  \ottkw{dom} ( \mathcal{C}  \ottsym{[}   \cdot   \ottsym{]} )   \ottsym{=} \emptyset  \\
            \ottnt{C'_{{\mathrm{1}}}}  \ottsym{,}  \ottnt{C_{{\mathrm{2}}}}  \vdash  \mathcal{E}_{{\mathrm{1}}}  \ottsym{[}  \ottnt{M'}  \ottsym{]} \, \ottnt{M_{{\mathrm{2}}}}  :  \ottnt{T}  \mid  \ottsym{0}\\
            \mathcal{C}  \ottsym{[}  \ottnt{C'_{{\mathrm{1}}}}  \ottsym{,}  \ottnt{C_{{\mathrm{2}}}}  \ottsym{]}  \vdash  \mathcal{H}'.
        \end{gather*}
        Having $  \ottkw{dom} ( \ottnt{C} )  \ottsym{=}  \ottkw{dom} ( \ottnt{C_{{\mathrm{1}}}}  \ottsym{,}  \ottnt{C_{{\mathrm{2}}}} )  $ by \propref{env/sub/runtime} and \hypref{2.1}, we can see the first subgoal follows by \hypref{2.4.2}.
        We can derive the second subgoal by \ruleref{T-App}, \ruleref{T-Weaken}, \hypref{2.3}, \hypref{2.4}, and \hypref{2.5}.
        The third subgoal is identical to \hypref{2.6}.

        \item[$ \mathcal{E} \ottsym{=} \ottkw{let} \, \ottmv{x}  \otimes  \ottmv{y}  \ottsym{=}  \mathcal{E}_{{\mathrm{1}}} \, \ottkw{in} \, \ottnt{M_{{\mathrm{2}}}} $]
        In this case, we have
        \begin{gather}
            \ottnt{C}  \lesssim  \mathcal{C}_{{\mathrm{1}}}  \ottsym{[}  \ottnt{C_{{\mathrm{1}}}}  \ottsym{]} \hyp{3.1}\\
             \ottnt{e} \le \ottsym{0}  \hyp{3.1.1}\\
            \ottnt{C_{{\mathrm{1}}}}  \vdash  \mathcal{E}_{{\mathrm{1}}}  \ottsym{[}  \ottnt{M}  \ottsym{]}  :  \ottnt{S_{{\mathrm{1}}}}  \otimes  \ottnt{S_{{\mathrm{2}}}}  \mid  \ottsym{0} \hyp{3.2}\\
            \mathcal{C}_{{\mathrm{1}}}  \ottsym{[}  \ottmv{x}  \mathord:  \ottnt{S_{{\mathrm{1}}}}  \parallel  \ottmv{y}  \mathord:  \ottnt{S_{{\mathrm{2}}}}  \ottsym{]}  \vdash  \ottnt{M_{{\mathrm{2}}}}  :  \ottnt{T}  \mid  \ottnt{e} \hyp{3.3}\\
            \ottsym{\{}  \ottmv{x}  \ottsym{\}}  \uplus  \ottsym{\{}  \ottmv{y}  \ottsym{\}}  \uplus   \ottkw{dom} ( \mathcal{C}  \ottsym{[}  \ottnt{C_{{\mathrm{1}}}}  \ottsym{]} )  \hyp{3.4}
        \end{gather}
        by \propref{typing/inv}, \propref{env/sub/runtime}, and \hypref{0.2}.
        We have $\mathcal{C}  \ottsym{[}  \ottnt{C}  \ottsym{]}  \lesssim  \mathcal{C}  \ottsym{[}  \mathcal{C}_{{\mathrm{1}}}  \ottsym{[}  \ottnt{C_{{\mathrm{1}}}}  \ottsym{]}  \ottsym{]}$ by \propref{env/sub/ctx} and \hypref{3.1}.
        So,
        \begin{gather*}
            \mathcal{C}  \ottsym{[}  \mathcal{C}_{{\mathrm{1}}}  \ottsym{[}  \ottnt{C_{{\mathrm{1}}}}  \ottsym{]}  \ottsym{]}  \vdash  \mathcal{H} \hyp{3.5}
        \end{gather*}
        by \propref{heap/weaken}.
        Here we notice that $ \mathcal{C}  \ottsym{[}  \mathcal{C}_{{\mathrm{1}}}  \ottsym{[}  \ottnt{C_{{\mathrm{1}}}}  \ottsym{]}  \ottsym{]} \ottsym{=} \ottsym{(}  \mathcal{C}  \ottsym{[}  \mathcal{C}_{{\mathrm{1}}}  \ottsym{]}  \ottsym{)}  \ottsym{[}  \ottnt{C_{{\mathrm{1}}}}  \ottsym{]} $.
        Therefore, we get some $\ottnt{C'_{{\mathrm{1}}}}$ such that
        \begin{gather}
              \ottsym{(}    \ottkw{dom} ( \ottnt{C'_{{\mathrm{1}}}} )  \setminus  \ottkw{dom} ( \ottnt{C_{{\mathrm{1}}}} )    \ottsym{)} \cap  \ottkw{dom} ( \mathcal{C}  \ottsym{[}  \mathcal{C}_{{\mathrm{1}}}  \ottsym{[}   \cdot   \ottsym{]}  \ottsym{]} )   \ottsym{=} \emptyset  \hyp{3.5.1}\\
            \ottnt{C'_{{\mathrm{1}}}}  \vdash  \mathcal{E}_{{\mathrm{1}}}  \ottsym{[}  \ottnt{M'}  \ottsym{]}  :  \ottnt{S_{{\mathrm{1}}}}  \otimes  \ottnt{S_{{\mathrm{2}}}}  \mid  \ottsym{0} \hyp{3.6}\\
            \mathcal{C}  \ottsym{[}  \mathcal{C}_{{\mathrm{1}}}  \ottsym{[}  \ottnt{C'_{{\mathrm{1}}}}  \ottsym{]}  \ottsym{]}  \vdash  \mathcal{H}' \hyp{3.7}
        \end{gather}
        by applying the induction hypothesis to \hypref{3.2} and \hypref{3.5}.
        To this end, we choose $ \ottnt{C'} \ottsym{=} \mathcal{C}_{{\mathrm{1}}}  \ottsym{[}  \ottnt{C'_{{\mathrm{1}}}}  \ottsym{]} $.
        The subgoals we need to show are
        \begin{gather*}
              \ottsym{(}    \ottkw{dom} ( \mathcal{C}_{{\mathrm{1}}}  \ottsym{[}  \ottnt{C'_{{\mathrm{1}}}}  \ottsym{]} )  \setminus  \ottkw{dom} ( \ottnt{C} )    \ottsym{)} \cap  \ottkw{dom} ( \mathcal{C}  \ottsym{[}   \cdot   \ottsym{]} )   \ottsym{=} \emptyset  \\
            \mathcal{C}_{{\mathrm{1}}}  \ottsym{[}  \ottnt{C'_{{\mathrm{1}}}}  \ottsym{]}  \vdash  \ottkw{let} \, \ottmv{x}  \otimes  \ottmv{y}  \ottsym{=}  \mathcal{E}_{{\mathrm{1}}}  \ottsym{[}  \ottnt{M'}  \ottsym{]} \, \ottkw{in} \, \ottnt{M_{{\mathrm{2}}}}  :  \ottnt{T}  \mid  \ottsym{0} \\
            \mathcal{C}  \ottsym{[}  \mathcal{C}_{{\mathrm{1}}}  \ottsym{[}  \ottnt{C'_{{\mathrm{1}}}}  \ottsym{]}  \ottsym{]}  \vdash  \mathcal{H}'.
        \end{gather*}
        Having $  \ottkw{dom} ( \ottnt{C} )  \ottsym{=}  \ottkw{dom} ( \mathcal{C}_{{\mathrm{1}}}  \ottsym{[}  \ottnt{C_{{\mathrm{1}}}}  \ottsym{]} )  $ by \propref{env/sub/runtime} and \hypref{3.1}, we can see the first subgoal follows by \hypref{3.5.1}.
        We can derive the second subgoal by \ruleref{T-ULet}, \ruleref{T-Weaken}, \hypref{3.1.1}, \hypref{3.3}, \hypref{3.4}, and \hypref{3.6}.
        The third subgoal is identical to \hypref{3.7}.

        \item[Other cases]
        We can show other cases in a similar way to the cases above.
    \end{match}
\end{prop}

\begin{prop}{subj/cmp/effect/ctx}
    If
    \begin{gather}
        \ottnt{M}  \mid  \mathcal{H}  \longrightarrow_\gamma  \ottnt{M'}  \mid  \mathcal{H}' \hyp{0.1}\\
        \ottnt{C}  \vdash  \mathcal{E}  \ottsym{[}  \ottnt{M}  \ottsym{]}  :  \ottnt{T}  \mid  \ottnt{e} \hyp{0.2}\\
        \ottnt{C}  \ottsym{,}  \ottnt{C''}  \vdash  \mathcal{H} \hyp{0.3},
    \end{gather}
    then there exists $\ottnt{C'}$ such that $  \ottsym{(}    \ottkw{dom} ( \ottnt{C'} )  \setminus  \ottkw{dom} ( \ottnt{C} )    \ottsym{)} \cap  \ottkw{dom} ( \ottnt{C''} )   \ottsym{=} \emptyset $, $\ottnt{C'}  \vdash  \mathcal{E}  \ottsym{[}  \ottnt{M'}  \ottsym{]}  :  \ottnt{T}  \mid  \ottnt{e}$, and $\ottnt{C'}  \ottsym{,}  \ottnt{C''}  \vdash  \mathcal{H}'$.

    \proof
    By structural induction on $\mathcal{E}$.
    \begin{match}
        \item[$ \mathcal{E} \ottsym{=} \ottsym{[]} $]
        This case directly follows by \propref{subj/cmp/effect}.

        \item[$ \mathcal{E} \ottsym{=} \mathcal{E}_{{\mathrm{1}}} \, \ottnt{M_{{\mathrm{2}}}} $]
        In this case, we have
        \begin{gather}
            \ottnt{C}  \lesssim  \ottnt{C_{{\mathrm{1}}}}  \ottsym{,}  \ottnt{C_{{\mathrm{2}}}} \hyp{1.1}\\
            \ottnt{C_{{\mathrm{1}}}}  \vdash  \mathcal{E}_{{\mathrm{1}}}  \ottsym{[}  \ottnt{M}  \ottsym{]}  :   \ottnt{S} \rightarrow _{ \ottnt{e_{{\mathrm{0}}}} } \ottnt{T}   \mid  \ottnt{e_{{\mathrm{1}}}} \hyp{1.2}\\
            \ottnt{C_{{\mathrm{2}}}}  \vdash  \ottnt{M_{{\mathrm{2}}}}  :  \ottnt{S}  \mid  \ottnt{e_{{\mathrm{2}}}} \hyp{1.3}\\
               \ottnt{e_{{\mathrm{0}}}}  \sqcup  \ottnt{e_{{\mathrm{1}}}}   \sqcup  \ottnt{e_{{\mathrm{2}}}}  \le \ottnt{e}  \hyp{1.4}
        \end{gather}
        by \propref{typing/inv}, \propref{env/sub/runtime}, and \hypref{0.2}.
        We have $\ottnt{C}  \ottsym{,}  \ottnt{C''}  \lesssim  \ottnt{C_{{\mathrm{1}}}}  \ottsym{,}  \ottnt{C_{{\mathrm{2}}}}  \ottsym{,}  \ottnt{C''}$ by \propref{env/sub/ctx} and \hypref{1.1}.
        So,
        \begin{gather*}
            \ottnt{C_{{\mathrm{1}}}}  \ottsym{,}  \ottnt{C_{{\mathrm{2}}}}  \ottsym{,}  \ottnt{C''}  \vdash  \mathcal{H} \hyp{1.4.1}
        \end{gather*}
        by \propref{heap/weaken} and \hypref{0.3}.
        Now, we can get some $\ottnt{C'_{{\mathrm{1}}}}$ such that
        \begin{gather}
              \ottsym{(}    \ottkw{dom} ( \ottnt{C'_{{\mathrm{1}}}} )  \setminus  \ottkw{dom} ( \ottnt{C_{{\mathrm{1}}}} )    \ottsym{)} \cap  \ottkw{dom} ( \ottnt{C_{{\mathrm{2}}}}  \ottsym{,}  \ottnt{C''} )   \ottsym{=} \emptyset  \hyp{1.4.2}\\
            \ottnt{C'_{{\mathrm{1}}}}  \vdash  \mathcal{E}_{{\mathrm{1}}}  \ottsym{[}  \ottnt{M'}  \ottsym{]}  :   \ottnt{S} \rightarrow _{ \ottnt{e_{{\mathrm{0}}}} } \ottnt{T}   \mid  \ottnt{e_{{\mathrm{1}}}} \hyp{1.5}\\
            \ottnt{C'_{{\mathrm{1}}}}  \ottsym{,}  \ottnt{C_{{\mathrm{2}}}}  \ottsym{,}  \ottnt{C''}  \vdash  \mathcal{H}' \hyp{1.6}
        \end{gather}
        by applying the induction hypothesis to \hypref{1.2} and \hypref{1.4.1}.
        To this end, we choose $ \ottnt{C'} \ottsym{=} \ottnt{C'_{{\mathrm{1}}}}  \ottsym{,}  \ottnt{C_{{\mathrm{2}}}} $.
        The subgoals we need to show are
        \begin{gather*}
              \ottsym{(}    \ottkw{dom} ( \ottnt{C'_{{\mathrm{1}}}}  \ottsym{,}  \ottnt{C_{{\mathrm{2}}}} )  \setminus  \ottkw{dom} ( \ottnt{C} )    \ottsym{)} \cap  \ottkw{dom} ( \ottnt{C''} )   \ottsym{=} \emptyset \\
            \ottnt{C'_{{\mathrm{1}}}}  \ottsym{,}  \ottnt{C_{{\mathrm{2}}}}  \vdash  \mathcal{E}_{{\mathrm{1}}}  \ottsym{[}  \ottnt{M'}  \ottsym{]} \, \ottnt{M_{{\mathrm{2}}}}  :  \ottnt{T}  \mid  \ottnt{e}\\
            \ottnt{C'_{{\mathrm{1}}}}  \ottsym{,}  \ottnt{C_{{\mathrm{2}}}}  \ottsym{,}  \ottnt{C''}  \vdash  \mathcal{H}'.
        \end{gather*}
        Having $  \ottkw{dom} ( \ottnt{C} )  \ottsym{=}  \ottkw{dom} ( \ottnt{C_{{\mathrm{1}}}}  \ottsym{,}  \ottnt{C_{{\mathrm{2}}}} )  $ by \propref{env/sub/runtime} and \hypref{1.1}, we can see the first subgoal follows by \hypref{1.4.2}.
        We can derive the second subgoal by \ruleref{T-App}, \ruleref{T-Weaken}, \hypref{1.3}, \hypref{1.4}, and \hypref{1.5}.
        The third subgoal is identical to \hypref{1.6}.

        \item[$ \mathcal{E} \ottsym{=} \ottnt{V_{{\mathrm{1}}}} \, \mathcal{E}_{{\mathrm{2}}} $]
        In this case, we have
        \begin{gather}
            \ottnt{C}  \lesssim  \ottnt{C_{{\mathrm{1}}}}  \ottsym{,}  \ottnt{C_{{\mathrm{2}}}} \hyp{2.1}\\
            \ottnt{C_{{\mathrm{1}}}}  \vdash  \ottnt{V_{{\mathrm{1}}}}  :   \ottnt{S} \rightarrow _{ \ottnt{e_{{\mathrm{0}}}} } \ottnt{T}   \mid  \ottnt{e_{{\mathrm{1}}}} \hyp{2.2}\\
            \ottnt{C_{{\mathrm{2}}}}  \vdash  \mathcal{E}_{{\mathrm{2}}}  \ottsym{[}  \ottnt{M}  \ottsym{]}  :  \ottnt{S}  \mid  \ottnt{e_{{\mathrm{2}}}} \hyp{2.3}\\
               \ottnt{e_{{\mathrm{0}}}}  \sqcup  \ottnt{e_{{\mathrm{1}}}}   \sqcup  \ottnt{e_{{\mathrm{2}}}}  \le \ottnt{e}  \hyp{2.4}.
        \end{gather}
        by \propref{typing/inv}, \propref{env/sub/runtime}, and \hypref{0.2}.
        We have
        \begin{gather}
            \ottnt{C_{{\mathrm{1}}}}  \lesssim   \cdot  \hyp{2.4.1}\\
             \cdot   \vdash  \ottnt{V_{{\mathrm{1}}}}  :   \ottnt{S} \rightarrow _{ \ottnt{e_{{\mathrm{0}}}} } \ottnt{T}   \mid  \ottsym{0} \hyp{2.4.2}
        \end{gather}
        by \propref{typing/value/unr} and \hypref{2.2}.
        So, we have $\ottnt{C}  \ottsym{,}  \ottnt{C''}  \lesssim  \ottnt{C_{{\mathrm{2}}}}  \ottsym{,}  \ottnt{C''}$ by \propref{env/sub/ctx} and \hypref{2.1}, and then
        \begin{gather*}
            \ottnt{C_{{\mathrm{2}}}}  \ottsym{,}  \ottnt{C''}  \vdash  \mathcal{H} \hyp{2.5}
        \end{gather*}
        by \propref{heap/weaken} and \hypref{0.3}.
        Now we can have some $\ottnt{C'_{{\mathrm{2}}}}$ such that
        \begin{gather*}
              \ottsym{(}    \ottkw{dom} ( \ottnt{C'_{{\mathrm{2}}}} )  \setminus  \ottkw{dom} ( \ottnt{C_{{\mathrm{2}}}} )    \ottsym{)} \cap  \ottkw{dom} ( \ottnt{C''} )   \ottsym{=} \emptyset  \hyp{2.5.1}\\
            \ottnt{C'_{{\mathrm{2}}}}  \vdash  \mathcal{E}_{{\mathrm{2}}}  \ottsym{[}  \ottnt{M'}  \ottsym{]}  :  \ottnt{S}  \mid  \ottnt{e_{{\mathrm{2}}}} \hyp{2.6}\\
            \ottnt{C'_{{\mathrm{2}}}}  \ottsym{,}  \ottnt{C''}  \vdash  \mathcal{H}' \hyp{2.7}
        \end{gather*}
        by applyin the induction hypothesis to \hypref{2.3} and \hypref{2.5}.
        To this end, we choose $ \ottnt{C'} \ottsym{=} \ottnt{C'_{{\mathrm{2}}}} $.
        The subgoals we need to show are
        \begin{gather*}
              \ottsym{(}    \ottkw{dom} ( \ottnt{C'_{{\mathrm{2}}}} )  \setminus  \ottkw{dom} ( \ottnt{C} )    \ottsym{)} \cap  \ottkw{dom} ( \ottnt{C''} )   \ottsym{=} \emptyset \\
            \ottnt{C'_{{\mathrm{2}}}}  \vdash  \ottnt{V_{{\mathrm{1}}}} \, \mathcal{E}_{{\mathrm{2}}}  \ottsym{[}  \ottnt{M'}  \ottsym{]}  :  \ottnt{T}  \mid  \ottnt{e} \\
            \ottnt{C'_{{\mathrm{2}}}}  \ottsym{,}  \ottnt{C''}  \vdash  \mathcal{H}'.
        \end{gather*}
        Having $  \ottkw{dom} ( \ottnt{C} )  \ottsym{=}  \ottkw{dom} ( \ottnt{C_{{\mathrm{2}}}} )  $ by \propref{env/sub/runtime} and \hypref{2.1}, we can see the first subgoal is identical to \hypref{2.5.1}.
        We can derive the second subgoal by \ruleref{T-App}, \ruleref{T-Weaken}, \hypref{2.4.2}, \hypref{2.4}, and \hypref{2.6}.
        The third subgoal is identical to \hypref{2.7}.

        \item[$ \mathcal{E} \ottsym{=} \mathcal{E}_{{\mathrm{1}}}  {}^\circ  \ottnt{M_{{\mathrm{2}}}} $]
        In this case, we have
        \begin{gather}
            \ottnt{C}  \lesssim  \ottnt{C_{{\mathrm{1}}}}  \parallel  \ottnt{C_{{\mathrm{2}}}} \hyp{3.1}\\
            \ottnt{C_{{\mathrm{1}}}}  \vdash  \mathcal{E}_{{\mathrm{1}}}  \ottsym{[}  \ottnt{M}  \ottsym{]}  :   \ottnt{S} \rightarrowtriangle _{ \ottnt{e_{{\mathrm{0}}}} } \ottnt{T}   \mid  \ottnt{e_{{\mathrm{1}}}} \hyp{3.2}\\
            \ottnt{C_{{\mathrm{2}}}}  \vdash  \ottnt{M_{{\mathrm{2}}}}  :  \ottnt{S}  \mid  \ottnt{e_{{\mathrm{2}}}} \hyp{3.3}\\
               \ottnt{e_{{\mathrm{0}}}}  \sqcup  \ottnt{e_{{\mathrm{1}}}}   \sqcup  \ottnt{e_{{\mathrm{2}}}}  \le \ottnt{e}  \hyp{3.4}
        \end{gather}
        by \propref{typing/inv}, \propref{env/sub/runtime}, and \hypref{0.2}.
        We have $\ottnt{C}  \ottsym{,}  \ottnt{C''}  \lesssim  \ottsym{(}  \ottnt{C_{{\mathrm{1}}}}  \parallel  \ottnt{C_{{\mathrm{2}}}}  \ottsym{)}  \ottsym{,}  \ottnt{C''}$ by \propref{env/sub/ctx} and \hypref{3.1}.
        So,
        \begin{gather*}
            \ottsym{(}  \ottnt{C_{{\mathrm{1}}}}  \parallel  \ottnt{C_{{\mathrm{2}}}}  \ottsym{)}  \ottsym{,}  \ottnt{C''}  \vdash  \mathcal{H}
        \end{gather*}
        by \propref{heap/weaken} and \hypref{0.3}.
        From that, we have
        \begin{gather*}
               \ottkw{dom} ( \ottnt{C_{{\mathrm{1}}}} )  \cap  \ottkw{dom} ( \ottnt{C_{{\mathrm{2}}}} )   \ottsym{=} \emptyset  \hyp{3.5.1}
        \end{gather*}
        by \propref{heap/par} and
        \begin{gather*}
            \ottnt{C_{{\mathrm{1}}}}  \ottsym{,}  \ottnt{C_{{\mathrm{2}}}}  \ottsym{,}  \ottnt{C''}  \vdash  \mathcal{H} \hyp{3.5}
        \end{gather*}
        by \propref{heap/weaken} since $\ottsym{(}  \ottnt{C_{{\mathrm{1}}}}  \parallel  \ottnt{C_{{\mathrm{2}}}}  \ottsym{)}  \ottsym{,}  \ottnt{C''}  \lesssim  \ottnt{C_{{\mathrm{1}}}}  \ottsym{,}  \ottnt{C_{{\mathrm{2}}}}  \ottsym{,}  \ottnt{C''}$.
        Now we can have some $\ottnt{C'_{{\mathrm{1}}}}$ such that
        \begin{gather*}
              \ottsym{(}    \ottkw{dom} ( \ottnt{C'_{{\mathrm{1}}}} )  \setminus  \ottkw{dom} ( \ottnt{C_{{\mathrm{1}}}} )    \ottsym{)} \cap  \ottkw{dom} ( \ottnt{C_{{\mathrm{2}}}}  \ottsym{,}  \ottnt{C''} )   \ottsym{=} \emptyset  \hyp{3.5.2}\\
            \ottnt{C'_{{\mathrm{1}}}}  \vdash  \mathcal{E}_{{\mathrm{1}}}  \ottsym{[}  \ottnt{M'}  \ottsym{]}  :   \ottnt{S} \rightarrowtriangle _{ \ottnt{e_{{\mathrm{0}}}} } \ottnt{T}   \mid  \ottnt{e_{{\mathrm{1}}}} \hyp{3.6}\\
            \ottnt{C'_{{\mathrm{1}}}}  \ottsym{,}  \ottnt{C_{{\mathrm{2}}}}  \ottsym{,}  \ottnt{C''}  \vdash  \mathcal{H}' \hyp{3.7}
        \end{gather*}
        by applying the induction hypothesis to \hypref{3.2} and \hypref{3.5}.
        To this end, we choose $ \ottnt{C'} \ottsym{=} \ottnt{C'_{{\mathrm{1}}}}  \parallel  \ottnt{C_{{\mathrm{2}}}} $.
        The subgoals we need to show are
        \begin{gather*}
              \ottsym{(}    \ottkw{dom} ( \ottnt{C'_{{\mathrm{1}}}}  \parallel  \ottnt{C_{{\mathrm{2}}}} )  \setminus  \ottkw{dom} ( \ottnt{C} )    \ottsym{)} \cap  \ottkw{dom} ( \ottnt{C''} )   \ottsym{=} \emptyset \\
            \ottnt{C'_{{\mathrm{1}}}}  \parallel  \ottnt{C_{{\mathrm{2}}}}  \vdash  \mathcal{E}_{{\mathrm{1}}}  \ottsym{[}  \ottnt{M'}  \ottsym{]}  {}^\circ  \ottnt{M_{{\mathrm{2}}}}  :  \ottnt{T}  \mid  \ottnt{e} \\
            \ottsym{(}  \ottnt{C'_{{\mathrm{1}}}}  \parallel  \ottnt{C_{{\mathrm{2}}}}  \ottsym{)}  \ottsym{,}  \ottnt{C''}  \vdash  \mathcal{H}'.
        \end{gather*}
        Having $  \ottkw{dom} ( \ottnt{C} )  \ottsym{=}  \ottkw{dom} ( \ottnt{C_{{\mathrm{1}}}}  \parallel  \ottnt{C_{{\mathrm{2}}}} )  $ by \propref{env/sub/runtime} and \hypref{3.1}, we can see the first subgoal follows by \hypref{3.5.2}.
        We can derive the second subgoal by \ruleref{T-UApp}, \ruleref{T-Weaken}, \hypref{3.3}, \hypref{3.4}, and \hypref{3.6}.
        As for the third subgoal, firstly we have $  \ottsym{(}    \ottkw{dom} ( \ottnt{C'_{{\mathrm{1}}}} )  \setminus  \ottkw{dom} ( \ottnt{C_{{\mathrm{1}}}} )    \ottsym{)} \cap  \ottkw{dom} ( \ottnt{C_{{\mathrm{2}}}} )   \ottsym{=} \emptyset $ by \hypref{3.5.2}.
Then, $   \ottkw{dom} ( \ottnt{C'_{{\mathrm{1}}}} )  \cap  \ottkw{dom} ( \ottnt{C_{{\mathrm{2}}}} )   \ottsym{=} \emptyset $ follows by \hypref{3.5.1}.
        Now, the subgoal follows by \propref{heap/par2} and \hypref{3.7}.

        \item[$ \mathcal{E} \ottsym{=} \ottnt{V_{{\mathrm{1}}}}  {}^\circ  \mathcal{E}_{{\mathrm{2}}} $]
        Similar to the case $ \mathcal{E} \ottsym{=} \mathcal{E}_{{\mathrm{1}}}  {}^\circ  \ottnt{M_{{\mathrm{2}}}} $ since $ \parallel $ is commutative.

        \item[$ \mathcal{E} \ottsym{=} \mathcal{E}_{{\mathrm{1}}}  {}^>  \ottnt{M_{{\mathrm{2}}}} $]
        In this case, we have
        \begin{gather}
            \ottnt{C}  \lesssim  \ottnt{C_{{\mathrm{1}}}}  \ottsym{,}  \ottnt{C_{{\mathrm{2}}}} \hyp{5.1}\\
            \ottnt{C_{{\mathrm{1}}}}  \vdash  \mathcal{E}_{{\mathrm{1}}}  \ottsym{[}  \ottnt{M}  \ottsym{]}  :   \ottnt{S} \twoheadrightarrow _{ \ottnt{e_{{\mathrm{0}}}} } \ottnt{T}   \mid  \ottnt{e_{{\mathrm{1}}}} \hyp{5.2}\\
            \ottnt{C_{{\mathrm{2}}}}  \vdash  \ottnt{M_{{\mathrm{2}}}}  :  \ottnt{S}  \mid  \ottsym{0} \hyp{5.3}\\
              \ottnt{e_{{\mathrm{0}}}}  \sqcup  \ottnt{e_{{\mathrm{1}}}}  \le \ottnt{e}  \hyp{5.4}
        \end{gather}
        by \propref{typing/inv}, \propref{env/sub/runtime}, and \hypref{0.2}.
        We have $\ottnt{C}  \ottsym{,}  \ottnt{C''}  \lesssim  \ottnt{C_{{\mathrm{1}}}}  \ottsym{,}  \ottnt{C_{{\mathrm{2}}}}  \ottsym{,}  \ottnt{C''}$ by \propref{env/sub/ctx} and \hypref{5.1}.
        So,
        \begin{gather}
            \ottnt{C_{{\mathrm{1}}}}  \ottsym{,}  \ottnt{C_{{\mathrm{2}}}}  \ottsym{,}  \ottnt{C''}  \vdash  \mathcal{H} \hyp{5.5}
        \end{gather}
        by \propref{heap/weaken} and \hypref{0.3}.
        Now we can have some $\ottnt{C'_{{\mathrm{1}}}}$ such that
        \begin{gather}
              \ottsym{(}    \ottkw{dom} ( \ottnt{C'_{{\mathrm{1}}}} )  \setminus  \ottkw{dom} ( \ottnt{C_{{\mathrm{1}}}} )    \ottsym{)} \cap  \ottkw{dom} ( \ottnt{C_{{\mathrm{2}}}}  \ottsym{,}  \ottnt{C''} )   \ottsym{=} \emptyset  \hyp{5.6}\\
            \ottnt{C'_{{\mathrm{1}}}}  \vdash  \mathcal{E}_{{\mathrm{1}}}  \ottsym{[}  \ottnt{M'}  \ottsym{]}  :   \ottnt{S} \twoheadrightarrow _{ \ottnt{e_{{\mathrm{0}}}} } \ottnt{T}   \mid  \ottnt{e_{{\mathrm{1}}}} \hyp{5.7}\\
            \ottnt{C'_{{\mathrm{1}}}}  \ottsym{,}  \ottnt{C_{{\mathrm{2}}}}  \ottsym{,}  \ottnt{C''}  \vdash  \mathcal{H}' \hyp{5.8}
        \end{gather}
        by applying the induction hypothesis to \hypref{5.2} and \hypref{5.5}.
        To this end, we choose $ \ottnt{C'} \ottsym{=} \ottnt{C'_{{\mathrm{1}}}}  \ottsym{,}  \ottnt{C_{{\mathrm{2}}}} $.
        The subgoals we need to show are
        \begin{gather*}
              \ottsym{(}    \ottkw{dom} ( \ottnt{C'_{{\mathrm{1}}}}  \ottsym{,}  \ottnt{C_{{\mathrm{2}}}} )  \setminus  \ottkw{dom} ( \ottnt{C} )    \ottsym{)} \cap  \ottkw{dom} ( \ottnt{C''} )   \ottsym{=} \emptyset  \\
            \ottnt{C'_{{\mathrm{1}}}}  \ottsym{,}  \ottnt{C_{{\mathrm{2}}}}  \vdash  \mathcal{E}_{{\mathrm{1}}}  \ottsym{[}  \ottnt{M'}  \ottsym{]}  {}^>  \ottnt{M_{{\mathrm{2}}}}  :  \ottnt{T}  \mid  \ottnt{e} \\
            \ottnt{C'_{{\mathrm{1}}}}  \ottsym{,}  \ottnt{C_{{\mathrm{2}}}}  \ottsym{,}  \ottnt{C''}  \vdash  \mathcal{H}'.
        \end{gather*}
        Having $  \ottkw{dom} ( \ottnt{C} )  \ottsym{=}  \ottkw{dom} ( \ottnt{C_{{\mathrm{1}}}}  \ottsym{,}  \ottnt{C_{{\mathrm{2}}}} )  $ by \propref{env/sub/runtime} and \hypref{5.1}, we can see the first subgoal follows by \hypref{5.6}.
        We can derive the second subgoal by \ruleref{T-RApp}, \ruleref{T-Weaken}, \hypref{5.3}, \hypref{5.4}, and \hypref{5.7}.
        The third subgoal is identical to \hypref{5.8}.

        \item[$ \mathcal{E} \ottsym{=} \ottnt{V_{{\mathrm{1}}}}  {}^>  \mathcal{E}_{{\mathrm{2}}} $]
        In this case, we have
        \begin{gather}
            \ottnt{C}  \lesssim  \ottnt{C_{{\mathrm{1}}}}  \ottsym{,}  \ottnt{C_{{\mathrm{2}}}} \hyp{6.1}\\
            \ottnt{C_{{\mathrm{1}}}}  \vdash  \ottnt{V_{{\mathrm{1}}}}  :   \ottnt{S} \twoheadrightarrow _{ \ottnt{e_{{\mathrm{0}}}} } \ottnt{T}   \mid  \ottnt{e_{{\mathrm{1}}}} \hyp{6.2}\\
            \ottnt{C_{{\mathrm{2}}}}  \vdash  \mathcal{E}_{{\mathrm{2}}}  \ottsym{[}  \ottnt{M}  \ottsym{]}  :  \ottnt{S}  \mid  \ottsym{0} \hyp{6.3}\\
              \ottnt{e_{{\mathrm{0}}}}  \sqcup  \ottnt{e_{{\mathrm{1}}}}  \le \ottnt{e}  \hyp{6.4}
        \end{gather}
        by \propref{typing/inv}, \propref{env/sub/runtime}, and \hypref{0.2}.
        We have $\ottnt{C}  \ottsym{,}  \ottnt{C''}  \lesssim  \ottnt{C_{{\mathrm{1}}}}  \ottsym{,}  \ottnt{C_{{\mathrm{2}}}}  \ottsym{,}  \ottnt{C''}$ by \propref{env/sub/ctx} and \hypref{6.1}.
        So,
        \begin{gather}
            \ottnt{C_{{\mathrm{1}}}}  \ottsym{,}  \ottnt{C_{{\mathrm{2}}}}  \ottsym{,}  \ottnt{C''}  \vdash  \mathcal{H} \hyp{6.5}
        \end{gather}
        by \propref{heap/weaken} and \hypref{0.3}.
        Here we can have some $\ottnt{C'_{{\mathrm{2}}}}$ such that
        \begin{gather}
              \ottsym{(}    \ottkw{dom} ( \ottnt{C'_{{\mathrm{2}}}} )  \setminus  \ottkw{dom} ( \ottnt{C_{{\mathrm{2}}}} )    \ottsym{)} \cap  \ottkw{dom} ( \ottnt{C_{{\mathrm{1}}}}  \ottsym{,}   \cdot   \ottsym{,}  \ottnt{C''} )   \ottsym{=} \emptyset  \hyp{6.6}\\
            \ottnt{C'_{{\mathrm{2}}}}  \vdash  \mathcal{E}_{{\mathrm{2}}}  \ottsym{[}  \ottnt{M'}  \ottsym{]}  :  \ottnt{S}  \mid  \ottsym{0} \hyp{6.7}\\
            \ottnt{C_{{\mathrm{1}}}}  \ottsym{,}  \ottnt{C'_{{\mathrm{2}}}}  \ottsym{,}  \ottnt{C''}  \vdash  \mathcal{H}' \hyp{6.8}
        \end{gather}
        by applying \propref{subj/cmp/noeffect/ctx} to \hypref{0.1}, \hypref{6.3}, and \hypref{6.5}.
        To this end, we choose $ \ottnt{C'} \ottsym{=} \ottnt{C_{{\mathrm{1}}}}  \ottsym{,}  \ottnt{C'_{{\mathrm{2}}}} $.
        The subgoals we need to show are
        \begin{gather}
              \ottsym{(}    \ottkw{dom} ( \ottnt{C_{{\mathrm{1}}}}  \ottsym{,}  \ottnt{C'_{{\mathrm{2}}}} )  \setminus  \ottkw{dom} ( \ottnt{C} )    \ottsym{)} \cap  \ottkw{dom} ( \ottnt{C''} )   \ottsym{=} \emptyset  \\
            \ottnt{C_{{\mathrm{1}}}}  \ottsym{,}  \ottnt{C'_{{\mathrm{2}}}}  \vdash  \ottnt{V_{{\mathrm{1}}}}  {}^>  \mathcal{E}_{{\mathrm{2}}}  \ottsym{[}  \ottnt{M'}  \ottsym{]}  :  \ottnt{T}  \mid  \ottnt{e} \\
            \ottnt{C_{{\mathrm{1}}}}  \ottsym{,}  \ottnt{C'_{{\mathrm{2}}}}  \ottsym{,}  \ottnt{C''}  \vdash  \mathcal{H}'.
        \end{gather}
        Having $  \ottkw{dom} ( \ottnt{C} )  \ottsym{=}  \ottkw{dom} ( \ottnt{C_{{\mathrm{1}}}}  \ottsym{,}  \ottnt{C_{{\mathrm{2}}}} )  $ by \propref{env/sub/runtime} and \hypref{6.1}, we can see the first subgoal follows by \hypref{6.6}.
        We can derive th second subgoal by \ruleref{T-RApp}, \ruleref{T-Weaken}, \hypref{6.2}, \hypref{6.4}, and \hypref{6.7}.
        The third subgoal is identical to \hypref{6.8}.

        \item[$ \mathcal{E} \ottsym{=} \mathcal{E}_{{\mathrm{1}}}  {}^<  \ottnt{V_{{\mathrm{2}}}} $] Similar to the case $\ottnt{V_{{\mathrm{1}}}}  {}^>  \mathcal{E}_{{\mathrm{2}}}$.
        \item[$ \mathcal{E} \ottsym{=} \ottnt{M_{{\mathrm{1}}}}  {}^<  \mathcal{E}_{{\mathrm{2}}} $] Similar to the case $\mathcal{E}_{{\mathrm{1}}}  {}^>  \ottnt{M_{{\mathrm{2}}}}$.
        \item[$ \mathcal{E} \ottsym{=} \mathcal{E}_{{\mathrm{1}}}  \otimes  \ottnt{M_{{\mathrm{2}}}} $] Similar to the case $\mathcal{E}_{{\mathrm{1}}}  {}^\circ  \ottnt{M_{{\mathrm{2}}}}$.
        \item[$ \mathcal{E} \ottsym{=} \ottnt{V_{{\mathrm{1}}}}  \otimes  \mathcal{E}_{{\mathrm{2}}} $] Similar to the case $\ottnt{V_{{\mathrm{1}}}}  {}^\circ  \ottnt{M_{{\mathrm{2}}}}$.
        \item[$ \mathcal{E} \ottsym{=} \mathcal{E}_{{\mathrm{1}}}  \odot  \ottnt{M_{{\mathrm{2}}}} $] Similar to the case $\mathcal{E}_{{\mathrm{1}}}  {}^>  \ottnt{M_{{\mathrm{2}}}}$.
        \item[$ \mathcal{E} \ottsym{=} \ottnt{V_{{\mathrm{1}}}}  \odot  \mathcal{E}_{{\mathrm{2}}} $] Similar to the case $\ottnt{V_{{\mathrm{1}}}} \, \mathcal{E}_{{\mathrm{2}}}$ if $\ottnt{V_{{\mathrm{1}}}}$ has an unrestricted type.
        Otherwise, similar to the case $\ottnt{V_{{\mathrm{1}}}}  {}^>  \mathcal{E}_{{\mathrm{2}}}$.

        \item[$ \mathcal{E} \ottsym{=} \ottkw{let} \, \ottmv{x}  \otimes  \ottmv{y}  \ottsym{=}  \mathcal{E}_{{\mathrm{1}}} \, \ottkw{in} \, \ottnt{M_{{\mathrm{2}}}} $]
        In this case, we have
        \begin{gather}
            \ottnt{C}  \lesssim  \mathcal{C}  \ottsym{[}  \ottnt{C_{{\mathrm{1}}}}  \ottsym{]} \hyp{13.1}\\
             \ottnt{e_{{\mathrm{2}}}} \le \ottnt{e}  \hyp{13.2}\\
            \ottsym{\{}  \ottmv{x}  \ottsym{\}}  \uplus  \ottsym{\{}  \ottmv{y}  \ottsym{\}}  \uplus   \ottkw{dom} ( \mathcal{C}  \ottsym{[}  \ottnt{C_{{\mathrm{1}}}}  \ottsym{]} )  \hyp{13.3}\\
            \ottnt{C_{{\mathrm{1}}}}  \vdash  \mathcal{E}_{{\mathrm{1}}}  \ottsym{[}  \ottnt{M}  \ottsym{]}  :  \ottnt{S_{{\mathrm{1}}}}  \otimes  \ottnt{S_{{\mathrm{2}}}}  \mid  \ottsym{0} \hyp{13.4}\\
            \mathcal{C}  \ottsym{[}  \ottmv{x}  \mathord:  \ottnt{S_{{\mathrm{1}}}}  \parallel  \ottmv{y}  \mathord:  \ottnt{S_{{\mathrm{2}}}}  \ottsym{]}  \vdash  \ottnt{M_{{\mathrm{2}}}}  :  \ottnt{T}  \mid  \ottnt{e_{{\mathrm{2}}}} \hyp{13.5}
        \end{gather}
        by \propref{typing/inv}, \propref{env/sub/runtime}, and \hypref{0.2}.
        We have $\ottnt{C}  \ottsym{,}  \ottnt{C''}  \lesssim  \mathcal{C}  \ottsym{[}  \ottnt{C_{{\mathrm{1}}}}  \ottsym{]}  \ottsym{,}  \ottnt{C''}$ by \propref{env/sub/ctx} and \hypref{13.1}.
        So,
        \begin{gather}
            \mathcal{C}  \ottsym{[}  \ottnt{C_{{\mathrm{1}}}}  \ottsym{]}  \ottsym{,}  \ottnt{C''}  \vdash  \mathcal{H} \hyp{13.6}
        \end{gather}
        by \propref{heap/weaken} and \hypref{13.1}.
        Now, we have some $\ottnt{C'_{{\mathrm{1}}}}$ such that
        \begin{gather}
              \ottsym{(}    \ottkw{dom} ( \ottnt{C'_{{\mathrm{1}}}} )  \setminus  \ottkw{dom} ( \ottnt{C_{{\mathrm{1}}}} )    \ottsym{)} \cap  \ottkw{dom} ( \mathcal{C}  \ottsym{,}  \ottnt{C''} )   \ottsym{=} \emptyset  \hyp{13.7}\\
            \ottnt{C'_{{\mathrm{1}}}}  \vdash  \mathcal{E}_{{\mathrm{1}}}  \ottsym{[}  \ottnt{M'}  \ottsym{]}  :  \ottnt{S_{{\mathrm{1}}}}  \otimes  \ottnt{S_{{\mathrm{2}}}}  \mid  \ottsym{0} \hyp{13.8}\\
            \mathcal{C}  \ottsym{[}  \ottnt{C'_{{\mathrm{1}}}}  \ottsym{]}  \ottsym{,}  \ottnt{C''}  \vdash  \mathcal{H}' \hyp{13.9}
        \end{gather}
        by \propref{subj/cmp/noeffect/ctx}, \hypref{0.1}, \hypref{13.4}, and \hypref{13.6}.
        To this end, we choose $ \ottnt{C'} \ottsym{=} \mathcal{C}  \ottsym{[}  \ottnt{C'_{{\mathrm{1}}}}  \ottsym{]} $.
        The subgoals we need to show are
        \begin{gather*}
             \ottsym{(}    \ottkw{dom} ( \mathcal{C}  \ottsym{[}  \ottnt{C'_{{\mathrm{1}}}}  \ottsym{]} )  \setminus  \ottkw{dom} ( \ottnt{C} )    \ottsym{)} \cap  \ottkw{dom} ( \ottnt{C''} )   \\
            \mathcal{C}  \ottsym{[}  \ottnt{C'_{{\mathrm{1}}}}  \ottsym{]}  \vdash  \ottkw{let} \, \ottmv{x}  \otimes  \ottmv{y}  \ottsym{=}  \mathcal{E}_{{\mathrm{1}}}  \ottsym{[}  \ottnt{M'}  \ottsym{]} \, \ottkw{in} \, \ottnt{M_{{\mathrm{2}}}}  :  \ottnt{T}  \mid  \ottnt{e} \\
            \mathcal{C}  \ottsym{[}  \ottnt{C'_{{\mathrm{1}}}}  \ottsym{]}  \ottsym{,}  \ottnt{C''}  \vdash  \mathcal{H}'.
        \end{gather*}
        Having $  \ottkw{dom} ( \ottnt{C} )  \ottsym{=}  \ottkw{dom} ( \mathcal{C}  \ottsym{[}  \ottnt{C_{{\mathrm{1}}}}  \ottsym{]} )  $ by \propref{env/sub/runtime} and \hypref{13.1}, we can see the first subgoal follows by \hypref{13.7}.
        We have $\ottsym{\{}  \ottmv{x}  \ottsym{\}}  \uplus  \ottsym{\{}  \ottmv{y}  \ottsym{\}}  \uplus   \ottkw{dom} ( \mathcal{C}  \ottsym{[}  \ottnt{C'_{{\mathrm{1}}}}  \ottsym{]} ) $ by \hypref{13.3} since $ \ottkw{dom} ( \mathcal{C}  \ottsym{[}  \ottnt{C'_{{\mathrm{1}}}}  \ottsym{]} ) $ has only locations.
        Therefore, we can derive the second subgoal by \ruleref{T-ULet}, \ruleref{T-Weaken}, \hypref{13.2}, \hypref{13.5}, and \hypref{13.8}.
        The third subgoal is identical to \hypref{13.9}.

        \item[$ \mathcal{E} \ottsym{=} \ottkw{let} \, \ottmv{x}  \odot  \ottmv{y}  \ottsym{=}  \mathcal{E}_{{\mathrm{1}}} \, \ottkw{in} \, \ottnt{M_{{\mathrm{2}}}} $]
        Similar to the case $ \mathcal{E} \ottsym{=} \ottkw{let} \, \ottmv{x}  \otimes  \ottmv{y}  \ottsym{=}  \mathcal{E}_{{\mathrm{1}}} \, \ottkw{in} \, \ottnt{M_{{\mathrm{2}}}} $.
    \end{match}
\end{prop}

\begin{prop}{subj}
    If
    \begin{gather}
        \ottnt{M}  \mid  \mathcal{H}  \longrightarrow  \ottnt{M'}  \mid  \mathcal{H}' \hyp{0.1}\\
        \ottnt{C}  \vdash  \ottnt{M}  :  \ottnt{T}  \mid  \ottnt{e} \hyp{0.2}\\
        \ottnt{C}  \vdash  \mathcal{H} \hyp{0.3},
    \end{gather}
    then there exists $\ottnt{C'}$ such that $\ottnt{C'}  \vdash  \ottnt{M'}  :  \ottnt{T}  \mid  \ottnt{e}$ and $\ottnt{C'}  \vdash  \mathcal{H}'$.

    \proof By case analysis on the last rule of the given derivation of \hypref{0.1}.
    \begin{match}
        \item[\ruleref{R-Exp}]
        In this case,
        \begin{gather}
             \ottnt{M} \ottsym{=} \mathcal{E}  \ottsym{[}  \ottnt{M_{{\mathrm{1}}}}  \ottsym{]}  \hyp{1.1}\\
             \ottnt{M'} \ottsym{=} \mathcal{E}  \ottsym{[}  \ottnt{M'_{{\mathrm{1}}}}  \ottsym{]}  \hyp{1.2}\\
             \mathcal{H} \ottsym{=} \mathcal{H}'  \hyp{1.3}\\
            \ottnt{M_{{\mathrm{1}}}}  \longrightarrow_\beta  \ottnt{M'_{{\mathrm{1}}}} \hyp{1.4}.
        \end{gather}
        So, we have $\ottnt{C}  \vdash  \mathcal{E}  \ottsym{[}  \ottnt{M'_{{\mathrm{1}}}}  \ottsym{]}  :  \ottnt{T}  \mid  \ottnt{e}$ by \propref{subj/beta/ctx}.
        To this end, we choose $ \ottnt{C'} \ottsym{=} \ottnt{C} $, which satisfies the required conditions.

        \item[\ruleref{R-Cfg}]
        In this case,
        \begin{gather}
             \ottnt{M} \ottsym{=} \mathcal{E}  \ottsym{[}  \ottnt{M_{{\mathrm{1}}}}  \ottsym{]}  \hyp{2.1}\\
             \ottnt{M'} \ottsym{=} \mathcal{E}  \ottsym{[}  \ottnt{M'_{{\mathrm{1}}}}  \ottsym{]}  \hyp{2.2}\\
            \ottnt{M_{{\mathrm{1}}}}  \mid  \mathcal{H}  \longrightarrow_\gamma  \ottnt{M'_{{\mathrm{1}}}}  \mid  \mathcal{H}' \hyp{2.3}.
        \end{gather}
        So, we have some $\ottnt{C'}$ such that
        \begin{gather}
            \ottnt{C'}  \vdash  \mathcal{E}  \ottsym{[}  \ottnt{M'_{{\mathrm{1}}}}  \ottsym{]}  :  \ottnt{T}  \mid  \ottnt{e} \hyp{2.4}\\
            \ottnt{C'}  \vdash  \mathcal{H}' \hyp{2.5}
        \end{gather}
        by \propref{subj/cmp/effect/ctx}.
        To this end, we choose $\ottnt{C'}$ as the one obtained above, which has satisfied the required conditions.
    \end{match}
\end{prop}

  \subsection{Properties for progress}
  \begin{prop}{progress/new}
    $ \ottkw{new} _{ \ottnt{m} }  \, \ottkw{unit}  \mid  \mathcal{H}  \longrightarrow  \ottnt{M'}  \mid  \mathcal{H}'$ for some $\ottnt{M'}$ and $\mathcal{H}'$.

    \proof We choose $ \ottnt{M'} \ottsym{=} \ottmv{l} $ such that $ \ottmv{l}  \notin   \ottkw{dom} ( \mathcal{H} )  $ --- this is always possible since $\mathcal{H}$ is a finite map --- and $ \mathcal{H}' \ottsym{=}  \mathcal{H} \cup \ottsym{\{} \ottmv{l}  \mapsto  \ottsym{(}  \ottsym{0}  \ottsym{,}  \ottnt{m}  \ottsym{,}  \varepsilon  \ottsym{)} \ottsym{\}}  $.
    Then, what we need to show is
    \begin{gather*}
         \ottkw{new} _{ \ottnt{m} }  \, \ottkw{unit}  \mid  \mathcal{H}  \longrightarrow  \ottmv{l}  \mid   \mathcal{H} \cup \ottsym{\{} \ottmv{l}  \mapsto  \ottsym{(}  \ottsym{0}  \ottsym{,}  \ottnt{m}  \ottsym{,}  \varepsilon  \ottsym{)} \ottsym{\}} ,
    \end{gather*}
    which can be derived by \ruleref{RC-Ne} and \ruleref{R-Cfg}.
\end{prop}

\begin{prop}{progress/split}
    If $\mathcal{C}  \ottsym{[}  \ottmv{l}  \mathord:  \ottsym{[}  \ottnt{m}  \ottsym{]}  \ottsym{]}  \vdash  \mathcal{H}$, then $ \ottkw{split} _{ \ottnt{m_{{\mathrm{1}}}} , \ottnt{m_{{\mathrm{2}}}} }  \, \ottmv{l}  \mid  \mathcal{H}  \longrightarrow  \ottnt{M'}  \mid  \mathcal{H}'$ for some $\ottnt{M'}$ and $\mathcal{H}'$.

    \proof We have some $ \ottsym{(}  \ottnt{n}  \ottsym{,}  \ottnt{m_{{\mathrm{0}}}}  \ottsym{,}  \ottnt{m'}  \ottsym{)} \ottsym{=} \mathcal{H}  \ottsym{(}  \ottmv{l}  \ottsym{)} $ by the assumption.
    To this end, we choose $ \ottnt{M'} \ottsym{=} \ottmv{l}  \odot  \ottmv{l} $ and $ \mathcal{H}' \ottsym{=}  \mathcal{H} \cup \ottsym{\{} \ottmv{l}  \mapsto  \ottsym{(}  \ottnt{n}  \ottsym{+}  \ottsym{1}  \ottsym{,}  \ottnt{m_{{\mathrm{0}}}}  \ottsym{,}  \ottnt{m'}  \ottsym{)} \ottsym{\}}  $.
    Then, what we need to show is
    \begin{gather}
         \ottkw{split} _{ \ottnt{m_{{\mathrm{1}}}} , \ottnt{m_{{\mathrm{2}}}} }  \, \ottmv{l}  \mid  \mathcal{H}  \longrightarrow  \ottmv{l}  \odot  \ottmv{l}  \mid   \mathcal{H} \cup \ottsym{\{} \ottmv{l}  \mapsto  \ottsym{(}  \ottnt{n}  \ottsym{+}  \ottsym{1}  \ottsym{,}  \ottnt{m_{{\mathrm{0}}}}  \ottsym{,}  \ottnt{m'}  \ottsym{)} \ottsym{\}} .
    \end{gather}
    Since $ \mathcal{H} \ottsym{=}  \mathcal{H} \cup \ottsym{\{} \ottmv{l}  \mapsto  \ottsym{(}  \ottnt{n}  \ottsym{,}  \ottnt{m_{{\mathrm{0}}}}  \ottsym{,}  \ottnt{m'}  \ottsym{)} \ottsym{\}}  $, we can derive that by \ruleref{RC-Sp} and \ruleref{R-Cfg}.
\end{prop}

\begin{prop}{progress/op}
    If
    \begin{gather}
        \ottmv{l}  \mathord:  \ottsym{[}  \ottnt{m}  \ottsym{]}  \ottsym{,}  \ottnt{C}  \vdash  \mathcal{H} \hyp{0.1}\\
          \ottnt{m_{{\mathrm{1}}}} \odot \ottnt{m_{{\mathrm{2}}}}  \le \ottnt{m}  \hyp{0.2},
    \end{gather}
    then $ \ottkw{op} _{ \ottnt{m_{{\mathrm{1}}}} }  \, \ottmv{l}  \mid  \mathcal{H}  \longrightarrow  \ottnt{M'}  \mid  \mathcal{H}'$ for some $\ottnt{M'}$ and $\mathcal{H}'$.

    \proof By definition, we have $ \ottsym{(}  \ottnt{n}  \ottsym{,}  \ottnt{m_{{\mathrm{0}}}}  \ottsym{,}  \ottnt{m'}  \ottsym{)} \ottsym{=} \mathcal{H}  \ottsym{(}  \ottmv{l}  \ottsym{)} $ such that
    \begin{gather}
           \ottnt{m'} \odot \ottnt{m}  \odot  \overline{  \langle \ottnt{C'} \rangle_{ \ottmv{l} }  }   \le \ottnt{m_{{\mathrm{0}}}} . \hyp{1.9}
    \end{gather}
    from \hypref{0.1}.
    To this end, we choose $ \ottnt{M'} \ottsym{=} \ottmv{l} $ and $ \mathcal{H}' \ottsym{=}  \mathcal{H} \cup \ottsym{\{} \ottmv{l}  \mapsto  \ottsym{(}  \ottnt{n}  \ottsym{,}  \ottnt{m_{{\mathrm{0}}}}  \ottsym{,}   \ottnt{m'} \odot \ottnt{m_{{\mathrm{1}}}}   \ottsym{)} \ottsym{\}}  $.
    Note that $  \mathcal{H} \cup \ottsym{\{} \ottmv{l}  \mapsto  \ottsym{(}  \ottnt{n}  \ottsym{,}  \ottnt{m_{{\mathrm{0}}}}  \ottsym{,}  \ottnt{m'}  \ottsym{)} \ottsym{\}}  \ottsym{=} \mathcal{H} $.
    So, what we need to show is
    \begin{gather*}
         \ottkw{op} _{ \ottnt{m_{{\mathrm{1}}}} }  \, \ottmv{l}  \mid   \mathcal{H} \cup \ottsym{\{} \ottmv{l}  \mapsto  \ottsym{(}  \ottnt{n}  \ottsym{,}  \ottnt{m_{{\mathrm{0}}}}  \ottsym{,}  \ottnt{m'}  \ottsym{)} \ottsym{\}}   \longrightarrow  \ottmv{l}  \mid   \mathcal{H} \cup \ottsym{\{} \ottmv{l}  \mapsto  \ottsym{(}  \ottnt{n}  \ottsym{,}  \ottnt{m_{{\mathrm{0}}}}  \ottsym{,}   \ottnt{m'} \odot \ottnt{m_{{\mathrm{1}}}}   \ottsym{)} \ottsym{\}} .
    \end{gather*}
    We have $    \ottnt{m'} \odot \ottnt{m_{{\mathrm{1}}}}  \odot \ottnt{m_{{\mathrm{2}}}}  \odot  \overline{  \langle \ottnt{C'} \rangle_{ \ottmv{l} }  }   \le \ottnt{m_{{\mathrm{0}}}} $ by \hypref{0.2} and \hypref{1.9}.
    So, we can derive the goal by \ruleref{RC-Op}, given $ \ottnt{m''} \ottsym{=}  \ottnt{m_{{\mathrm{2}}}} \odot  \overline{  \langle \ottnt{C'} \rangle_{ \ottmv{l} }  }   $, and \ruleref{R-Cfg}.
\end{prop}

\begin{prop}{progress/close}
    If
    \begin{gather}
        \mathcal{C}  \ottsym{[}  \ottmv{l}  \mathord:  \ottsym{[}  \ottnt{m}  \ottsym{]}  \ottsym{]}  \vdash  \mathcal{H} \hyp{0.1}\\
         \varepsilon \le \ottnt{m}  \hyp{0.2},
    \end{gather}
    then $\ottkw{drop} \, \ottmv{l}  \mid  \mathcal{H}  \longrightarrow  \ottnt{M'}  \mid  \mathcal{H}'$ for some $\ottnt{M'}$ and $\mathcal{H}'$.

    \proof We have some $ \ottsym{(}  \ottnt{n}  \ottsym{,}  \ottnt{m_{{\mathrm{0}}}}  \ottsym{,}  \ottnt{m'}  \ottsym{)} \ottsym{=} \mathcal{H}  \ottsym{(}  \ottmv{l}  \ottsym{)} $ such that
    \begin{gather}
          |  \lBrack  \langle \mathcal{C}  \ottsym{[}   \cdot   \ottsym{]} \rangle_{ \ottmv{l} }  \rBrack  |_{\bullet}  \ottsym{=} \ottnt{n}  \hyp{1.1}\\
          \ottnt{m'} \odot  \overline{  \langle \mathcal{C} \rangle_{ \ottmv{l} }   \ottsym{[}  \ottmv{l}  \mathord:  \ottsym{[}  \ottnt{m}  \ottsym{]}  \ottsym{]} }   \le \ottnt{m_{{\mathrm{0}}}}  \hyp{1.2}
    \end{gather}
    by \hypref{0.1}.
    Let $ \mathcal{H}_{{\mathrm{1}}} \ottsym{=}  \mathcal{H} \setminus \ottsym{\{} \ottmv{l}  \mapsto  \ottsym{(}  \ottnt{n}  \ottsym{,}  \ottnt{m_{{\mathrm{0}}}}  \ottsym{,}  \ottnt{m'}  \ottsym{)} \ottsym{\}}  $.
    Note that $  \mathcal{H}_{{\mathrm{1}}} \cup \ottsym{\{} \ottmv{l}  \mapsto  \ottsym{(}  \ottnt{n}  \ottsym{,}  \ottnt{m_{{\mathrm{0}}}}  \ottsym{,}  \ottnt{m'}  \ottsym{)} \ottsym{\}}  \ottsym{=} \mathcal{H} $, and $ \ottmv{l}  \notin   \ottkw{dom} ( \mathcal{H}_{{\mathrm{1}}} )  $.
    We consider whether $ \ottnt{n} \ottsym{=} \ottsym{0} $ or not.
    \begin{match}
        \item[$ \ottnt{n} \ottsym{=} \ottsym{0} $]
        We have $  \ottnt{m'} \odot \ottnt{m}  \le \ottnt{m_{{\mathrm{0}}}} $ by \hypref{1.1} and \hypref{1.2}.
        So, we can see
        \begin{gather}
             \ottnt{m'} \le \ottnt{m_{{\mathrm{0}}}}  \hyp{2.1}
        \end{gather}
        from \hypref{0.2}.
        To this end, we choose $ \ottnt{M'} \ottsym{=} \ottkw{unit} $ and $ \mathcal{H}' \ottsym{=} \mathcal{H}_{{\mathrm{1}}} $.
        What we need to show is
        \begin{gather}
            \ottkw{drop} \, \ottmv{l}  \mid   \mathcal{H}_{{\mathrm{1}}} \cup \ottsym{\{} \ottmv{l}  \mapsto  \ottsym{(}  \ottsym{0}  \ottsym{,}  \ottnt{m_{{\mathrm{0}}}}  \ottsym{,}  \ottnt{m'}  \ottsym{)} \ottsym{\}}   \longrightarrow  \ottkw{unit}  \mid  \mathcal{H}_{{\mathrm{1}}}.
        \end{gather}
        We can derived that by \ruleref{RC-Cl2}, \ruleref{R-Cfg}, and \hypref{2.1}.

        \item[$ \ottnt{n} \ottsym{=} \ottnt{n'}  \ottsym{+}  \ottsym{1} $ for some $\ottnt{n'}$]
        To this end, we choose $ \ottnt{M'} \ottsym{=} \ottkw{unit} $ and $ \mathcal{H}' \ottsym{=}  \mathcal{H}_{{\mathrm{1}}} \cup \ottsym{\{} \ottmv{l}  \mapsto  \ottsym{(}  \ottnt{n'}  \ottsym{,}  \ottnt{m_{{\mathrm{0}}}}  \ottsym{,}  \ottnt{m'}  \ottsym{)} \ottsym{\}}  $.
        What we need to show is
        \begin{gather}
            \ottkw{drop} \, \ottmv{l}  \mid   \mathcal{H}_{{\mathrm{1}}} \cup \ottsym{\{} \ottmv{l}  \mapsto  \ottsym{(}  \ottnt{n'}  \ottsym{+}  \ottsym{1}  \ottsym{,}  \ottnt{m_{{\mathrm{0}}}}  \ottsym{,}  \ottnt{m'}  \ottsym{)} \ottsym{\}}   \longrightarrow  \ottkw{unit}  \mid   \mathcal{H}_{{\mathrm{1}}} \cup \ottsym{\{} \ottmv{l}  \mapsto  \ottsym{(}  \ottnt{n'}  \ottsym{,}  \ottnt{m_{{\mathrm{0}}}}  \ottsym{,}  \ottnt{m'}  \ottsym{)} \ottsym{\}} ,
        \end{gather}
        which can be derived by \ruleref{RC-Cl1} and \ruleref{R-Cfg}.
    \end{match}
\end{prop}

\begin{prop}{progress/ctx}
    If $ \ottnt{M} \ottsym{=} \mathcal{E}  \ottsym{[}  \ottnt{M_{{\mathrm{1}}}}  \ottsym{]} $ and $\ottnt{M_{{\mathrm{1}}}}  \mid  \mathcal{H}  \longrightarrow  \ottnt{M'_{{\mathrm{1}}}}  \mid  \mathcal{H}'$, then there exists $\ottnt{M'}$ such that $\ottnt{M}  \mid  \mathcal{H}  \longrightarrow  \ottnt{M'}  \mid  \mathcal{H}'$.

    \proof We get some $\mathcal{E}'$, $\ottnt{M_{{\mathrm{2}}}}$, and $\ottnt{M'_{{\mathrm{2}}}}$ such that $ \ottnt{M_{{\mathrm{1}}}} \ottsym{=} \mathcal{E}'  \ottsym{[}  \ottnt{M_{{\mathrm{2}}}}  \ottsym{]} $, $ \ottnt{M'_{{\mathrm{1}}}} \ottsym{=} \mathcal{E}'  \ottsym{[}  \ottnt{M'_{{\mathrm{2}}}}  \ottsym{]} $, and (1) $\ottnt{M_{{\mathrm{2}}}}  \longrightarrow_\beta  \ottnt{M'_{{\mathrm{2}}}}$ and $ \mathcal{H} \ottsym{=} \mathcal{H}' $ or (2) $\ottnt{M_{{\mathrm{2}}}}  \mid  \mathcal{H}  \longrightarrow_\gamma  \ottnt{M'_{{\mathrm{2}}}}  \mid  \mathcal{H}'$ by the second assumption.
    To this end, we choose $ \ottnt{M'} \ottsym{=} \mathcal{E}  \ottsym{[}  \mathcal{E}'  \ottsym{[}  \ottnt{M'_{{\mathrm{2}}}}  \ottsym{]}  \ottsym{]} $.
    What we need to show is $\mathcal{E}  \ottsym{[}  \mathcal{E}'  \ottsym{[}  \ottnt{M_{{\mathrm{2}}}}  \ottsym{]}  \ottsym{]}  \mid  \mathcal{H}  \longrightarrow  \mathcal{E}  \ottsym{[}  \mathcal{E}'  \ottsym{[}  \ottnt{M'_{{\mathrm{2}}}}  \ottsym{]}  \ottsym{]}  \mid  \mathcal{H}'$, which can be derived by using \ruleref{R-Exp} or \ruleref{R-Cfg} depending on whether (1) or (2) holds.
\end{prop}

\begin{prop}{progress/noeffect}
    If
    \begin{gather}
        \ottnt{C}  \vdash  \ottnt{M}  :  \ottnt{T}  \mid  \ottsym{0} \hyp{0.1}\\
        \mathcal{C}  \ottsym{[}  \ottnt{C}  \ottsym{]}  \vdash  \mathcal{H} \hyp{0.2},
    \end{gather}
    then
    \begin{itemize}
        \item there exist $\ottnt{M'}$ and $\mathcal{H}'$ such that $\ottnt{M}  \mid  \mathcal{H}  \longrightarrow  \ottnt{M'}  \mid  \mathcal{H}'$, or
        \item $ \ottnt{M} \ottsym{=} \ottnt{V} $ for some $\ottnt{V}$.
    \end{itemize}

    \proof By induction on the given derivation of \hypref{0.1}.
    \begin{match}
        \item[]\ruleref{T-Unit}, \ruleref{T-New}, \ruleref{T-Op}, \ruleref{T-Split}, \ruleref{T-Close}, \ruleref{T-Loc}, \ruleref{T-Abs}, \ruleref{T-UAbs}, \ruleref{T-RAbs}, and \ruleref{T-LAbs}.
        $\ottnt{M}$ is a value in these cases.
        Hence, the second item of the goal holds.

        \item[\ruleref{T-Var}]
        This case is vacuously true because $ \ottnt{C} \ottsym{=} \ottmv{x}  \mathord:  \ottnt{T} $, but this is impossible.

        \item[\ruleref{T-App}]
        In this case,
        \begin{gather}
             \ottnt{C} \ottsym{=} \ottnt{C_{{\mathrm{1}}}}  \ottsym{,}  \ottnt{C_{{\mathrm{2}}}}  \hyp{1.1}\\
             \ottnt{M} \ottsym{=} \ottnt{M_{{\mathrm{1}}}} \, \ottnt{M_{{\mathrm{2}}}}  \hyp{1.1.1}\\
               \ottnt{e_{{\mathrm{0}}}}  \sqcup  \ottnt{e_{{\mathrm{1}}}}   \sqcup  \ottnt{e_{{\mathrm{2}}}}  \ottsym{=} \ottsym{0}  \hyp{2.3}\\
            \ottnt{C_{{\mathrm{1}}}}  \vdash  \ottnt{M_{{\mathrm{1}}}}  :   \ottnt{S} \rightarrow _{ \ottnt{e_{{\mathrm{0}}}} } \ottnt{T}   \mid  \ottnt{e_{{\mathrm{1}}}} \hyp{1.2}\\
            \ottnt{C_{{\mathrm{2}}}}  \vdash  \ottnt{M_{{\mathrm{2}}}}  :  \ottnt{S}  \mid  \ottnt{e_{{\mathrm{2}}}} \hyp{2.2}.
        \end{gather}
        We have $ \ottnt{e_{{\mathrm{0}}}} \ottsym{=} \ottnt{e_{{\mathrm{1}}}}  =  \ottnt{e_{{\mathrm{2}}}} \ottsym{=} \ottsym{0} $ by \hypref{2.3}.
        So, we have the following cases by the induction hypothesis for \hypref{1.2}.
        \begin{match}
            \item[$\ottnt{M_{{\mathrm{1}}}}  \mid  \mathcal{H}  \longrightarrow  \ottnt{M'_{{\mathrm{1}}}}  \mid  \mathcal{H}'$ for some $\ottnt{M'_{{\mathrm{1}}}}$ and $\mathcal{H}'$]
            The first item of the goal follows by \propref{progress/ctx} given $ \mathcal{E} \ottsym{=} \ottsym{[]} \, \ottnt{M_{{\mathrm{2}}}} $.

            \item[$ \ottnt{M_{{\mathrm{1}}}} \ottsym{=} \ottnt{V_{{\mathrm{1}}}} $ for some $\ottnt{V_{{\mathrm{1}}}}$]
            We have the following cases by the induction hypothesis for \hypref{2.2}.
            \begin{match}
                \item[$\ottnt{M_{{\mathrm{2}}}}  \mid  \mathcal{H}  \longrightarrow  \ottnt{M'_{{\mathrm{2}}}}  \mid  \mathcal{H}'$ for some $\ottnt{M'_{{\mathrm{2}}}}$ and $\mathcal{H}'$]
                The first item of the goal follows by \propref{progress/ctx} given $ \mathcal{E} \ottsym{=} \ottnt{V_{{\mathrm{1}}}} \, \ottsym{[]} $.

                \item[$ \ottnt{M_{{\mathrm{2}}}} \ottsym{=} \ottnt{V_{{\mathrm{2}}}} $ for some $\ottnt{V_{{\mathrm{2}}}}$]
                Remember that we have had
                \begin{gather*}
                    \ottnt{C_{{\mathrm{1}}}}  \vdash  \ottnt{V_{{\mathrm{1}}}}  :   \ottnt{S} \rightarrow _{ \ottsym{0} } \ottnt{T}   \mid  \ottsym{0} \hyp{2.5}\\
                    \ottnt{C_{{\mathrm{2}}}}  \vdash  \ottnt{V_{{\mathrm{2}}}}  :  \ottnt{S}  \mid  \ottsym{0}. \hyp{2.6}
                \end{gather*}
                We have the following cases by \propref{typing/canonical(arrow)} and \hypref{2.5}.
                \begin{match}
                    \item[$ \ottnt{V_{{\mathrm{1}}}} \ottsym{=}  \ottkw{new} _{ \ottnt{m} }  $]
                    We have $ \ottnt{S} \ottsym{=}  \mathtt{Unit}  $ by \propref{typing/inv} and \hypref{2.5}.
                    Then, $ \ottnt{V_{{\mathrm{2}}}} \ottsym{=} \ottkw{unit} $ by \propref{typing/canonical(unit)}.
                    Now, the first item of the goal follows by \propref{progress/new}.

                    \item[$ \ottnt{V_{{\mathrm{1}}}} \ottsym{=}  \ottkw{op} _{ \ottnt{m_{{\mathrm{1}}}} }  $]
                    We have $  \ottnt{S} \rightarrow _{ \ottsym{0} } \ottnt{T}  \ottsym{=}  \ottsym{[}  \ottnt{m}  \ottsym{]} \rightarrow _{ \ottsym{1} } \ottsym{[}  \ottnt{m_{{\mathrm{2}}}}  \ottsym{]}  $ and $  \ottnt{m_{{\mathrm{1}}}} \odot \ottnt{m_{{\mathrm{2}}}}  \le \ottnt{m} $ for some $\ottnt{m}$ and $\ottnt{m_{{\mathrm{2}}}}$ by \propref{typing/inv} and \hypref{2.5}, but this is impossible since $ \ottsym{0} \neq \ottsym{1} $.

                    \item[$ \ottnt{V_{{\mathrm{1}}}} \ottsym{=}  \ottkw{split} _{ \ottnt{m_{{\mathrm{1}}}} , \ottnt{m_{{\mathrm{2}}}} }  $]
                    We have $ \ottnt{S} \ottsym{=} \ottsym{[}  \ottnt{m_{{\mathrm{0}}}}  \ottsym{]} $ for some $\ottnt{m_{{\mathrm{0}}}}$ by \propref{typing/inv(split)} and \hypref{2.5}.
                    Then, $ \ottnt{V_{{\mathrm{2}}}} \ottsym{=} \ottmv{l} $ for some $\ottmv{l}$ by \propref{typing/canonical(resource)} and \hypref{2.6}.
                    Furthermore, we have $\ottnt{C_{{\mathrm{2}}}}  \lesssim  \ottmv{l}  \mathord:  \ottsym{[}  \ottnt{m_{{\mathrm{0}}}}  \ottsym{]}$ by \propref{typing/inv} and \hypref{2.6}.
                    So, we have $\mathcal{C}  \ottsym{[}  \ottnt{C}  \ottsym{]}  \lesssim  \mathcal{C}  \ottsym{[}  \ottnt{C_{{\mathrm{1}}}}  \ottsym{,}  \ottmv{l}  \mathord:  \ottsym{[}  \ottnt{m_{{\mathrm{0}}}}  \ottsym{]}  \ottsym{]}$ by \propref{env/sub/ctx}.
                    From that, $\mathcal{C}  \ottsym{[}  \ottnt{C_{{\mathrm{1}}}}  \ottsym{,}  \ottmv{l}  \mathord:  \ottsym{[}  \ottnt{m_{{\mathrm{0}}}}  \ottsym{]}  \ottsym{]}  \vdash  \mathcal{H}$ follows by \propref{heap/weaken} and \hypref{0.2}.
                    Now, the first item of the goal follows by \propref{progress/split}.

                    \item[$ \ottnt{V_{{\mathrm{1}}}} \ottsym{=} \ottkw{drop} $]
                    We have
                    \begin{gather}
                         \ottnt{S} \ottsym{=} \ottsym{[}  \ottnt{m}  \ottsym{]}  \hyp{2.7}\\
                         \varepsilon \le \ottnt{m}  \hyp{2.8}
                    \end{gather}
                    for some $\ottnt{m}$ by \propref{typing/inv(drop)} and \hypref{2.5}.
                    So, $ \ottnt{V_{{\mathrm{2}}}} \ottsym{=} \ottmv{l} $ for some $\ottmv{l}$ by \propref{typing/canonical(resource)} and \hypref{2.6}.
                    Furthermore, we have $\ottnt{C_{{\mathrm{2}}}}  \lesssim  \ottmv{l}  \mathord:  \ottsym{[}  \ottnt{m}  \ottsym{]}$ by \propref{typing/inv(loc)}.
                    So, we have $\mathcal{C}  \ottsym{[}  \ottnt{C}  \ottsym{]}  \lesssim  \mathcal{C}  \ottsym{[}  \ottnt{C_{{\mathrm{1}}}}  \ottsym{,}  \ottmv{l}  \mathord:  \ottsym{[}  \ottnt{m}  \ottsym{]}  \ottsym{]}$ by \propref{env/sub/ctx}.
                    From that, $\mathcal{C}  \ottsym{[}  \ottnt{C_{{\mathrm{1}}}}  \ottsym{,}  \ottmv{l}  \mathord:  \ottsym{[}  \ottnt{m}  \ottsym{]}  \ottsym{]}  \vdash  \mathcal{H}$ follows by \propref{heap/weaken} and \hypref{0.2}.
                    Now, the first item of the goal follows by \propref{progress/close} and \hypref{2.8}.

                    \item[$ \ottnt{V_{{\mathrm{1}}}} \ottsym{=} \lambda  \ottmv{x}  \ottsym{.}  \ottnt{M_{{\mathrm{11}}}} $]
                    To show the first item of the goal, we choose $ \ottnt{M'} \ottsym{=} \ottnt{M_{{\mathrm{11}}}}  \ottsym{[}  \ottnt{V_{{\mathrm{2}}}}  \ottsym{/}  \ottmv{x}  \ottsym{]} $ and $ \mathcal{H}' \ottsym{=} \mathcal{H} $.
                    Then, what we need to show is
                    \begin{gather*}
                        \ottsym{(}  \lambda  \ottmv{x}  \ottsym{.}  \ottnt{M_{{\mathrm{11}}}}  \ottsym{)} \, \ottnt{V_{{\mathrm{2}}}}  \mid  \mathcal{H}  \longrightarrow  \ottnt{M_{{\mathrm{11}}}}  \ottsym{[}  \ottnt{V_{{\mathrm{2}}}}  \ottsym{/}  \ottmv{x}  \ottsym{]}  \mid  \mathcal{H},
                    \end{gather*}
                    which can be derived by \ruleref{RE-Beta} and \ruleref{R-Exp}.
                \end{match}
            \end{match}
        \end{match}

        \item[\ruleref{T-UApp}, \ruleref{T-RApp}, \ruleref{T-LApp}]
        Similar to the case \ruleref{T-App}.

        \item[\ruleref{T-UPair} and \ruleref{T-OPair}]
        Similar to the case \ruleref{T-App} except, when we know both elements of a pair are value, the second item of the goal follows.

        \item[\ruleref{T-ULet}]
        In this case,
        \begin{gather}
             \ottnt{C} \ottsym{=} \mathcal{C}_{{\mathrm{1}}}  \ottsym{[}  \ottnt{C_{{\mathrm{1}}}}  \ottsym{]}  \hyp{5.1}\\
             \ottnt{M} \ottsym{=} \ottkw{let} \, \ottmv{x}  \otimes  \ottmv{y}  \ottsym{=}  \ottnt{M_{{\mathrm{1}}}} \, \ottkw{in} \, \ottnt{M_{{\mathrm{2}}}}  \hyp{5.2}\\
            \ottnt{C_{{\mathrm{1}}}}  \vdash  \ottnt{M_{{\mathrm{1}}}}  :  \ottnt{S_{{\mathrm{1}}}}  \otimes  \ottnt{S_{{\mathrm{2}}}}  \mid  \ottsym{0} \hyp{5.3}.
        \end{gather}
        So, we have the following cases by the induction hypothesis for \hypref{5.3}.
        \begin{match}
            \item[$\ottnt{M_{{\mathrm{1}}}}  \mid  \mathcal{H}  \longrightarrow  \ottnt{M'_{{\mathrm{1}}}}  \mid  \mathcal{H}'$ for some $\ottnt{M'_{{\mathrm{1}}}}$ and $\mathcal{H}'$]
            The first item of the goal follows by \propref{progress/ctx} given $ \mathcal{E} \ottsym{=} \ottkw{let} \, \ottmv{x}  \otimes  \ottmv{y}  \ottsym{=}  \ottsym{[]} \, \ottkw{in} \, \ottnt{M_{{\mathrm{2}}}} $.

            \item[$ \ottnt{M_{{\mathrm{1}}}} \ottsym{=} \ottnt{V_{{\mathrm{1}}}} $ for some $\ottnt{V_{{\mathrm{1}}}}$]
            We have
            \begin{gather*}
                 \ottnt{V_{{\mathrm{1}}}} \ottsym{=} \ottnt{V_{{\mathrm{11}}}}  \otimes  \ottnt{V_{{\mathrm{12}}}} 
            \end{gather*}
            by \propref{typing/canonical(uprod)} and \hypref{5.3}.
            To show the first item of the goal, we choose $ \ottnt{M'} \ottsym{=} \ottnt{M_{{\mathrm{2}}}}  \ottsym{[}  \ottnt{V_{{\mathrm{11}}}}  \ottsym{/}  \ottmv{x}  \ottsym{]}  \ottsym{[}  \ottnt{V_{{\mathrm{12}}}}  \ottsym{/}  \ottmv{y}  \ottsym{]} $ and $ \mathcal{H}' \ottsym{=} \mathcal{H} $.
            What we need to show is
            \begin{gather*}
                \ottkw{let} \, \ottmv{x}  \otimes  \ottmv{y}  \ottsym{=}  \ottnt{V_{{\mathrm{11}}}}  \otimes  \ottnt{V_{{\mathrm{12}}}} \, \ottkw{in} \, \ottnt{M_{{\mathrm{2}}}}  \mid  \mathcal{H}  \longrightarrow  \ottnt{M_{{\mathrm{2}}}}  \ottsym{[}  \ottnt{V_{{\mathrm{11}}}}  \ottsym{/}  \ottmv{x}  \ottsym{]}  \ottsym{[}  \ottnt{V_{{\mathrm{12}}}}  \ottsym{/}  \ottmv{y}  \ottsym{]}  \mid  \mathcal{H},
            \end{gather*}
            which can be derived by \ruleref{RE-ULet} and \ruleref{R-Exp}.
        \end{match}

        \item[\ruleref{T-OLet}] Similar to the case \ruleref{T-ULet}.

        \item[\ruleref{T-Weaken}]
        In this case,
        \begin{gather}
            \ottnt{C}  \lesssim  \ottnt{C_{{\mathrm{1}}}} \hyp{4.1}\\
             \ottnt{e_{{\mathrm{1}}}} \ottsym{<} \ottsym{0}  \hyp{4.2}\\
            \ottnt{C_{{\mathrm{1}}}}  \vdash  \ottnt{M}  :  \ottnt{T}  \mid  \ottnt{e_{{\mathrm{1}}}} \hyp{4.3}.
        \end{gather}
        We have $ \ottnt{e_{{\mathrm{1}}}} \ottsym{=} \ottsym{0} $ by \hypref{4.2}.
        We have $\mathcal{C}  \ottsym{[}  \ottnt{C}  \ottsym{]}  \lesssim  \mathcal{C}  \ottsym{[}  \ottnt{C_{{\mathrm{1}}}}  \ottsym{]}$ by \propref{env/sub/ctx} and \hypref{4.1}, and so $\mathcal{C}  \ottsym{[}  \ottnt{C_{{\mathrm{1}}}}  \ottsym{]}  \vdash  \mathcal{H}$ by \propref{heap/weaken} and \hypref{0.2}.
        Now, the goal follows by the induction hypothesis.
    \end{match}
\end{prop}

\begin{prop}{progress/effect}
    If
    \begin{gather}
        \ottnt{C}  \vdash  \ottnt{M}  :  \ottnt{T}  \mid  \ottnt{e} \hyp{0.1}\\
        \ottnt{C}  \ottsym{,}  \ottnt{C'}  \vdash  \mathcal{H} \hyp{0.2},
    \end{gather}
    then
    \begin{itemize}
        \item $\ottnt{M}  \mid  \mathcal{H}  \longrightarrow  \ottnt{M'}  \mid  \mathcal{H}'$ for some $\ottnt{M'}$ and $\mathcal{H}'$, or
        \item $ \ottnt{M} \ottsym{=} \ottnt{V} $ for some $\ottnt{V}$.
    \end{itemize}

    \proof By induction on the given derivation of \hypref{0.1}.
    \begin{match}
        \item[]\ruleref{T-Unit}, \ruleref{T-New}, \ruleref{T-Op}, \ruleref{T-Split}, \ruleref{T-Close}, \ruleref{T-Loc}, \ruleref{T-Abs}, \ruleref{T-UAbs}, \ruleref{T-RAbs}, and \ruleref{T-LAbs}.
        $\ottnt{M}$ is a value in these cases.
        Hence, the second item of the goal holds.

        \item[\ruleref{T-Var}]
        This case is vacuously true because $ \ottnt{C} \ottsym{=} \ottmv{x}  \mathord:  \ottnt{T} $, but this is impossible.

        \item[\ruleref{T-App}]
        In this case,
        \begin{gather}
             \ottnt{C} \ottsym{=} \ottnt{C_{{\mathrm{1}}}}  \ottsym{,}  \ottnt{C_{{\mathrm{2}}}}  \hyp{1.1}\\
             \ottnt{M} \ottsym{=} \ottnt{M_{{\mathrm{1}}}} \, \ottnt{M_{{\mathrm{2}}}}  \hyp{1.2}\\
             \ottnt{e} \ottsym{=}   \ottnt{e_{{\mathrm{0}}}}  \sqcup  \ottnt{e_{{\mathrm{1}}}}   \sqcup  \ottnt{e_{{\mathrm{2}}}}   \hyp{1.3}\\
            \ottnt{C_{{\mathrm{1}}}}  \vdash  \ottnt{M_{{\mathrm{1}}}}  :   \ottnt{S} \rightarrow _{ \ottnt{e_{{\mathrm{0}}}} } \ottnt{T}   \mid  \ottnt{e_{{\mathrm{1}}}} \hyp{1.4}\\
            \ottnt{C_{{\mathrm{2}}}}  \vdash  \ottnt{M_{{\mathrm{2}}}}  :  \ottnt{S}  \mid  \ottnt{e_{{\mathrm{2}}}} \hyp{1.5}.
        \end{gather}
        We have $\ottnt{C_{{\mathrm{1}}}}  \ottsym{,}  \ottnt{C_{{\mathrm{2}}}}  \ottsym{,}  \ottnt{C'}  \vdash  \mathcal{H}$ by \hypref{0.2}.
        So, we have the following cases by the induction hypothesis of \hypref{1.4}.
        \begin{match}
            \item[$\ottnt{M_{{\mathrm{1}}}}  \mid  \mathcal{H}  \longrightarrow  \ottnt{M'_{{\mathrm{1}}}}  \mid  \mathcal{H}'$]
            The first item of the goal follows by \propref{progress/ctx} given $ \mathcal{E} \ottsym{=} \ottsym{[]} \, \ottnt{M_{{\mathrm{2}}}} $.

            \item[$ \ottnt{M_{{\mathrm{1}}}} \ottsym{=} \ottnt{V_{{\mathrm{1}}}} $ for some $\ottnt{V_{{\mathrm{1}}}}$]
            We have $\ottnt{C_{{\mathrm{1}}}}  \lesssim   \cdot $ by \propref{typing/value/unr} and \hypref{1.4}.
            So, $\ottnt{C_{{\mathrm{2}}}}  \ottsym{,}  \ottnt{C'}  \vdash  \mathcal{H}$ follows by \propref{heap/weaken} and \hypref{0.2}.
            Then, we have the following cases by the induction hypothesis of \hypref{1.5}.
            \begin{match}
                \item[$\ottnt{M_{{\mathrm{2}}}}  \mid  \mathcal{H}  \longrightarrow  \ottnt{M'_{{\mathrm{2}}}}  \mid  \mathcal{H}'$]
                The first item of the goal follows by \propref{progress/ctx} given $ \mathcal{E} \ottsym{=} \ottnt{V_{{\mathrm{1}}}} \, \ottsym{[]} $.

                \item[$ \ottnt{M_{{\mathrm{2}}}} \ottsym{=} \ottnt{V_{{\mathrm{2}}}} $ for some $\ottnt{V_{{\mathrm{2}}}}$]
                Remember that here, we have had
                \begin{gather}
                    \ottnt{C_{{\mathrm{1}}}}  \vdash  \ottnt{V_{{\mathrm{1}}}}  :   \ottnt{S} \rightarrow _{ \ottnt{e_{{\mathrm{0}}}} } \ottnt{T}   \mid  \ottnt{e_{{\mathrm{1}}}} \hyp{1.6}\\
                    \ottnt{C_{{\mathrm{2}}}}  \vdash  \ottnt{V_{{\mathrm{2}}}}  :  \ottnt{S}  \mid  \ottnt{e_{{\mathrm{2}}}} \hyp{1.7}\\
                    \ottnt{C_{{\mathrm{2}}}}  \ottsym{,}  \ottnt{C'}  \vdash  \mathcal{H}. \hyp{1.8}
                \end{gather}
                So, $ \ottnt{e_{{\mathrm{1}}}} \ottsym{=} \ottnt{e_{{\mathrm{2}}}}  = 0$ by \propref{typing/value/unr}.
                We have the following cases by \propref{typing/canonical(arrow)} and \hypref{1.6}.
                \begin{match}
                    \item[$ \ottnt{V_{{\mathrm{1}}}} \ottsym{=}  \ottkw{new} _{ \ottnt{m} }  $ for some $\ottnt{m}$]
                    We have $ \ottnt{e_{{\mathrm{0}}}} \ottsym{=} \ottsym{0} $ by \propref{typing/inv(new)} and \hypref{1.6}.
                    So, we can have $\ottnt{C}  \vdash  \ottnt{M}  :  \ottnt{T}  \mid  \ottsym{0}$ and the goal follows by \propref{progress/noeffect}.

                    \item[$ \ottnt{V_{{\mathrm{1}}}} \ottsym{=}  \ottkw{op} _{ \ottnt{m_{{\mathrm{1}}}} }  $ for some $\ottnt{m_{{\mathrm{1}}}}$]
                    We have $ \ottnt{S} \ottsym{=} \ottsym{[}  \ottnt{m_{{\mathrm{0}}}}  \ottsym{]} $ and $  \ottnt{m_{{\mathrm{1}}}} \odot \ottnt{m_{{\mathrm{2}}}}  \le \ottnt{m_{{\mathrm{0}}}} $ for some $\ottnt{m_{{\mathrm{0}}}}$ and $\ottnt{m_{{\mathrm{2}}}}$ by \propref{typing/inv(op)} and \hypref{1.6}.
                    So, $ \ottnt{V_{{\mathrm{2}}}} \ottsym{=} \ottmv{l} $ for some $\ottmv{l}$ by \propref{typing/canonical(resource)} and \hypref{1.7}.
                    Furthermore, $\ottnt{C_{{\mathrm{2}}}}  \lesssim  \ottmv{l}  \mathord:  \ottsym{[}  \ottnt{m_{{\mathrm{0}}}}  \ottsym{]}$ by \propref{typing/inv}.
                    From that, we have $\ottmv{l}  \mathord:  \ottsym{[}  \ottnt{m_{{\mathrm{0}}}}  \ottsym{]}  \ottsym{,}  \ottnt{C'}  \vdash  \mathcal{H}$ by \hypref{1.8}.
                    Now, the first item of the goal follows by \propref{progress/op}.

                    \item[$ \ottnt{V_{{\mathrm{1}}}} \ottsym{=}  \ottkw{split} _{ \ottnt{m_{{\mathrm{1}}}} , \ottnt{m_{{\mathrm{2}}}} }  $ for some $\ottnt{m_{{\mathrm{1}}}}$ and $\ottnt{m_{{\mathrm{2}}}}$]
                    We have $ \ottnt{e_{{\mathrm{0}}}} \ottsym{=} \ottsym{0} $ by \propref{typing/inv(split)} and \hypref{1.6}.
                    So, we can have $\ottnt{C}  \vdash  \ottnt{M}  :  \ottnt{T}  \mid  \ottsym{0}$ and the goal follows by \propref{progress/noeffect}.

                    \item[$ \ottnt{V_{{\mathrm{1}}}} \ottsym{=} \ottkw{drop} $]
                    We have $ \ottnt{e_{{\mathrm{0}}}} \ottsym{=} \ottsym{0} $ by \propref{typing/inv(drop)} and \hypref{1.6}.
                    So, we can have $\ottnt{C}  \vdash  \ottnt{M}  :  \ottnt{T}  \mid  \ottsym{0}$ and the goal follows by \propref{progress/noeffect}.

                    \item[$ \ottnt{V_{{\mathrm{1}}}} \ottsym{=} \lambda  \ottmv{x}  \ottsym{.}  \ottnt{M'_{{\mathrm{1}}}} $]
                    To show the first item of the goal, we choose $ \ottnt{M'} \ottsym{=} \ottnt{M'_{{\mathrm{1}}}}  \ottsym{[}  \ottnt{V_{{\mathrm{2}}}}  \ottsym{/}  \ottmv{x}  \ottsym{]} $ and $ \mathcal{H}' \ottsym{=} \mathcal{H} $.
                    What we need to show is
                    \begin{gather*}
                        \ottsym{(}  \lambda  \ottmv{x}  \ottsym{.}  \ottnt{M'_{{\mathrm{1}}}}  \ottsym{)} \, \ottnt{V_{{\mathrm{2}}}}  \mid  \mathcal{H}  \longrightarrow  \ottnt{M'_{{\mathrm{1}}}}  \ottsym{[}  \ottnt{V_{{\mathrm{2}}}}  \ottsym{/}  \ottmv{x}  \ottsym{]}  \mid  \mathcal{H},
                    \end{gather*}
                    which can be derived by \ruleref{RE-Beta} and \ruleref{R-Exp}.
                \end{match}
            \end{match}
        \end{match}

        \item[\ruleref{T-UApp}]
        In this case,
        \begin{gather}
             \ottnt{C} \ottsym{=} \ottnt{C_{{\mathrm{1}}}}  \parallel  \ottnt{C_{{\mathrm{2}}}}  \hyp{2.1}\\
             \ottnt{M} \ottsym{=} \ottnt{M_{{\mathrm{1}}}}  {}^\circ  \ottnt{M_{{\mathrm{2}}}}  \hyp{2.2}\\
             \ottnt{e} \ottsym{=}   \ottnt{e_{{\mathrm{0}}}}  \sqcup  \ottnt{e_{{\mathrm{1}}}}   \sqcup  \ottnt{e_{{\mathrm{2}}}}   \hyp{2.3}\\
            \ottnt{C_{{\mathrm{1}}}}  \vdash  \ottnt{M_{{\mathrm{1}}}}  :   \ottnt{S} \rightarrowtriangle _{ \ottnt{e_{{\mathrm{0}}}} } \ottnt{T}   \mid  \ottnt{e_{{\mathrm{1}}}} \hyp{2.4}\\
            \ottnt{C_{{\mathrm{2}}}}  \vdash  \ottnt{M_{{\mathrm{2}}}}  :  \ottnt{S}  \mid  \ottnt{e_{{\mathrm{2}}}} \hyp{2.5}.
        \end{gather}
        So, we have $\ottnt{C_{{\mathrm{1}}}}  \ottsym{,}  \ottnt{C_{{\mathrm{2}}}}  \ottsym{,}  \ottnt{C'}  \vdash  \mathcal{H}$ by \propref{heap/weaken} and \hypref{0.2}.
        Now, we have the following cases by the induction hypothesis of \hypref{2.4}.
        \begin{match}
            \item[$\ottnt{M_{{\mathrm{1}}}}  \mid  \mathcal{H}  \longrightarrow  \ottnt{M'_{{\mathrm{1}}}}  \mid  \mathcal{H}'$]
            The first item of the goal follows by \propref{progress/ctx} given $ \mathcal{E} \ottsym{=} \ottsym{[]}  {}^\circ  \ottnt{M_{{\mathrm{2}}}} $.

            \item[$ \ottnt{M_{{\mathrm{1}}}} \ottsym{=} \ottnt{V_{{\mathrm{1}}}} $ for some $\ottnt{V_{{\mathrm{1}}}}$]
            We have $\ottnt{C_{{\mathrm{2}}}}  \ottsym{,}  \ottnt{C_{{\mathrm{1}}}}  \ottsym{,}  \ottnt{C'}  \vdash  \mathcal{H}$ by \propref{heap/weaken} and \hypref{0.2}.
            So, we have the following cases by the induction hypothesis of \hypref{2.5}.
            \begin{match}
                \item[$\ottnt{M_{{\mathrm{2}}}}  \mid  \mathcal{H}  \longrightarrow  \ottnt{M'_{{\mathrm{2}}}}  \mid  \mathcal{H}'$]
                The first item of the goal follows by \propref{progress/ctx} given $ \mathcal{E} \ottsym{=} \ottnt{V_{{\mathrm{1}}}}  {}^\circ  \ottsym{[]} $.

                \item[$ \ottnt{M_{{\mathrm{2}}}} \ottsym{=} \ottnt{V_{{\mathrm{2}}}} $ for some $\ottnt{V_{{\mathrm{2}}}}$]
                We have $ \ottnt{V_{{\mathrm{1}}}} \ottsym{=} \lambda^\circ  \ottmv{x}  \ottsym{.}  \ottnt{M'_{{\mathrm{1}}}} $ by \propref{typing/canonical(uarrow)}.
                To show the first item of the goal, we choose $ \ottnt{M'} \ottsym{=} \ottnt{M'_{{\mathrm{1}}}}  \ottsym{[}  \ottnt{V_{{\mathrm{2}}}}  \ottsym{/}  \ottmv{x}  \ottsym{]} $ and $ \mathcal{H}' \ottsym{=} \mathcal{H} $.
                What we need to show is
                \begin{gather*}
                    \ottsym{(}  \lambda^\circ  \ottmv{x}  \ottsym{.}  \ottnt{M'_{{\mathrm{1}}}}  \ottsym{)}  {}^\circ  \ottnt{V_{{\mathrm{2}}}}  \mid  \mathcal{H}  \longrightarrow  \ottnt{M'_{{\mathrm{1}}}}  \ottsym{[}  \ottnt{V_{{\mathrm{2}}}}  \ottsym{/}  \ottmv{x}  \ottsym{]}  \mid  \mathcal{H},
                \end{gather*}
                which can be derived by \ruleref{RE-UBeta} and \ruleref{R-Exp}.
            \end{match}
        \end{match}

        \item[\ruleref{T-RApp}]
        In this case,
        \begin{gather}
             \ottnt{C} \ottsym{=} \ottnt{C_{{\mathrm{1}}}}  \ottsym{,}  \ottnt{C_{{\mathrm{2}}}}  \hyp{3.1}\\
             \ottnt{M} \ottsym{=} \ottnt{M_{{\mathrm{1}}}}  {}^>  \ottnt{M_{{\mathrm{2}}}}  \hyp{3.2}\\
             \ottnt{e} \ottsym{=}  \ottnt{e_{{\mathrm{0}}}}  \sqcup  \ottnt{e_{{\mathrm{1}}}}   \hyp{3.3}\\
            \ottnt{C_{{\mathrm{1}}}}  \vdash  \ottnt{M_{{\mathrm{1}}}}  :   \ottnt{S} \twoheadrightarrow _{ \ottnt{e_{{\mathrm{0}}}} } \ottnt{T}   \mid  \ottnt{e_{{\mathrm{1}}}} \hyp{3.4}\\
            \ottnt{C_{{\mathrm{2}}}}  \vdash  \ottnt{M_{{\mathrm{2}}}}  :  \ottnt{S}  \mid  \ottsym{0} \hyp{3.5}.
        \end{gather}
        We have the following cases by the induction hypothesis of \hypref{3.4}.
        \begin{match}
            \item[$\ottnt{M_{{\mathrm{1}}}}  \mid  \mathcal{H}  \longrightarrow  \ottnt{M'_{{\mathrm{1}}}}  \mid  \mathcal{H}'$]
            The first item of the goal follows by \propref{progress/ctx} given $ \mathcal{E} \ottsym{=} \ottsym{[]}  {}^>  \ottnt{M_{{\mathrm{2}}}} $.

            \item[$ \ottnt{M_{{\mathrm{1}}}} \ottsym{=} \ottnt{V_{{\mathrm{1}}}} $ for some $\ottnt{V_{{\mathrm{1}}}}$]
            We have the following cases by \propref{progress/noeffect} and \hypref{3.5}.
            \begin{match}
                \item[$\ottnt{M_{{\mathrm{2}}}}  \mid  \mathcal{H}  \longrightarrow  \ottnt{M'_{{\mathrm{2}}}}  \mid  \mathcal{H}'$]
                The first item of the goal follows by \propref{progress/ctx} given $ \mathcal{E} \ottsym{=} \ottnt{V_{{\mathrm{1}}}}  {}^>  \ottsym{[]} $.

                \item[$ \ottnt{M_{{\mathrm{2}}}} \ottsym{=} \ottnt{V_{{\mathrm{2}}}} $ for some $\ottnt{V_{{\mathrm{2}}}}$]
                We have $ \ottnt{V_{{\mathrm{1}}}} \ottsym{=} \lambda^>  \ottmv{x}  \ottsym{.}  \ottnt{M'_{{\mathrm{1}}}} $ by \propref{typing/canonical(rarrow)} and \hypref{3.4}.
                To show the first item of the goal, we choose $ \ottnt{M'} \ottsym{=} \ottnt{M'_{{\mathrm{1}}}}  \ottsym{[}  \ottnt{V_{{\mathrm{2}}}}  \ottsym{/}  \ottmv{x}  \ottsym{]} $ and $ \mathcal{H}' \ottsym{=} \mathcal{H} $.
                What we need to show is
                \begin{gather*}
                    \ottsym{(}  \lambda^>  \ottmv{x}  \ottsym{.}  \ottnt{M'_{{\mathrm{1}}}}  \ottsym{)}  {}^>  \ottnt{V_{{\mathrm{2}}}}  \mid  \mathcal{H}  \longrightarrow  \ottnt{M'_{{\mathrm{1}}}}  \ottsym{[}  \ottnt{V_{{\mathrm{2}}}}  \ottsym{/}  \ottmv{x}  \ottsym{]}  \mid  \mathcal{H},
                \end{gather*}
                which can be derived by \ruleref{RE-RBeta} and \ruleref{R-Exp}.
            \end{match}
        \end{match}

        \item[\ruleref{T-LApp}] Similar to the case \ruleref{T-RApp}.
        \item[\ruleref{T-UPair}]
        In this case,
        \begin{gather*}
             \ottnt{C} \ottsym{=} \ottnt{C_{{\mathrm{1}}}}  \parallel  \ottnt{C_{{\mathrm{2}}}}  \\
             \ottnt{M} \ottsym{=} \ottnt{M_{{\mathrm{1}}}}  \otimes  \ottnt{M_{{\mathrm{2}}}}  \\
             \ottnt{T} \ottsym{=} \ottnt{T_{{\mathrm{1}}}}  \otimes  \ottnt{T_{{\mathrm{2}}}}  \\
             \ottnt{e} \ottsym{=}  \ottnt{e_{{\mathrm{1}}}}  \sqcup  \ottnt{e_{{\mathrm{2}}}}   \\
            \ottnt{C_{{\mathrm{1}}}}  \vdash  \ottnt{M_{{\mathrm{1}}}}  :  \ottnt{T_{{\mathrm{1}}}}  \mid  \ottnt{e_{{\mathrm{1}}}} \\
            \ottnt{C_{{\mathrm{2}}}}  \vdash  \ottnt{M_{{\mathrm{2}}}}  :  \ottnt{T_{{\mathrm{2}}}}  \mid  \ottnt{e_{{\mathrm{2}}}}.
        \end{gather*}
        The rest of the proof is similar to the case \ruleref{T-UApp} except, when we know both $\ottnt{M_{{\mathrm{1}}}}$ and $\ottnt{M_{{\mathrm{2}}}}$ are values, the second item of the goal follows.

        \item[\ruleref{T-OPair}]
        In this case,
        \begin{gather*}
             \ottnt{C} \ottsym{=} \ottnt{C_{{\mathrm{1}}}}  \ottsym{,}  \ottnt{C_{{\mathrm{2}}}}  \\
             \ottnt{M} \ottsym{=} \ottnt{M_{{\mathrm{1}}}}  \odot  \ottnt{M_{{\mathrm{2}}}}  \\
             \ottnt{T} \ottsym{=} \ottnt{T_{{\mathrm{1}}}}  \odot  \ottnt{T_{{\mathrm{2}}}}  \\
             \ottnt{e} \ottsym{=}  \ottnt{e_{{\mathrm{1}}}}  \sqcup  \ottnt{e_{{\mathrm{2}}}}   \\
            \ottnt{C_{{\mathrm{1}}}}  \vdash  \ottnt{M_{{\mathrm{1}}}}  :  \ottnt{T_{{\mathrm{1}}}}  \mid  \ottnt{e_{{\mathrm{1}}}} \\
            \ottnt{C_{{\mathrm{2}}}}  \vdash  \ottnt{M_{{\mathrm{2}}}}  :  \ottnt{T_{{\mathrm{2}}}}  \mid  \ottnt{e_{{\mathrm{2}}}} \\
            \text{$\ottkw{ord} \, \ottnt{T_{{\mathrm{1}}}}$ implies $ \ottnt{e_{{\mathrm{2}}}} \ottsym{=} \ottsym{0} $}.
        \end{gather*}
        We consider wether $\ottkw{ord} \, \ottnt{T_{{\mathrm{1}}}}$ or not, namely $\ottkw{unr} \, \ottnt{T_{{\mathrm{1}}}}$.
        \begin{match}
            \item[$\ottkw{unr} \, \ottnt{T_{{\mathrm{1}}}}$]
            We can prove this case in a similar way to the case \ruleref{T-App} except, when we know both $\ottnt{M_{{\mathrm{1}}}}$ and $\ottnt{M_{{\mathrm{2}}}}$ are value, the second item of the goal follows.

            \item[$\ottkw{ord} \, \ottnt{T_{{\mathrm{1}}}}$]
            We can prove this case in a similar way to the case \ruleref{T-RApp} except, when we know both $\ottnt{M_{{\mathrm{1}}}}$ and $\ottnt{M_{{\mathrm{2}}}}$ are value, the second item of the goal follows.
        \end{match}

        \item[\ruleref{T-ULet}]
        In this case,
        \begin{gather}
             \ottnt{C} \ottsym{=} \mathcal{C}_{{\mathrm{1}}}  \ottsym{[}  \ottnt{C_{{\mathrm{1}}}}  \ottsym{]}  \hyp{5.1}\\
             \ottnt{M} \ottsym{=} \ottkw{let} \, \ottmv{x}  \otimes  \ottmv{y}  \ottsym{=}  \ottnt{M_{{\mathrm{1}}}} \, \ottkw{in} \, \ottnt{M_{{\mathrm{2}}}}  \hyp{5.2}\\
            \ottnt{C_{{\mathrm{1}}}}  \vdash  \ottnt{M_{{\mathrm{1}}}}  :  \ottnt{S_{{\mathrm{1}}}}  \otimes  \ottnt{S_{{\mathrm{2}}}}  \mid  \ottsym{0}. \hyp{5.3}
        \end{gather}
        We have $\mathcal{C}_{{\mathrm{1}}}  \ottsym{[}  \ottnt{C_{{\mathrm{1}}}}  \ottsym{]}  \ottsym{,}  \ottnt{C'}  \vdash  \mathcal{H}$ by \hypref{0.2}.
        So, we have the following cases by \propref{progress/noeffect} given $ \mathcal{C} \ottsym{=} \mathcal{C}_{{\mathrm{1}}}  \ottsym{,}  \ottnt{C'} $.
        \begin{match}
            \item[$\ottnt{M_{{\mathrm{1}}}}  \mid  \mathcal{H}  \longrightarrow  \ottnt{M'_{{\mathrm{1}}}}  \mid  \mathcal{H}'$ for some $\ottnt{M'_{{\mathrm{1}}}}$ and $\mathcal{H}'$]
            The first item of the goal follows by \propref{progress/ctx} given $ \mathcal{E} \ottsym{=} \ottkw{let} \, \ottmv{x}  \otimes  \ottmv{y}  \ottsym{=}  \ottsym{[]} \, \ottkw{in} \, \ottnt{M_{{\mathrm{2}}}} $.

            \item[$ \ottnt{M_{{\mathrm{1}}}} \ottsym{=} \ottnt{V_{{\mathrm{1}}}} $ for some $\ottnt{V_{{\mathrm{1}}}}$]
            We have
            \begin{gather*}
                 \ottnt{V_{{\mathrm{1}}}} \ottsym{=} \ottnt{V_{{\mathrm{11}}}}  \otimes  \ottnt{V_{{\mathrm{12}}}} 
            \end{gather*}
            for some $\ottnt{V_{{\mathrm{11}}}}$ and $\ottnt{V_{{\mathrm{12}}}}$ by \propref{typing/canonical(uprod)} and \hypref{5.3}.
            To show the first item of the goal, we choose $ \ottnt{M'} \ottsym{=} \ottnt{M_{{\mathrm{2}}}}  \ottsym{[}  \ottnt{V_{{\mathrm{11}}}}  \ottsym{/}  \ottmv{x}  \ottsym{]}  \ottsym{[}  \ottnt{V_{{\mathrm{12}}}}  \ottsym{/}  \ottmv{y}  \ottsym{]} $ and $ \mathcal{H}' \ottsym{=} \mathcal{H} $.
            What we need to show is
            \begin{gather*}
                \ottkw{let} \, \ottmv{x}  \otimes  \ottmv{y}  \ottsym{=}  \ottnt{V_{{\mathrm{11}}}}  \otimes  \ottnt{V_{{\mathrm{12}}}} \, \ottkw{in} \, \ottnt{M_{{\mathrm{2}}}}  \mid  \mathcal{H}  \longrightarrow  \ottnt{M_{{\mathrm{2}}}}  \ottsym{[}  \ottnt{V_{{\mathrm{11}}}}  \ottsym{/}  \ottmv{x}  \ottsym{]}  \ottsym{[}  \ottnt{V_{{\mathrm{12}}}}  \ottsym{/}  \ottmv{y}  \ottsym{]}  \mid  \mathcal{H},
            \end{gather*}
            which can be derived by \ruleref{RE-ULet} and \ruleref{R-Exp}.
        \end{match}

        \item[\ruleref{T-OLet}] Similar to the case \ruleref{T-ULet}.
        \item[\ruleref{T-Weaken}]
        In this case,
        \begin{gather}
            \ottnt{C}  \lesssim  \ottnt{C_{{\mathrm{1}}}} \hyp{6.1}\\
             \ottnt{e_{{\mathrm{1}}}} \ottsym{<} \ottnt{e}  \hyp{6.2}\\
            \ottnt{C_{{\mathrm{1}}}}  \vdash  \ottnt{M}  :  \ottnt{T}  \mid  \ottnt{e_{{\mathrm{1}}}}. \hyp{6.3}
        \end{gather}
        We have $\ottnt{C}  \ottsym{,}  \ottnt{C'}  \lesssim  \ottnt{C_{{\mathrm{1}}}}  \ottsym{,}  \ottnt{C'}$ by \propref{env/sub/ctx} and \hypref{6.1}.
        So, $\ottnt{C_{{\mathrm{1}}}}  \ottsym{,}  \ottnt{C'}  \vdash  \mathcal{H}$ by \propref{heap/weaken} and \hypref{0.2}.
        Consequently, this case follows by the induction hypothesis.
    \end{match}
\end{prop}

\begin{prop}{progress}
    If
    \begin{gather*}
        \ottnt{C}  \vdash  \ottnt{M}  :  \ottnt{T}  \mid  \ottnt{e} \\
        \ottnt{C}  \vdash  \mathcal{H},
    \end{gather*}
    then
    \begin{itemize}
        \item $\ottnt{M}  \mid  \mathcal{H}  \longrightarrow  \ottnt{M'}  \mid  \mathcal{H}'$ for some $\ottnt{M'}$ and $\mathcal{H}'$ , or
        \item $ \ottnt{M} \ottsym{=} \ottnt{V} $ for some $\ottnt{V}$.
    \end{itemize}

    \proof This follows by \propref{progress/effect}.
\end{prop}

\end{FULLVERSION}

\end{document}